\documentclass[a4paper, amsfonts, amssymb, amsmath, reprint, showkeys, nofootinbib, onecolumn, superscriptaddress]{revtex4-2}

\usepackage[english]{babel}
\usepackage[utf8]{inputenc}
\usepackage{siunitx}
\usepackage{amsthm}
\usepackage{mathtools}
\usepackage{physics}
\usepackage{xcolor}
\usepackage{graphicx}
\usepackage[left=23mm,right=13mm,top=35mm,columnsep=15pt]{geometry} 
\usepackage{adjustbox}
\usepackage{placeins}
\usepackage[T1]{fontenc}
\usepackage{lipsum}
\usepackage{csquotes}
\usepackage{caption}
\usepackage{wrapfig}
\usepackage{gensymb}
\usepackage{xparse} 
\usepackage{bbm}
\usepackage{makecell}
\usepackage{array}
\usepackage{float}
\usepackage{tikz}
\usetikzlibrary{
    arrows.meta,
    positioning,
    shapes,
    calc,
    fit
}
\usepackage{subcaption}

\renewcommand{\thesection}{\arabic{section}}
\renewcommand{\thesubsection}{\thesection.\arabic{subsection}}
\renewcommand{\thesubsubsection}{\thesubsection.\arabic{subsubsection}}

\usepackage{todonotes}

\makeatletter
\DeclareRobustCommand{\headingcite}[1]{%
  \begingroup
    \@fileswfalse
    \cite{#1}%
  \endgroup
}
\makeatother

\usepackage{hyperref} 
\usepackage[capitalize]{cleveref} 

\crefname{enumi}{}{} 
\hypersetup{
	colorlinks = true,
	linkcolor =blue,
	citecolor=blue, 
	urlcolor=blue 
}

\DeclareCaptionType{protocol}[Protocol][List of Protocols]
\crefname{appendix}{App.}{Apps.}
\crefname{section}{Sec.}{Secs.}
\crefname{footnote}{footnote}{footnotes}

\theoremstyle{definition}
\newtheorem{assumption}{Assumption}
\crefname{assumption}{Assumption}{Assumptions}

\makeatletter
\renewcommand{\p@section}{}
\renewcommand{\p@subsection}{}
\renewcommand{\p@subsubsection}{}
\makeatother

\newtheorem{theorem}{Theorem}
\newtheorem{lemma}{Lemma}

\theoremstyle{definition}
\newtheorem{definition}{Definition}
\newtheorem{remark}{Remark}

\newtheorem{assumptionbis}{Assumption}

\crefname{assumptionbis}{Assumption}{Assumptions}
\Crefname{assumptionbis}{Assumption}{Assumptions}

\usepackage{comment}

\usepackage[most]{tcolorbox}

\newtcolorbox{constraintblock}[1]{
  enhanced,
  colback=white,
  colframe=black,
  boxrule=0.5pt,
  arc=1pt,
  left=4pt,
  right=4pt,
  top=-4pt,
  bottom=4pt,
  title={#1},
  coltitle=black,
  fonttitle=\bfseries,
  attach boxed title to top left={
    xshift=5mm,
    yshift=-2mm
  },
  boxed title style={
    colback=white,
    colframe=white,
    boxrule=0pt,
    left=2pt,
    right=2pt,
    top=0pt,
    bottom=0pt
  }
}

\newlength\colsep   
\newlength\offsetX  
\newlength\offsetY  

\newcommand{\totRounds}{n}
\newcommand{\fssCutoff}{{N_\flagSpaceReg}}
\newcommand{\tagCutoff}{{N_\mathrm{ph}}} 
\newcommand{\blockVar}{m}
\newcommand{\blockVarReg}{M}

\newcommand{\epsCor}{\varepsilon_\mathrm{cor}}
\newcommand{\epsSecr}{\varepsilon_\mathrm{sec}}
\newcommand{\epsSecu}{\varepsilon}

\newcommand{\keyl}{\ell}
\newcommand{\ECcost}{\lambda_\mathrm{EC}}
\newcommand{\fFull}{\hat{f}_{\mathrm{full}}^\mathrm{QKD}}
\newcommand{\varDecReg}{\widehat{C}}
\newcommand{\interReg}{\bar{C}} 
\newcommand{\interAlph}{\bar{\mathcal{C}}}
\newcommand{\cppReg}{C_\mathrm{CPP}}
\newcommand{\varDecVal}{\hat{c}}
\newcommand{\varDecGen}{\hat{c}_\mathtt{gen}} 
\newcommand{\ftradeoff}[2]{f_{\hat c_{#1}^{#2}}}
\newcommand{\varDecAlph}{\mathcal{C}}

\newcommand{\kappaQKDVirt}{\tilde{\kappa}^\mathrm{QKD}} 

\newcommand{\annFunc}[1][]{f_\mathrm{ann}^{({#1})}}
\newcommand{\keymapFunc}[1][]{f_\mathrm{kmap}^{({#1})}}
\newcommand{\annKeyMap}[1][]{\mathcal{T}_{#1}^{\mathrm{clas}}} 

\newcommand{\gMap}{\mathcal{T}}
\newcommand{\gMapTildeFull}{\widetilde{\mathcal{T}}^\mathrm{full}} 

\newcommand{\gMapMod}{\widetilde{\mathcal{T}}} 
\newcommand{\attackCh}{\mathcal{A}} 
\newcommand{\setAttackCh}{\boldsymbol{\mathcal{A}}^\mathrm{set}} 
\newcommand{\setAttackChVirt}{\boldsymbol{\widehat{\mathcal{A}}}^\mathrm{set}} 

\newcommand{\secretReg}{S} 
\newcommand{\secretVal}{s} 
\newcommand{\preampString}{\boldsymbol{S}} 
\newcommand{\secretAlph}{\mathcal{\secretReg}} 

\newcommand{\qkdMap}{\mathcal{M}} 
\newcommand{\qkdMapMod}{\widetilde{\mathcal{M}}} 

\newcommand{\varDecFunc}{\phi} 
\newcommand{\secretFunc}{\psi} 

\newcommand{\ppMap}{\mathcal{M}_\mathrm{CPP}} 

\newcommand{\timeA}{t^A}
\newcommand{\timeB}{t^B}
\newcommand{\timeAnn}{t^\mathrm{ann}}

\newcommand{\perfectReplaceCh}{\mathcal{R}_\mathrm{ideal}} 

\newcommand{\Aprime}{A'}
\newcommand{\Amarg}{A} 
\newcommand{\Ameas}{\bar{A}}
\newcommand{\Ashield}{\hat{A}} 
\newcommand{\Aclassical}{X}
\newcommand{\AclassicalVal}{x}
\newcommand{\AclassicalAlph}{\mathcal{X}}
\newcommand{\Astate}{\sigma}
\newcommand{\AkeyReg}{K_A}
\newcommand{\AstateVirt}{\tau} 
\newcommand{\Apovmel}[1][]{{M_\AclassicalVal^{\Amarg_{#1}}}}

\newcommand{\Bmeas}{B} 
\newcommand{\BmeasSquash}{B'} 
\newcommand{\BclassicalReg}{Y} 
\newcommand{\BclassicalVal}{y}
\newcommand{\BclassicalAlph}{\mathcal{Y}} 
\newcommand{\Bpovmel}[1][]{{M_\BclassicalVal^{\Bmeas_{#1}}}}
\newcommand{\BpovmelTilde}[1][]{{\widetilde{M}_y^{\BmeasSquash_{#1}}}}
\newcommand{\BpovmelBlock}[2]{{M_{\BclassicalVal}^{\Bmeas_{#1, #2}}}}
\newcommand{\BpovmelBlockTilde}[2]{{\widetilde{M}_{\BclassicalVal}^{\Bmeas_{#1, #2}}}}

\newcommand{\BpovmelTarg}[1][]{{F_\BclassicalVal^{\BmeasSquash_{#1}}}}
\newcommand{\BkeyReg}{K_B}
\newcommand{\lambdaMin}{\lambda_\mathrm{min}}

\newcommand{\Ereg}{E} 
\newcommand{\eveCopyReg}{\widetilde{C}} 
\newcommand{\evePurReg}{\widehat{E}} 

\newcommand{\srcImpBnd}[1][]{{\delta_{#1}}}

\newcommand{\dtImpBnd}[1][]{{q_{#1}}}
\newcommand{\dtImpBndMax}[1][]{{q_{#1}^\mathrm{max}}} 
\newcommand{\srcMap}{\Psi}
\newcommand{\squashMap}{\Lambda}
\newcommand{\fssMap}{\Lambda}
\newcommand{\noiseChMap}{\Phi}
\newcommand{\flagSpaceReg}{F} 

\newcommand{\setAliceMarginalsBar}[1][]{\boldsymbol{\bar \sigma}_{\Amarg_{#1}}^\mathrm{set}}
\newcommand{\setAliceMarginalsPrime}[1][]{\boldsymbol{\sigma}_{\Amarg_{#1}}^{\mathrm{set}'}} 

\newcommand{\setAliceMarginalsVirt}[1][]{\boldsymbol{\tau}_{\Amarg_{#1}}^\mathrm{set}} 

\newcommand{\gMapVirt}[1][]{\widehat{\mathcal{T}}_{#1}}

\newcommand{\srcImpBndEnc}[1][]{{\delta_{#1}^{\mathrm{enc}}}}
\newcommand{\srcImpBndMode}[1][]{{\delta_{#1}^{\mathrm{mod}}}}

\newcommand{\characCutoff}{{N_\mathrm{ch}}} 

\newcommand{\intensityAlph}{\mathcal{I}}

\newcommand{\AstatePrepared}{\psi}

\newcommand{\srcImpBndTHA}[1][]{{\delta_{#1}^{\mathrm{tha}}}}

\newcommand{\srcImpBndInt}[1][]{{\delta_{#1}^{\mathrm{int}}}}
\newcommand{\srcImpBndPN}[1][]{{\delta_{#1}^{\mathrm{pn}}}} 

\newcommand{\phaseCutoff}{N_A}

\newcommand{\nbDetec}{N_\mathrm{det}} 
\newcommand{\dtImpBndDCR}[1][]{{q_{#1}^{\mathbf{d_B}}}} 
\newcommand{\dtImpBndEff}[1][]{{q_{#1}^{\boldsymbol{\eta}}}} 

\newcommand{\genDensity}{\rho} 
\newcommand{\genDensityAlt}{\omega} 
\newcommand{\genQReg}{Q} 
\newcommand{\genClassIndex}{z} 
\newcommand{\genClassReg}{Z} 
\newcommand{\genClassAlph}{\mathcal{Z}}
\newcommand{\genpovmel}{E_{\genClassIndex}} 
\newcommand{\genpovmelTilde}{\widetilde{E}_{\genClassIndex}} 
\newcommand{\genpovmelBlock}[1][]{E_{{\genClassIndex},{#1}}} 
\newcommand{\genpovmelBlockTilde}[1][]{\widetilde{E}_{\genClassIndex, {#1}}} 
\newcommand{\genState}{\varphi} 
\newcommand{\genPurReg}{R} 

\newcommand{\cptp}{\mathrm{CPTP}}
\newcommand{\hilbert}{\mathcal{H}}
\newcommand{\setDensity}{\mathcal{S}} 

\newcommand{\identity}{\mathbbm{1}}
\newcommand{\proj}{\Pi} 
\newcommand{\purFunc}{\mathtt{Pur}} 

\newcommand{\fRenyiUpGen}{\widetilde{H}_\alpha^{\uparrow, f}}
\newcommand{\fRenyiUpEnt}[2]{\widetilde{H}_\alpha^{\uparrow, f_{\hat c_{#1}^{#2}}}}
\newcommand{\fRenyiUpEntFull}{\widetilde{H}_\alpha^{\uparrow, \fFull}}

\newcommand{\renyidown}{\widetilde{H}_\alpha^\downarrow}
\newcommand{\renyiup}{\widetilde{H}_\alpha^\uparrow}

\newcommand{\renyidiv}{\widetilde{D}_\alpha}

\newcommand{\gramMatrix}{G}
\newcommand{\probMatrix}{P}
\newcommand{\probMatrixVec}{\boldsymbol{s}} 

\newcommand{\prob}{P}
\newcommand{\ptest}{p_\mathtt{test}}

\newcommand{\marginalMap}{\mathcal{Q}} 
\newcommand{\marginalSet}{\boldsymbol{\mathcal{Q}}^\mathrm{set}} 
\newcommand{\marginalSetVirt}{\boldsymbol{\widehat{\mathcal{Q}}}^\mathrm{set}} 
\newcommand{\marginalSetTilde}{\boldsymbol{\widetilde{\mathcal{Q}}}^\mathrm{set}} 

\newcommand{\stateSet}{\Sigma} 

\newcommand{\yield}{\boldsymbol{Y}_j}
\newcommand{\probSimplex}{\mathbb{P}}

\newcommand{\setPos}{\mathrm{Pos}} 

\newcommand{\renyiMblock}{v_\alpha}
\newcommand{\unitStatsVec}[1][]{\hat{e}_{\varDecFunc_j({#1})}}

\usepackage{tocvsec2}

\begin{document}
\title{Security framework for practical quantum key distribution with imperfect devices}

\author{Jerome Wiesemann}
\affiliation{Institute for Quantum Computing and Department of Physics and Astronomy, University of Waterloo, Waterloo, Ontario, Canada, N2L 3G1}

\author{John Burniston}
\affiliation{Institute for Quantum Computing and Department of Physics and Astronomy, University of Waterloo, Waterloo, Ontario, Canada, N2L 3G1}

\author{Devashish Tupkary}
\affiliation{Institute for Quantum Computing and Department of Physics and Astronomy, University of Waterloo, Waterloo, Ontario, Canada, N2L 3G1}

\author{Norbert Lütkenhaus}
\affiliation{Institute for Quantum Computing and Department of Physics and Astronomy, University of Waterloo, Waterloo, Ontario, Canada, N2L 3G1}

\date{\today} 

\begin{abstract}
    Practical quantum key distribution (QKD) systems inevitably exhibit imperfections in both the source and detector. At the same time, the behavior of these imperfect devices is never exactly known due to characterization uncertainty, parameter fluctuations, and potential influence by an adversary. In this work, we present a security proof for generic prepare-and-measure QKD protocols, including decoy-state BB84, with imperfect and imperfectly characterized sources and detectors using  the marginal-constrained entropy accumulation theorem (MEAT). Our approach uses a sequence of proof technique independent source maps and squashing maps, yielding a very modular framework. We show that practical key rates can be achieved even when multiple imperfections are combined. More broadly, our work provides a unified foundation that avoids the need for dedicated protocol-specific arguments and can be readily extended to other protocols and device imperfections.
\end{abstract}

\maketitle
\tableofcontents

\section{Introduction} 
\label{sec:introduction}
\noindent Quantum key distribution (QKD) allows two distant parties to establish a secure key by means of an information-theoretically secure protocol \cite{Portmann22, Mayers01, Scarani09, benor2004universalcomposablesecurityquantum, ferradini2025definingsecurityquantumkey, lmcs:9979,tupkary2025qkdsecurityproofsdecoystate}. However, practical (device-dependent) QKD systems inevitably exhibit imperfections in both the source and the detector \cite{BSI24, makarov_preparing_2024, PhysRevA.78.042333, qi2006timeshiftattackpracticalquantum,PhysRevLett.107.110501,Weier_2011,PhysRevLett.85.1330,PhysRevA.73.022320,PhysRevA.74.022313}, and these imperfections must be incorporated into the security proof to obtain rigorous security guarantees and close potential loopholes. Over the past two decades, significant progress has been made in developing security proofs for practical implementations with imperfect devices \cite{gottesman2004securityquantumkeydistribution,curras-lorenzo_security_2026, navarrete2026numericalsecurityanalysispractical, pereira_optimal_2025, Curras-Lorenzo25Optica, wang2025phaseerrorestimationpassive, nahar_imperfect_2026,sixto2026finitekeysecurityanalysisdecoystate,Tupkary2025phaseerrorrate,kamin_renyi_2025}. 

At the same time, the precise behavior of these imperfect devices is never known exactly. The characterization process inherently carries uncertainty, device parameters may fluctuate during the protocol, and some imperfections may even be partially under adversarial control. It is therefore not sufficient to account only for device imperfections. One must also account for uncertainty about which device description is realized, which we refer to as \textit{imperfect knowledge} of the devices. In this work, we address both aspects. We use the marginal-constrained entropy accumulation theorem (MEAT) \cite{arqand_marginal-constrained_2025}, which has been shown to be highly modular with respect to protocol variations \cite{tupkary_rigorous_2026} while yielding tight key rates \cite{kamin_renyi_2025}.

This work builds upon Ref.~\cite{tupkary_rigorous_2026} which provides a rigorous and complete security proof for a broad class of prepare-and-measure QKD protocols, including decoy-state BB84 \cite{Hwang03, wang2005, Lo05}, and addresses the gaps present in previous analyses~\cite{tupkary2025qkdsecurityproofsdecoystate} pertaining to protocol security of QKD. However, Ref.~\cite{tupkary_rigorous_2026} assumes that the source and detector are perfectly characterized, and does not cover imperfections. In this work, we extend the analysis to incorporate practical source and detector imperfections together with imperfect characterization. Our work can therefore be viewed as a direct extension of Ref.~\cite{tupkary_rigorous_2026}, largely following the same notation and inheriting much of its level of rigor. The key contributions can be summarized as follows:
\begin{itemize}
\renewcommand{\labelitemi}{\small$\bullet$}
    \item Develop a general framework combining a wide range of source and detector imperfections, while allowing for imperfect characterization of the devices.
    \item Show that practical key rates can still be achieved under realistic experimental parameters, even when combining multiple imperfections simultaneously.
    \item Provide a modular approach that can readily be extended to other protocols and additional device imperfections beyond those considered in this work.
\end{itemize}

Typically, QKD security proofs start with the protocol description and device parameters as input and output a valid secure key rate. If the devices are imperfectly characterized, however, their description is no longer known exactly. A natural approach is therefore to prove security for the worst case over the set of device descriptions compatible with the characterization and model assumptions. This is, in essence, the approach taken in this work. More specifically, instead of assuming that the source prepares a fixed set of states and the detectors measure with a fixed POVM, we assume that they belong to known sets of possible states and POVMs determined by the characterization of the devices. Directly optimizing over such sets may be impractical: the states and POVMs may be infinite-dimensional, and the resulting optimization problem may be numerically intractable. To resolve this, we apply a sequence of source maps \cite{gottesman2004securityquantumkeydistribution, nahar_imperfect_2023, Curras-Lorenzo_2025} and squashing maps \cite{PhysRevLett.101.093601, PhysRevA.84.020303, PhysRevA.89.012325, PRXQuantum.5.040315, PhysRevA.81.012328, PhysRevA.78.032302, PRXQuantum.2.020325, zhang_security_2021, nahar_imperfect_2026} at the protocol level before applying the MEAT.  

An advantage of this approach is that the treatment of imperfections is largely proof-technique independent. This reduction is also the key feature that makes our analysis modular, in the sense that it allows us to directly accommodate future improvements of the MEAT (e.g. potentially with a secret memory register), as well as to easily extend the techniques to other protocols. On the source side, device imperfections enter the analysis as an optimization over the Gram matrix of purifications of the signal states \cite{pereira_optimal_2025}. On the detector side, we combine the flag-state squasher \cite{zhang_security_2021} with the noise channel construction from Ref.~\cite{nahar_imperfect_2026} to reduce imperfectly characterized detectors to perfectly characterized detectors. Crucially, the type of imperfections considered in this work is made explicit through a small set of core assumptions on the source and detector, stated in \cref{sec:model_assumptions}. These assumptions specify the precise properties required and thereby state the model assumptions for the physical devices. The security analysis then builds directly on these core assumptions.

We illustrate the versatility of our work by incorporating a wide range of practical source and detector imperfections, listed in \cref{tab:device_imperfections_list}, including state preparation flaws, detector efficiency and dark-count rate mismatch, imperfect beamsplitter, Trojan-horse attacks, imperfectly block-diagonal sources (or imperfect phase randomization), intensity fluctuations and intensity leakage in decoy-state protocols \cite{BSI24}. Notably, our analysis also allows each of these imperfections to be only imperfectly characterized. To the best of our knowledge, no previous finite-size decoy-state analysis combines generic source and detector imperfections, intensity fluctuations and imperfect phase randomization. Moreover, the protocol parameters and parameters related to the device imperfections can be different for each round, and the protocol can support certain aspects of adaptive and time-dependent behavior relevant for satellite implementations \cite[Section 8.7.1]{tupkary_rigorous_2026}. 

We show that practical key rates can be achieved with realistic experimental parameters even when many of these imperfections are simultaneously combined (\cref{fig:keyRate_example}). For the numerical evaluation of the key rates, we rely on the approach of Ref.~\cite{kamin_renyi_2025} (which extends Ref.~\cite{winick_reliable_2018}, see also Ref.~\cite{navarro_finitesize_2026}), and which provides a security analysis for decoy-state protocols using the MEAT, and incorporates intensity fluctuations and phase imperfections. Our work extends the treatment to generic device uncertainties, source imperfections, detector imperfections, and combinations thereof. While we restrict our attention to source and detector imperfections that vary independently from round to round, certain correlated imperfections \cite{wang2025phaseerrorestimationpassive,marwah2025provingsecuritybb84source,curraslorenzo2026rigorousphaseerrorestimationsecurityframework} can be reduced to independent imperfections and in principle be incorporated \cite{meat_correlations}. We leave the explicit treatment of correlations to future work. We also compare our key rates with those obtained in recent works using phase-error estimation \cite{navarrete2026numericalsecurityanalysispractical,sixto2026finitekeysecurityanalysisdecoystate}, which have a similar scope in their combination of source and detector imperfections (\cref{fig:comparison}).

The remainder of this work is structured as follows. In \cref{sec:background}, we introduce the protocol, notation, and the notions of perfectly and imperfectly characterized devices. In \cref{sec:toolbox}, we develop the main mathematical tools used to incorporate device imperfections throughout this work, including the various source maps, squashing maps, and the MEAT framework. In \cref{sec:security_imperfect_devices}, we combine these tools to derive the general security statement for prepare-and-measure QKD with imperfect and imperfectly characterized devices, resulting in \cref{th:security_imperfect_devices}. In \cref{sec:device_imperfections}, we show how to use the framework in practice and explicitly construct the constraints corresponding to a broad range of source and detector imperfections. We illustrate the framework for a polarization-encoded decoy-state BB84 setup in \cref{sec:example_pol_encoding} (see \cref{fig:keyRate_example}). The convex optimization problem used for the numerics is derived in \cref{ap:numerics}, resulting in \cref{th:convex_opti_numerics_imperfect} and \cref{fig:opti_problem_final}.

For readers interested primarily in applying the framework rather than the technical derivations, it is sufficient to read the protocol description in \cref{sec:protocol_description} and the guide in \cref{sec:device_imperfections}. To facilitate practical use, we provide the numerical software used in this work \cite{numericalanalysis}. The code allows users to evaluate key rates for arbitrary combinations of imperfections and can easily be extended to incorporate additional imperfections or protocol variants.  
\section{Background and notation} \label{sec:background}

\noindent We denote registers by uppercase letters such as $\genClassReg$ and $\genQReg$. For any registers $\genQReg$ and $\genQReg'$, we use the shorthand notation $\genQReg \oplus \genQReg'$ to denote the register corresponding to the Hilbert space $\hilbert_\genQReg \oplus \hilbert_{\genQReg'}$. When handling multiple registers, we write $\genQReg_i^j$ to denote the sequence of registers $\genQReg_i \ldots \genQReg_j$. We denote by $\setDensity_=(\genQReg)$ and $\setDensity_\leq(\genQReg)$ the sets of normalized and sub-normalized density operators acting on register $\genQReg$, respectively. Accordingly, we write $\genDensity_\genQReg \in \setDensity_=(\genQReg)$ to indicate that $\genDensity_\genQReg$ is a density operator on register $\genQReg$. For a state $\genDensity_{\genQReg\genQReg'}$, we use $\genDensity_\genQReg$ to denote the marginal state on register $\genQReg$, obtained by tracing out $\genQReg'$. Conversely, for a state $\genDensity_\genQReg$, we use $\genDensity_{\genQReg\genQReg'}$ to denote an extension of $\genDensity_\genQReg$ to the larger register $\genQReg\genQReg'$. We denote by $\cptp(\genQReg, \genQReg')$ the set of completely positive and trace-preserving (CPTP) maps from $\genQReg$ to $\genQReg'$. Some additional information-theoretic background is provided in \cref{app:misc}. Throughout this work, we use the squared fidelity convention. A list of the symbols and notation used throughout this work is provided in \cref{app:notation}, i.e. \cref{tab:generic_notation,tab:notation_alice_bob_eve,tab:protocol_notation}.

\subsection{Protocol description}
\label{sec:protocol_description}

\noindent We consider the \nameref{prot:generic_qkd_protocol} described below, which is based on \cite[Protocol~1]{tupkary_rigorous_2026}, and covers a broad class of prepare-and-measure QKD protocols, including both decoy-state and non-decoy protocols.  Throughout this work, the handling of device imperfections is concerned only with the first few steps, including state preparation, measurement, public announcements and sifting. In particular, the subsequent post-processing steps are identical to \cite[Protocol~1]{tupkary_rigorous_2026}. This allows us to directly build on top of the security analysis of Ref.~\cite{tupkary_rigorous_2026}, inheriting both its structure and level of rigor.

For a first reading, the round index $j$ may largely be ignored. It is included to allow for the most general setting where the protocol parameters, signal states, measurements and device imperfections may vary (in an independent manner) from round to round. In the special case where the protocol is identical in every round, the corresponding quantities can simply be chosen independently of $j$. A summary of the notation for the various registers and operators associated with Alice, Bob, and Eve, as well as the QKD protocol parameters, is given in \cref{tab:protocol_notation,tab:notation_alice_bob_eve}.

{\captionsetup{justification=raggedright, singlelinecheck=false, labelfont=bf}
\renewcommand{\figurename}{Protocol}
\captionof{protocol}[Generic QKD Protocol]{Generic QKD Protocol}
\label{prot:generic_qkd_protocol}
}

\noindent \textbf{Parameters:}\footnote{Technically, there are additional implicit parameters, such as the security parameters. We only list the most important parameters here and refer to Ref.~\cite[Table~VIII]{tupkary_rigorous_2026} for an exhaustive list. The implicit parameters here are contained in $\ppMap$. }

\noindent\(\displaystyle
\begin{aligned}
\totRounds\in\mathbb{N} & : && \quad \text{Total number of rounds} \\
\Big\{\prob_{\Aclassical_j}\in\probSimplex_{|\AclassicalAlph|}\Big\}_{j=1}^{\totRounds} & : && \quad \text{Alice's setting choice probability distributions} \\
\Big\{\big\{(\Astate_\AclassicalVal^{(j)})_{\Aprime_j}\in\setDensity_=(\Aprime_j)\big\}_{\AclassicalVal\in\AclassicalAlph}\Big\}_{j=1}^{\totRounds} & : && \quad \text{Alice's signal states} \\
\Big\{\big\{\Bpovmel[j]\succeq 0\big\}_{\BclassicalVal\in\BclassicalAlph}:\sum_{\BclassicalVal\in\BclassicalAlph}\Bpovmel[j]=\identity_{\Bmeas_j}\Big\}_{j=1}^{\totRounds} & : && \quad \text{POVM elements describing Bob's measurements} \\
\Big\{\annFunc[j]:\AclassicalAlph\times\BclassicalAlph\rightarrow\probSimplex_{|\varDecAlph|}\Big\}_{j=1}^{\totRounds} & : && \quad \text{Probabilistic announcement maps} \\
\Big\{\keymapFunc[j]:\AclassicalAlph\times\varDecAlph\rightarrow\secretAlph\Big\}_{j=1}^{\totRounds} & : && \quad \text{Key maps} \\
\ppMap\in\cptp(\Aclassical_1^\totRounds\BclassicalReg_1^\totRounds\secretReg_1^\totRounds\varDecReg_1^\totRounds, \AkeyReg\BkeyReg\cppReg\varDecReg_1^\totRounds) & : && \quad \text{Classical post-processing map}
\end{aligned}
\) \vspace{0.3cm}

\noindent \textbf{Protocol steps:}
\begin{enumerate}
    \item For every round $j\in \{1,\ldots,\totRounds\}$, Alice and Bob perform the following operations:
    \begin{enumerate}
        \item \textbf{State preparation and transmission:} Alice prepares one of the states in $\big\{(\Astate_\AclassicalVal^{(j)})_{\Aprime_j}\big\}_{\AclassicalVal \in \AclassicalAlph}$ according to the probability distribution $\prob_{\Aclassical_j}$ where $\Astate_\AclassicalVal^{(j)} \in \setDensity_=(\Aprime_j)$. She stores the setting choice $\AclassicalVal \in \AclassicalAlph$ in a classical register $\Aclassical_j$ and sends the signal state to Bob via the quantum channel. Let $\timeA_j$ denote the time this state leaves Alice's lab. 
        \item \textbf{Measurements:} Bob measures the state received with a POVM $\big\{\Bpovmel[j]\big\}_{y\in \mathcal{Y}}$ and stores the outcome $\BclassicalVal \in \BclassicalAlph$ in a classical register $\BclassicalReg_j$. Let $\timeB_j$ denote the time this measurement is completed. 
        \item \textbf{Public announcements:} Alice and Bob make public announcements based on their registers $\Aclassical_j$ and $\BclassicalReg_j$, and store them in $\varDecReg_j$ with alphabet $\varDecAlph$. This can be represented by a stochastic map $\annFunc[j]:\AclassicalAlph\times\BclassicalAlph\rightarrow\probSimplex_{|\varDecAlph|}$.\footnote{This stochastic mapping represents the fact that Alice and Bob may make random decisions based on their outcomes, for example when randomly deciding whether a round is a test round or a generation round. Every implementable public announcement procedure can be described via such a stochastic map. The converse does not hold; an arbitrary stochastic map  need not be implementable through public communication. We therefore implicitly restrict $\annFunc[j]$ to implementable announcement procedures. If the announcements are performed deterministically, then $\annFunc[j]$ can be treated as a function which outputs an element of $\varDecAlph$.} Let $\timeAnn_j$ denote the time the announcements begin.\footnote{The register $\varDecReg_j$ stores the decision whether the round is a test round or generation round. This decision can be made by either Alice or Bob, see \cref{rem:test_gen_decision}.}
        \item \textbf{Sifting and key map:} Alice maps her private data $\Aclassical_j$ and public announcements $\varDecReg_j$ to a classical value $\secretReg_j$ with alphabet $\secretAlph$. This can be represented by a function $\keymapFunc[j]:\AclassicalAlph\times\varDecAlph\rightarrow\secretAlph$.\footnote{Similarly to the announcement map $\annFunc[j]$, the key map can also be performed probabilistically by choosing $\keymapFunc[j]$ to be a stochastic map.} She then applies a deterministic rule, based on $\varDecReg_j$, to discard certain $\secretReg_j$. This produces the pre-amplification string $\preampString$. We assume discarded rounds are encoded by setting the corresponding value of $\secretReg_j$ to some special symbol. 
    \end{enumerate}
    \item \textbf{Variable-length decision:} Based on the public announcements ${\varDecVal}_1^{\totRounds}$ stored in the registers $\varDecReg_1^\totRounds$, Alice and Bob determine the number of bits used for error correction $\ECcost({\varDecVal}_1^\totRounds)$ and the length of the key $\keyl({\varDecVal}_1^\totRounds)$ produced. If $\keyl(\varDecVal_1^\totRounds) = 0$, they abort the protocol.
    \item \textbf{Error correction:} Alice and Bob perform an error-correction protocol where the number of possible transcripts can be stored in at most $\ECcost(\varDecVal_1^\totRounds)$ bits, resulting in Bob outputting a guess for Alice's pre-amplification string.
    \item \textbf{Error verification:} Alice and Bob use two-universal hashing with a hash of length $\left\lceil \log\left(\frac{1}{\epsCor}\right)\right\rceil$ to perform error verification by comparing the hashes of the pre-amplification string and Bob's guess.
    \item \textbf{Privacy amplification:} Alice and Bob generate the output key pair $\AkeyReg, \BkeyReg$ by using two-universal hashing on their pre-amplification strings.
    
   The sifting and post-processing in steps 2-5 is described by the map $\ppMap\in\cptp(\Aclassical_1^\totRounds\BclassicalReg_1^\totRounds\secretReg_1^\totRounds\varDecReg_1^\totRounds, \AkeyReg\BkeyReg\cppReg \varDecReg_1^\totRounds)$, where the register $\cppReg$ stores all public announcements made in these  steps.\footnote{For example, the register stores  the choice of hash function for error verification and privacy amplification, along with the communication required for error correction, or for communicating the key length decision etc. See \cite[Protocol~1]{tupkary_rigorous_2026} for a more detailed discussion.} 
\end{enumerate}

\begin{remark}
    We assume that the announcements for round $j$ occur after state preparation and measurement for round $j$. That is, $\timeAnn_j \geq \timeA_j, \timeB_j$ for all $j$. Such a requirement can be enforced by appropriate protocol design \cite{tupkary_rigorous_2026} and under  the assumption that authentication behaves honestly, i.e, all messages sent are received correctly by the receiving party, some time after they are sent. This assumption can be relaxed \cite{tupkary_authentication} to more realistic assumptions on authentication, where messages sent may be rejected by the receiving party (due to tampering by Eve), and message timings may be affected. In particular, Ref.~\cite{tupkary_authentication} shows that the security of a QKD protocol that includes a small final authentication post-processing step, with realistic assumptions on authentication, can be reduced to that of the QKD protocol without the post-processing step, with the assumption that authentication behaves honestly. We refer the reader to Refs.~\cite{tupkary_rigorous_2026,tupkary_authentication} for additional details. In this work, we focus on the end result of this reduction, and assume that the above requirement on timing is satisfied, and authentication behaves honestly. We stress that this is an implicit assumption in almost all works on QKD security proofs. 
\end{remark}

For the purpose of the security analysis, given  $\timeAnn_j \geq \timeA_j, \timeB_j$ for all $j$, we may equivalently consider a modified protocol where Alice prepares and sends all signal states before Bob performs the first measurement, and public announcements occur right after Bob's measurement. This is illustrated in \cref{fig:modified_protocol_ch}, and follows from \cite[Section 5]{arqand_marginal-constrained_2025} (see also \cite[Section 8.1]{tupkary_rigorous_2026}), which argues that doing so only gives Eve more power.\footnote{Note that directly applying the MEAT to this protocol yields trivial key rates, as Alice’s marginal is classical. Consequently, a purification of the Alice–Bob state directly gives Eve a copy of Alice’s setting choice. We therefore first apply the source-replacement scheme before applying the MEAT (see \cref{sec:security_using_MEAT}).} The equivalence between the original and modified timings is formalized in \cref{lem:modified_timing_protocol}.

\begin{lemma}[Modified protocol {\protect\cite[Lemmas~8.1 and 8.2]{tupkary_rigorous_2026}}]
\label{lem:modified_timing_protocol}
    Any Alice-Bob-Eve state that can be obtained in the \nameref{prot:generic_qkd_protocol} can also be obtained in the following modified protocol. The modified protocol takes the same parameters as the \nameref{prot:generic_qkd_protocol} and is illustrated in \cref{fig:modified_protocol_ch}.
    \begin{enumerate}
    \item Alice prepares the $\totRounds$-round state 
\begin{align}
\label{eq:n_round_state_rho_AX}
\Astate_{\Aclassical_1^\totRounds(\Aprime)_1^\totRounds} &= \bigotimes_{j=1}^\totRounds \Astate_{\Aclassical_j\Aprime_j}^{(j)}\,, \\
    \Astate_{\Aclassical_j\Aprime_j}^{(j)} &= \sum_{\AclassicalVal\in\AclassicalAlph} \prob_{\Aclassical_j}(\AclassicalVal) \ketbra{\AclassicalVal}{\AclassicalVal}_{\Aclassical_j} \otimes (\Astate_\AclassicalVal^{(j)})_{\Aprime_j}\,,
\end{align}
and sends $(\Aprime)_1^\totRounds$ to Bob.
 We denote Eve's starting register as $\Ereg_0 = (\Aprime)_1^\totRounds$.
    \item For every round $j\in \{1,\ldots,\totRounds\}$, the following maps are applied:
    \begin{enumerate}
        \item Eve implements her attack $\attackCh_j\in \setAttackCh_j \subseteq \cptp(\Ereg_{j-1}, \Bmeas_j \Ereg_j')$ and forwards $\Bmeas_j$ to Bob.
        
        The measurement, public announcement, sifting and key map steps are identical to the \nameref{prot:generic_qkd_protocol}. We denote $\widetilde C_j$ a copy of the announcements, which is available to Eve. These steps can be described by a map $\gMapTildeFull_j\in\cptp(\Aclassical_j\Bmeas_j, \secretReg_j\Aclassical_j\BclassicalReg_j\varDecReg_j\eveCopyReg_j)$. We also define $\gMapMod_j\coloneq \Tr_{\Aclassical_j\BclassicalReg_j}\circ \gMapTildeFull_j$, i.e. the map $\gMapTildeFull_j$ with outputs restricted to the secret register and the announcements.
        
        \item The registers $\Ereg_j'$ and the copy of the classical announcements $\widetilde C_j$ are combined into a register $\Ereg_j$, which is forwarded to the next round.

        The state output is given by
        \begin{align}
        \label{eq:state_output_protocol}
            \genDensity_{\secretReg_1^\totRounds \Aclassical_1^\totRounds \BclassicalReg_1^\totRounds \varDecReg_1^\totRounds \Ereg_\totRounds} &= \gMapTildeFull_\totRounds\circ \attackCh_\totRounds\circ \ldots \circ \gMapTildeFull_1\circ \attackCh_1\left[\Astate_{\Aclassical_1^\totRounds(\Aprime)_1^\totRounds} \right] \\
            \genDensity_{\secretReg_1^\totRounds \varDecReg_1^\totRounds \Ereg_\totRounds} &= \gMapMod_\totRounds\circ \attackCh_\totRounds\circ \ldots \circ \gMapMod_1\circ \attackCh_1\left[\Astate_{\Aclassical_1^\totRounds(\Aprime)_1^\totRounds} \right]\,,
        \end{align}
        and each round can be described by a map $\qkdMapMod_j \coloneqq \gMapMod_j\circ \attackCh_j$ (see \cref{fig:modified_protocol_ch}).
    \end{enumerate}
    \item Alice and Bob perform the remaining post-processing steps from the \nameref{prot:generic_qkd_protocol} described by $\ppMap\in\cptp(\Aclassical_1^\totRounds\BclassicalReg_1^\totRounds\secretReg_1^\totRounds\varDecReg_1^\totRounds, \AkeyReg\BkeyReg\cppReg \varDecReg_1^\totRounds)$.
\end{enumerate}
\end{lemma}

\newcommand{\ModifiedProtocolChFig}{

\begin{tikzpicture}[>=Stealth,auto,
  box/.style  ={draw,minimum width=1cm,minimum height=1cm,
                align=center,fill=red!20},
  cpbox/.style ={draw,minimum width=1.4cm,minimum height=0.8cm,
                align=center,fill=gray!20},
  sig/.style  ={->,rounded corners=4pt},
  merge/.style={circle,fill,inner sep=1.2pt}
]


\node[box]   (N1) {$\attackCh_1$};
\node[cpbox, below right=\offsetY and \offsetX of N1] (L1) {$\gMapMod_1$};

\node[box,   right=\colsep of N1] (N2) {$\attackCh_2$};
\node[cpbox, below right=\offsetY and \offsetX of N2] (L2) {$\gMapMod_2$};

\node[right=\colsep of N2] (dots) {$\cdots$};
\node[box,   right= of dots] (Nn) {$\attackCh_n$};
\node[cpbox, below right=\offsetY and \offsetX of Nn] (Ln) {$\gMapMod_n$};

\node[left =1.4cm of N1] (E0) {};
\node[right=4cm   of Nn] (En) {};

\coordinate (join1) at ($(L1.east |- N1) + (1.0,0)$);
\coordinate (join2) at ($(L2.east |- N2) + (1.0,0)$);
\coordinate (joinn) at ($(Ln.east |- Nn) + (1.0,0)$);

\draw[sig] (E0) node[above=2pt, xshift=10pt]{$E_{0}=(A')_1^n$} -- (N1.west);

\draw[sig] (N1.east) node[above=2pt, xshift=10pt] {$E'_1$} -- (join1);
\draw[sig] (L1.east) node[above=2pt, xshift=10pt]   {$\eveCopyReg_1$}
           -- ++(0.6,0) |- (join1);
\node[merge] at (join1) {};
\draw[sig] (join1) -- node[above=2pt,xshift=1pt] {$E_{1}$} (N2.west);

\draw[sig] (N2.east) node[above=2pt, xshift=10pt] {$E'_2$} -- (join2);
\draw[sig] (L2.east) node[above=2pt, xshift=10pt]   {$\eveCopyReg_2$}
           -- ++(0.6,0) |- (join2);
\node[merge] at (join2) {};
\draw[sig] (join2) -- node[above=2pt,xshift=1pt] {$E_{2}$} (dots.west);


\draw[sig] (Nn.east) node[above=2pt, xshift=10pt] {$E'_n$} -- (joinn);
\draw[sig] (Ln.east) node[above=2pt, xshift=10pt]   {$\eveCopyReg_n$}
           -- ++(0.6,0) |- (joinn);
\node[merge] at (joinn) {};

\draw[sig] (dots.east) node[above=2pt,xshift = 3pt] {$E_{n-1}$} -- (Nn.west);
\draw[sig] (joinn) -- node[above=2pt,xshift=1pt] {$E_{n}$} (En);

\foreach \i/\N/\L in {1/N1/L1, 2/N2/L2, n/Nn/Ln}{
  \draw[sig] (\N.south) -- ++(0,-0.6)
             |- node[above right=3pt] {$B_{\i}$} (\L.west);
}

\foreach \i/\L in {1/L1, 2/L2, n/Ln}{
  \node[below=1.5cm of \L] (S\i) {$S_{\i}\widehat{C}_{\i}$};
  \draw[sig] (\L.south) -- (S\i);

  \node[above=2.5cm of \L] (A\i) {$\Aclassical_{\i}$};
  \draw[sig,<-] (\L.north) -- (A\i);
}

\draw[red,thick,rounded corners]
  ($ (N1.north west) + (-3pt,3pt)$)  
  rectangle
  ($ (join1 |- L1.south) + (3pt,-6pt)$);   
\node[text=red] at ($($(N1.north west) + (0pt, 10pt)$)$) {\(\qkdMapMod_1\)};
\draw[red,thick,rounded corners]
  ($ (N2.north west) + (-3pt,3pt)$)
  rectangle
  ($ (join2 |- L2.south) + (3pt,-6pt)$);
\node[text=red] at ($($(N2.north west) + (0pt, 10pt)$)$) {\(\qkdMapMod_2\)};
\draw[red,thick,rounded corners]
  ($ (Nn.north west) + (-3pt,3pt)$)
  rectangle
  ($ (joinn |- Ln.south) + (3pt,-6pt)$);
\node[text=red] at ($($(Nn.north west) + (0pt, 10pt)$)$) {\(\qkdMapMod_n\)};
\end{tikzpicture}

}

\begin{figure}
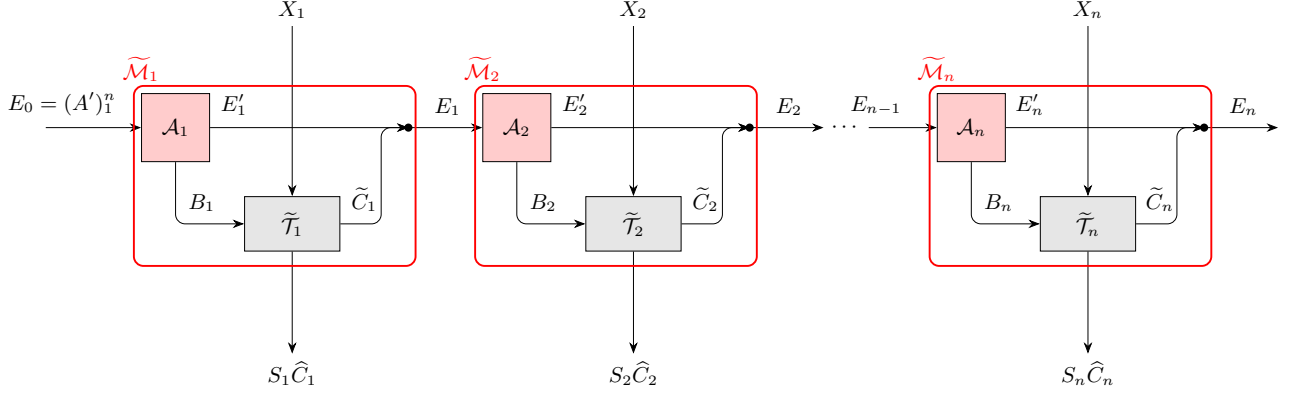

    \centering
    \scalebox{0.9}{\ModifiedProtocolChFig}
    \caption{Sequence of channels describing the evolution of the state described in the modified protocol from \cref{lem:modified_timing_protocol}.}
    \label{fig:modified_protocol_ch}
\end{figure}

\begin{definition}
\label{def:tuple_protocol}
    We denote by $\big\{\big\{\Astate^{(j)}_{\Aclassical_j\Aprime_j}, \big\{\Bpovmel[j]\big\}_{\BclassicalVal\in\BclassicalAlph}, \annKeyMap[j]\big\}_{j=1}^\totRounds, \ppMap\big\}$
    an instance of the modified protocol described in \cref{lem:modified_timing_protocol}, where the elements of the tuple specify the input parameters. Here, for notational convenience, we let $\annKeyMap[j]$ denote the map corresponding to the announcement map $\annFunc[j]$ and key map $\keymapFunc[j]$ in the $j$th round, and implements the round-by-round classical processing.
\end{definition}

We now define security in terms of Eve's round-by-round attack channels $\attackCh_j$. The notion of restricted set of attack channels becomes relevant when discussing squashing maps as we can restrict Eve's attack to prove security for a protocol where Eve is unrestricted. We therefore formulate security with respect to a family of attack sets $\{\setAttackCh_j\}_j$. Ultimately, however, our goal remains to prove security of the \nameref{prot:generic_qkd_protocol} against arbitrary attacks, i.e. $\attackCh_j \in \cptp(\Ereg_{j-1},\Bmeas_j\Ereg_j')$ for all rounds $j$.

\begin{definition}[QKD security {\protect\cite[Def.~8.2]{tupkary_rigorous_2026} \cite{Portmann22,benor2004universalcomposablesecurityquantum,ferradini2025definingsecurityquantumkey}}]
\label{def:epsilon_security}
    Consider the protocol $\big\{\big\{\Astate_{\Aclassical_j\Aprime_j}^{(j)}, \big\{\Bpovmel[j]\big\}_{\BclassicalVal\in\BclassicalAlph}, \annKeyMap[j]\big\}_{j=1}^\totRounds, \ppMap\big\}$, cf. \cref{def:tuple_protocol}. The protocol is said to be $\epsSecu$-secure against all attacks in $\{\setAttackCh_j\}_j$ if
    \begin{equation}
        \sup_{\{\attackCh_j\in\setAttackCh_j\}_j} \norm{\left(\identity - \perfectReplaceCh\right)\ppMap\circ \gMapTildeFull_\totRounds \circ\attackCh_\totRounds\ldots\gMapTildeFull_1 \circ\attackCh_1[\Astate_{ \Aclassical_1^\totRounds (\Aprime)_1^\totRounds}]}_1 \leq \epsSecu\,,
    \end{equation}
    where $\perfectReplaceCh\in\cptp(\AkeyReg\BkeyReg,\AkeyReg\BkeyReg)$ is the ideal replacement channel defined in \cite[Eq.~(13)]{tupkary_rigorous_2026}, which replaces the output keys by perfect secure keys of the appropriate length, and $\gMapTildeFull_j$ is defined in \cref{lem:modified_timing_protocol}. Whenever the attack sets $\{\setAttackCh_j\}_j$ are omitted, we implicitly mean unrestricted attacks, i.e. $\setAttackCh_j = \cptp(\Ereg_{j-1},\Bmeas_j\Ereg_j')$ for all rounds $j$.
    
\end{definition}

\subsection{Imperfectly characterized devices}
\label{sec:imperfect_characterization_discussion}

\noindent Consider the \nameref{prot:generic_qkd_protocol}. The source is said to be \textit{imperfectly characterized} if the signal states are not exactly known, but are only known to lie in a given set of possible signal states. Similarly, the detector is said to be \textit{imperfectly characterized} if the POVMs describing the measurements are not exactly known, but are only known to lie in a given set of possible POVMs.

The MEAT framework \cite{arqand_marginal-constrained_2025, kamin_renyi_2025, tupkary_rigorous_2026} naturally handles perfectly characterized (yet potentially imperfect) sources and detectors, as the resulting optimization problem for computing the key rate can directly be evaluated for the (imperfect, known) states Alice prepares and the (imperfect, known) POVMs describing Bob's measurement. However, the devices used in the protocol can never be perfectly characterized. This can be due to various reasons, including inherent uncertainty in the characterization process, partial control by an adversary, or parameter fluctuations occurring during the protocol. It is therefore crucial to incorporate this uncertainty about the devices into the security analysis, which is precisely the aim of this work.

Throughout this work, we choose particular metrics to define the sets of possible signal states and POVMs. These choices are motivated by previous analyses of imperfectly characterized devices \cite{pereira_optimal_2025,nahar_imperfect_2026}, where they were shown to be operationally meaningful as they can be related to experimental device parameters. Moreover, they allow us to make use of the corresponding tools developed in those works. We emphasize, however, that these particular metrics are not fundamental. Rather, they constitute one natural choice for which the required mathematical tools are currently available.

More specifically, on the source side, we define the set of possible signal states by bounding the fidelity to some known \textit{target states}. This will be formalized in \cref{as:imperfectly_characterized_source} (\cref{sec:model_assumptions}) when specifying the assumptions on the devices. On the detector side, we quantify the uncertainty in the POVMs through a specific operator inequality relative to some known \textit{target POVMs}. This will be formalized in \cref{as:imperfectly_characterized_detector}.

Note that if the characterization process itself is subject to uncertainty, such that the signal states and POVMs are only known to lie in the corresponding sets up to a given confidence level, this uncertainty can be rigorously incorporated into the analysis using the approach from Ref.~\cite{tan2025incorporatingdevicecharacterizationsecurity}.

\section{Toolbox}
\label{sec:toolbox}

\noindent In this section, we introduce the main mathematical tools used throughout this work. We use source maps (\cref{sec:source_maps}) and squashing maps (\cref{sec:squashing_maps}) to handle source and detector imperfections, respectively, in \cref{sec:security_imperfect_devices}. An overview of the various source and squashing maps introduced is given in \cref{tab:overview_tools_imperfections}. Informally, they allow us to reduce the problem of proving security for a protocol with given signal states and detector POVMs to that of a new protocol with more convenient states and more convenient POVMs, as formalized in \cref{sec:security_source_squashing}. Finally, in \cref{sec:security_using_MEAT}, we introduce an equivalent entanglement-based picture and provide an overview of the MEAT framework based on Ref.~\cite{tupkary_rigorous_2026} to introduce the relevant notation and, importantly, identify where device imperfections and uncertainty about the devices enter the analysis.

\begin{table}[t]
\centering
\renewcommand{\arraystretch}{1.4}
\setlength{\tabcolsep}{12pt}
\begin{tabular}{lll}
\hline
 & \makecell[c]{\textbf{Infinite dimensions}} & \makecell[c]{\textbf{Imperfect characterization}} \\
\hline

\textbf{Source} & Tagging source map (\cref{lem:tagging_source_map})* & Partial charac. source map (\cref{lem:source_map_imperfect_source}) \\

\textbf{Detector} & Flag-state squasher (\cref{lem:flag_state_squasher})* & Noise channel (\cref{lem:noise_channel})\\

\hline
\end{tabular}
\caption{Overview of the various tools used to incorporate imperfect devices in the MEAT. (*) assumes block-diagonal structure. We handle imperfectly block-diagonal sources in \cref{sec:imperfect_block_diagonal}. }
\label{tab:overview_tools_imperfections}
\end{table}

\subsection{Source maps}
\label{sec:source_maps}

\noindent A source is described by a set of states $\{\genDensity_\genClassIndex\}_{\genClassIndex\in\genClassAlph}$ and the probability distribution $\prob_{\genClassReg}$ according to which the states are prepared. The states may be infinite-dimensional (e.g. describing optical pulses) and they may be imperfectly characterized (see \cref{sec:imperfect_characterization_discussion}). We handle both aspects using the idea of \textit{source maps} \cite{gottesman2004securityquantumkeydistribution, nahar_imperfect_2023, Curras-Lorenzo_2025} (\cref{def:source_map}). Source maps are powerful tools that allow us to relate the security of a given protocol where the source prepares $\{\genDensity_\genClassIndex\}_{\genClassIndex\in\genClassAlph}$ to that of a new protocol where the source prepares $\{\tilde \genDensity_\genClassIndex\}_{\genClassIndex\in\genClassAlph}$ instead. The new set of states $\{\tilde \genDensity_\genClassIndex\}_{\genClassIndex\in\genClassAlph}$ is then typically chosen to be more convenient to work with. The security aspects will be formalized in \cref{lem:epsilon_security_source_maps}. For now, we describe the two source maps used to (a) perform the dimension reduction and (b) handle imperfect characterization.

\begin{definition}[Source map]
\label{def:source_map}
    Let $\genDensity_{\genClassReg\genQReg}\in \setDensity_\leq(\genClassReg\genQReg)$ and $\tilde \genDensity_{\genClassReg\genQReg'} \in \setDensity_\leq(\genClassReg\genQReg')$ be classical in $\genClassReg$ and let $\srcMap\in\cptp(\genQReg', \genQReg)$. If $\genDensity_{\genClassReg\genQReg} = \srcMap[\tilde \genDensity_{\genClassReg\genQReg'}]$, then $\srcMap$ is called a \textit{source map}. Equivalently, if $\genDensity_{\genQReg|\genClassReg=\genClassIndex} = \srcMap[\tilde \genDensity_{\genQReg'|\genClassReg=\genClassIndex}]$ for all $\genClassIndex$, then $\srcMap$ is a source map.
\end{definition}

A particular construction of a source map, which always exists when the states $\genDensity_\genClassIndex$ are block-diagonal in the same basis, is the \textit{tagging source map} (\cref{lem:tagging_source_map}). It allows us to remove a block and replace it with classical flags storing the label $\genClassIndex\in\genClassAlph$. Given that this does not depend on the dimensionality of the block removed, in \cref{sec:problem_reductions} we use the tagging source map to remove infinite-dimensional blocks from a set of infinite-dimensional states to reduce the problem to a new set of states which are finite-dimensional. An example of infinite-dimensional states includes weak-coherent pulses which are commonly used for decoy-state protocols. Intuitively, the cost for performing this replacement is upper-bounded by the amount of information leaked through the classical flags, which is given by the probability weight of the removed block.

\begin{lemma}[Tagging source map \cite{gottesman2004securityquantumkeydistribution}]
\label{lem:tagging_source_map}
    Let $\genDensity_{\genClassReg\genQReg}\in \setDensity_\leq(\genClassReg\genQReg)$ be classical in $\genClassReg$ and of the block-diagonal form
    \begin{equation}
        \genDensity_{\genClassReg\genQReg} \coloneqq \sum_{\genClassIndex\in \genClassAlph} \prob_{\genClassReg}(\genClassIndex) \ketbra{\genClassIndex}{\genClassIndex}_\genClassReg \otimes \left(\prob(\blockVar=0|\genClassIndex) \genDensity_{\genQReg_0|\genClassReg=\genClassIndex}^{\blockVar=0} \oplus (1 -\prob(\blockVar=0|\genClassIndex)) \genDensity_{\genQReg_1|\genClassReg=\genClassIndex}^{\blockVar=1} \right)\,,
    \end{equation}
    where $\genQReg = \genQReg_0 \oplus \genQReg_1$. Let
    \begin{equation}
        \tilde \genDensity_{\genClassReg\genQReg'} \coloneqq \sum_{\genClassIndex\in \genClassAlph} \prob_{\genClassReg}(\genClassIndex) \ketbra{\genClassIndex}{\genClassIndex}_\genClassReg \otimes \left(\prob(\blockVar=0|\genClassIndex) \genDensity_{\genQReg_0|\genClassReg=\genClassIndex}^{\blockVar=0} \oplus (1 -\prob(\blockVar=0|\genClassIndex)) \ketbra{\genClassIndex}{\genClassIndex}_\flagSpaceReg\right)\,,
    \end{equation}
    where $\{\ket{\genClassIndex}\}_{\genClassIndex\in \genClassAlph}$ forms an orthonormal basis and $\genQReg' = \genQReg_0 \oplus \flagSpaceReg$. Then there exists a source map $\srcMap\in\cptp(\genQReg',\genQReg)$, also called \textit{tagging source map}, such that $\srcMap\left[\tilde\genDensity_{\genClassReg\genQReg'}\right] = \genDensity_{\genClassReg\genQReg}.$
\end{lemma}
\begin{proof}
    Let the source map be given by the channel that projects onto $\ketbra{\genClassIndex}{\genClassIndex}_\flagSpaceReg$ and prepares $\genDensity_{\genQReg_1|\genClassReg=\genClassIndex}^{\blockVar=1}$, then the statement directly follows.
\end{proof}

When the source is imperfectly characterized (see \cref{sec:imperfect_characterization_discussion}) with bounded fidelity between the signal states and known target states, then the construction described in \cref{lem:source_map_imperfect_source} below allows the problem to be reduced to pure state preparation, with partially known overlaps. In \cref{sec:problem_reductions}, this construction is used to incorporate imperfectly characterized sources by introducing an additional optimization over the Gram matrix containing the unknown overlaps, following the approach of Ref.~\cite{pereira_optimal_2025}.

\begin{lemma}[Partial characterization source map \cite{Curras-Lorenzo_2025}]
\label{lem:source_map_imperfect_source}
    Let $\genDensity_{\genClassReg\genQReg} = \sum_\genClassIndex \prob_{\genClassReg}(\genClassIndex) \ketbra{\genClassIndex}{\genClassIndex}_\genClassReg \otimes \genDensity_{\genQReg|\genClassReg=\genClassIndex}$ and target states $\tilde \genDensity_{\genClassReg\genQReg} = \sum_\genClassIndex \prob_{\genClassReg}(\genClassIndex) \ketbra{\genClassIndex}{\genClassIndex}_\genClassReg \otimes \tilde \genDensity_{\genQReg|\genClassReg=\genClassIndex}$ be classical in $\genClassReg$ with $F(\genDensity_{\genQReg|\genClassReg=\genClassIndex}, \tilde \genDensity_{\genQReg|\genClassReg=\genClassIndex}) \geq 1- \srcImpBnd_\genClassIndex$ for all $\genClassIndex\in \genClassAlph$ and $\ket{\tilde \genDensity^{\genClassIndex}}_{\genQReg\genPurReg}$ be any purification of $\tilde \genDensity_{\genQReg}^{\genClassIndex}$. Let
    \begin{equation}
        \rho'_{\genClassReg\genQReg\genPurReg\genPurReg'} \coloneqq \sum_{\genClassIndex\in\genClassAlph} \prob_{\genClassReg}(\genClassIndex) \ketbra{\genClassIndex}{\genClassIndex}_\genClassReg \otimes \ketbra{\genState^\genClassIndex}{\genState^\genClassIndex}_{\genQReg\genPurReg\genPurReg'}\,,
    \end{equation}
    with 
    \begin{equation}
    \label{eq:new_state_gram_matrix}
        \ket{\genState^\genClassIndex}_{\genQReg\genPurReg\genPurReg'} \coloneqq \sqrt{1-\srcImpBnd_\genClassIndex}\ket{\tilde \genDensity^{\genClassIndex}}_{\genQReg\genPurReg\genPurReg'} + \sqrt{\srcImpBnd_\genClassIndex}\ket{\tilde \genDensity^{\genClassIndex, \perp}}_{\genQReg\genPurReg\genPurReg'}\,,
    \end{equation}
    where $\ket{\tilde \genDensity^{\genClassIndex}}_{\genQReg\genPurReg\genPurReg'} = \ket{\tilde \genDensity^{\genClassIndex}}_{\genQReg\genPurReg}\otimes\ket{0}_{\genPurReg'}$. Then there exists a choice of normalized states $\ket{\tilde \genDensity^{\genClassIndex,\perp}}_{\genQReg\genPurReg\genPurReg'}$ satisfying $\bra{\tilde \genDensity^{\genClassIndex}}\ket{\tilde \genDensity^{\genClassIndex,\perp}}=0$ such that $\srcMap=\Tr_{\genPurReg\genPurReg'}$ is a source map satisfying $\srcMap[\rho'_{\genClassReg\genQReg\genPurReg\genPurReg'}]=\genDensity_{\genClassReg\genQReg}$.
\end{lemma}

\begin{proof}
    The result is that of \cite[Lemma~1]{Curras-Lorenzo_2025} and we briefly summarize the main idea. By Uhlmann's theorem, for every $\genClassIndex\in\genClassAlph$ there exists a purification $\ket{\genDensity^{\genClassIndex}}_{\genQReg\genPurReg}$ of $\genDensity_{\genQReg|\genClassReg=\genClassIndex}$ and a value $\srcImpBnd'_{\genClassIndex}\leq\srcImpBnd_{\genClassIndex}$ such that, after fixing an irrelevant global phase, $\braket{\tilde{\genDensity}^{\genClassIndex}}{\genDensity^{\genClassIndex}}=\sqrt{1-\srcImpBnd'_{\genClassIndex}}$. The difference between $\srcImpBnd'_{\genClassIndex}$ and $\srcImpBnd_{\genClassIndex}$ can be absorbed by appending a fictitious two-dimensional flag register $\genPurReg'$ with orthogonal states $\ket{0}_{\genPurReg'}$ and $\ket{1}_{\genPurReg'}$. The corresponding amplitudes can be chosen such that the overlap with $\ket{\tilde{\genDensity}^{\genClassIndex}}_{\genQReg\genPurReg}\otimes\ket{0}_{\genPurReg'}$ is exactly $\sqrt{1-\srcImpBnd_{\genClassIndex}}$, while leaving the marginal on $\genQReg$ unchanged. Decomposing the resulting purification into a component parallel to $\ket{\tilde{\genDensity}^{\genClassIndex}}_{\genQReg\genPurReg\genPurReg'}$ and an orthogonal component gives \cref{eq:new_state_gram_matrix}. Tracing out $\genPurReg\genPurReg'$ therefore recovers the original state, proving that $\srcMap=\Tr_{\genPurReg\genPurReg'}$ is a source map.

    \cref{lem:source_map_imperfect_source} is an existence statement. For every source compatible with the fidelity bounds, there exists a choice of the states $\ket{\tilde \genDensity^{\genClassIndex,\perp}}$ for which the source map is exact. Therefore, the uncertainty is in the state $\ket{\tilde \genDensity^{\genClassIndex,\perp}}$, rather than in the source map itself. These states need not be known explicitly as one can later optimize over all possible states compatible with the fidelity bounds. 
\end{proof}

\subsection{Squashing maps}
\label{sec:squashing_maps}

\noindent A detector is described by a POVM $\{\genpovmel\}_{\genClassIndex\in\genClassAlph}$ with outcomes labeled as $\genClassIndex\in\genClassAlph$. The POVM elements may be infinite-dimensional (e.g. describing threshold detectors) and they may be imperfectly characterized (see \cref{sec:imperfect_characterization_discussion}). We handle both aspects using the idea of \textit{squashing maps} \cite{PhysRevLett.101.093601, PhysRevA.84.020303, PhysRevA.89.012325, PRXQuantum.5.040315, PhysRevA.81.012328, PhysRevA.78.032302, PRXQuantum.2.020325, zhang_security_2021, nahar_imperfect_2026}, which can be viewed as the analogue of source maps on the detector side. Similarly to source maps, squashing maps are powerful tools as they allow us to relate the security of a given protocol where the measurement is described by a POVM $\{\genpovmel\}_{\genClassIndex\in\genClassAlph}$ to that of a new protocol where the measurement is instead described by $\{\genpovmel'\}_{\genClassIndex\in\genClassAlph}$. The new POVM $\{\genpovmel'\}_{\genClassIndex\in\genClassAlph}$ is then typically chosen to be more convenient to work with. The security aspects will be formalized in \cref{lem:epsilon_security_squashing}. For now, we describe the two squashing maps used to (a) perform the dimension reduction and (b) handle imperfect characterization.

\begin{definition}[Squashing map]
\label{def:squashing_map}
    Let $\{\genpovmel\}_{\genClassIndex}$ be a POVM acting on the register $\genQReg$, $\{\genpovmel'\}_{\genClassIndex}$ be a POVM acting on the register $\genQReg'$ and $\squashMap \in \cptp(\genQReg, \genQReg')$. If $\genpovmel = \squashMap^\dagger[\genpovmel']$ for all $\genClassIndex$, then $\squashMap$ is called a \textit{squashing map}. Equivalently, if $\Tr[\genpovmel'\squashMap[\genDensity]] = \Tr[\genpovmel\genDensity]$ for all $\genDensity\in\setDensity_=(\genQReg)$ and all $\genClassIndex$, then $\squashMap$ is a squashing map.
\end{definition}

A particular construction of a squashing map, which always exists when the POVM elements are block-diagonal in the same basis, is the so-called \textit{flag-state squasher} (\cref{lem:flag_state_squasher}). It allows us to remove a block and replace it with classical flags storing the label $\genClassIndex \in \genClassAlph$, analogous to the tagging source map on the source side. For example, POVMs describing measurements with threshold detectors are infinite-dimensional as they act on the full Fock space. Then, the flag-state squasher allows the problem to be reduced to finite-dimensional POVMs analogously to the tagging source map, which will be used in \cref{sec:problem_reductions}. 
\begin{lemma}[Flag-state squasher {\protect\cite[Theorem~1]{zhang_security_2021}}]
\label{lem:flag_state_squasher}
    Let $\{\genpovmel\}_{\genClassIndex}$ be a POVM acting on the register $\genQReg_0 \oplus \genQReg_1$ with block-diagonal elements, i.e. we can write $\genpovmel = \genpovmel^{\genQReg_0} \oplus \genpovmel^{\genQReg_1}$. Let $\{\genpovmel'\}_{\genClassIndex}$ be a POVM with elements defined as 
    \begin{equation}
        \genpovmel' \coloneqq \genpovmel^{\genQReg_0} \oplus \ketbra{\genClassIndex}{\genClassIndex}_{\flagSpaceReg}\,,
    \end{equation}
    with orthonormal basis $\{\ket{\genClassIndex}\}_\genClassIndex$ on the \textit{flag space} $\hilbert_{\flagSpaceReg}$ orthogonal to $\hilbert_{\genQReg_0}$, i.e. where the second block is replaced by classical flags. Then there exists a squashing map $\fssMap$, also called \textit{flag-state squasher}, such that $\genpovmel = \fssMap^\dagger[\genpovmel']$ for all $\genClassIndex$.
    The space $\hilbert_{\genQReg_0}$ is also called the \textit{preserved subspace} and $\hilbert_{\genQReg_1}$ the \textit{non-preserved subspace}.
\end{lemma}

Intuitively, the flag-state squasher first measures the input state with projector $\proj_{\genQReg_0}$ onto the preserved subspace and projector $\proj_{\genQReg_1}$ onto the non-preserved subspace. If the measurement outcome corresponds to the preserved subspace, the input state is left unchanged. If the measurement outcome corresponds to the non-preserved subspace, a measurement $\{\genpovmel^{\genQReg_1}\}_{\genClassIndex}$ is performed and the outcome $\genClassIndex$ forwarded as a classical flag. In this way, any weight in the non-preserved subspace is transferred to the flag space. The consequence of this is that if the weight of the input state in the non-preserved subspace is unbounded, then the weight in the flag space is unbounded as well. In the extreme case, the entire input could lie in the non-preserved subspace and be fully converted into flags, and Eve knows the measurement outcomes deterministically (and only trivial key rates are obtained). Therefore, we must derive bounds on the weight in the non-preserved subspace.

This is commonly achieved by choosing an outcome $W^{\genQReg}$ with non-zero eigenvalue $\lambda_\mathrm{min}$ in the non-preserved subspace where $\genQReg = \genQReg_0 \oplus \genQReg_1$.\footnote{The outcome $W^{\genQReg}$ may correspond to any of Bob’s POVM elements or a coarse-graining thereof. More generally, it may be any outcome whose statistics are observed during the protocol.} This implies that $W^{\genQReg} \geq \lambda_\mathrm{min} \proj_{\genQReg_1}$. Therefore, for all input states $\genDensity\in\setDensity_\leq(\genQReg)$ we can upper bound the weight in the flag space by
\begin{equation}
\label{eq:bound_subspace}
    \Tr[W^{\genQReg}\genDensity] \geq \lambda_\mathrm{min} \Tr[\proj_{\genQReg_1}\genDensity] = \lambda_\mathrm{min} \Tr[\proj_{\flagSpaceReg}\squashMap(\genDensity)]\,,
\end{equation}
where $\proj_{\flagSpaceReg}$ is the projector onto the flag space.
Hence, the weight in the flag space can be bounded directly from the statistics of the outcome $W^{\genQReg}$. This bound will be used in \cref{sec:security_source_squashing}.

When the detector is imperfectly characterized (see \cref{sec:imperfect_characterization_discussion}), the \textit{noise channel} construction (\cref{lem:noise_channel}) can be used to reduce the problem to that of a perfectly characterized detector. 

\begin{lemma}[Noise channel {\protect\cite[Theorem~3]{nahar_imperfect_2026}}]
\label{lem:noise_channel}
    Let $\{\genpovmel\}_{\genClassIndex}$ and $\{\genpovmelTilde\}_{\genClassIndex}$ be POVMs acting on the register $\genQReg$ with the same block-diagonal structure, i.e. $\genpovmel = \oplus_\blockVar \genpovmelBlock[\blockVar]$ and $\genpovmelTilde = \oplus_\blockVar \genpovmelBlockTilde[\blockVar]$, and $1>q_\blockVar\geq 0$ such that
    \begin{equation}
        \genpovmelBlock[\blockVar] - (1-\dtImpBnd_\blockVar) \genpovmelBlockTilde[\blockVar] \geq 0
    \end{equation}
    for all $\genClassIndex$ and $\blockVar$. Let $\{\genpovmel'\}_{\genClassIndex}$ be the target POVM with elements defined as 
    \begin{equation}
        \genpovmel' \coloneqq \left(\oplus_\blockVar \genpovmelBlockTilde[\blockVar]\right) \oplus \ketbra{\genClassIndex}{\genClassIndex}_{\flagSpaceReg}\,.
    \end{equation}
    Then there exists a squashing map $\noiseChMap \in \cptp(Q, Q\oplus \flagSpaceReg)$, also called \textit{noise channel}, such that
    \begin{align}
         \noiseChMap^\dagger[\genpovmel'] &= \genpovmel \\
         \noiseChMap^\dagger[\proj_\blockVar] &= (1-\dtImpBnd_\blockVar)\proj_\blockVar
    \end{align}
    for all $\genClassIndex$, where $\proj_\blockVar$ is a projector onto the block $\blockVar$.
\end{lemma}

Similarly to the flag-state squasher, the noise channel first performs a measurement onto the blocks with projectors $\proj_\blockVar$. Based on the outcome $\blockVar$, the input state is preserved with probability $1-\dtImpBnd_\blockVar$ and, with probability $\dtImpBnd_\blockVar$, a measurement is performed and the outcome forwarded as a classical flag. As with the flag-state squasher, this means that part of the input weight is transferred to the flag space. Consequently, when the flag-state squasher and the noise channel are combined, the total weight mapped to the flag space is determined by the combination of both squashers. A bound on the combined weight in the flag space, similarly to \cref{eq:bound_subspace}, is given by \cref{lem:combination_FSS_noise_channel}.

\subsection{Security with source and squashing maps}
\label{sec:security_source_squashing}

\noindent Informally, if a source map relating two sets of states exists, then proving security for a protocol where Alice prepares one set of states implies security for the protocol where Alice instead prepares the other set of states. Similarly, if a squashing map relating two POVMs exists, then proving security for a protocol where the measurements are described by one POVM implies security for a protocol where the measurements are instead described by the other POVM. Notably, in the case of squashing maps, this implication holds even when the proof for the new protocol only considers a restricted set of attack channels, while security of the original protocol is given against arbitrary attacks. We formalize this for source maps in \cref{lem:epsilon_security_source_maps} and for squashing maps in \cref{lem:epsilon_security_squashing}.

\begin{lemma}[Security with source maps {\protect\cite[Lemma~9.1]{tupkary_rigorous_2026}}]
\label{lem:epsilon_security_source_maps}
    Consider the protocol $\big\{\big\{\Astate^{(j)}_{\Aclassical_j\Aprime_j}, \big\{\Bpovmel[j]\big\}_{\BclassicalVal\in\BclassicalAlph}, \annKeyMap[j]\big\}_{j=1}^\totRounds,$ $\ppMap\big\}$. Let $\big\{\big\{\tilde{\Astate}^{(j)}_{\Aclassical_j\Aprime'_j}, \big\{\Bpovmel[j]\big\}_{\BclassicalVal\in\BclassicalAlph}, \annKeyMap[j]\big\}_{j=1}^\totRounds, \ppMap\big\}$ be a new protocol with different signal states. Assume there exists a source map $\srcMap\in\cptp((\Aprime')_1^\totRounds, (\Aprime)_1^\totRounds)$ such that
    \begin{equation}
        \srcMap[\tilde \Astate_{\Aclassical_1^\totRounds(\Aprime')_1^\totRounds}] = \Astate_{\Aclassical_1^n(\Aprime)_1^\totRounds}\,.
    \end{equation}
    Then, if the new protocol is $\epsSecu$-secure, then the original protocol is $\epsSecu$-secure.
\end{lemma}

\begin{lemma}[Security with squashing maps {\protect\cite[Lemma~9.3]{tupkary_rigorous_2026}}]
\label{lem:epsilon_security_squashing}
    Consider the protocol $\big\{\big\{\Astate^{(j)}_{\Aclassical_j\Aprime_j}, \big\{\Bpovmel[j]\big\}_{\BclassicalVal\in\BclassicalAlph}, \annKeyMap[j]\big\}_{j=1}^\totRounds,$ $\ppMap\big\}$. Let $\big\{\big\{\Astate^{(j)}_{\Aclassical_j\Aprime_j}, \big\{\BpovmelTilde[j]\big\}_{\BclassicalVal\in\BclassicalAlph}, \annKeyMap[j]\big\}_{j=1}^\totRounds, \ppMap\big\}$ be a new protocol with different POVMs. Assume there exist squashing maps $\squashMap_j \in \cptp(\Bmeas_j, \BmeasSquash_j)$ such that
    \begin{equation}
        \Bpovmel[j] = \squashMap_j^\dagger\left[\BpovmelTilde[j]\right]
    \end{equation}
    for all $j$ and $\BclassicalVal \in \BclassicalAlph$. Then, if the new protocol is $\epsSecu$-secure against all attacks in $\{\setAttackCh_j\}_j$ where
    \begin{equation}
    \label{eq:condition_restricted_attack_set_squashing}
        \setAttackCh_j \supset \squashMap_j \circ \cptp(\Ereg_{j-1}, \Bmeas_j \Ereg_j)
    \end{equation}
    for all $j$, then the original protocol is $\epsSecu$-secure against all attacks in $\{\cptp(\Ereg_{j-1}, \Bmeas_j \Ereg_j)\}_j$.
\end{lemma}

As discussed in the previous section, after applying the flag-state squasher, if Eve is unrestricted, she may place arbitrary weight of the input state into the non-preserved subspace. However, by using the statistics of an outcome $W^{\Bmeas_j}$ with non-zero minimum eigenvalue $\lambdaMin$ in the non-preserved subspace, we can bound this weight via \cref{eq:bound_subspace}. Equivalently, this restriction can be incorporated into the restricted attack sets as
\begin{equation}
    \setAttackCh_j =\left\{\attackCh_j \in \cptp(\Ereg_{j-1}, \BmeasSquash_j \Ereg_j) : \attackCh_j^\dagger\left[\widetilde{W}^{\BmeasSquash_j} \otimes \identity_{\Ereg_j} \right] \geq  \lambdaMin   \attackCh_j^\dagger\left[\proj_\flagSpaceReg^{\BmeasSquash_j} \otimes \identity_{\Ereg_j} \right]\right\}\,,
\end{equation}
where $\widetilde{W}^{\BmeasSquash_j}$ is the squashed outcome corresponding to $W^{\Bmeas_j}$, i.e. $\squashMap^\dagger[\widetilde{W}^{\BmeasSquash_j}] = W^{\Bmeas_j}$. Then, \cref{eq:condition_restricted_attack_set_squashing} trivially holds for this construction of the restricted attack sets \cite{nahar_proof-technique-independent_2026,tupkary_rigorous_2026}. 

We later use the combination of flag-state squasher and noise channel, and therefore formalize the construction of the combined restricted set of attack channels in \cref{lem:combination_FSS_noise_channel}. 
First, the flag-state squasher isolates the finite-dimensional block 
$\blockVar \leq \fssCutoff$, where the detector is imperfectly characterized, and replaces the non-preserved subspace with classical flags. Second, the detector imperfections within the preserved subspace are probabilistically absorbed as additional flags through the noise channel construction. These additional flags are stored in the same flag space introduced by the flag-state squasher. This reduces the detector POVM to a \textit{finite-dimensional} and \textit{perfectly characterized} POVM.

\begin{lemma}
\label{lem:combination_FSS_noise_channel}
    Consider the protocol $\big\{\big\{\Astate^{(j)}_{\Aclassical_j\Aprime_j}, \big\{\Bpovmel[j]\big\}_{\BclassicalVal\in\BclassicalAlph}, \annKeyMap[j]\big\}_{j=1}^\totRounds,$ $\ppMap\big\}$ where Bob's measurement POVMs are block-diagonal in Fock space, i.e. $\Bpovmel[j] = \big(\oplus_{\blockVar=0}^{\fssCutoff}\BpovmelBlock{j}{\blockVar}\big)\oplus\BpovmelBlock{j}{\blockVar>\fssCutoff}$. Let $W^{\Bmeas_j}$ be an outcome and $\lambdaMin>0$ such that $ W^{\Bmeas_j} \geq \lambda_\mathrm{min} \proj_{>\fssCutoff}^{\Bmeas_j}$ with photon-number cutoff $\fssCutoff$, where $\proj_{>\fssCutoff}^{\Bmeas_j}$ is a projector onto the space $\blockVar>\fssCutoff$. Assume that there exists a set of known POVM elements $\{\BpovmelBlockTilde{j}{\blockVar}\}_{j, \BclassicalVal\in \BclassicalAlph, \blockVar\leq\fssCutoff}$, such that
    \begin{equation}
    \label{eq:op_ineq_combination_fss}
        \BpovmelBlock{j}{\blockVar} - (1 - \dtImpBnd[\blockVar, j]) \BpovmelBlockTilde{j}{\blockVar} \geq 0\,,
    \end{equation}
    with $1>\dtImpBnd[\blockVar, j]\geq 0$ for all $\blockVar \leq \fssCutoff$ and all $\BclassicalVal\in\BclassicalAlph$. Define a new set of POVMs $\big\{\BpovmelTarg[j]\big\}_{j, \BclassicalVal\in \BclassicalAlph}$ where
    \begin{equation}
    \label{eq:povm_virt}
        \BpovmelTarg[j] \coloneqq \left(\oplus_{\blockVar = 0}^\fssCutoff \BpovmelBlockTilde{j}{\blockVar} \right) \oplus \ketbra{\BclassicalVal}{\BclassicalVal}_\flagSpaceReg\,.
    \end{equation}
    Then there exist squashing maps $\squashMap_j\in\cptp(\Bmeas_j, \BmeasSquash_j)$ such that
    \begin{align}
        \squashMap_j^\dagger \left[\BpovmelTarg[j]\right] &= \Bpovmel[j]
    \end{align}
    and if the restricted attack channel is constructed as
    \begin{align}
    \label{eq:restricted_attack_set_fss_noise}
         \setAttackCh_j = \left\{\attackCh_j \in \cptp(\Ereg_{j-1}, \BmeasSquash_j \Ereg_j) : \attackCh_j^\dagger\left[\widetilde{W}^{\Bmeas_j'} \otimes \identity_{\Ereg_j} \right] \geq  \lambdaMin  \left( \attackCh_j^\dagger\left[\identity_{\BmeasSquash_j\Ereg_j} \right] - \sum_{\blockVar=0}^{\fssCutoff} \frac{1}{1 - \dtImpBnd[\blockVar, j]} \attackCh_j^\dagger\left[\proj_\blockVar^{\BmeasSquash_j} \otimes \identity_{\Ereg_j} \right]\right)\right\}
    \end{align}
    where $\squashMap_j^\dagger \left[\widetilde{W}^{\Bmeas_j'}\right] = W^{\Bmeas_j}$, then
    \begin{equation}
        \setAttackCh_j \supset \squashMap_j \circ \cptp(\Ereg_{j-1}, \Bmeas_j \Ereg_j)
    \end{equation}
    for all rounds $j$.
\end{lemma}
\begin{proof}
    While the construction is technically inferred by Refs.~\cite{nahar_imperfect_2026,tupkary_rigorous_2026}, to the best of our knowledge no formal proof has previously appeared in the literature. We therefore provide a proof in \cref{ap:proof_combination_fss_noise_channel} for completeness. 
\end{proof}

Having introduced the various source and squashing maps, we provide a brief overview of the sequence of source  and squashing maps used to incorporate imperfect devices. By combining the tagging source map (\cref{lem:tagging_source_map}) with the partial characterization source map (\cref{lem:source_map_imperfect_source}) we reduce the problem of proving security for a protocol with imperfectly characterized and infinite-dimensional source to that of a protocol with finite-dimensional source and an additional optimization for the imperfect characterization (through \cref{lem:epsilon_security_source_maps}). Similarly, for the detector, combining the flag-state squasher (\cref{lem:flag_state_squasher}) and the noise channel construction (\cref{lem:noise_channel}) we are able to reduce the problem of proving security for a protocol with imperfectly characterized and infinite-dimensional detector to that of a protocol with perfectly characterized and finite-dimensional detector (through \cref{lem:epsilon_security_squashing,lem:combination_FSS_noise_channel}). Thus, we can combine these source maps and squashing maps to handle imperfectly characterized sources and detectors, which we argue rigorously in \cref{sec:problem_reductions}. 
Recall that an overview of the various tools used is given in \cref{tab:overview_tools_imperfections}. Before doing so, we consider some other necessary elements of the security proof.

\subsection{Security using the MEAT}
\label{sec:security_using_MEAT}

\noindent Consider the protocol $\big\{\big\{\Astate^{(j)}_{\Aclassical_j\Aprime_j}, \big\{\Bpovmel[j]\big\}_{\BclassicalVal\in\BclassicalAlph}, \annKeyMap[j]\big\}_{j=1}^\totRounds, \ppMap\big\}$. Following \cite[Lemma~8.5]{tupkary_rigorous_2026}, $\epsSecu$-security is first decomposed into $\epsCor$-correctness and $\epsSecr$-secrecy. Correctness follows directly from the error verification step in the \nameref{prot:generic_qkd_protocol}, which is identical to that considered in \cite[Lemma~8.6]{tupkary_rigorous_2026}. Therefore, the main task is to prove $\epsSecr$-secrecy. In \cref{sec:entanglement_based_picture}, we construct an equivalent entanglement-based protocol required for our analysis. Then, in \cref{sec:introduction_meat}, we apply the MEAT to the constructed protocol.

\subsubsection{Entanglement-based picture}
\label{sec:entanglement_based_picture}

\noindent Informally, instead of preparing states in register $\Aprime$ according to a probability distribution over $\Aclassical$, Alice can be viewed as preparing an entangled state in registers $\Ameas\Ashield\Aprime$ with shield system $\Ashield$, keeping $\Amarg\coloneqq \Ameas\Ashield$, and sending $\Aprime$ to Bob. At a later stage, Alice measures $\Amarg$ and stores the outcome in the classical register $\Aclassical$. This is a standard argument referred to as the \textit{source-replacement scheme} (which we describe in greater detail in \cref{lem:source_replacement_scheme}).
Using this, we can describe the evolution of the state from the \nameref{prot:generic_qkd_protocol}, after the time-commuting argument from \cref{lem:modified_timing_protocol}, using an equivalent \nameref{prot:entanglement_qkd_protocol}. Therefore, for the purpose of proving security, we may equivalently consider this entanglement-based protocol, cf. \cref{lem:relation_entanglement_picture}. The resulting sequence of channels is illustrated in \cref{fig:MEAT_channels}. 

\newcommand{\MEATChFig}{

\begin{tikzpicture}[>=Stealth,auto,
  box/.style  ={draw,minimum width=1cm,minimum height=1cm,
                align=center,fill=red!20},
  cpbox/.style ={draw,minimum width=1.4cm,minimum height=0.8cm,
                align=center,fill=gray!20},
  sig/.style  ={->,rounded corners=4pt},
  merge/.style={circle,fill,inner sep=1.2pt}
]


\node[box]   (N1) {$\attackCh_1$};
\node[cpbox, below right=\offsetY and \offsetX of N1] (L1) {$\gMap_1$};

\node[box,   right=\colsep of N1] (N2) {$\attackCh_2$};
\node[cpbox, below right=\offsetY and \offsetX of N2] (L2) {$\gMap_2$};

\node[right=\colsep of N2] (dots) {$\cdots$};
\node[box,   right= of dots] (Nn) {$\attackCh_n$};
\node[cpbox, below right=\offsetY and \offsetX of Nn] (Ln) {$\gMap_n$};

\node[left =1.4cm of N1] (E0) {};
\node[right=4cm   of Nn] (En) {};

\coordinate (join1) at ($(L1.east |- N1) + (1.0,0)$);
\coordinate (join2) at ($(L2.east |- N2) + (1.0,0)$);
\coordinate (joinn) at ($(Ln.east |- Nn) + (1.0,0)$);

\draw[sig] (E0) node[above=2pt, xshift=10pt]{$E_{0}=(A')_1^n$} -- (N1.west);

\draw[sig] (N1.east) node[above=2pt, xshift=10pt] {$E'_1$} -- (join1);
\draw[sig] (L1.east) node[above=2pt, xshift=10pt]   {$\eveCopyReg_1$}
           -- ++(0.6,0) |- (join1);
\node[merge] at (join1) {};
\draw[sig] (join1) -- node[above=2pt,xshift=1pt] {$E_{1}$} (N2.west);

\draw[sig] (N2.east) node[above=2pt, xshift=10pt] {$E'_2$} -- (join2);
\draw[sig] (L2.east) node[above=2pt, xshift=10pt]   {$\eveCopyReg_2$}
           -- ++(0.6,0) |- (join2);
\node[merge] at (join2) {};
\draw[sig] (join2) -- node[above=2pt,xshift=1pt] {$E_{2}$} (dots.west);


\draw[sig] (Nn.east) node[above=2pt, xshift=10pt] {$E'_n$} -- (joinn);
\draw[sig] (Ln.east) node[above=2pt, xshift=10pt]   {$\eveCopyReg_n$}
           -- ++(0.6,0) |- (joinn);
\node[merge] at (joinn) {};

\draw[sig] (dots.east) node[above=2pt,xshift = 3pt] {$E_{n-1}$} -- (Nn.west);
\draw[sig] (joinn) -- node[above=2pt,xshift=1pt] {$E_{n}$} (En);

\foreach \i/\N/\L in {1/N1/L1, 2/N2/L2, n/Nn/Ln}{
  \draw[sig] (\N.south) -- ++(0,-0.6)
             |- node[above right=3pt] {$B_{\i}$} (\L.west);
}

\foreach \i/\L in {1/L1, 2/L2, n/Ln}{
  \node[below=1.5cm of \L] (S\i) {$S_{\i}\widehat{C}_{\i}$};
  \draw[sig] (\L.south) -- (S\i);

  \node[above=2.5cm of \L] (A\i) {$A_{\i}$};
  \draw[sig,<-] (\L.north) -- (A\i);
}

\draw[red,thick,rounded corners]
  ($ (N1.north west) + (-3pt,3pt)$)  
  rectangle
  ($ (join1 |- L1.south) + (3pt,-6pt)$);   
\node[text=red] at ($($(N1.north west) + (0pt, 10pt)$)$) {\(\mathcal{M}_1\)};
\draw[red,thick,rounded corners]
  ($ (N2.north west) + (-3pt,3pt)$)
  rectangle
  ($ (join2 |- L2.south) + (3pt,-6pt)$);
\node[text=red] at ($($(N2.north west) + (0pt, 10pt)$)$) {\(\mathcal{M}_2\)};
\draw[red,thick,rounded corners]
  ($ (Nn.north west) + (-3pt,3pt)$)
  rectangle
  ($ (joinn |- Ln.south) + (3pt,-6pt)$);
\node[text=red] at ($($(Nn.north west) + (0pt, 10pt)$)$) {\(\mathcal{M}_n\)};
\end{tikzpicture}

}

\begin{figure}
    \centering
    \scalebox{0.9}{\MEATChFig}
    \caption{Sequence of channels describing the evolution of the state described in the \nameref{prot:entanglement_qkd_protocol} \cite[Fig.~4]{tupkary_rigorous_2026}.}
    \label{fig:MEAT_channels}
\end{figure}

{\captionsetup{justification=raggedright, singlelinecheck=false, labelfont=bf}
\renewcommand{\figurename}{Protocol}
\captionof{protocol}[Virtual entanglement-based Protocol]{Virtual entanglement-based Protocol}
\label{prot:entanglement_qkd_protocol}
}

\noindent \textbf{Parameters:}

\noindent\(\displaystyle
\begin{aligned}
\totRounds\in\mathbb{N} & : && \quad \text{Total number of rounds} \\
\Big\{\prob_{\Aclassical_j}\in\probSimplex_{|\AclassicalAlph|}\Big\}_{j=1}^{\totRounds} & : && \quad \text{Alice's setting choice probability distributions} \\
\Big\{\big\{(\Astate_\AclassicalVal^{(j)})_{\Aprime_j}\in\setDensity_=(\Aprime_j)\big\}_{\AclassicalVal\in\AclassicalAlph}\Big\}_{j=1}^{\totRounds} & : && \quad \text{Alice's signal states} \\
\Big\{\big\{\Bpovmel[j]\succeq 0\big\}_{\BclassicalVal\in\BclassicalAlph}:\sum_{\BclassicalVal\in\BclassicalAlph}\Bpovmel[j]=\identity_{\Bmeas_j}\Big\}_{j=1}^{\totRounds} & : && \quad \text{POVM elements describing Bob's measurements} \\
\Big\{\annFunc[j]:\AclassicalAlph\times\BclassicalAlph\rightarrow\probSimplex_{|\varDecAlph|}\Big\}_{j=1}^{\totRounds} & : && \quad \text{Probabilistic announcement maps} \\
\Big\{\keymapFunc[j]:\AclassicalAlph\times\varDecAlph\rightarrow\secretAlph\Big\}_{j=1}^{\totRounds} & : && \quad \text{Key maps} \\
\ppMap\in\cptp(\Aclassical_1^\totRounds\BclassicalReg_1^\totRounds\secretReg_1^\totRounds\varDecReg_1^\totRounds, \AkeyReg\BkeyReg\cppReg\varDecReg_1^\totRounds) & : && \quad \text{Classical post-processing map}
\end{aligned}
\) \vspace{0.3cm}

\noindent \textbf{Protocol steps:}

\begin{enumerate}
    \item Alice prepares the source-replaced state $\genDensity_{\Ameas_1^\totRounds \Ashield_1^\totRounds (\Aprime)_1^\totRounds} = \bigotimes_{j=1}^\totRounds \genDensity_{\Ameas_j \Ashield_j \Aprime_j}^{(j)}$ (\cref{lem:source_replacement_scheme}), sends the registers $(\Aprime)_1^\totRounds$ to Eve and keeps $\Ameas_1^\totRounds \Ashield_1^\totRounds$. We denote Eve's starting register as $\Ereg_0$.
    \item For every round $j\in \{1,\ldots,\totRounds\}$, the following maps are applied:
    \begin{enumerate}
        \item Eve implements her attack $\attackCh_j\in \setAttackCh_j \subseteq \cptp(\Ereg_{j-1}, \Bmeas_j \Ereg_j')$.
        \item \textbf{Measurements:} Eve forwards $\Bmeas_j$ to Bob. Bob measures the quantum register $\Bmeas_j$ with a POVM $\{\Bpovmel[j]\}_{\BclassicalVal\in \BclassicalAlph}$ and stores the outcome in a register $\BclassicalReg_j$ with alphabet $\BclassicalAlph$. Alice measures the quantum registers $\Ameas_j\Ashield_j$ with a POVM $\{\Apovmel[j]\}_{\AclassicalVal\in \AclassicalAlph}$ with $\Apovmel[j] = \ketbra{\AclassicalVal}{\AclassicalVal}_{\Ameas_j}\otimes \identity_{\Ashield_j}$ and stores the outcome in a register $\Aclassical_j$ with alphabet $\AclassicalAlph$. 
        \item \textbf{Public announcements:} Alice and Bob make public announcements based on their registers $\Aclassical_j$ and $\BclassicalReg_j$. We denote $\varDecReg_j$ the register storing all public announcements in the $j$th round and $\varDecAlph$ its alphabet. The mapping between their outcomes and the announcements may be probabilistic and can be represented by a stochastic map $\annFunc[j]:\AclassicalAlph \times \BclassicalAlph \rightarrow \probSimplex_{|\varDecAlph|}$. In particular, the register $\varDecReg_j$ stores the decision whether the round is a test or generation round, cf. \cref{rem:test_gen_decision}. A copy $\widetilde C_j$ of all announcements is given to Eve. 
        \item \textbf{Sifting and key map:} Alice maps her private data $\Aclassical_j$ and public announcements $\varDecReg_j$ to a classical value $\secretReg_j$ with alphabet $\secretAlph$. This can be represented by a function $\keymapFunc[j]:\AclassicalAlph\times\varDecAlph\rightarrow\secretAlph$. The last three steps can be described by a map $\gMap_j \in \cptp(\Amarg_j\Bmeas_j, \secretReg_j \varDecReg_j \widetilde C_j)$ where $\Amarg_j = \Ameas_j\Ashield_j$.
        \item The registers $\Ereg_j'$ and the copy of the classical announcements $\widetilde C_j$ are combined into a register $\Ereg_j$, which is forwarded to the next round.
        
        Each round can therefore be described by a map $\qkdMap_j\in\cptp(\Ereg_{j-1}\Amarg_j, \secretReg_j\varDecReg_j \Ereg_j)$ given by $\qkdMap_j \coloneqq \gMap_j\circ \attackCh_j$ such that the state output is can be written as
        \begin{equation}
        \label{eq:state_output_protocol}
            \genDensity_{\secretReg_1^\totRounds \varDecReg_1^\totRounds \Ereg_\totRounds} = \qkdMap_\totRounds\circ \ldots \circ \qkdMap_1\left(\genDensity_{\Ameas_1^\totRounds \Ashield_1^\totRounds (\Aprime)_1^\totRounds} \right)\,.
        \end{equation}
    \end{enumerate}
    \item Alice and Bob perform the remaining post-processing steps from the \nameref{prot:generic_qkd_protocol} described by $\ppMap\in\cptp(\Aclassical_1^\totRounds\BclassicalReg_1^\totRounds\secretReg_1^\totRounds\varDecReg_1^\totRounds, \AkeyReg\BkeyReg\cppReg\varDecReg_1^\totRounds)$.
\end{enumerate}

\begin{lemma}[{\protect\cite[Lemma~8.4]{tupkary_rigorous_2026}}]
\label{lem:relation_entanglement_picture}
    For any protocol $\big\{\big\{\Astate^{(j)}_{\Aclassical_j\Aprime_j}, \big\{\Bpovmel[j]\big\}_{\BclassicalVal\in\BclassicalAlph}, \annKeyMap[j]\big\}_{j=1}^\totRounds, \ppMap\big\}$, the evolution of the state can be described by a \nameref{prot:entanglement_qkd_protocol}, where the inputs are given by the elements of the tuple, cf. \cref{def:tuple_protocol}.
\end{lemma}

Applying the MEAT to the \nameref{prot:entanglement_qkd_protocol} reduces the security analysis to a single-round optimization problem, as described in the next section. For the security analysis, we need to consider all attack channels $\attackCh_j \in \setAttackCh_j$. Instead, we would like to consider a related set of maps $\marginalSet_j \subseteq \cptp(A'_j, B_j)$, introduced in \cite[Sec.~8.5]{tupkary_rigorous_2026}. Doing so allows us to evaluate the optimization problem in the MEAT (\cref{eq:opti_original_meat}) by optimizing over a set of states obtained by applying the channels $\marginalMap_j \in \marginalSet_j$ to the source-replaced state (which has fixed dimensions) instead of the channels $\attackCh_j$ (which involve the systems $E_j$ of unknown dimension). 
To do so, we require that applying maps in $\marginalSet_j$ and then giving a purification to Eve is at least as powerful as applying maps in $\setAttackCh_j$.  

\begin{definition}[{\protect Marginal of an attack set \cite[Definition~8.4]{tupkary_rigorous_2026}}]
\label{def:marginal}
    Let $\setAttackCh_j\subseteq \cptp(\Ereg_{j-1}, \Bmeas_j\Ereg_j')$ be an attack set and $\marginalSet_j\subseteq \cptp(\Aprime_j, \Bmeas_j)$. The set $\marginalSet_j$ is said to be a \textit{marginal} of $\setAttackCh_j$ if for every $\omega_{\Amarg_j\Ereg_{j-1}\evePurReg}$ such that $\omega_{\Amarg_j} = \sigma^{(j)}_{\Amarg_j}$ and $\evePurReg$ is a purifying register for $\Amarg_j\Ereg_{j-1}$ and every $\attackCh_j \in \setAttackCh_j$, there exist $\marginalMap_j \in \marginalSet_j$ such that
    \begin{equation}
        \marginalMap_j\left[\sigma^{(j)}_{\Amarg_j\Aprime_j}\right] = \Tr_{\Ereg_j \evePurReg} \circ \attackCh_j\left[\omega_{\Amarg_j\Ereg_{j-1}\evePurReg}\right]\,,
    \end{equation}
    where $\sigma^{(j)}_{\Amarg_j\Aprime_j}$ is the source-replaced state, cf. \cref{lem:source_replacement_scheme}.
\end{definition}

\subsubsection{Application of the MEAT}
\label{sec:introduction_meat}

\noindent Recall that we are now only concerned with proving $\epsSecr$-secrecy. In the following, we provide an overview of the security proof based on \cite[Theorem~8.3]{tupkary_rigorous_2026}, where perfectly characterized and finite-dimensional devices are assumed. This allows us to introduce the relevant notation while identifying where imperfect knowledge about the devices enters the analysis. The extension of the security proof to imperfectly characterized devices is deferred to \cref{sec:security_imperfect_devices}. We refer the interested reader to Ref.~\cite{tupkary_rigorous_2026} for a detailed discussion of the steps that remain unchanged when considering imperfectly characterized devices.

We consider the equivalent \nameref{prot:entanglement_qkd_protocol}, whose sequence of channels is illustrated in \cref{fig:MEAT_channels}. Recall that $\genDensity_{\Amarg_j\Aprime_j}$ is Alice's source-replaced state in round $j$ with marginal constraint $\genDensity_{\Amarg_j}=\Astate^{(j)}_{\Amarg_j}$ given by \cref{lem:source_replacement_scheme}, where $\Amarg_j \coloneqq \Ameas_j \Ashield_j$. Secrecy is established via \cite[Theorem~8.1]{tupkary_rigorous_2026} by constructing a global tradeoff function $\fFull$ satisfying
\begin{equation}
    \fRenyiUpEntFull(\secretReg^{\totRounds}_1|\varDecReg_1^\totRounds\Ereg_{\totRounds}) \geq 0
\end{equation}
for the output state $\genDensity_{\secretReg_1^\totRounds\varDecReg_1^\totRounds\Ereg_\totRounds}$, cf.~\cref{eq:state_output_protocol}, and the $f$-weighted Rényi entropy is defined in \cref{def:f_weighted_renyi_entropy}. The global tradeoff function is then obtained using the MEAT as follows \cite[Sec.~8.7.1]{tupkary_rigorous_2026}.

\begin{enumerate}
    \item For every round $j$ and every realization $\varDecVal_1^{j-1}$, choose a tradeoff function $\ftradeoff{1}{j-1}:\varDecReg_j\rightarrow\mathbb{R}.$ The choice of tradeoff function is otherwise arbitrary.

    \item Construct the restricted set of attack channels $\setAttackCh_j\subseteq\cptp(\Ereg_{j-1},\Bmeas_j\Ereg_j)$ and a corresponding marginal set $\marginalSet_j$. These sets incorporate constraints resulting from potential squashing map reductions. 

    \item Let $\kappaQKDVirt\left(\ftradeoff{1}{j-1},\Astate^{(j)}_{\Amarg_j},\gMap_j\right)$ be any value satisfying
    \begin{align}
    \label{eq:kappa_without_uncertainty}
    \kappaQKDVirt\left(\ftradeoff{1}{j-1},\Astate^{(j)}_{\Amarg_j},\gMap_j\right)
    &\leq \tilde{\kappa}\left(\ftradeoff{1}{j-1},\Astate^{(j)}_{\Amarg_j},\gMap_j, \marginalSet_j\right) \\
    &\coloneqq \inf_{\genDensity_{\Amarg_j\Bmeas_j}\in\stateSet_j(\marginalSet_j)}\fRenyiUpEnt{1}{j-1}(\secretReg_j|\eveCopyReg_j\varDecReg_j\evePurReg)_{\gMap_j[\purFunc(\genDensity_{\Amarg_j\Bmeas_j})]} \label{eq:opti_original_meat}
    \end{align}
    where the set of output states is given by
    \begin{equation}
    \label{eq:def_set_out_marginal}
    \stateSet_j(\marginalSet_j)\coloneqq\left\{\marginalMap_j[\genDensity_{\Amarg_j\Aprime_j}]:\genDensity_{\Amarg_j}=\Astate^{(j)}_{\Amarg_j},\marginalMap_j\in\marginalSet_j\right\}\,,
    \end{equation}
    with purifying function $\purFunc$ from $\Amarg_j\Bmeas_j$ onto $\evePurReg$ (see \cite[Def.~6.2]{tupkary_rigorous_2026}), and the map $\gMap_j$ is defined in the \nameref{prot:entanglement_qkd_protocol} (and depends on Bob's POVM $\{\Bpovmel[j]\}_{\BclassicalVal\in \BclassicalAlph}$).

    \item Define the global tradeoff function
    \begin{equation}
    \fFull\coloneqq\sum_{j=1}^{\totRounds}\left(\ftradeoff{1}{j-1}(\varDecVal_j)+\kappaQKDVirt\left(\ftradeoff{1}{j-1},\Astate^{(j)}_{\Amarg_j},\gMap_j\right)\right)\,.
    \end{equation}

    \item Finally, the protocol is $\epsilon$-secure when the output key length satisfies
    \begin{equation}
    \keyl(\varDecVal^\totRounds_1)=\max\left\{0,\left\lfloor\fFull(\varDecVal^\totRounds_1)-\ECcost(\varDecVal^\totRounds_1)-\left\lceil\log\left(\frac{1}{\epsCor}\right)\right\rceil-\frac{\alpha}{\alpha-1}\log\left(\frac{1}{\epsSecr}\right)+2\right\rfloor\right\}\,.
    \end{equation}
\end{enumerate}

If we now consider a protocol with uncertainty about the devices, the issue is that we can no longer directly evaluate step 3 above, since the marginal states $\big\{\sigma_{\Amarg_j}^{(j)}\big\}_j$ and the protocol maps $\big\{\gMap_j\big\}_j$ are not known exactly (recall that the protocol maps are determined by the measurement POVMs, which are now imperfectly characterized). Instead, they are only known to belong to corresponding sets of possible marginal states and protocol maps, determined by the characterization and modeling of the devices. Therefore, we need to additionally optimize \cref{eq:kappa_without_uncertainty} over all elements in these sets.

While this reduces the problem to an optimization over compatible marginals and protocol maps, the resulting optimization problem may not be numerically tractable. Moreover, the signal states and detector POVMs may be infinite-dimensional, in which case \cite[Theorem~8.3]{tupkary_rigorous_2026} cannot even be applied, cf.~\cref{sec:security_using_MEAT}. To resolve these issues, we
\begin{enumerate}
    \item[(a)] apply a suitable sequence of squashing maps to reduce imperfect and potentially infinite-dimensional detectors to finite-dimensional, perfectly characterized POVMs,
    \item[(b)] apply a suitable sequence of source maps to reduce imperfect and potentially infinite-dimensional sources to finite-dimensional source states (which are not perfectly known), and
    \item[(c)] construct the resulting set of possible marginals describing the remaining uncertainty about the virtual source.
\end{enumerate}

Having identified the points at which imperfect devices enter the security proof of \cite[Theorem~8.3]{tupkary_rigorous_2026}, and having introduced all the required tools in the previous sections, we are now in a position to explicitly carry out these steps in the following section.
\section{Security with imperfect devices}
\label{sec:security_imperfect_devices}

\noindent In this section, we combine the tools from \cref{sec:toolbox} to prove security with imperfect and imperfectly characterized devices. We first state the model assumptions on the physical source and detector and apply a sequence of source maps and squashing maps, cf. \cref{sec:source_maps,sec:squashing_maps}, to reduce security of the original protocol to security of a virtual finite-dimensional protocol with constraints induced by the device imperfections. Then, in \cref{sec:construction_set_output_states} we construct the set of output states entering the $\tilde{\kappa}$ optimization in \cref{eq:opti_original_meat}, and derive the resulting security statement for the \nameref{prot:generic_qkd_protocol} with imperfect and imperfectly characterized devices in \cref{sec:security_statement}.

\subsection{Model assumptions}
\label{sec:model_assumptions}

\noindent We denote $\AclassicalVal \in \AclassicalAlph$ Alice's setting choice where $\AclassicalVal=(a,\mu)$ includes the encoding $a\in\tilde{\AclassicalAlph}$ (basis choice and bit value) and, for decoy-state protocols, the intensity choice $\mu\in \intensityAlph$. For non-decoy state protocols, e.g. without intensity choice or when only one intensity is used, $\intensityAlph$ is a singleton. This allows us to describe both decoy and non-decoy protocols uniformly in the following. 

Consider the protocol $\big\{\big\{\Astate^{(j)}_{\Aclassical_j\Aprime_j}, \big\{\Bpovmel[j]\big\}_{\BclassicalVal\in\BclassicalAlph}, \annKeyMap[j]\big\}_{j=1}^\totRounds, \ppMap\big\}$ as defined in \cref{def:tuple_protocol}. We make the following assumptions on the inputs of the protocol, which translate into requirements on the physical devices. 

\begin{assumption}[Block-diagonal source]
    \label{as:block_diagonal_source} Alice's signal states are block-diagonal in the Fock basis at least up to a photon-number cutoff $\tagCutoff$, i.e. we can write 
    \begin{equation}
        \label{eq:assumption_state_prepared}
        (\Astate_\AclassicalVal^{(j)})_{\Aprime_j} = \sum_{\blockVar=0}^\tagCutoff \prob_{\blockVarReg_j|\Aclassical_j}(\blockVar|\AclassicalVal) (\Astate_{\AclassicalVal,\blockVar}^{(j)})_{\Aprime_j} + \left(1-\sum_{\blockVar=0}^\tagCutoff \prob_{\blockVarReg_j|\Aclassical_j}(\blockVar|\AclassicalVal)\right) (\Astate_{\AclassicalVal,\blockVar>\tagCutoff}^{(j)})_{\Aprime_j}
    \end{equation}
    where $(\Astate_{\AclassicalVal,\blockVar}^{(j)})_{\Aprime_j}$ is finite-dimensional and supported on the $\blockVar$-photon subspace for all $\blockVar \leq \tagCutoff$, and $(\Astate_{\AclassicalVal,\blockVar>\tagCutoff}^{(j)})_{\Aprime_j}$ is supported on the subspace with more than $\tagCutoff$ photons. In \cref{sec:imperfect_block_diagonal}, we relax this assumption by considering the case where the states are not perfectly block-diagonal, e.g. due to imperfect phase randomization. 
\end{assumption}

\begin{assumption}[Imperfectly characterized signal states]
    \label{as:imperfectly_characterized_source} Alice's signal states are imperfectly characterized for each $\blockVar$-photon block with $\blockVar\leq \tagCutoff$, cf. \cref{sec:imperfect_characterization_discussion}. That is, there exists a known set of states $\{(\tilde \Astate_{\AclassicalVal,\blockVar}^{(j)})_{\Aprime_j}\}_{j, \AclassicalVal\in \AclassicalAlph, \blockVar\leq\tagCutoff}$ such that
        \begin{equation}
        \label{eq:fidelity_bound_assumption}
            F(\Astate_{\AclassicalVal,\blockVar}^{(j)}, \tilde \Astate_{\AclassicalVal,\blockVar}^{(j)}) \geq 1-\srcImpBnd[\AclassicalVal,\blockVar,j]
        \end{equation}
        for all $\AclassicalVal\in\AclassicalAlph$, $\blockVar\leq\tagCutoff$ and all rounds $j$.
\end{assumption}

\begin{assumption}[Imperfectly characterized photon-number distribution]
\label{as:probability_bounds} The probability distribution $\prob_{\blockVarReg_j|\Aclassical_j}$ is imperfectly characterized, i.e. there exist known upper and lower bounds such that
    \begin{equation}
        \prob_{\blockVarReg_j|\Aclassical_j}^{\mathrm{L}}(\blockVar|\AclassicalVal) \leq \prob_{\blockVarReg_j|\Aclassical_j}(\blockVar|\AclassicalVal) \leq \prob_{\blockVarReg_j|\Aclassical_j}^{\mathrm{U}}(\blockVar|\AclassicalVal)
    \end{equation}
    for all $\AclassicalVal\in\AclassicalAlph$, $\blockVar\leq\tagCutoff$ and all rounds $j$.
\end{assumption}

\begin{assumption}[Block-diagonal detector]
    \label{as:block_diagonal_detector} Bob's POVMs are block-diagonal in the Fock basis at least up to a photon-number cutoff $\fssCutoff$, i.e. we can write
        \begin{equation}
            \label{eq:assumption_detector_POVM_structure}
            \Bpovmel[j] = \left(\oplus_{\blockVar=0}^\fssCutoff\BpovmelBlock{j}{\blockVar}\right)\oplus\BpovmelBlock{j}{\blockVar>\fssCutoff}
        \end{equation}
        for all $\BclassicalVal\in\BclassicalAlph$ and all rounds $j$, and the POVM elements up to $\fssCutoff$ are finite-dimensional.
\end{assumption}

\begin{assumption}[Imperfectly characterized detector]
    \label{as:imperfectly_characterized_detector} Bob's POVMs are imperfectly characterized for each $\blockVar$-photon block with $\blockVar\leq \fssCutoff$, cf. \cref{sec:imperfect_characterization_discussion}. That is, there exists a known set of POVM elements $\{\BpovmelBlockTilde{j}{\blockVar}\}_{j, \BclassicalVal\in \BclassicalAlph, \blockVar\leq\fssCutoff}$ and $0\leq \dtImpBnd[\blockVar, j]<1$, such that
        \begin{equation}
        \label{eq:povm_deviation_assumption}
            \BpovmelBlock{j}{\blockVar} - (1 - \dtImpBnd[\blockVar, j]) \BpovmelBlockTilde{j}{\blockVar} \geq 0\,,
        \end{equation}
        for all $\BclassicalVal\in\BclassicalAlph$, $\blockVar \leq \fssCutoff$ and all rounds $j$.
\end{assumption}

The assumptions above are intentionally formulated in a way that is agnostic to both the dimensionality and the structure of the underlying physical systems. Overall, these assumptions encompass a broad class of device imperfections, as demonstrated in \cref{sec:device_imperfections}, as well as prepare-and-measure protocols and implementations, including both decoy-state protocols and non-decoy protocols such as single-photon protocols. We briefly comment on the assumptions above.

\begin{enumerate}
    \item \cref{as:block_diagonal_source,as:block_diagonal_detector} do not impose any restriction on the dimension of the blocks with photon number $\blockVar>\tagCutoff$ and $\blockVar>\fssCutoff$. In particular, the states in block $\blockVar>\tagCutoff$ may act on infinite-dimensional Hilbert spaces, as is the case for weak coherent pulses, or even be trivial (i.e. vanish entirely). 
    
    \item While \cref{as:block_diagonal_source,as:block_diagonal_detector} are stated in terms of a block-diagonal decomposition, this structure is not restrictive in practice. For example, in the special case where only the $m=1$ block is present (e.g. an ideal single-photon source), the decomposition becomes trivial by choosing $\tagCutoff = 1$ and setting $\prob_{\blockVarReg_j|\Aclassical_j}(\blockVar|\AclassicalVal) = 0$ for all $\blockVar\neq 1$.

    \item \cref{as:imperfectly_characterized_source,as:probability_bounds,as:imperfectly_characterized_detector} on the characterization of the devices are highly general in that they allow for arbitrary target models of the devices. In particular, the sets of target states and POVM elements can be chosen freely based on the physical implementation, as long as a valid bound on the deviation (e.g. in fidelity or operator inequality) is known, cf. \cref{sec:imperfect_characterization_discussion}. This enables us to incorporate a wide range of device imperfections in a unified manner\footnote{As discussed in \cref{sec:imperfect_characterization_discussion}, if the devices are imperfectly characterized and there is an additional probability that the conditions in \cref{eq:fidelity_bound_assumption,eq:povm_deviation_assumption} do not hold, then the approach from Ref.~\cite{tan2025incorporatingdevicecharacterizationsecurity} can be used to incorporate this into the security parameter.}.

    \item \cref{as:imperfectly_characterized_source,as:probability_bounds,as:imperfectly_characterized_detector} naturally include the case of perfectly characterized devices. In particular, if some or all device parameters are known exactly, the corresponding bounds can be chosen to be trivial (e.g. zero deviation or exact probabilities), and the analysis reduces to the perfectly characterized setting for the known parameters.

    \item \cref{as:block_diagonal_source} is a natural assumption for single-photon sources and phase-randomized weak coherent pulses commonly used in decoy-state protocols. This assumption is relaxed to imperfectly block-diagonal states in \cref{sec:imperfect_block_diagonal}, which, e.g., may result from imperfect phase randomization of weak coherent pulses or imperfect single-photon sources. \cref{as:block_diagonal_detector} is a natural assumption for threshold detectors, which are commonly used in QKD implementations.
\end{enumerate}

\subsection{Handling device imperfections}
\label{sec:problem_reductions}

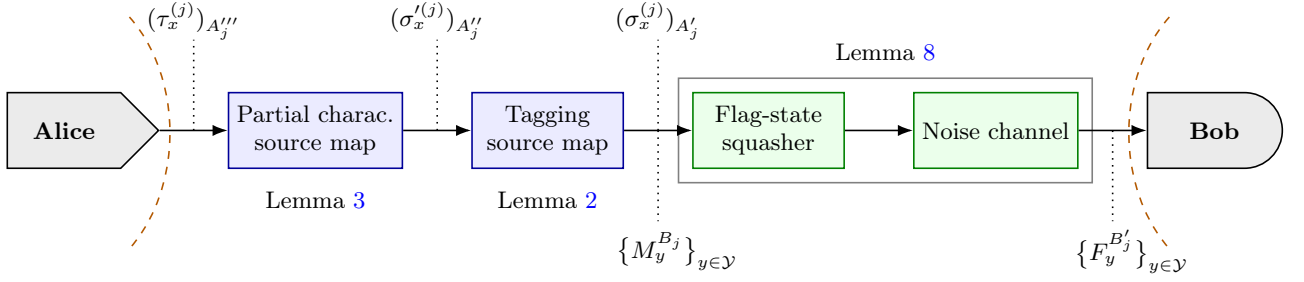
\begin{figure}
    \centering
    \begin{tikzpicture}[
        >=Latex,
        node distance=0.9cm,
        every node/.style={font=\small},
        alicebox/.style={
            signal,
            signal to=east,
            signal from=nowhere,
            draw,
            semithick,
            minimum width=2cm,
            minimum height=1cm,
            align=center,
            fill=black!8
        },
        midbox/.style={
            draw,
            semithick,
            minimum width=2cm,
            minimum height=1cm,
            align=center
        },
        bobbox/.style={
            rounded rectangle,
            rounded rectangle west arc=180,
            draw,
            semithick,
            minimum width=2cm,
            minimum height=1cm,
            align=center,
            fill=black!8
        },
        sourcemap/.style={
            midbox,
            fill=blue!8,
            draw=blue!60!black
        },
        squasher/.style={
            midbox,
            fill=green!8,
            draw=green!50!black
        }
    ]
    
    \node[alicebox] (alice) {\textbf{Alice}};
    \node[sourcemap, right=of alice] (b1) {Partial charac. \\ source map};
    \node[sourcemap, right=of b1] (b2) {Tagging \\ source map};
    \node[squasher, right=of b2] (b3) {Flag-state \\ squasher};
    \node[squasher, right=of b3] (b4) {Noise channel};
    \node[bobbox, right=of b4] (bob) {\textbf{Bob}};

    \node[below=0.18cm of b2] {\cref{lem:tagging_source_map}};

    \node[below=0.18cm of b1] {\cref{lem:source_map_imperfect_source}};

    \node[
        draw=gray,
        semithick,
        fit=(b3)(b4),
        inner sep=0.18cm
    ] (detbox) {};
    
    \node[
        above=0.12cm of detbox
    ]
    {\cref{lem:combination_FSS_noise_channel}};
    
    \draw[->, semithick] (alice) -- (b1);
    \draw[dotted, semithick]
    ($(alice.east)!0.5!(b1.west)$)
    --
    ++(0,1.2);
    
    \node
    at ($($(alice.east)!0.5!(b1.west)$)+(0,1.45)$)
    {$(\AstateVirt_\AclassicalVal^{(j)})_{\Aprime''_j}$};

    \draw[->, semithick] (b1) -- (b2);
    \draw[dotted, semithick]
    ($(b1.east)!0.5!(b2.west)$)
    --
    ++(0,1.2);
    
    \node
    at ($($(b1.east)!0.5!(b2.west)$)+(0,1.45)$)
    {$(\Astate_\AclassicalVal'^{(j)})_{\Aprime'_j}$};

    \draw[->, semithick] (b2) -- (b3);

    \draw[dotted, semithick]
    ($(b2.east)!0.5!(b3.west)$)
    --
    ++(0,1.2);
    
    \node
    at ($($(b2.east)!0.5!(b3.west)$)+(0,1.45)$)
    {$(\Astate_\AclassicalVal^{(j)})_{\Aprime_j}$};
    
    \draw[dotted, semithick]
    ($(b2.east)!0.5!(b3.west)$)
    --
    ++(0,-1.2);
    
    \node
    at ($($(b2.east)!0.5!(b3.west)$)+(0.25,-1.6)$)
    {$\big\{\Bpovmel[j]\big\}_{\BclassicalVal\in\BclassicalAlph}$};
    
    \draw[->, semithick] (b3) -- (b4);
    \draw[dotted, semithick]
    ($(b4.east)!0.5!(bob.west)$)
    --
    ++(0,-1.2);
    
    \node
    at ($($(b4.east)!0.5!(bob.west)$)+(0.25,-1.6)$)
    {$\big\{\BpovmelTarg[j]\big\}_{\BclassicalVal\in\BclassicalAlph}$};
    \draw[->, semithick] (b4) -- (bob);

    \draw[
        semithick,
        orange!70!black,
        dashed
    ]
    ([xshift=-0.4cm, yshift=1.5cm]alice.east)
    arc[
        start angle=50,
        end angle=-50,
        x radius=1.5cm,
        y radius=2.0cm
    ];
    
    \draw[
        semithick,
        orange!70!black,
        dashed
    ]
    ([xshift=0.3cm, yshift=1.5cm]bob.west)
    arc[
        start angle=130,
        end angle=230,
        x radius=1.5cm,
        y radius=2.0cm
    ];
    
    \end{tikzpicture}
    \caption{Sequence of source maps and squashing maps used in \cref{sec:problem_reductions}.}
    \label{fig:sequence_source_squashing_maps}
\end{figure}

\noindent Based on the assumptions from the previous section, we apply the sequence of source maps and squashing maps illustrated in \cref{fig:sequence_source_squashing_maps} to reduce the original protocol to a new protocol with finite-dimensional states and POVMs. The imperfect characterization of the devices appears as corresponding constraints. 

\textbf{Tagging source map:} Using the block-diagonality of Alice's signal states (\cref{as:block_diagonal_source}) we apply the tagging source map (\cref{lem:tagging_source_map}) to reduce the analysis on the source side to that of the blocks from $0$ to $\tagCutoff$, and flags. The entire $\blockVar>\tagCutoff$ block in \cref{eq:assumption_state_prepared}, irrespective of its dimension, is replaced by classical flags. Hence, the resulting signal states are finite-dimensional and given by\footnote{In this work, we use $\Aprime_j$, $\Aprime'_j$, and $\Aprime''_j$ to distinguish the signal registers before and after the successive source maps. We emphasize that the distinction is purely notational as all three denote the signal states at the corresponding stage of the reduction.}
\begin{equation}
\label{eq:application_tagging_source_map}
    (\Astate_\AclassicalVal'^{(j)})_{\Aprime'_j} = \sum_{\blockVar=0}^\tagCutoff \prob_{\blockVarReg_j|\Aclassical_j}(\blockVar|\AclassicalVal) (\Astate_{\AclassicalVal,\blockVar}^{(j)})_{\Aprime'_j} + \left(1 - \sum_{\blockVar=0}^\tagCutoff \prob_{\blockVarReg_j|\Aclassical_j}(\blockVar|\AclassicalVal)\right)\ketbra{\AclassicalVal}{\AclassicalVal}_{\Aprime'_j}\,,
\end{equation}
with orthonormal basis $\{\ket{\AclassicalVal}\}_{\AclassicalVal\in\AclassicalAlph}$ orthogonal to the span of the states $\big\{\Astate_{\AclassicalVal,\blockVar}^{(j)}\big\}_{j, \AclassicalVal\in \AclassicalAlph, \blockVar\leq\tagCutoff}$. As mentioned in \cref{sec:source_maps}, the cost of this replacement is given by the weight in the block $> \tagCutoff$, which is converted to classical flags. By increasing the cutoff, this step can be made arbitrarily tight. 

\textbf{Partial characterization source map:} Additionally, using \cref{as:imperfectly_characterized_source} on the characterization of the source, we apply the partial characterization source map (\cref{lem:source_map_imperfect_source}) on the blocks from $0$ to $\tagCutoff$, which yields new states
\begin{align}
\label{eq:def_new_virtual_states_alice}
    (\AstateVirt_\AclassicalVal^{(j)})_{\Aprime''_j} &= \sum_{\blockVar=0}^\tagCutoff \prob_{\blockVarReg_j|\Aclassical_j}(\blockVar|\AclassicalVal) \ketbra{\AstateVirt_{\AclassicalVal,\blockVar}^{(j)}}{\AstateVirt_{\AclassicalVal,\blockVar}^{(j)}}_{\Aprime''_j} + \left(1 - \sum_{\blockVar=0}^\tagCutoff \prob_{\blockVarReg_j|\Aclassical_j}(\blockVar|\AclassicalVal)\right)\ketbra{\AclassicalVal}{\AclassicalVal}_{\Aprime''_j}\,, \\
    \ket{\AstateVirt_{\AclassicalVal,\blockVar}^{(j)}}_{\Aprime''_j} &= \sqrt{1-\srcImpBnd[\AclassicalVal,\blockVar,j]}\ket{\tilde \Astate_{\AclassicalVal,\blockVar}^{(j)}}_{\Aprime''_j} +\sqrt{\srcImpBnd[\AclassicalVal,\blockVar,j]} \ket{\tilde \Astate_{\AclassicalVal,\blockVar}^{(j)\perp}}_{\Aprime''_j}\,, \label{eq:definition_tau_x_m}
\end{align}
where $\left\{\ket{\tilde \Astate_{\AclassicalVal,\blockVar}^{(j)}}_{\Aprime''_j}\right\}_{j, \AclassicalVal\in \AclassicalAlph, \blockVar\leq\tagCutoff}$ is defined in terms of any set of purifications of the target states $\big\{\tilde \Astate_{\AclassicalVal,\blockVar}^{(j)}\big\}_{j, \AclassicalVal\in \AclassicalAlph, \blockVar\leq\tagCutoff}$ from \cref{as:imperfectly_characterized_source}, while the states $\ket{\tilde \Astate_{\AclassicalVal,\blockVar}^{(j)\perp}}_{\Aprime''_j}$ are unknown and $\bra{\tilde \Astate_{\AclassicalVal,\blockVar}^{(j)}}\ket{\tilde \Astate_{\AclassicalVal,\blockVar}^{(j)\perp}} = 0$ for all $\AclassicalVal\in\AclassicalAlph$, $\blockVar\leq\tagCutoff$, and all rounds $j$. We stress that we do not know the states $(\AstateVirt_\AclassicalVal^{(j)})_{\Aprime''_j}$ perfectly. Later, we express the marginal constraint after the source replacement with $(\AstateVirt_\AclassicalVal^{(j)})_{\Aprime''_j}$ as an optimization over all the possible Gram matrices describing the overlaps with the unknown states \cite{pereira_optimal_2025}. This replacement captures the worst-case scenario. In particular, they may occupy a larger Hilbert space and be mutually orthogonal across different values of $\AclassicalVal$ and $\blockVar$, even if the original states did not.

\textbf{Flag-state squasher:} We proceed analogously on the detector side. Using the block-diagonality of Bob's POVMs (\cref{as:block_diagonal_detector}) we apply the flag-state squasher (\cref{lem:flag_state_squasher}), to reduce the analysis on the detector side to that of the blocks from $0$ to $\fssCutoff$. The entire $\blockVar>\fssCutoff$ block in \cref{eq:assumption_detector_POVM_structure}, irrespective of its dimension, is replaced by classical flags. Hence, the resulting POVMs are finite-dimensional. Similarly to the tagging source map, and as discussed in \cref{sec:squashing_maps}, the cost of this replacement is given by the weight in the block $\blockVar> \fssCutoff$, which is converted to classical flags. By increasing the cutoff, this step can in principle be made arbitrarily tight. We use an outcome $W^{\Bmeas_j}$ satisfying $W^{\Bmeas_j} \geq \lambda_\mathrm{min} \proj_{>\fssCutoff}^{\Bmeas_j}$ with $\lambda_\mathrm{min}> 0$ for the weight estimation in the block $\blockVar> \fssCutoff$ (see \cref{sec:squashing_maps}), where $\proj_{>\fssCutoff}^{\Bmeas_j}$ is a projector onto the space $\blockVar>\fssCutoff$.\footnote{As an example, the multi-click outcome can be used when considering a passive detection setup, for which the corresponding $\lambda_\mathrm{min}$ can be computed for detectors with imperfect efficiency and dark counts \cite{wang2025phaseerrorestimationpassive,PhysRevResearch.6.043223}.}

\textbf{Noise channel:} Using the imperfect characterization assumption (\cref{as:imperfectly_characterized_detector}) we can apply the noise channel construction from \cref{lem:noise_channel}. The construction of the two squashers combined is given by \cref{lem:combination_FSS_noise_channel}. Therefore, the resulting POVMs are given by the elements
\begin{equation}
\label{eq:def_new_virtual_povm_elements}
    \BpovmelTarg[j] \coloneqq \left(\oplus_{\blockVar = 0}^\fssCutoff \BpovmelBlockTilde{j}{\blockVar} \right) \oplus \ketbra{\BclassicalVal}{\BclassicalVal}_\flagSpaceReg\,,
\end{equation}
with orthonormal basis $\{\ket{\BclassicalVal}\}_{\BclassicalVal\in\BclassicalAlph}$ on the flag space $\hilbert_\flagSpaceReg$ and the known set of POVM elements $\{\BpovmelBlockTilde{j}{\blockVar}\}_{j, \BclassicalVal\in \BclassicalAlph, \blockVar\leq\fssCutoff}$ is given by \cref{as:imperfectly_characterized_detector}. Additionally, following \cref{lem:combination_FSS_noise_channel}, after combining the flag-state squasher and noise channel, Eve's restricted attack sets can be constructed as in \cref{eq:restricted_attack_set_fss_noise}. We denote $\big\{\setAttackChVirt_j\big\}_j$ Eve's restricted attack sets for the specific construction of flag-state squasher and noise channel above.

Overall, we may therefore restrict our attention to the protocol with the new states and POVMs obtained by applying the source and squashing maps described above (see \cref{fig:sequence_source_squashing_maps}) to prove security of the original protocol, as formalized in \cref{lem:security_virtual_after_reductions}.

\begin{lemma}
\label{lem:security_virtual_after_reductions}
    Consider the protocol $\big\{\big\{\Astate^{(j)}_{\Aclassical_j\Aprime_j}, \big\{\Bpovmel[j]\big\}_{\BclassicalVal\in\BclassicalAlph}, \annKeyMap[j]\big\}_{j=1}^\totRounds, \ppMap\big\}$ where \cref{as:block_diagonal_source,as:imperfectly_characterized_source,as:probability_bounds,as:block_diagonal_detector,as:imperfectly_characterized_detector} are satisfied by the inputs. Let $\big\{\big\{\AstateVirt_{\Aclassical_j\Aprime''_j}^{(j)}, \big\{\BpovmelTarg[j]\big\}_{\BclassicalVal\in\BclassicalAlph}, \annKeyMap[j]\big\}_{j=1}^\totRounds, \ppMap\big\}$ be a new protocol with (known) POVMs given by \cref{eq:def_new_virtual_povm_elements} and imperfectly characterized signal states satisfying \cref{eq:def_new_virtual_states_alice}. Then, $\epsSecu$-security against $\{\setAttackChVirt_j\}_j$ of the new protocol for all possible signal states satisfying \cref{eq:def_new_virtual_states_alice} implies $\epsSecu$-security against $\{\cptp(\Ereg_{j-1}, \Bmeas_j \Ereg_j)\}_j$ of the original protocol, where $\setAttackChVirt_j$ is given by \cref{eq:restricted_attack_set_fss_noise}.
\end{lemma}
\begin{proof}
        The states $(\AstateVirt_\AclassicalVal^{(j)})_{\Aprime''_j}$ given by \cref{eq:def_new_virtual_states_alice} are obtained by applying the tagging source map (\cref{lem:tagging_source_map}) and the partial characterization source map (\cref{lem:source_map_imperfect_source}) as discussed above. Then, there exist source maps $\srcMap_j$ such that $\srcMap_j\big[(\AstateVirt_\AclassicalVal^{(j)})_{\Aprime''_j}\big] = (\Astate_\AclassicalVal^{(j)})_{\Aprime_j}$ for all $\AclassicalVal\in\AclassicalAlph$ and all rounds $j$. Additionally, the POVMs $\big\{\BpovmelTarg[j]\big\}_{\BclassicalVal\in\BclassicalAlph}$ given by \cref{eq:def_new_virtual_povm_elements} are obtained by applying the flag-state squasher (\cref{lem:flag_state_squasher}) and noise channel (\cref{lem:noise_channel}) as discussed above. Then, there exist squashing maps $\Lambda_j$ such that $\squashMap_j^\dagger \big[\BpovmelTarg[j]\big] = \Bpovmel[j]$ for all $\BclassicalVal\in\BclassicalAlph$ and all rounds $j$. Given the existence of the source map and squashing map, the statement directly follows from \cref{lem:epsilon_security_source_maps,lem:epsilon_security_squashing}. 
\end{proof}

\subsection{Construction of the set of states output by the single-round attack channels}
\label{sec:construction_set_output_states}

\noindent In the following, we consider the equivalent \nameref{prot:entanglement_qkd_protocol} corresponding to the protocol $\big\{\big\{\AstateVirt_{\Aclassical_j\Aprime''_j}^{(j)}, \big\{\BpovmelTarg[j]\big\}_{\BclassicalVal\in\BclassicalAlph}, \annKeyMap[j]\big\}_{j=1}^\totRounds, \ppMap\big\}$ from \cref{lem:security_virtual_after_reductions}, obtained after applying the source map and squashing map reductions from \cref{sec:problem_reductions}. Since both the signal states and POVMs are now finite-dimensional, the security analysis from \cite[Theorem~8.3]{tupkary_rigorous_2026} applies directly, and so does the discussion in \cref{sec:introduction_meat}. 

The POVM elements $\BpovmelTarg[j]$ are perfectly characterized and given by \cref{eq:def_new_virtual_povm_elements}. Therefore, the corresponding protocol map from the \nameref{prot:entanglement_qkd_protocol} is fixed, which we denote $\gMapVirt[j]\in\cptp(\Amarg_j\BmeasSquash_j, \secretReg_j \varDecReg_j \widetilde C_j)$. In contrast, the signal states in $\Aprime''_j$ are imperfectly characterized, therefore the corresponding source-replaced states $\AstateVirt_{\Amarg_j\Aprime''_j}^{(j)}$ (\cref{lem:source_replacement_scheme}) are also imperfectly characterized, where $\Amarg_j = \Ameas_j \Ashield_j$. We can nevertheless apply the MEAT for any specific choice of signal states, which results in an optimization problem (see \cref{eq:opti_original_meat})
\begin{align}
    \tilde{\kappa}\big(\ftradeoff{1}{j-1},\AstateVirt^{(j)}_{\Amarg_j},\gMapVirt[j], \marginalSetVirt_j\big) &= \inf_{\genDensity_{\Amarg_j\BmeasSquash_j}\in\stateSet_j(\marginalSetVirt_j)}\fRenyiUpEnt{1}{j-1}(\secretReg_j|\eveCopyReg_j\varDecReg_j\evePurReg)_{\gMapVirt[j][\purFunc(\genDensity_{\Amarg_j\BmeasSquash_j})]}\label{eq:opti_problem_virtual} \\
    \stateSet_j(\marginalSetVirt_j)&=\left\{\marginalMap_j[\genDensity_{\Amarg_j\Aprime''_j}]:\genDensity_{\Amarg_j}=\AstateVirt^{(j)}_{\Amarg_j},\marginalMap_j\in\marginalSetVirt_j\right\}\,, \label{eq:state_set_virt}
\end{align}
where $\purFunc$ is a purifying function from $\Amarg_j\BmeasSquash_j$ onto $\evePurReg$ and $\marginalSetVirt_j$ denotes the marginal of the restricted set of attack channels $\setAttackChVirt_j$ given by
\begin{equation}
\label{eq:marginal_definition_fss_noise}
    \marginalSetVirt_j \coloneqq \left\{\marginalMap_j\in\cptp(\Aprime''_j,\BmeasSquash_j): \marginalMap_j^\dagger\left[\widetilde{W}^{\Bmeas_j'} \right] \geq  \lambdaMin  \left( \marginalMap_j^\dagger\left[\identity_{\BmeasSquash_j} \right] - \sum_{\blockVar=0}^{\fssCutoff} \frac{1}{1 - \dtImpBnd[\blockVar, j]} \marginalMap_j^\dagger\left[\proj_\blockVar^{\BmeasSquash_j} \right]\right)\right\}\,.
\end{equation}
That \cref{eq:marginal_definition_fss_noise} defines a valid
marginal of $\setAttackChVirt_j$ follows from \cite[Lemma~9.4]{tupkary_rigorous_2026}. The
statement there is given for the flag-state squasher alone, but its proof applies verbatim when
replacing $F_o^{(Q)}$ by $\widetilde{W}^{\BmeasSquash_j}$ and $\Pi^{\mathrm{Sq}}$ by the operator
appearing on the r.h.s. of the inequality above.

Recall that we wish to lower bound \cref{eq:opti_problem_virtual} and that $\AstateVirt^{(j)}_{\Amarg_j}$ is not known exactly. Let $\setAliceMarginalsVirt[j]$ be a known set of possible marginals such that $\AstateVirt^{(j)}_{\Amarg_j}\in \setAliceMarginalsVirt[j]$. Since we do not know $\AstateVirt^{(j)}_{\Amarg_j}$, we will obtain a lower bound by additionally optimizing over all possible marginals
\begin{equation}
    \inf_{\tilde{\AstateVirt}^{(j)}_{\Amarg_j} \in \setAliceMarginalsVirt[j]}\tilde{\kappa}\big(\ftradeoff{1}{j-1},\tilde{\AstateVirt}^{(j)}_{\Amarg_j},\gMapVirt[j], \marginalSetVirt_j\big) \leq \tilde{\kappa}\big(\ftradeoff{1}{j-1},\AstateVirt^{(j)}_{\Amarg_j},\gMapVirt[j], \marginalSetVirt_j\big)\,. \label{eq:lower_bound_with_marginal_set}
\end{equation}
For the purpose of the security statement, we do not need to explicitly construct $\setAliceMarginalsVirt[j]$, although we require it to be convex to preserve convexity of the resulting optimization problem. We defer the explicit construction of $\setAliceMarginalsVirt[j]$ to \cref{ap:construction_marginal_constraint}, as it is required to compute the key rates in \cref{ap:numerics}.\footnote{Informally, the set $\setAliceMarginalsVirt[j]$ of possible marginals is constructed from \cref{eq:def_new_virtual_states_alice} together with \cref{as:block_diagonal_source,as:imperfectly_characterized_source,as:probability_bounds}.}

We can replace the optimization over the marginal attack channels $\marginalMap_j\in\marginalSetVirt_j$ in \cref{eq:state_set_virt} by an optimization directly over the corresponding output states $\genDensity_{\Amarg_j\BmeasSquash_j}$. The necessary conditions that the output states must satisfy are given by the explicit form of $\marginalSetVirt_j$ in \cref{eq:marginal_definition_fss_noise}, together with the fact that $\marginalMap_j$ acts only on $\Aprime''_j$. This directly yields a lower bound on $\tilde{\kappa}\big(\ftradeoff{1}{j-1},\AstateVirt^{(j)}_{\Amarg_j},\gMapVirt[j], \marginalSetVirt_j\big)$ without optimization over the marginal attack channels, formalized in \cref{lem:lower_bound_meat_output_states} below. Intuitively, this can be thought of as rewriting the constraint from the flag-state squasher and noise channel directly into $\rho_{A_j B'_j}$, rather than expressing it as a part of Eve's attack channel.

\begin{lemma}
\label{lem:lower_bound_meat_output_states}
    Consider the \nameref{prot:entanglement_qkd_protocol} corresponding to the protocol $\big\{\big\{\AstateVirt_{\Aclassical_j\Aprime''_j}^{(j)}, \big\{\BpovmelTarg[j]\big\}_{\BclassicalVal\in\BclassicalAlph},$ $ \annKeyMap[j]\big\}_{j=1}^\totRounds, \ppMap\big\}$ from \cref{lem:security_virtual_after_reductions} with (known) POVMs given by \cref{eq:def_new_virtual_povm_elements} and imperfectly characterized signal states satisfying \cref{eq:def_new_virtual_states_alice}. Let $\AstateVirt_{\Amarg_j\Aprime''_j}^{(j)}$ denote the corresponding source-replaced state in the $j$th round and $\setAliceMarginalsVirt[j]$ be a set of possible marginal states, such that $\AstateVirt^{(j)}_{\Amarg_j} \in \setAliceMarginalsVirt[j]$. Define 
     \begin{align}
\widehat{\stateSet}_j(\setAliceMarginalsVirt[j]) \coloneqq& \left\{\genDensity_{\Amarg_j\BmeasSquash_j} \in \setDensity_=(\Amarg_j\BmeasSquash_j) : \Tr_{\BmeasSquash_j}\left[\genDensity_{\Amarg_j\BmeasSquash_j}\right] \in \setAliceMarginalsVirt[j]\,,\right. \label{eq:restricted_set_output_states} \\
    &\left. \Tr\left[\genDensity_{\Amarg_j\BmeasSquash_j}\widetilde{W}^{\Bmeas_j'}\right] \geq \lambdaMin\left(1 - \sum_{\blockVar =0}^{\fssCutoff} \frac{1}{1 - \dtImpBnd[\blockVar, j]} \Tr\left[\genDensity_{\Amarg_j\BmeasSquash_j}\Pi_\blockVar^{\BmeasSquash_j}\right]\right)\right\}\,, \nonumber
\end{align}
    for all rounds $j$, where $\gMapVirt[j]$ is the protocol map corresponding to the known POVM $\big\{\BpovmelTarg[j]\big\}_{\BclassicalVal\in\BclassicalAlph}$ and $\marginalSetVirt_j$ is the marginal (\cref{def:marginal}) of the restricted set of attack channels $\setAttackChVirt_j$, given by \cref{eq:marginal_definition_fss_noise}. Then,
    \begin{equation}
        \inf_{\genDensity_{\Amarg_j\BmeasSquash_j}\in\widehat{\stateSet}_j(\setAliceMarginalsVirt[j])}\fRenyiUpEnt{1}{j-1}(\secretReg_j|\eveCopyReg_j\varDecReg_j\evePurReg)_{\gMapVirt[j]\left[\purFunc(\genDensity_{\Amarg_j\BmeasSquash_j})\right]} \leq \tilde{\kappa}\left(\ftradeoff{1}{j-1},\AstateVirt^{(j)}_{\Amarg_j},\gMapVirt[j], \marginalSetVirt_j\right)\,.
        \label{eq:lower_bound_kappa_virt}
    \end{equation}
   
\end{lemma}

\begin{proof}
    \cref{eq:lower_bound_with_marginal_set} first provides a lower bound by optimizing over all possible marginal states in $\setAliceMarginalsVirt[j]$. It remains to relax the corresponding optimization over the marginal attack channels $\marginalMap_j\in\marginalSetVirt_j$ to an optimization over output states. Since $\marginalMap_j$ only acts on $\Aprime''_j$, every corresponding output state $\genDensity_{\Amarg_j\BmeasSquash_j}$ satisfies $\Tr_{\BmeasSquash_j}\left[\genDensity_{\Amarg_j\BmeasSquash_j}\right]\in\setAliceMarginalsVirt[j]$. Moreover, the explicit form of $\marginalSetVirt_j$ in \cref{eq:marginal_definition_fss_noise} implies that every such output state satisfies
    \begin{equation}
        \Tr\left[\genDensity_{\Amarg_j\BmeasSquash_j}\widetilde{W}^{\Bmeas_j'}\right] \geq \lambdaMin\left(1 - \sum_{\blockVar =0}^{\fssCutoff} \frac{1}{1 - \dtImpBnd[\blockVar, j]} \Tr\left[\genDensity_{\Amarg_j\BmeasSquash_j}\Pi_\blockVar^{\BmeasSquash_j}\right]\right)\,.
    \end{equation}
    Hence, all output states appearing in the optimization on the l.h.s. of \cref{eq:lower_bound_with_marginal_set} are contained in $\widehat{\stateSet}_j(\setAliceMarginalsVirt[j])$. Enlarging the set over which the infimum is taken can only decrease its value, therefore the l.h.s. of \cref{eq:lower_bound_kappa_virt} is smaller or equal to the l.h.s. of \cref{eq:lower_bound_with_marginal_set}, from which \cref{eq:lower_bound_kappa_virt} follows. 
\end{proof}

\begin{remark}
\label{rem:we_do_output_state}
    Since Alice’s signal states are not perfectly known, even in the new protocol, the optimization in \cref{eq:lower_bound_kappa_virt} cannot directly be formulated as an optimization over just the attack channels, as done for example in Ref.~\cite{kamin_renyi_2025}. Indeed, such a formulation requires an explicit description of the source-replaced state which is unknown under \cref{as:imperfectly_characterized_source}. We therefore instead optimize over the corresponding set of compatible output states induced by $\setAliceMarginalsVirt[j]$. Under suitable assumptions, both formulations are equivalent, as discussed in Ref.~\cite[Sec.~8]{tupkary_rigorous_2026}.
\end{remark}

\subsection{Security statement}
\label{sec:security_statement}

\noindent Combining the reduction of the \nameref{prot:generic_qkd_protocol} with imperfect, potentially infinite-dimensional and imperfectly characterized devices from \cref{sec:problem_reductions} with the construction of the set of output states from \cref{sec:construction_set_output_states} yields the following security statement for the \nameref{prot:generic_qkd_protocol} with device uncertainty.

\begin{theorem}[Security with imperfect and imperfectly characterized devices]
\label{th:security_imperfect_devices}
    Consider the \nameref{prot:generic_qkd_protocol} given by $\big\{\big\{\Astate^{(j)}_{\Aclassical_j\Aprime_j}, \big\{\Bpovmel[j]\big\}_{\BclassicalVal\in\BclassicalAlph}, \annKeyMap[j]\big\}_{j=1}^\totRounds, \ppMap\big\}$ where \cref{as:block_diagonal_source,as:imperfectly_characterized_source,as:probability_bounds,as:block_diagonal_detector,as:imperfectly_characterized_detector} are satisfied by the inputs. Let $W^{\Bmeas_j}$ be an outcome and $\lambdaMin> 0$ satisfying $W^{\Bmeas_j} \geq \lambdaMin \proj_{>\fssCutoff}^{\Bmeas_j}$, where $\proj_{>\fssCutoff}^{\Bmeas_j}$ is the projector onto the subspace with $\blockVar>\fssCutoff$ photons. For each round $j$ and every value $\varDecVal_1^{j-1}$, let $\ftradeoff{1}{j-1}$ be a tradeoff function on the register $\varDecReg_j$. Then the protocol is $(\epsSecr + \epsCor)$-secure for all $\alpha\in(1,2)$ if the final key length $\keyl$ satisfies
    \begin{equation}
        \keyl(\varDecVal^\totRounds_1)= \max \left\{0, \left\lfloor \fFull(\varDecVal^\totRounds_1) - \ECcost(\varDecVal^\totRounds_1) - \left\lceil \log\left(\frac{1}{\epsCor}\right)\right\rceil - \frac{\alpha}{\alpha-1} \log\left(\frac{1}{\epsSecr}\right) + 2 \right\rfloor \right\}\,,
    \end{equation}
    where
    \begin{equation}
        \fFull(\varDecVal^\totRounds_1) \coloneqq \sum_{j=1}^\totRounds\left(\ftradeoff{1}{j-1}(\varDecVal_j) + \kappaQKDVirt_j\left(\ftradeoff{1}{j-1},\setAliceMarginalsVirt[j],\gMapVirt[j]\right)\right)\,,
    \end{equation}
    and
    \begin{align}
    \label{eq:optimization_imperfect_devices}
        \kappaQKDVirt_j\left(\ftradeoff{1}{j-1},\setAliceMarginalsVirt[j],\gMapVirt[j]\right) &\leq \inf_{\genDensity_{\Amarg_j\BmeasSquash_j}\in\widehat{\stateSet}_j(\setAliceMarginalsVirt[j])}\fRenyiUpEnt{1}{j-1}(\secretReg_j|\eveCopyReg_j\varDecReg_j\evePurReg)_{\gMapVirt[j]\left[\purFunc(\genDensity_{\Amarg_j\BmeasSquash_j})\right]} \\
        \widehat{\stateSet}_j(\setAliceMarginalsVirt[j]) &= \left\{\genDensity_{\Amarg_j\BmeasSquash_j} \in \setDensity_=(\Amarg_j\BmeasSquash_j) : \Tr_{\BmeasSquash_j}\left[\genDensity_{\Amarg_j\BmeasSquash_j}\right] \in \setAliceMarginalsVirt[j]\,,\right. \label{eq:set_output_states_security_statement} \\
        &\left. \Tr\left[\genDensity_{\Amarg_j\BmeasSquash_j}\widetilde{W}^{\Bmeas_j'}\right] \geq \lambdaMin\left(1 - \sum_{\blockVar =0}^{\fssCutoff} \frac{1}{1 - \dtImpBnd[\blockVar, j]} \Tr\left[\genDensity_{\Amarg_j\BmeasSquash_j}\Pi_\blockVar^{\BmeasSquash_j}\right]\right)\right\}\,. \nonumber
    \end{align}
    Here, $\AstateVirt_{\Amarg_j\Aprime''_j}^{(j)}$ denotes the source-replaced state in the $j$th round, corresponding to $\AstateVirt_{\Aclassical_j\Aprime''_j}^{(j)}$, cf. \cref{eq:def_new_virtual_states_alice}, with marginal constraint $\Tr_{\Aprime''_j}[\AstateVirt_{\Amarg_j\Aprime''_j}^{(j)}] = \AstateVirt_{\Amarg_j}^{(j)}$. $\purFunc$ is a purifying function of $\Amarg_j\BmeasSquash_j$ onto $\evePurReg$. The protocol maps $\{\gMapVirt[j]\}_j$ from the \nameref{prot:entanglement_qkd_protocol} are given in terms of the squashed POVMs $\big\{\BpovmelTarg[j]\big\}_{j,\BclassicalVal\in\BclassicalAlph}$ from \cref{eq:def_new_virtual_povm_elements} (see \cref{def:qkd_protocol_map}). Furthermore, the set of marginals $\setAliceMarginalsVirt[j]$ is defined such that $\AstateVirt_{\Amarg_j}^{(j)} \in \setAliceMarginalsVirt[j]$ and explicitly constructed in \cref{ap:construction_marginal_constraint}. Finally, $\squashMap_j^\dagger \big[\widetilde{W}^{\Bmeas_j'}\big] = W^{\Bmeas_j}$, where $\squashMap_j$ is given by \cref{eq:explicit_noise_channel_construction}.
\end{theorem}

\begin{proof}
    First, we apply the sequence of source maps and squashing maps from \cref{lem:security_virtual_after_reductions} to reduce security of the original protocol to security of a new protocol $\big\{\big\{\AstateVirt_{\Aclassical_j\Aprime''_j}^{(j)}, \big\{\BpovmelTarg[j]\big\}_{\BclassicalVal\in\BclassicalAlph}, \annKeyMap[j]\big\}_{j=1}^\totRounds, \ppMap\big\}$ against attacks $\{\setAttackChVirt_j\}_j$, where $\setAttackChVirt_j$ is given by \cref{eq:restricted_attack_set_fss_noise}, the POVM elements are given by \cref{eq:def_new_virtual_povm_elements}, and the signal states are given by \cref{eq:def_new_virtual_states_alice}.

    Since the new protocol is finite-dimensional, we apply \cite[Theorem~8.3]{tupkary_rigorous_2026}, where, instead of choosing $\marginalSetVirt_j = \cptp(\Aprime''_j,\BmeasSquash_j)$, we choose $\marginalSetVirt_j$ to be a marginal of the restricted attack set $\setAttackChVirt_j$. The discussion from \cref{sec:security_using_MEAT} then applies. Finally, we apply \cref{lem:lower_bound_meat_output_states} which provides a computable lower bound for $\kappa$, from which the statement follows.
\end{proof}

\begin{remark}
    The statement of \cref{th:security_imperfect_devices} holds for arbitrary choices of the tradeoff functions $\ftradeoff{1}{j-1}$ in each round $j$. In practice, however, obtaining a positive key rate typically requires these functions to be chosen carefully. An optimal choice can be determined using \cite[Lemma~4.12]{arqand_marginal-constrained_2025} (see also \cite[Sec.~VI.B]{kamin_renyi_2025} and \cite[Sec.~10.2]{tupkary_rigorous_2026}).
\end{remark}

The formulation in \cref{th:security_imperfect_devices} is intentionally generic and abstract, providing a unified description of a broad class of device imperfections, all of which enter the analysis through the constraints defining $\widehat{\stateSet}_j(\setAliceMarginalsVirt[j])$. The resulting optimization problem is made numerically tractable in \cref{ap:numerics}, where we also explicitly construct the constraints defining $\widehat{\stateSet}_j(\setAliceMarginalsVirt[j])$ and derive the corresponding convex optimization problem formalized in \cref{th:convex_opti_numerics_imperfect}. Combining the tools developed here to account for simultaneous source and detector imperfections and device uncertainty requires several non-trivial steps, particularly to preserve convexity of the final optimization, which is discussed in detail in \cref{ap:numerics}.

In the next section, we illustrate the modularity of the framework by considering various real-world device imperfections.
\section{Device imperfections}
\label{sec:device_imperfections}

\noindent We first provide a brief overview for how to use the framework in \cref{sec:how_to_use_framework}. Then, in \cref{sec:example_pol_encoding}, we compute key rates for a polarization-encoded decoy-state BB84 implementation combining various device imperfections for practical device parameters. We also compare our key rates with those obtained using the analyses of Refs.~\cite{sixto2026finitekeysecurityanalysisdecoystate,navarrete2026numericalsecurityanalysispractical} for device imperfections covered by all three approaches.

\subsection{How to use the framework}
\label{sec:how_to_use_framework}

\begin{table}[t]
\centering
\renewcommand{\arraystretch}{1.4}
\setcellgapes{4pt}
\makegapedcells
\setlength{\tabcolsep}{12pt} 

\begin{tabular}{c|c|c}
\hline
\textbf{Device imperfection} & \textbf{Affects} & \textbf{Parameter/constraint} \\
\hline \hline

\multicolumn{3}{c}{\textbf{Source imperfections}} \\
\hline

Imperfect state encoding & \cref{as:imperfectly_characterized_source} & $\srcImpBndEnc[\blockVar]$ \\
\hline

Trojan-horse attack & \cref{as:imperfectly_characterized_source} & $\srcImpBndTHA[\blockVar]$ \\
\hline

Intensity fluctuations & \cref{as:probability_bounds} & $\prob_{\blockVarReg_j|\Aclassical_j}^{\mathrm{L}}(\blockVar|\AclassicalVal), \prob_{\blockVarReg_j|\Aclassical_j}^{\mathrm{U}}(\blockVar|\AclassicalVal)$ \\
\hline

Imperfect photon-number distribution & \cref{as:probability_bounds}   & $\prob_{\blockVarReg_j|\Aclassical_j}^{\mathrm{L}}(\blockVar|\AclassicalVal), \prob_{\blockVarReg_j|\Aclassical_j}^{\mathrm{U}}(\blockVar|\AclassicalVal)$ \\
\hline

Imperfectly block-diagonal source & \cref{as:block_diagonal_source} & Approx. block diagonalization \\
\hline\hline

\multicolumn{3}{c}{\textbf{Detector imperfections}} \\
\hline

Dark-count rate mismatch & \cref{as:imperfectly_characterized_detector} & $\dtImpBndDCR[\blockVar]$ \\
\hline

Detection-efficiency mismatch & \cref{as:imperfectly_characterized_detector} & $\dtImpBndEff[\blockVar, \eta_*]$ \\
\hline

Imperfect passive basis choice & \cref{as:imperfectly_characterized_detector} & $q_{\blockVar}^{s}$ \\
\hline\hline

\end{tabular}

\caption{Examples of device imperfections to illustrate the framework in \cref{sec:device_imperfections}. The models for the device imperfections are detailed in \cref{ap:models_device_imperfections}.}
\label{tab:device_imperfections_list}

\end{table}

\noindent The aim of the following sections is to provide an overview for incorporating device imperfections into the security framework developed in \cref{sec:security_imperfect_devices}. For each physical imperfection considered, we require two ingredients: a set of \textit{target} states or POVM elements, and a bound on the \textit{deviation} of the real devices from these targets. At a high level, incorporating a device imperfection consists of the following steps: 

\begin{enumerate}
    \item identify the assumptions in \cref{th:security_imperfect_devices} (\cref{sec:model_assumptions}) through which the device imperfection enters the analysis,
    \item choose suitable known target states or POVM elements,
    \item derive bounds in the form of \cref{as:imperfectly_characterized_source,as:probability_bounds,as:imperfectly_characterized_detector}, quantifying the deviation between the physical devices and the targets (this step may require an experimental characterization),
    \item solve the resulting optimization problem for $\kappaQKDVirt_j\left(\ftradeoff{1}{j-1},\setAliceMarginalsVirt[j],\gMapVirt[j]\right)$ from \cref{th:security_imperfect_devices} with the chosen target states, target POVMs and the determined bounds. This is formulated as a numerically tractable convex optimization problem in \cref{th:convex_opti_numerics_imperfect}.
\end{enumerate}

Importantly, the target states and POVM elements can be chosen arbitrarily and do not need to correspond to ``perfect'' devices.\footnote{In the sense that they need to be known exactly, but do not need to involve exact $H,V,A,D$ state preparation, or unit efficiency of detectors, etc.} The only requirement is that the targets themselves are known. The analysis therefore naturally distinguishes between two different effects: imperfections of the physical devices, and uncertainty in the characterization of these devices.

Both contribute to the final key rate, but they enter the analysis in different ways, as illustrated in \cref{rem:two_different_imperfection_types}. In \cref{sec:imperfect_block_diagonal}, we relax the exact block-diagonal structure assumed in \cref{as:block_diagonal_source} and show how approximate block-diagonal structure, which may arise for example from imperfect phase randomization, is incorporated into the analysis.

\begin{remark}
\label{rem:two_different_imperfection_types}
    As an example, one may choose perfectly encoded states as the targets in \cref{th:security_imperfect_devices}. In this case, the target model itself is ideal, but the real implementation may deviate significantly from it. The corresponding penalty is then captured entirely through the source fidelity bounds, i.e. through the Gram matrix optimization.

    On the other hand, one may choose target states that are themselves imperfectly encoded, but perfectly characterized. In this case, the source fidelity bounds vanish, since the real states coincide with the targets, and the penalty instead comes from the fact that the target states appearing in \cref{th:security_imperfect_devices} are already imperfect.
    
    Given that both the Gram matrix optimization and the noise channel construction consider worst-case scenarios in terms of device uncertainty, it is in general best to choose target states and POVMs that minimize the source fidelity bounds and detector POVM deviations (even if they are chosen to be imperfect).
\end{remark}

\subsection{Example: polarization-encoded decoy-state BB84}
\label{sec:example_pol_encoding}

\subsubsection{Passive decoy-state BB84}

\begin{figure}
\centering
\begin{tikzpicture}[
    >=Latex,
    every node/.style={font=\small},
    optic/.style={
        draw,
        semithick,
        minimum width=1.15cm,
        minimum height=0.72cm,
        align=center,
        fill=black!8
    },
    det/.style={
        draw,
        semithick,
        shape=semicircle,
        shape border rotate=270,
        minimum width=0.58cm,
        minimum height=0.55cm,
        inner sep=1pt,
        align=center,
        fill=black!8
    }
]

\coordinate (laser) at (0,0);
\coordinate (polenc) at (2.0,0);
\coordinate (im) at (4.5,0);
\coordinate (iso) at (6.5,0);
\coordinate (chanstart) at (8.0,0);
\coordinate (bs1) at (10.0,0);
\coordinate (pbsZ) at (10.0,1.9);
\coordinate (rot) at (12.0,0);
\coordinate (pbsX) at (14.0,0);

\node[optic] (L) at (laser) {Laser};
\node[optic] (P) at (polenc) {Pol. enc.};
\node[optic] (I) at (im) {IM};
\node[optic] (O) at (iso) {Iso.};

\node[anchor=north east] at (0.25,1.85) {\textbf{Alice}};

\coordinate (aset) at ($(P.north)+(0,0.5)$);
\draw[->, semithick] (aset) -- (P.north)
    node[pos=0, above] {$a \in \{H,V,D,A\}$};

\node[below=0.12cm of P] {$\delta_{\mathrm{SPF}}$};

\coordinate (mset) at ($(I.north)+(0,0.5)$);
\draw[->, semithick] (mset) -- (I.north)
node[pos=0, above] {$\mu \in \{\mu_{\mathrm S}, \mu_{\mathrm D}\}$};
\node[below=0.09cm of I] {$\srcImpBndInt$};

\node[below=0.15cm of O] {$\eta_\mathrm{ISO}$};

\draw[->, semithick] (L.east) -- (P.west);
\draw[->, semithick] (P.east) -- (I.west);
\draw[->, semithick] (I.east) -- (O.west);

\draw[->, semithick] (O.east) -- (bs1);

\begin{scope}[shift={($(O.east)!0.38!(bs1)$)}, scale=0.45]
    \draw[red, semithick, smooth, domain=-0.9:0.9, samples=80]
        plot (\x,{1.2*exp(-10*(\x+0.10)^2)});
    \draw[blue, semithick, smooth, domain=-0.9:0.9, samples=80]
        plot (\x,{1.2*exp(-10*(\x-0.10)^2)});
\end{scope}

\draw[semithick] ($(bs1)+(-0.18,-0.18)$) -- ($(bs1)+(0.18,0.18)$);
\node[below=0.15cm of bs1] {$s/(1-s)$};

\draw[semithick] (bs1) -- (pbsZ);
\draw[semithick] ($(pbsZ)+(-0.18,-0.18)$) -- ($(pbsZ)+(0.18,0.18)$);
\node[left=0.18cm of pbsZ] {PBS};

\draw[semithick] (pbsZ) -- ++(0,1.05) node[det, above] (detV) {$V$};
\node[right=0.12cm of detV] {$\eta_{\mathrm V}, \mathrm d_{\mathrm V}$};

\draw[semithick] (pbsZ) -- ++(1.35,0) node[det, right] (detH) {$H$};
\node[above=0.12cm of detH] {$\eta_{\mathrm H}, \mathrm d_{\mathrm H}$};

\node[optic] (R) at (rot) {Pol. rot.};

\draw[semithick] (bs1) -- (R.west);
\draw[semithick] (R.east) -- (pbsX);

\draw[semithick] ($(pbsX)+(-0.18,-0.18)$) -- ($(pbsX)+(0.18,0.18)$);
\node[below=0.12cm of pbsX] {PBS};

\draw[semithick] (pbsX) -- ++(0,1.05) node[det, above] (detA) {$A$};
\node[right=0.12cm of detA] {$\eta_{\mathrm A}, \mathrm d_{\mathrm A}$};

\draw[semithick] (pbsX) -- ++(1.35,0) node[det, right] (detD) {$D$};
\node[above=0.12cm of detD] {$\eta_{\mathrm D}, \mathrm d_{\mathrm D}$};

\node[anchor=north east] at (16.1,4.1) {\textbf{Bob}};

\draw[dashed, semithick, orange!70!black, rounded corners]
    (-1.0,-1.1) rectangle (7.5,2.0);

\draw[dashed, semithick, orange!70!black, rounded corners]
    (8.9,-1.1) rectangle (16.3,4.3);

\end{tikzpicture}
\caption{Passive polarization-encoded decoy-state BB84 setup with various device imperfections, including imperfect state encoding ($\delta_{\mathrm{SPF}}$) with angular uncertainty ($\Delta\theta$), intensity fluctuations ($\srcImpBndInt$), Trojan-horse attacks ($\eta_\mathrm{ISO}$), imperfect detection efficiencies ($\eta_{i}$) with uncertainty $\Delta_{\eta}$, dark counts ($\mathrm d_{i}$) and imperfect passive basis choice beamsplitter $(s)$ with uncertainty ($\Delta_s$), cf. \cref{tab:device_imperfections_list}.}
\label{fig:pol_enc_setup}
\end{figure}
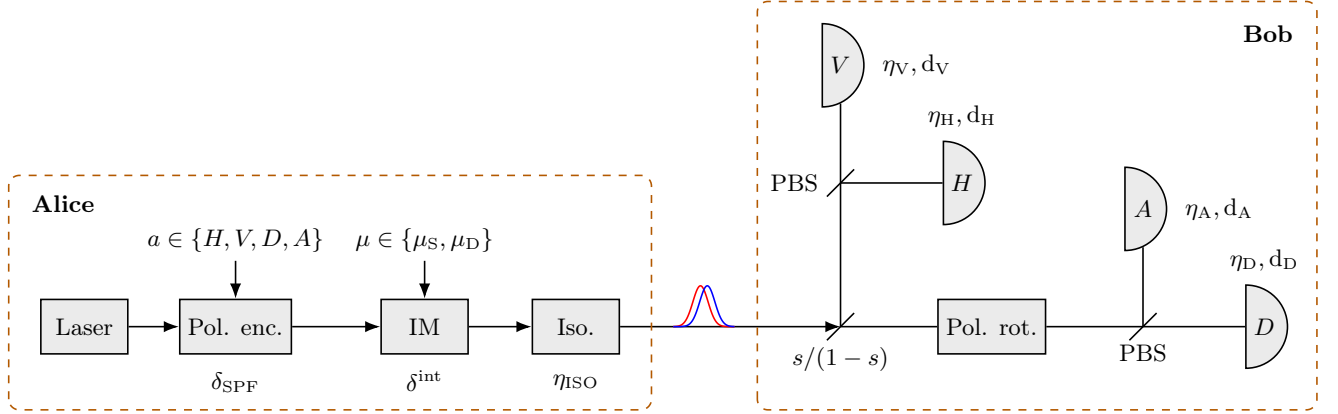

\noindent We now illustrate the framework for a polarization-encoded decoy-state BB84 setup shown in \cref{fig:pol_enc_setup}. A non-exhaustive list of practical source and detector imperfections, based on Ref.~\cite{BSI24}, is given in \cref{tab:device_imperfections_list}. We briefly describe the different imperfections in the main text, while the models used are detailed in \cref{ap:models_device_imperfections}. \cref{tab:device_imp_parameters} relates representative experimental parameters to the corresponding theoretical metrics. The remaining protocol and device parameters are listed in \cref{tab:protocol_parameters_example}.

To illustrate the combined impact of these imperfections for practical scenarios, we evaluate the key rate by considering combinations of device imperfection parameters, listed in \cref{tab:device_imp_parameters}, and present the results in \cref{fig:keyRate_example}.

\textbf{Alice} prepares phase-randomized weak coherent pulses with two intensity choices, the signal intensity $\mu_\mathrm{S}$ and decoy intensity $\mu_\mathrm{D}$. The intensity $\tilde{\mu}_i$ may fluctuate within $\tilde{\mu}_i \in [\mu_i(1-\srcImpBndInt),\mu_i(1+\srcImpBndInt)]$, where $i\in\{\mathrm{S},\mathrm{D}\}$\footnote{We emphasize that the probability distribution of the actual intensity within this interval need not be known. The interval may account for both physical intensity fluctuations and uncertainty in the characterization of the intensity.}. Whenever the decoy intensity is chosen, the round is used for testing, while the signal intensity is used for both key generation and testing (see $p_{\mathtt{test}|Z}$ in \cref{tab:protocol_parameters_example}). For illustration purposes, Alice's target polarization states are given by \cref{eq:state_preparation_flaw_model} and include state preparation flaws quantified by $\delta_{\mathrm{SPF}}$. The signal states are, in turn, imperfectly characterized relative to these targets, with the uncertainty in the polarization encoding quantified by the angular uncertainty $\Delta\theta$ (which is related to the source fidelity bound $\srcImpBndEnc$ via \cref{eq:relation_source_fidelity_delta_enc}). Thus, $\delta_{\mathrm{SPF}}$ describes a known imperfection included in the target states, whereas $\Delta\theta$ quantifies the uncertainty of the signal states around that model. Alice's source is not perfectly isolated, so Eve may inject light into Alice's setup and analyse the
back-reflections in a Trojan-horse attack. We quantify this
leakage through the corresponding source fidelity bound $\srcImpBndTHA$, defined such that the mean
photon number of the light back-reflected out of Alice's setup is at most $\srcImpBndTHA/2$ in every
round (which is directly related to Alice's attenuation $\eta_\mathrm{ISO}$, see \cref{eq:fidelity_bound_THA}). We assume that Alice does not leak information about her intensity choice. The total source fidelity bound is then given by \cref{eq:formula_combination_source_fidelity_bound} combining $\srcImpBndEnc$ and $\srcImpBndTHA$.

\textbf{Bob} has a passive detection setup with efficiencies characterized by $\tilde{\eta}_i \in [\eta_i(1-\Delta_{\eta}),\eta_i(1+\Delta_{\eta})]$, with known $\eta_i$ and characterization uncertainty $\Delta_{\eta}$. The detector dark-count probabilities are upper bounded by $\mathrm{d}_i^\mathrm{U}$. Bob performs the basis choice using a beamsplitter with splitting ratio $\tilde{s}\in[s(1-\Delta_s),s(1+\Delta_s)]$ with known $s$ (which we optimize over) and characterization uncertainty $\Delta_s$. We use the multi-click outcome for the subspace-weight estimation in \cref{th:security_imperfect_devices}, cf. \cref{ap:numerics}. The resulting detector POVM deviations $q_0$ and $q_1$ are given by \cref{eq:q_param_0_full,eq:q_param_1_full}. 

\textbf{Channel.} We model the channel by its transmittance and a misalignment angle $\theta_\mathrm{mis}$, which corresponds to a rotation of Bob's polarizing measurement relative to Alice's signal preparation.

\begin{table}[t]
\centering
\renewcommand{\arraystretch}{1.4}
\begin{tabular}{lllll}
\hline
\makecell[c]{\rule{0pt}{2.5ex}\textbf{Experimental} \\ \textbf{parameter}} & \makecell[c]{\rule{0pt}{2.5ex}\textbf{Value}} & \makecell[c]{\rule{0pt}{2.5ex}\textbf{Theoretical} \\ \textbf{parameter}}  & \makecell[c]{\rule{0pt}{2.5ex}\textbf{Value}} & \makecell[c]{\rule{0pt}{2.5ex}\textbf{Description}}\\ [0.2cm]

\hline
$\Delta\theta$ & $0.06\degree$, $0.5\degree,2\degree$ & $\srcImpBndEnc$ (\cref{eq:relation_source_fidelity_delta_enc}) & $\approx 10^{-6}, 10^{-4}, 10^{-3}$  & Uncertainty about Alice's polarization encoding \\

$\eta_\mathrm{ISO}$\footnote{When relating $\eta_\mathrm{ISO}$ to the corresponding source fidelity bound $\srcImpBndTHA$ via \cref{eq:fidelity_bound_THA}, we assume for illustration a laser-induced damage threshold of $50\,\mathrm{W}$, a repetition rate of $500\,\mathrm{MHz}$, and an operating wavelength of $1550\,\mathrm{nm}$.} & $200\,$dB & $\srcImpBndTHA$ (\cref{eq:fidelity_bound_THA}) & $\approx 10^{-8}$  & Two-way attenuation from Alice's isolators \\

$\srcImpBndInt$ & $10\,$\% & $\srcImpBndInt$ & $0.1$  & Alice's intensity fluctuations \\

$\Delta_{\eta}$ & $2.5\%$, $5\,$\%, $10\,$\% & $\dtImpBndEff[\blockVar = 1]$ (\cref{eq:q_param_efficiency}) & $\approx 0.04, 0.07, 0.15$  & Uncertainty about Bob's detector efficiencies \\

$\mathrm{d}_{i}^\mathrm{U}, i\in\{\mathrm{H,V,D,A}\}$ & $10^{-8}$ \cite{pittaluga_600-km_2021} & $\dtImpBndDCR[\blockVar=0,1]$ (\cref{eq:q_param_dark_counts}) & $\approx 4\cdot 10^{-8}$  & Upper bound on Bob's dark-count probability\\

$\Delta_{s}$ & $2.5\%$, $5\,$\%, $10\,$\% & $q_{\blockVar=1}^{s}$ (\cref{eq:q_param_splitting}) & - & Uncertainty about Bob's beamsplitting ratio \\
\hline
\end{tabular}
\caption{Example relation between device imperfection parameters and the corresponding theoretical metrics for the setup described in \cref{sec:example_pol_encoding} with protocol and device parameters given by \cref{tab:protocol_parameters_example}. The theoretical metrics are computed using the models described in \cref{ap:models_device_imperfections}.}
\label{tab:device_imp_parameters}
\end{table}

\begin{table}[t]
\centering
\renewcommand{\arraystretch}{1.4}
\begin{tabular}{lll}
\hline
\makecell[c]{\textbf{Parameter}} & \makecell[c]{\textbf{Value}} & \makecell[c]{\textbf{Description}}\\
\hline

$\totRounds$ & $10^{12}$ & Total number of signals sent \\

$\epsSecu$ & $10^{-15}$ & Security parameter \\

$p_Z^A$ & 0.95 & Probability for Alice to choose the Z-basis \\

$s$ & \textit{optimized} & Bob's beamsplitting ratio\footnote{For the active setup considered below, we set Bob's Z-basis probability $p_Z^B =p_Z^A=0.95$.} \\

$p_{\mathtt{test}|Z}$ & 0.05 & Probability for a test round given Z-basis \\

$\mu_\mathrm{S}, \mu_\mathrm{D}$ & \textit{optimized} & Signal and decoy intensities \\

$p_{\mu_\mathrm{S}}, p_{\mu_\mathrm{D}}$ & 0.95, 0.05 & Probability for signal and decoy intensities\footnote{For the active setup considered below, two decoy intensities are chosen and we set $p_{\mu_\mathrm{S}}=0.9$, $ p_{\mu_\mathrm{D1}}= p_{\mu_\mathrm{D2}}=0.05$.} \\

$\delta_{\mathrm{SPF}}$ & 0.063\footnote{For illustration we take $\delta_{\mathrm{SPF}} = 0.063$, corresponding to an extinction ratio of
$30$\,dB, a level achievable with a
phase-stabilized source \cite{PhysRevA.92.032305}.} & State preparation flaws \\

$\eta_\mathrm{i}$ & 0.73 \cite{pittaluga_600-km_2021} & Bob's detector efficiencies \\

$\theta_\mathrm{mis}$ & 0.1 & Misalignment angle \\

$\alpha$ & \textit{optimized} & Rényi $\alpha$-parameter \\

$\tagCutoff$ & 8 & Tagging source map cutoff \\

$\fssCutoff$ & 1 & Flag-state squasher cutoff \\

\hline
\end{tabular}
\caption{Protocol and device parameters used for the example in \cref{sec:example_pol_encoding}.}
\label{tab:protocol_parameters_example}
\end{table}

\begin{figure}
    \centering
    \includegraphics[]{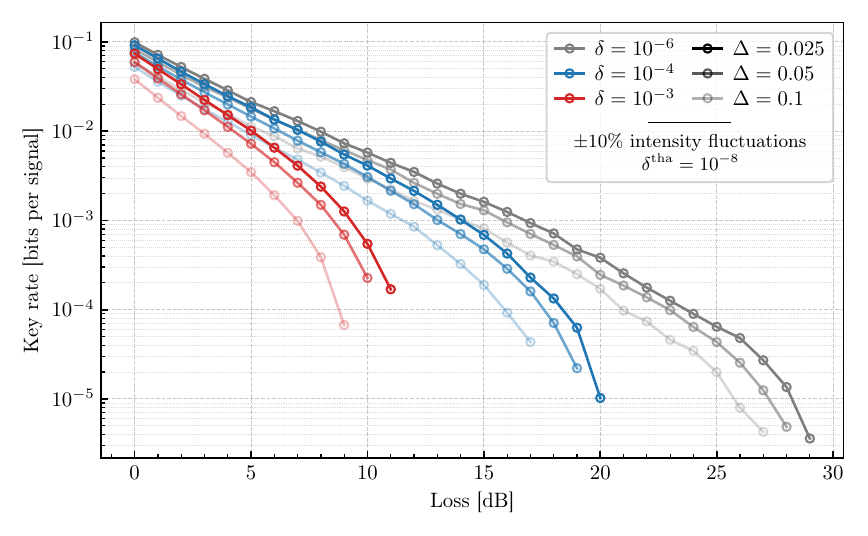}
    \caption{Finite-size secure key rate for the passive decoy-state BB84 setup depicted in \cref{fig:pol_enc_setup} with the imperfections listed in \cref{tab:device_imp_parameters} and the parameters listed in \cref{tab:protocol_parameters_example}. We show various combinations of single-photon source fidelity bounds $\delta$ (e.g. encoding uncertainty $\srcImpBndEnc$) and detector characterization uncertainty $\Delta$. For simplicity, the detector efficiency and beamsplitting ratio uncertainties are taken to be equal, $\Delta_{\eta}=\Delta_s=\Delta$. Every curve additionally includes $\pm10\%$ intensity fluctuations, Trojan-horse leakage $\srcImpBndTHA=10^{-8}$, state-preparation flaws $\delta_{\mathrm{SPF}} = 0.063$, and upper-bounded detector dark-count probabilities $\mathrm{d}_{i}^\mathrm{U}=10^{-8}$.}
    \label{fig:keyRate_example}
\end{figure}

The secure key rate for the setup described above and illustrated in \cref{fig:pol_enc_setup} is shown in \cref{fig:keyRate_example}, including state-preparation flaws and encoding uncertainty, intensity fluctuations, Trojan-horse attacks, imperfectly characterized detector efficiencies, dark counts, and beamsplitting ratio. 

We observe that positive key rates can be achieved for source fidelity bounds $\srcImpBnd$ on the order of $10^{-3}$, which, for example, corresponds to $\Delta\theta\approx2\degree$ angular uncertainty in the polarization encoding. Nevertheless, source imperfections have a significant impact on the key rate. This is due to the reduction used to incorporate imperfect characterization, which only constrains the overlaps of the known components of the signal states through the Gram matrix. The remaining unknown components are treated in the worst case and may lie in an arbitrarily large Hilbert space while being mutually orthogonal for different signal states. Consequently, the unknown component may become perfectly distinguishable, allowing Eve to effectively perform an unambiguous state discrimination attack whose effectiveness increases with channel loss. Intuitively, for source imperfections on the order of $\srcImpBnd$, such an attack can give full information about the key once the channel transmittance is also on the order of $\srcImpBnd$.

In contrast, the results are comparatively robust against detector characterization uncertainty. Relative uncertainties of $10\%$ (and more) in the detector still yield high key rates, whereas small source fidelity bound values already have a significant impact on the tolerable loss. Within the device models and approach considered here, accurate source characterization therefore appears to have the largest impact on performance.

\subsubsection{Comparison with previous work}

\begin{figure}[t]
    \centering
    \includegraphics[width=\linewidth]{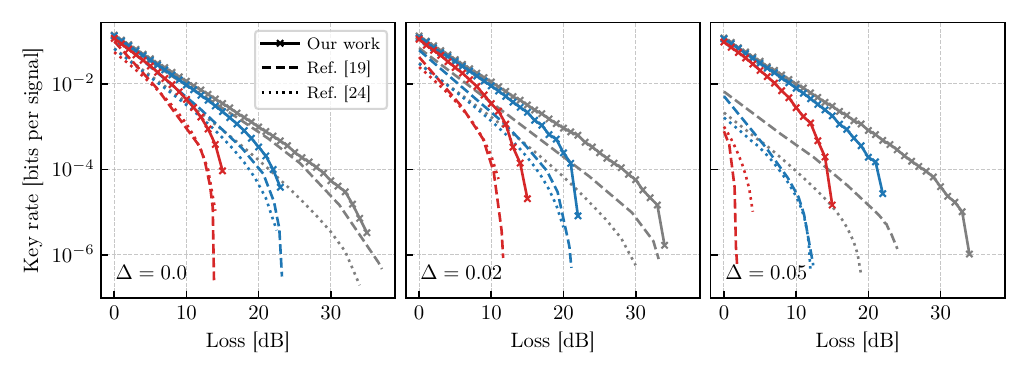}
    \caption{Key rates for active decoy-state BB84 obtained with the present work and the analyses of Ref.~\cite{navarrete2026numericalsecurityanalysispractical,sixto2026finitekeysecurityanalysisdecoystate}. The detector characterization uncertainty is $\Delta_{\mathrm{d}_\mathrm{B}} = \Delta_\eta = \Delta$, with $\Delta = 0$, $0.02$, and $0.05$ from left to right. The source fidelity bounds are $\delta = 10^{-3}$ (red), $10^{-4}$ (blue), and $10^{-6}$ (grey). The device imperfections are described in the main text, and the protocol parameters are listed in \cref{tab:protocol_parameters_example} (with two decoy intensities instead).}
    \label{fig:comparison}
\end{figure}

\noindent For comparison, we present the decoy-state BB84 key rates obtained using the recent analyses of both Refs.~\cite{navarrete2026numericalsecurityanalysispractical,sixto2026finitekeysecurityanalysisdecoystate} in \cref{fig:comparison} for the same device imperfections. Both works have a scope similar to ours as they combine a broad range of source and detector imperfections for finite-size decoy-state protocols, but follow the phase-error estimation approach, using numerical and analytical techniques, respectively. For example, both works account for state encoding flaws and source leakage satisfying the fidelity bounds in \cref{as:imperfectly_characterized_source}, as well as detector efficiency and dark count rate mismatches. Nevertheless, there are several important differences in scope with respect to the present work, in addition to their different proof technique. Therefore, in \cref{fig:comparison} we only include the imperfections covered by all three analyses.

Refs.~\cite{navarrete2026numericalsecurityanalysispractical,sixto2026finitekeysecurityanalysisdecoystate} account for correlations between successive signal states. Such source correlations are not treated explicitly in the present work, although they could be incorporated (together with detector correlations) \cite{meat_correlations}. We leave this extension for future work and, for a fair comparison, do not include source correlations in \cref{fig:comparison}. Ref.~\cite{sixto2026finitekeysecurityanalysisdecoystate} considers intensity fluctuations, but assumes that the complete underlying intensity distribution is known. Together with perfect phase randomization, this determines the corresponding photon-number probabilities exactly.\footnote{Thus, the intensity fluctuations merely replace the Poisson distribution by a different photon-number distribution, which is still assumed to be known exactly.} This assumption may be restrictive when the fluctuations originate, for example, from imperfect device characterization as the prepared intensity may only be known to lie within a specified interval, without any reliable model for its probability distribution. The present work, similarly to Ref.~\cite{kamin_renyi_2025}, requires only a known support for the intensity fluctuations and makes no assumption about the distribution within that support (see \cref{as:probability_bounds}). It directly covers situations in which the distribution is known and remains applicable when only characterization bounds are available.

Furthermore, Refs.~\cite{navarrete2026numericalsecurityanalysispractical,sixto2026finitekeysecurityanalysisdecoystate} assume perfect phase randomization, whereas the present work explicitly treats imperfect phase randomization (see \cref{sec:imperfect_block_diagonal}). Ref.~\cite{sixto2026finitekeysecurityanalysisdecoystate} also assumes that Alice's emitted states do not leak information about the selected intensity. This assumption can be violated, for example, by a Trojan-horse attack, in which Eve probes Alice's preparation device and the returning light carries information about the intensity setting (see \cref{ap:trojan_horse_attack}). We account for such intensity leakage through the relaxed decoy-state constraints introduced in \cref{ap:relaxed_decoy}. In addition, the protocols analyzed in Refs.~\cite{navarrete2026numericalsecurityanalysispractical,sixto2026finitekeysecurityanalysisdecoystate} do not cover on-the-fly announcements, in which Alice and Bob communicate during the distribution and measurement stage, which are permitted in our work. Such announcements significantly reduce the classical memory requirements of practical implementations.

Finally, Refs.~\cite{navarrete2026numericalsecurityanalysispractical,sixto2026finitekeysecurityanalysisdecoystate} consider active basis choice, whereas our work applies to both passive and active basis choice. The passive receiver is treated in the main text, while the active receiver is treated in \cref{ap:active_basis_choice}. On the other hand, Ref.~\cite{sixto2026finitekeysecurityanalysisdecoystate} obtains analytical key rate expressions, which makes the numerical certification of its bounds comparatively direct. For both Ref.~\cite{navarrete2026numericalsecurityanalysispractical} and the present work, numerical certification requires the construction of a valid dual-feasible point, which can be more challenging to ensure.

To obtain a direct numerical comparison, we consider the same active decoy-state BB84 protocol and include only the imperfections common to both analyses. Specifically, we use the protocol parameters listed in \cref{tab:protocol_parameters_example} (with two decoy intensities instead) and consider several combinations of source imperfections, parametrized by the source fidelity bound $\srcImpBnd$,\footnote{We set $\delta_\blockVar = \blockVar\delta$ (see \cref{ap:models_source_imperfections}). As in the present work, Ref.~\cite{navarrete2026numericalsecurityanalysispractical} uses fidelity bounds for individual photon-number blocks, whereas Ref.~\cite{sixto2026finitekeysecurityanalysisdecoystate} uses only the single-photon overlap.} and detector characterization uncertainties, parametrized by $\Delta_{\mathrm{d}_\mathrm{B}} = \Delta_{\eta}$ (see \cref{rem:active_basis_choice_q_params}). Following the discussion above, source correlations, intensity fluctuations, intensity leakage, and imperfect phase randomization are excluded from this comparison.

\cref{fig:comparison} presents the resulting key rates for several values of the detector characterization uncertainty and source fidelity bounds. Our work consistently produces higher key rates than the analyses of Refs.~\cite{navarrete2026numericalsecurityanalysispractical,sixto2026finitekeysecurityanalysisdecoystate}. This improvement over EUR-based approaches has been previously noted in Ref.~\cite{kamin_renyi_2025}. 

We stress that any conclusion drawn from such comparisons is subject to change as these proof techniques are constantly being improved upon. More concretely, we expect that the EUR key rates can be improved by utilizing recent results on tight Leftover Hashing Lemmas \cite{regula2026rethinkingquantumsmoothentropies}. Furthermore, our superior performance against detector imperfections is largely due to the techniques from Ref.~\cite{nahar_imperfect_2026}, which can also be applied to EUR-based proofs \cite{nahar_proof-technique-independent_2026}. Doing so would likely result in an improved performance of those techniques against detector imperfections.

\section{Conclusion}
\label{sec:conclusion}
\noindent In this work, we present a rigorous finite-size security proof for a wide class of prepare-and-measure QKD protocols, including qubit and decoy-state BB84, incorporating a broad range of practical source and detector imperfections in a modular manner. Importantly, we incorporate the case where the devices are not perfectly characterized (see \cref{sec:imperfect_characterization_discussion}), as characterization procedures inherently involve uncertainty and some imperfections may be under Eve’s influence. This extends the security proof from Ref.~\cite{tupkary_rigorous_2026} to the setting with imperfect and imperfectly characterized devices.

We have shown that practical key rates can be obtained even when several imperfections are present simultaneously (\cref{fig:keyRate_example}), obtaining higher key rates than those reported in previous works. As examples, we consider state encoding flaws, detector efficiency and dark-count rate mismatch, imperfectly characterized beamsplitter, Trojan-horse attacks, imperfectly block-diagonal sources (or imperfect phase randomization), intensity fluctuations, and intensity leakage. We directly compare our results with the recent phase-error estimation analyses of Refs.~\cite{navarrete2026numericalsecurityanalysispractical,sixto2026finitekeysecurityanalysisdecoystate}, which have a similar scope, for the imperfections covered by all three works (\cref{fig:comparison}). 

Our approach involves the application of a sequence of source maps and squashing maps to reduce the original problem with infinite-dimensional and imperfect devices to an equivalent finite-dimensional virtual protocol. Using the MEAT yields the general security statement in \cref{th:security_imperfect_devices}. We then explicitly expand the result into the convex optimization problem of \cref{th:convex_opti_numerics_imperfect}, making the results numerically tractable. On the source side, imperfections are incorporated through an optimization over the Gram matrix of Alice’s signal states \cite{pereira_optimal_2025}, while on the detector side they are incorporated as an additional constraint similar to the flag-state squasher constraint \cite{nahar_imperfect_2026}.

A key feature of the framework is its modularity. In contrast to many existing security proofs, where individual imperfections require dedicated and highly protocol-specific arguments, the approach presented here allows different source and detector imperfections, together with imperfect characterization, to be incorporated and combined in a proof-technique independent manner.

While many steps of our analysis are tight, we make specific choices pertaining to the imperfection parameters and their bounds, which may be improved upon. For instance, we characterize the signal states through their fidelity to known target states, cf. \cref{as:imperfectly_characterized_source}, motivated by previous work \cite{pereira_optimal_2025}. However, the source-replaced marginal is determined directly by the overlaps between Alice's signal states (see \cref{ap:construction_marginal_constraint}). It therefore appears natural to characterize the source directly through bounds on these overlaps, rather than through their fidelity to arbitrary target states. Similarly, the metric used for the characterization of the detectors from \cref{as:imperfectly_characterized_detector} is motivated by previous work \cite{nahar_imperfect_2026}.

There are also some device imperfections that are not currently covered by our approach. For example, light-injection attacks where Eve partially controls Alice's state preparation (e.g. through injection-locking) can, without on-the-fly announcements, be treated similarly to the Trojan-horse attack (\cref{ap:trojan_horse_attack}). With on-the-fly announcements, however, the modified timing argument (\cref{lem:modified_timing_protocol}) and the source-replacement scheme no longer directly apply, since Alice's state preparation depends on previous announcements. To the best of our knowledge, there is currently no security proof covering this scenario. Similarly, active attacks on the detector, such as detector-blinding attacks \cite{BSI24}, are not straightforward to model through \cref{as:imperfectly_characterized_detector}. In fact, handling such attacks is currently a general open problem in the field and may require a more detailed model of the detector behavior together with suitable experimental countermeasures.

An interesting direction for future work is to better understand the fundamental limits imposed by source and detector imperfections on the achievable key rates. While we restrict the analysis to source and detector imperfections that vary independently from round to round, correlated imperfections can in principle also be incorporated as they reduce to independent imperfections \cite{meat_correlations}. An explicit treatment is left to future work. More broadly, our work provides a general approach for incorporating practical device imperfections into QKD security proofs and can be extended to other protocols, such as entanglement-based \cite{PhysRevLett.67.661} and measurement-device-independent QKD \cite{PhysRevLett.108.130503}.
\section*{Acknowledgments} 
\label{sec:acknowledgements}

\noindent We particularly thank Xoel Sixto and Víctor Zapatero for providing the data used in our comparison with Ref.~\cite{sixto2026finitekeysecurityanalysisdecoystate}, and Álvaro Navarrete and Guillermo Currás-Lorenzo for the comparison with Ref.~\cite{navarrete2026numericalsecurityanalysispractical}. We also thank Mariana Navarro, Carlos Pascual-García, and Lars Kamin for stimulating discussions on the numerical analysis, and Shlok Nahar, Aodhán Corrigan, and Zhiyao Wang for helpful discussions on the theory. This work was funded by the NSERC Alliance QUINT and DND-MicroNet. It was conducted at the Institute for Quantum Computing, University of Waterloo, which is funded by the Government of Canada through ISED. D.T. was partially funded by the Mike and Ophelia Lazaridis Fellowship.
\section*{Author contributions}
\noindent J.W. carried out the research. J.B. implemented the numerics with the support of J.W. D.T. and N.L. supervised the research. All authors contributed to discussions and reviewed the manuscript.

\section*{Code availability}
\label{sec:code_availability}

\noindent We provide the numerical software used to compute the key rates in Ref.~\cite{numericalanalysis}, using the open source OpenQKDSecurity package \cite{burniston_2024_14262569}, which is designed to be easily adapted to different QKD implementations and combinations of device imperfections.

\bibliographystyle{apsrev4-2}
\bibliography{lit}{}

\appendix

\renewcommand{\thesection}{\Alph{section}}
\renewcommand{\thesubsection}{\thesection.\arabic{subsection}}
\renewcommand{\thesubsubsection}{\thesubsection.\arabic{subsubsection}}
\crefalias{section}{appendix}
\crefalias{subsection}{appendix}
\crefalias{subsubsection}{appendix}

\makeatletter
\let\oldaddcontentsline\addcontentsline
\renewcommand{\addcontentsline}[3]{%
  \def\argone{#1}%
  \def\argtwo{#2}%
  \def\tocname{toc}%
  \def\subsecname{subsection}%
  \def\subsubsecname{subsubsection}%
  \ifx\argone\tocname
    \ifx\argtwo\subsecname
    \else
      \ifx\argtwo\subsubsecname
      \else
        \oldaddcontentsline{#1}{#2}{#3}%
      \fi
    \fi
  \else
    \oldaddcontentsline{#1}{#2}{#3}%
  \fi
}
\makeatother

\addtocontents{toc}{\protect\setcounter{tocdepth}{1}}

\section{Combination of flag-state squasher and noise channel}
\label{ap:proof_combination_fss_noise_channel}

\noindent In this section we provide a proof of \cref{lem:combination_FSS_noise_channel} from \cref{sec:security_source_squashing}. While the construction is technically inferred by Refs.~\cite{nahar_imperfect_2026,tupkary_rigorous_2026}, to the best of our knowledge no formal proof has previously appeared in the literature. 
\begingroup
\renewcommand{\thelemma}{\ref{lem:combination_FSS_noise_channel}}
\begin{lemma}
    Consider the protocol $\big\{\big\{\Astate^{(j)}_{\Aclassical_j\Aprime_j}, \big\{\Bpovmel[j]\big\}_{\BclassicalVal\in\BclassicalAlph}, \annKeyMap[j]\big\}_{j=1}^\totRounds,$ $\ppMap\big\}$ where Bob's measurement POVMs are block-diagonal in Fock space, i.e. $\Bpovmel[j] = \big(\oplus_{\blockVar=0}^{\fssCutoff}\BpovmelBlock{j}{\blockVar}\big)\oplus\BpovmelBlock{j}{\blockVar>\fssCutoff}$. Let $W^{\Bmeas_j}$ be an outcome and $\lambdaMin>0$ such that $ W^{\Bmeas_j} \geq \lambda_\mathrm{min} \proj_{>\fssCutoff}^{\Bmeas_j}$ with photon-number cutoff $\fssCutoff$, where $\proj_{>\fssCutoff}^{\Bmeas_j}$ is a projector onto the space $\blockVar>\fssCutoff$. Assume that there exists a set of known POVM elements $\{\BpovmelBlockTilde{j}{\blockVar}\}_{j, \BclassicalVal\in \BclassicalAlph, \blockVar\leq\fssCutoff}$, such that
    \begin{equation}
        \BpovmelBlock{j}{\blockVar} - (1 - \dtImpBnd[\blockVar, j]) \BpovmelBlockTilde{j}{\blockVar} \geq 0\,,
    \end{equation}
    with $1>\dtImpBnd[\blockVar, j]\geq 0$ for all $\blockVar \leq \fssCutoff$ and all $\BclassicalVal\in\BclassicalAlph$. Define a new set of POVMs $\big\{\BpovmelTarg[j]\big\}_{j, \BclassicalVal\in \BclassicalAlph}$ where
    \begin{equation}
    \label{eq:povm_virt_ap}
        \BpovmelTarg[j] \coloneqq \left(\oplus_{\blockVar = 0}^\fssCutoff \BpovmelBlockTilde{j}{\blockVar} \right) \oplus \ketbra{\BclassicalVal}{\BclassicalVal}_\flagSpaceReg\,.
    \end{equation}
    Then there exist squashing maps $\squashMap_j\in\cptp(\Bmeas_j, \BmeasSquash_j)$ such that
    \begin{align}
        \squashMap_j^\dagger \left[\BpovmelTarg[j]\right] &= \Bpovmel[j]
    \end{align}
    and if the restricted attack channel is constructed as
    \begin{align}
         \setAttackCh_j = \left\{\attackCh_j \in \cptp(\Ereg_{j-1}, \BmeasSquash_j \Ereg_j) : \attackCh_j^\dagger\left[\widetilde{W}^{\Bmeas_j'} \otimes \identity_{\Ereg_j} \right] \geq  \lambdaMin  \left( \attackCh_j^\dagger\left[\identity_{\BmeasSquash_j\Ereg_j} \right] - \sum_{\blockVar=0}^{\fssCutoff} \frac{1}{1 - \dtImpBnd[\blockVar, j]} \attackCh_j^\dagger\left[\proj_\blockVar^{\BmeasSquash_j} \otimes \identity_{\Ereg_j} \right]\right)\right\}
    \end{align}
    where $\squashMap_j^\dagger \left[\widetilde{W}^{\Bmeas_j'}\right] = W^{\Bmeas_j}$, then
    \begin{equation}
        \setAttackCh_j \supset \squashMap_j \circ \cptp(\Ereg_{j-1}, \Bmeas_j \Ereg_j)
    \end{equation}
    for all rounds $j$.
\end{lemma}
\endgroup
\begin{proof}
    For convenience, in this section, whenever we write block operators such as $\BpovmelBlock{j}{\blockVar}$ or $\BpovmelBlockTilde{j}{\blockVar}$, we implicitly identify them with their embeddings into the full Hilbert space through the corresponding projectors, i.e.
    \begin{equation}
        \BpovmelBlock{j}{\blockVar} = \proj_\blockVar^{\Bmeas_j} \Bpovmel[j] \proj_\blockVar^{\Bmeas_j}\,,
    \end{equation}
    and similarly for all other block operators. We assume $\dtImpBnd[\blockVar,j]<1$ for all $\blockVar\leq\fssCutoff$. For every $\blockVar \leq \fssCutoff$, define
    \begin{equation}
    \label{eq:def_r_tilde}
        \widetilde{R}_{\BclassicalVal,j,\blockVar} \coloneqq \frac{\BpovmelBlock{j}{\blockVar} - (1-\dtImpBnd[\blockVar,j])\BpovmelBlockTilde{j}{\blockVar}}{\dtImpBnd[\blockVar,j]}
    \end{equation}
    whenever $\dtImpBnd[\blockVar,j]>0$, and choose any POVM $\{\widetilde{R}_{\BclassicalVal,j,\blockVar}\}_{\BclassicalVal\in\BclassicalAlph}$ on the $\blockVar$-photon subspace otherwise. By definition, $\widetilde{R}_{\BclassicalVal,j,\blockVar}\geq 0$. Moreover, since both $\big\{\BpovmelBlock{j}{\blockVar}\big\}_{\BclassicalVal\in\BclassicalAlph}$ and $\big\{\BpovmelBlockTilde{j}{\blockVar}\big\}_{\BclassicalVal\in\BclassicalAlph}$ are POVMs on the $\blockVar$-photon subspace, we have
    \begin{equation}
        \sum_{\BclassicalVal\in\BclassicalAlph}\widetilde{R}_{\BclassicalVal,j,\blockVar} = \proj_\blockVar^{\Bmeas_j}\,.
    \end{equation}
    We explicitly construct the squashing map $\squashMap_j\in\cptp(\Bmeas_j,\BmeasSquash_j)$ as follows
    \begin{align}
    \label{eq:explicit_noise_channel_construction}
        \squashMap_j[\genDensity] &\coloneqq \sum_{\blockVar=0}^{\fssCutoff}\left((1-\dtImpBnd[\blockVar,j])\proj_\blockVar^{\Bmeas_j}\genDensity\proj_\blockVar^{\Bmeas_j} + \dtImpBnd[\blockVar,j]\sum_{\BclassicalVal\in\BclassicalAlph}\Tr\left[\widetilde{R}_{\BclassicalVal,j,\blockVar}\proj_\blockVar^{\Bmeas_j}\genDensity\proj_\blockVar^{\Bmeas_j}\right]\left(0_{\hilbert_{\leq\fssCutoff}} \oplus \ketbra{\BclassicalVal}{\BclassicalVal}_{\flagSpaceReg}\right)\right) \\
        &+ \sum_{\BclassicalVal\in\BclassicalAlph}\Tr\left[\BpovmelBlock{j}{\blockVar>\fssCutoff}\proj_{>\fssCutoff}^{\Bmeas_j}\genDensity\proj_{>\fssCutoff}^{\Bmeas_j}\right]\left(0_{\hilbert_{\leq\fssCutoff}} \oplus \ketbra{\BclassicalVal}{\BclassicalVal}_{\flagSpaceReg}\right)\,.\nonumber
    \end{align}
    The map first measures the block label $\blockVar$. Photon-number blocks with $\blockVar>\fssCutoff$ are mapped to classical flags. For blocks with $\blockVar\leq\fssCutoff$, a fraction $\dtImpBnd[\blockVar,j]$ is mapped to flags, while the remaining fraction is left unchanged. Therefore, it is easy to see that this map is completely positive and trace-preserving. We now show that $\squashMap_j^\dagger[\BpovmelTarg[j]]=\Bpovmel[j]$. It suffices to prove that
    \begin{equation}
    \label{eq:trace_equivalence_proof}
        \Tr\left[\BpovmelTarg[j]\squashMap_j[\genDensity]\right] = \Tr\left[\Bpovmel[j]\genDensity\right]
    \end{equation}
    for all $\genDensity\in\setDensity_\leq(\Bmeas_j)$. Using \cref{eq:povm_virt_ap,eq:explicit_noise_channel_construction}, the left-hand side is
    \begin{align}
        \Tr\left[\BpovmelTarg[j]\squashMap_j[\genDensity]\right]
        =& \sum_{\blockVar=0}^{\fssCutoff}(1-\dtImpBnd[\blockVar,j])\Tr\left[\BpovmelBlockTilde{j}{\blockVar}\proj_\blockVar^{\Bmeas_j}\genDensity\proj_\blockVar^{\Bmeas_j}\right] 
        + \sum_{\blockVar=0}^{\fssCutoff}\dtImpBnd[\blockVar,j]\Tr\left[\widetilde{R}_{\BclassicalVal,j,\blockVar}\proj_\blockVar^{\Bmeas_j}\genDensity\proj_\blockVar^{\Bmeas_j}\right] \nonumber\\
        & + \Tr\left[\BpovmelBlock{j}{\blockVar>\fssCutoff}\proj_{>\fssCutoff}^{\Bmeas_j}\genDensity\proj_{>\fssCutoff}^{\Bmeas_j}\right] \\
        =& \sum_{\blockVar=0}^{\fssCutoff}\Tr\left[\BpovmelBlock{j}{\blockVar}\proj_\blockVar^{\Bmeas_j}\genDensity\proj_\blockVar^{\Bmeas_j}\right]
        \quad + \Tr\left[\BpovmelBlock{j}{\blockVar>\fssCutoff}\proj_{>\fssCutoff}^{\Bmeas_j}\genDensity\proj_{>\fssCutoff}^{\Bmeas_j}\right] \\
        =& \Tr\left[\Bpovmel[j]\genDensity\right].
    \end{align}
    The second equality follows from \cref{eq:def_r_tilde}, while the final equality follows from the block-diagonal structure of $\Bpovmel[j]$. Hence $\squashMap_j^\dagger[\BpovmelTarg[j]]=\Bpovmel[j]$. It remains to prove the inclusion of the set of attack channels. Let $\attackCh_j=\squashMap_j\circ\attackCh'_j$ with $\attackCh'_j\in\cptp(\Ereg_{j-1},\Bmeas_j\Ereg_j)$. By construction,
    \begin{equation}
        \squashMap_j^\dagger\left[\proj_\blockVar^{\BmeasSquash_j}\right] = (1-\dtImpBnd[\blockVar,j])\proj_\blockVar^{\Bmeas_j}
    \end{equation}
    for all $\blockVar\leq\fssCutoff$. Moreover, using $\squashMap_j^\dagger[\widetilde{W}^{\Bmeas_j'}]=W^{\Bmeas_j}$ and $W^{\Bmeas_j}\geq\lambdaMin\proj_{>\fssCutoff}^{\Bmeas_j}$, we have
    \begin{equation}
        \squashMap_j^\dagger[\widetilde{W}^{\Bmeas_j'}] \geq \lambdaMin\left(\identity_{\Bmeas_j}-\sum_{\blockVar=0}^{\fssCutoff}\proj_\blockVar^{\Bmeas_j}\right)\,,
    \end{equation}
    and therefore
    \begin{equation}
        \squashMap_j^\dagger[\widetilde{W}^{\Bmeas_j'}] \geq \lambdaMin\left(\identity_{\Bmeas_j}-\sum_{\blockVar=0}^{\fssCutoff}\frac{1}{1-\dtImpBnd[\blockVar,j]}\squashMap_j^\dagger\left[\proj_\blockVar^{\BmeasSquash_j}\right]\right)\,.
    \end{equation}
    Applying $\attackCh_j^{\prime\dagger}$ to both sides and using $\attackCh_j^\dagger=\attackCh_j^{\prime\dagger}\circ\squashMap_j^\dagger$ gives
    \begin{equation}
        \attackCh_j^\dagger\left[\widetilde{W}^{\Bmeas_j'}\otimes\identity_{\Ereg_j}\right] \geq \lambdaMin\left(\attackCh_j^\dagger\left[\identity_{\BmeasSquash_j}\otimes\identity_{\Ereg_j}\right]-\sum_{\blockVar=0}^{\fssCutoff}\frac{1}{1-\dtImpBnd[\blockVar,j]}\attackCh_j^\dagger\left[\proj_\blockVar^{\BmeasSquash_j}\otimes\identity_{\Ereg_j}\right]\right)\,.
    \end{equation}
    Hence $\attackCh_j\in\setAttackCh_j$, and therefore $\setAttackCh_j \supset \squashMap_j\circ\cptp(\Ereg_{j-1},\Bmeas_j\Ereg_j)$ as this holds for all $\attackCh_j=\squashMap_j\circ\attackCh'_j$.
\end{proof}

\section{Numerics}
\label{ap:numerics}

\noindent We consider the optimization problem for $\kappaQKDVirt_j$ from \cref{th:security_imperfect_devices} (\cref{eq:optimization_imperfect_devices}), where the source and detector are imperfect and imperfectly characterized according to \cref{as:block_diagonal_source,as:imperfectly_characterized_source,as:probability_bounds,as:block_diagonal_detector,as:imperfectly_characterized_detector}. In this section, we reformulate this problem as a numerically tractable convex optimization problem. We first explicitly construct the protocol map $\gMap_j$ from the \nameref{prot:entanglement_qkd_protocol} (see \cref{fig:MEAT_channels}) in \cref{ap:protocol_map}. Then, in \cref{ap:expanding_objective,sec:expanding_constraints}, we expand the objective function and the constraints, leading to the convex optimization problem stated in \cref{th:convex_opti_numerics_imperfect}. The definitions of the relevant entropic quantities can be found in \cref{app:misc}. 

The analysis in this section relies strongly on the numerical analysis presented in Ref.~\cite{kamin_renyi_2025} (which generalizes Ref.~\cite{winick_reliable_2018}), including several of its entropy reductions and convexity results.  We refer the reader to Ref.~\cite[Section 10]{tupkary_rigorous_2026} for a discussion on various subtleties regarding these techniques, and Ref.~\cite{navarro_finitesize_2026} for another approach to solving the optimization problem.
We note that Ref.~\cite{kamin_renyi_2025} formulates the corresponding optimization over Eve's attack channel. Here, since the signal states are imperfectly characterized (\cref{as:imperfectly_characterized_source}), we instead optimize over the set of compatible output states, as discussed in \cref{rem:we_do_output_state}.

For computing key rates, we require that the tradeoff function $f(\hat c)$ takes the same value for all $\hat c\in\varDecReg_j$ associated with generation rounds. (This simplifies certain technical calculations from Ref.~\cite{kamin_renyi_2025}, and is made more apparent in \cref{ap:conditioning_classical_registers}).
However, this condition is  difficult to directly impose numerically. We therefore implement this condition via an additional register, as we explain below.

\subsection{Protocol map}
\label{ap:protocol_map}

\noindent In this section, we explicitly construct the protocol map $\gMap_j\in\cptp(\Amarg_j\Bmeas_j,  \secretReg_j\varDecReg_j \eveCopyReg_j)$ for the \nameref{prot:entanglement_qkd_protocol}, involving Alice and Bob's measurements, announcements and the key map. 
Alice and Bob make probabilistic announcements, which are represented by the stochastic announcement map $\annFunc[j]:\AclassicalAlph \times \BclassicalAlph \rightarrow \probSimplex_{|\varDecAlph|}$ (see \cref{rem:test_gen_decision}).  To implement the restriction that $f(\hat c)$ takes the same value for all $\hat c$ in generation rounds, we will, solely for the purpose of the numerical calculations, introduce an intermediate register $\interReg_j$ with alphabet $\interAlph$, which stores all of Alice and Bob's classical announcements in round $j$. We write the announcement map $\annFunc[j]:\AclassicalAlph \times \BclassicalAlph \rightarrow \probSimplex_{|\interAlph|}$ on this new register and partition $\interAlph =  \interAlph_{\mathtt{test}} \cup \interAlph_{\mathtt{gen}}$ into the set of announcements corresponding to test and generation rounds. Then, the value of the register $\varDecReg_j$ is computed with a function $\varDecFunc_j: \interAlph \rightarrow \varDecAlph$ which enforces that the register $\varDecReg_j$ is fixed to $\varDecGen$ in all generation rounds, i.e. 
\begin{equation}
    \varDecFunc_j(\bar c) = 
    \begin{cases}
        \varDecGen & \text{if } \bar c \in \interAlph_{\mathtt{gen}}\\
        \bar c & \text{if } \bar c \in \interAlph_{\mathtt{test}}.
    \end{cases}
\end{equation}
Therefore, the alphabet of $\varDecReg_j$ can be partitioned as $\varDecAlph=\varDecAlph_\mathtt{test} \cup \varDecAlph_\mathtt{gen}$ with $\varDecAlph_\mathtt{test} = \interAlph_{\mathtt{test}}$ and $\varDecAlph_\mathtt{gen} = \{\varDecGen\}$. We emphasize that this construction is solely for numerical computation purposes and that it does not restrict Alice and Bob in the protocol, nor does it change the value of the objective function. This is because a copy of the actual public announcements is anyway given to Eve. 

Finally, the mapping from Alice’s setting choice, Bob’s measurement outcome and the register $\interReg_j$ to the secret register can be represented by a function $\secretFunc_j:\AclassicalAlph\times\BclassicalAlph\times\interAlph \rightarrow \secretAlph$ which includes the key map. With these definitions, we explicitly construct the protocol map $\gMap_j$ from the \nameref{prot:entanglement_qkd_protocol}.

\begin{definition}[Protocol map $\gMap_j$]
    \label{def:qkd_protocol_map}    
    Consider the \nameref{prot:entanglement_qkd_protocol}. We define the \textit{protocol map} $\gMap_j\in\cptp(\Amarg_j\Bmeas_j,  \secretReg_j\varDecReg_j \eveCopyReg_j)$ as
    \begin{equation}
        \gMap_j[\genDensity] \coloneqq \sum_{\secretVal\in\secretAlph}\sum_{\bar c \in \interAlph} \ketbra{\secretVal}{\secretVal}_{\secretReg_j}\otimes \ketbra{\varDecFunc_j(\bar c)}{\varDecFunc_j(\bar c)}_{\varDecReg_j}\otimes \ketbra{\bar c}{\bar c}_{\eveCopyReg_j} \Tr\left[M_{\secretVal,\bar c}^{\Amarg_j\Bmeas_j} \genDensity\right]\,,
    \label{eq:definition_G_map}
    \end{equation}
    where
    \begin{equation}
    \label{eq:povm_element_grouped}
        M_{\secretVal,\bar c}^{\Amarg_j\Bmeas_j} \coloneqq \sum_{\substack{(\AclassicalVal, \BclassicalVal)\in\AclassicalAlph\times\BclassicalAlph \\ \text{s.t. } (\AclassicalVal, \BclassicalVal, \bar c) \in \secretFunc_j^{-1}(\secretVal)}} \annFunc[j](\bar c|\AclassicalVal,\BclassicalVal) \Apovmel[j] \otimes \Bpovmel[j]
    \end{equation}
    is the effective POVM element corresponding to announcement $\bar c \in \interAlph$ and secret register value $\secretVal\in\secretAlph$. 
\end{definition}

\begin{remark}
\label{rem:test_gen_decision}
    The public announcements in round $j$ may be probabilistic according to $\annFunc[j]$. While the security proof supports any probabilistic announcement, we only require that the probability of declaring a test round (and hence also generation round) is fixed in advance (see \cite[Section~10.1]{tupkary_rigorous_2026}).\footnote{We emphasize that this is by no means a fundamental limitation of the approach, and the numerics can be extended to handle the case where the test/generation split is not fixed, although no work has done so.}  More precisely, let $\varDecAlph = \varDecAlph_{\mathtt{test}}\cup \varDecAlph_{\mathtt{gen}}$, where $\varDecAlph_{\mathtt{test}}$ ($\varDecAlph_{\mathtt{gen}}$) denotes the set of announcements corresponding to a test (generation) round. Then for every $\AclassicalVal \in \AclassicalAlph$, there exists a predetermined value $p_{\mathtt{test}}^{(j)}(\AclassicalVal)$ such that
    \begin{equation}
        \sum_{\varDecVal\in \mathcal \varDecAlph_{\mathtt{test}}}
    \annFunc[j](\varDecVal|\AclassicalVal,\BclassicalVal) = p_{\mathtt{test}}^{(j)}(\AclassicalVal)
    \end{equation}
    for all $\BclassicalVal \in \BclassicalAlph$. Then the test round probability is fixed and given by $p_{\mathtt{test}}^{(j)} = \sum_{\AclassicalVal\in\AclassicalAlph}\prob_{\Aclassical_j}(\AclassicalVal)  p_{\mathtt{test}}^{(j)}(\AclassicalVal)$. Notably, both Alice and Bob can make the decision, but the test round probability cannot depend on Bob’s measurement outcome and hence cannot be influenced by Eve’s attack. 
\end{remark}

\subsection{Expansion of the objective function}
\label{ap:expanding_objective}

\noindent In this section, we exploit the block-diagonality of the source and apply the duality relation \cite[Lemma~16, Eq.~(A2)]{kamin_renyi_2025} for down-arrow sandwiched Rényi entropies to expand the objective function from \cref{eq:optimization_imperfect_devices}. As discussed in \cref{sec:model_assumptions}, we handle a broad class of protocols uniformly, including decoy-state and non-decoy protocols. In particular, any block-diagonal structure (or decoy constraints) naturally become trivial whenever these properties are not present.

\subsubsection{Conditioning on the classical registers}
\label{ap:conditioning_classical_registers}

\noindent Alice’s signal states are block diagonal by \cref{as:block_diagonal_source}. Therefore, in each round, Eve may perform a QND measurement of the block label without disturbing the state emitted by Alice, and store the outcome (corresponding to either the photon number $\blockVar\in\{0,\ldots,\tagCutoff\}$ or the flag label $\AclassicalVal\in\AclassicalAlph$) in a classical register $\blockVarReg_j$ \cite[Sec.~IX]{kamin_renyi_2025}. We absorb this QND measurement into Eve’s attack, resulting in a state $\genDensity_{\Amarg_j\BmeasSquash_j\blockVarReg_j}$ in each round, without loss of generality. 

Let $\tagCutoff$ denote the photon-number cutoff introduced by the tagging source map, cf. \cref{lem:tagging_source_map}. Here, we introduce an additional cutoff $\characCutoff \leq \tagCutoff$ to reduce the numerical complexity of the optimization similarly to Ref.~\cite{kamin_renyi_2025}. More precisely, we explicitly characterize the photon-number blocks $\blockVar\in\{0,\ldots,\characCutoff\}$, while the remaining blocks $\blockVar\in\{\characCutoff+1,\ldots,\tagCutoff\}$ are represented only through (looser) yields describing the relevant measurement statistics. This allows us to retain partial information about higher photon-number blocks without the numerical cost of characterizing the full states in the optimization. The yields are introduced in \cref{ap:simplifications}.

Let $\nu_{\secretReg_j\varDecReg_j\eveCopyReg_j\blockVarReg_j\evePurReg} \coloneqq \gMapVirt[j]\left[\purFunc\left(\genDensity_{\Amarg_j\BmeasSquash_j\blockVarReg_j}\right)\right]$ denote the state output in the $j$th round of the protocol, where we recall that $\gMapVirt[j]$ is the protocol map of the protocol resulting from the source map and squashing map reductions from \cref{sec:problem_reductions}. Concretely, $\gMapVirt[j]$ is given by \cref{def:qkd_protocol_map} with appropriate replacement of the POVM elements with Bob's squashed POVM elements (\cref{eq:def_new_virtual_povm_elements}). We begin by expanding the objective function from \cref{eq:optimization_imperfect_devices} in terms of the classical register $\varDecReg_j$ using \cref{def:f_weighted_renyi_entropy}, yielding
\begin{equation}
\label{eq:expanded_entropy_f}
    \fRenyiUpEnt{1}{j-1}(\secretReg_j|\eveCopyReg_j\varDecReg_j\blockVarReg_j\evePurReg)_\nu \geq \frac{\alpha}{1-\alpha}\log\left((1-\ptest^{(j)})2^{\frac{1-\alpha}{\alpha} \left(\renyidown(\secretReg_j|\eveCopyReg_j\blockVarReg_j\evePurReg)_{\nu_{|\mathtt{gen}}} - \ftradeoff{1}{j-1}(\varDecGen)\right)} + \sum_{\varDecVal\in\varDecAlph_\mathtt{test}}\nu_{\varDecReg_j}(\varDecVal)2^{-\frac{1-\alpha}{\alpha}\ftradeoff{1}{j-1}(\varDecVal)}\right)
\end{equation}
where we split the sum up into generation rounds ($\varDecAlph_\mathtt{gen}$) and test rounds ($\varDecAlph_\mathtt{test}$), and used the fact that $\renyidown(\secretReg_j|\eveCopyReg_j\evePurReg) = 0$ in test rounds. Given that the register $\varDecReg_j$ is set to $\varDecGen$ in generation rounds and the generation round probability is fixed (see \cref{rem:test_gen_decision}), this implies $\nu_{\varDecReg_j}(\varDecGen) = 1 - \ptest^{(j)}$. This explains why we require the stochastic announcement map $\annFunc[j]$ to have fixed generation and testing probabilities (and the register $\varDecReg_j$ to take a fixed value in generation rounds) as this allows the generation round contribution to be factored out, simplifying the numerical analysis. This is also the approach taken in Ref.~\cite{kamin_renyi_2025}. However, we note that this assumption is not fundamental, and the numerical optimization can be extended to handle announcement maps where the generation and testing probabilities are not fixed, and depend on Eve's actions. Lastly, we used the fact that $\renyiup \geq \renyidown$ \cite[Sec.~5.2]{tomamichel_quantum_2016}, yielding the inequality above.\footnote{It is also possible to keep $\renyiup$, apply the conditioning step in \cref{eq:expand_entropy_block}, and later use $\renyiup \geq \renyidown$. This may lead to a slight improvement in key rates due to a slightly different scaling in the Rényi $\alpha$ parameter.} For the classical distribution on $\varDecReg_j$ generated by $\nu_{\secretReg_j\varDecReg_j\eveCopyReg_j}$, we have
\begin{equation}
\label{eq:def_nu_c_hat}
    \nu_{\varDecReg_j} = \Tr_{\secretReg_j\eveCopyReg_j} \left[\gMapVirt[j][\rho_{\Amarg_j\BmeasSquash_j}]\right] = \sum_{\bar c \in \interAlph}\sum_{\substack{\AclassicalVal\in \AclassicalAlph \\ \BclassicalVal\in \BclassicalAlph}} \annFunc[j](\bar c|\AclassicalVal,\BclassicalVal) \Tr[\left(\Apovmel[j] \otimes \BpovmelTarg[j]\right) \rho_{\Amarg_j\BmeasSquash_j}] \unitStatsVec[\bar c]\,,\end{equation}
where $\unitStatsVec[\bar c] \coloneqq \ketbra{\varDecFunc_j(\bar c)}{\varDecFunc_j(\bar c)}_{\varDecReg_j}$ and the function $\varDecFunc_j: \interAlph \rightarrow \varDecAlph$ is described in \cref{ap:protocol_map}. We expand the Rényi entropy in \cref{eq:expanded_entropy_f} further by using the conditioning on the block register $\blockVarReg_j$ \cite[Proposition~5.4]{tomamichel_quantum_2016}, yielding 
\begin{equation}
\label{eq:expand_entropy_block}
    \renyidown(\secretReg_j|\blockVarReg_j\eveCopyReg_j\evePurReg)_{\nu_{|\mathtt{gen}}} \geq \frac{1}{1-\alpha}\log\left(\sum_{\blockVar=0}^\characCutoff \prob_{\blockVarReg_j|\varDecReg_j}(\blockVar|\mathtt{gen}) \left(2^{(1-\alpha) \renyidown(\secretReg_j|\eveCopyReg_j\evePurReg)_{\nu_{|\mathtt{gen}\land \blockVarReg_j=\blockVar}}} - 1\right) + 1\right) 
\end{equation}
where we only evaluate the entropy term for $\blockVar\leq \characCutoff$ and, for the blocks $\blockVar> \characCutoff$ and blocks containing flags, we used the fact that $\nu_{|\mathtt{gen}}$ is classical on $\secretReg_j$, therefore $\renyidown(\secretReg_j|\eveCopyReg_j\evePurReg)_{\nu_{|\mathtt{gen}\land \blockVarReg_j=\blockVar/\AclassicalVal}}\geq 0$. Here,
\begin{equation}
\label{eq:nu_from_G_map_with_proj}
    \nu_{\secretReg_j \eveCopyReg_j \evePurReg|\mathtt{gen}\land \blockVarReg_j=\blockVar} = \Tr_{\varDecReg_j}\left[\frac{\proj_\mathtt{gen}^{\varDecReg_j}\gMapVirt[j]\left[\genDensity_{\Amarg_j\BmeasSquash_j\evePurReg|\blockVarReg_j=\blockVar}\right]\proj_\mathtt{gen}^{\varDecReg_j}}{1-p_{\mathtt{test}|\blockVarReg_j=\blockVar}^{(j)}}\right]\,,
\end{equation}
with projector $\proj_\mathtt{gen}^{\varDecReg_j} \coloneqq \ketbra{\varDecGen}{\varDecGen}_{\varDecReg_j}$ onto generation rounds.

\subsubsection{Reduction of the Rényi entropy for generation rounds}

\noindent Next, we wish to evaluate the Rényi entropy conditioned on generation rounds appearing in \cref{eq:expand_entropy_block}. Similarly to \cite[App.~H.1]{kamin_renyi_2025}, we simplify \cref{eq:nu_from_G_map_with_proj} by conditioning $\genDensity_{\Amarg_j\BmeasSquash_j\evePurReg|\blockVarReg_j=\blockVar}$ on generation rounds and replacing the protocol map $\gMapVirt[j]$ by the corresponding generation round protocol map $\gMapVirt[j]^\mathtt{gen}$, which is CPTP. This reformulation will allow us to apply the duality relation from \cite[Lemma~16, Eq.~(A2)]{kamin_renyi_2025} and remove the dependency on Eve's register $\evePurReg$. For this purpose, we define the partition operator for generation rounds
\begin{equation}
    \Lambda_\mathtt{gen}^{(j)} \coloneqq \sum_{\secretVal\in\secretAlph}\sum_{\bar c \in \interAlph_\mathtt{gen}} F_{\secretVal,\bar c}^{\Amarg_j\BmeasSquash_j}\,,
\end{equation}
where $F_{\secretVal,\bar c}^{\Amarg_j\BmeasSquash_j}$ is defined analogously to \cref{eq:povm_element_grouped} with Bob's squashed POVM elements given by \cref{eq:def_new_virtual_povm_elements}.
Then, it holds that $\sqrt{\left(\Lambda_\mathtt{gen}^{(j)}\right)^{-1}}\sqrt{\Lambda_\mathtt{gen}^{(j)}} = \proj_\mathtt{gen}^{\Amarg_j\BmeasSquash_j}$ is a projector onto the generation round subspace, where the inverse denotes the Penrose pseudo-inverse. Following \cite[App.~H]{kamin_renyi_2025}, we write the conditional state
\begin{equation}
    \genDensity_{\Amarg_j\BmeasSquash_j\evePurReg|\mathtt{gen}\land \blockVarReg_j=\blockVar} = \frac{\sqrt{\Lambda_\mathtt{gen}^{(j)}} \genDensity_{\Amarg_j\BmeasSquash_j\evePurReg|\blockVarReg_j=\blockVar} \sqrt{\Lambda_\mathtt{gen}^{(j)}}}{1-p_{\mathtt{test}|\blockVarReg_j=\blockVar}^{(j)}}
\end{equation}
and the conditional POVM elements
\begin{equation}
    F_{\secretVal,\bar c}^{\Amarg_j\BmeasSquash_j|\mathtt{gen}} \coloneqq \sqrt{\left(\Lambda_\mathtt{gen}^{(j)}\right)^{-1}} F_{\secretVal,\bar c}^{\Amarg_j\BmeasSquash_j}\sqrt{\left(\Lambda_\mathtt{gen}^{(j)}\right)^{-1}}
\end{equation}
as well as the generation round protocol map 
\begin{equation}
    \gMapVirt[j]^\mathtt{gen}[\genDensity] \coloneq \sum_{\secretVal\in\secretAlph}\sum_{\bar c \in \interAlph_\mathtt{gen}} \ketbra{\secretVal}{\secretVal}_{\secretReg_j}\otimes \ketbra{\varDecFunc_j(\bar c)}{\varDecFunc_j(\bar c)}_{\varDecReg_j}\otimes \ketbra{\bar c}{\bar c}_{\eveCopyReg_j} \Tr\left[F_{\secretVal,\bar c}^{\Amarg_j\BmeasSquash_j|\mathtt{gen}} \genDensity\right]\,,
\end{equation}
which is CPTP on $\mathrm{supp}(\Lambda_\mathtt{gen}^{(j)})$.
Then, it is easy to verify that
\begin{equation}
\label{eq:nu_from_G_map}
    \nu_{\secretReg_j \eveCopyReg_j \evePurReg|\mathtt{gen}\land \blockVarReg_j=\blockVar} = \Tr_{\varDecReg_j}\left[\gMapVirt[j]^\mathtt{gen}\left[\genDensity_{\Amarg_j\BmeasSquash_j\evePurReg|\mathtt{gen}\land\blockVarReg_j=\blockVar}\right]\right]
\end{equation}
corresponds to \cref{eq:nu_from_G_map_with_proj}. We can now use the duality argument from \cite[Corollary~9]{kamin_renyi_2025} and remove the dependency on Eve's register $\evePurReg$. 
We apply \cite[Lemma~16, Eq.~(A2)]{kamin_renyi_2025} where we identify the map given by \cref{eq:nu_from_G_map} as $\mathcal{M}$ in \cite[Lemma~16, Eq.~(A2)]{kamin_renyi_2025} and we define its Stinespring dilation as
\begin{equation}
    V_j \coloneqq \sum_{\secretVal\in\secretAlph}\sum_{\bar c \in \interAlph_{\mathtt{gen}}} \ket{\secretVal}_{\secretReg_j}\ket{\secretVal}_{\tilde \secretReg_j} \ket{\varDecFunc_j(\bar c)}_{\varDecReg_j} \ket{\bar c}_{\eveCopyReg_j} \ket{\bar c}_{\interReg_j} \sqrt{F_{\secretVal,\bar c}^{\Amarg_j\BmeasSquash_j|\mathtt{gen}}}
\end{equation}
yielding
\begin{equation}
    \renyidown(\secretReg_j|\eveCopyReg_j\evePurReg)_{\nu_{|\mathtt{gen}\land\blockVarReg_j=\blockVar}} = \frac{1}{1-\alpha}\log\Tr\left[\left(\Tr_{\secretReg_j}\left[\Tr_{\eveCopyReg_j}\left[\left(V_j \genDensity_{\Amarg_j\BmeasSquash_j|\mathtt{gen}\land \blockVarReg_j=\blockVar} V_j^\dagger \right)\right]^{\frac{1}{\alpha}}\right]\right)^\alpha\right]\,,
\end{equation}
following \cite[Eq.~(A2)]{kamin_renyi_2025}. Having expressed the objective function in terms of the conditional output states for each photon-number block $\blockVar$, we now turn to the constraints of the optimization problem.

\subsection{Constraints}
\label{sec:expanding_constraints}

\noindent In this section, we consider the constraints appearing in the set of output states $\widehat \stateSet_j(\setAliceMarginalsVirt[j])$ from \cref{th:security_imperfect_devices} (\cref{eq:set_output_states_security_statement}) and, again, exploit the block-diagonal structure of Alice's signal states to expand the constraints.

\subsubsection{Source constraints}

\noindent By \cref{as:block_diagonal_source}, Alice's signal states are block diagonal. Consequently, every marginal in $\setAliceMarginalsVirt[j]$ is block diagonal with respect to the shield system $\Ashield_j$ (see \cref{ap:construction_marginal_constraint}). Since Eve's attack only acts on $\Aprime''_j$, the block structure is preserved in the state after her attack, which can be written as
\begin{equation}
\label{eq:block_diagonal_output_state}
    \genDensity_{\Amarg_j\BmeasSquash_j\evePurReg} = \sum_{\blockVar=0}^\tagCutoff \prob_{\blockVarReg_j}(\blockVar) \ketbra{\blockVar}{\blockVar}_{\Ashield_j}\otimes\genDensity_{\Ameas_j\BmeasSquash_j\evePurReg| \Ashield_j=\blockVar} + \sum_{\AclassicalVal\in\AclassicalAlph} \prob_{\blockVarReg_j}(\AclassicalVal) \ketbra{\AclassicalVal}{\AclassicalVal}_{\Ashield_j} \otimes \genDensity_{\Ameas_j\BmeasSquash_j\evePurReg| \Ashield_j=\AclassicalVal}\,,
\end{equation}
where the first term corresponds to the $\blockVar$-photon blocks and the second term corresponds to the flags resulting from applying the tagging source map (\cref{eq:application_tagging_source_map}). Let us consider the marginal constraint in \cref{eq:set_output_states_security_statement}, which is given by $\setAliceMarginalsVirt[j]$, cf. \cref{eq:explicit_construction_marginal_constraint}. We can write the marginal constraint on each block in the shield system $\Ashield_j$ as
\begin{align}
     \frac{\bra{\AclassicalVal}_{\Ameas_j}\genDensity_{\Ameas_j\land\Ashield_j=\blockVar} \ket{\AclassicalVal'}_{\Ameas_j}}{\sqrt{\prob_{\Aclassical_j}(\AclassicalVal)\prob_{\Aclassical_j}(\AclassicalVal')}} &= \sqrt{\prob_{\blockVarReg_j|\Aclassical_j}(\blockVar|\AclassicalVal)\prob_{\blockVarReg_j|\Aclassical_j}(\blockVar|\AclassicalVal')}\braket{\AstateVirt_{\AclassicalVal',\blockVar}^{(j)}}{\AstateVirt_{\AclassicalVal,\blockVar}^{(j)}} \quad \forall \AclassicalVal, \AclassicalVal' \in\AclassicalAlph, \forall \blockVar\leq \characCutoff\,, \label{eq:marginal_for_opti}\\
    \bra{\AclassicalVal'}_{\Ameas_j}\genDensity_{\Ameas_j \land\Ashield_j=\AclassicalVal} \ket{\AclassicalVal''}_{\Ameas_j} &= \bar \delta_{\AclassicalVal\AclassicalVal'}\bar \delta_{\AclassicalVal\AclassicalVal''}\prob_{\Aclassical_j}(\AclassicalVal)\left(1 -\sum_{\blockVar=0}^\tagCutoff \prob_{\blockVarReg_j|\Aclassical_j}(\blockVar|\AclassicalVal)\right) \quad \forall \AclassicalVal,\AclassicalVal', \AclassicalVal'' \in\AclassicalAlph\,, \label{eq:marginal_for_opti_flags}
\end{align}
with Kronecker deltas $\bar \delta_{\AclassicalVal\AclassicalVal'}$, $\bar \delta_{\AclassicalVal\AclassicalVal''}$, and where 
\begin{align}
    \braket{\AstateVirt_{\AclassicalVal,\blockVar}^{(j)}}{\AstateVirt_{\AclassicalVal',\blockVar}^{(j)}} &= \sqrt{\left(1-\srcImpBnd[\AclassicalVal,\blockVar,j]\right)\left(1-\srcImpBnd[\AclassicalVal',\blockVar,j]\right)}\bra{\tilde \Astate_{\AclassicalVal,\blockVar}^{(j)}}\ket{\tilde \Astate_{\AclassicalVal',\blockVar}^{(j)}} + \sqrt{\left(1-\srcImpBnd[\AclassicalVal,\blockVar,j]\right)\srcImpBnd[\AclassicalVal',\blockVar,j]} \bra{\tilde \Astate_{\AclassicalVal,\blockVar}^{(j)}}\ket{\tilde \Astate_{\AclassicalVal',\blockVar}^{(j)\perp}} \nonumber\\
    &+ \sqrt{\srcImpBnd[\AclassicalVal,\blockVar,j]\left(1-\srcImpBnd[\AclassicalVal',\blockVar,j]\right)} \bra{\tilde \Astate_{\AclassicalVal,\blockVar}^{(j)\perp}}\ket{\tilde \Astate_{\AclassicalVal',\blockVar}^{(j)}} + \sqrt{\srcImpBnd[\AclassicalVal,\blockVar,j]\srcImpBnd[\AclassicalVal',\blockVar,j]} \bra{\tilde \Astate_{\AclassicalVal,\blockVar}^{(j)\perp}}\ket{\tilde \Astate_{\AclassicalVal',\blockVar}^{(j)\perp}} \label{eq:marginal_with_overlaps}
\end{align}
with known overlaps $\bra{\tilde \Astate_{\AclassicalVal,\blockVar}^{(j)}}\ket{\tilde \Astate_{\AclassicalVal',\blockVar}^{(j)}}$ and $\bra{\tilde \Astate_{\AclassicalVal,\blockVar}^{(j)}}\ket{\tilde \Astate_{\AclassicalVal,\blockVar}^{(j)\perp}} = 0$ for all $\AclassicalVal\in\AclassicalAlph$, and the remaining overlaps are unknown due to imperfect source characterization (see \cref{sec:problem_reductions}).  We can rewrite the constraint given by \cref{eq:marginal_with_overlaps} in terms of Gram matrices over which we later optimize the unknown overlaps, thereby extending the approach from Ref.~\cite{pereira_optimal_2025} to decoy-state protocols. To achieve this, we introduce the Gram matrices $\gramMatrix^{(\blockVar, j)}$ of the unions
\begin{equation}
    \left\{\sqrt{\prob_{\blockVarReg_j|\Aclassical_j}(\blockVar|\AclassicalVal)}\ket{\tilde \Astate_{\AclassicalVal, \blockVar}^{(j)}}\right\}_{\AclassicalVal\in\AclassicalAlph}\cup\left\{\sqrt{\prob_{\blockVarReg_j|\Aclassical_j}(\blockVar|\AclassicalVal)}\ket{\tilde \Astate_{\AclassicalVal, \blockVar}^{(j)\perp}}\right\}_{\AclassicalVal\in\AclassicalAlph}\,,
\end{equation}
for each block $\blockVar\leq \characCutoff$. The number of setting choices is $|\AclassicalAlph|$. By definition, the following constraints hold
\begin{align}
    \gramMatrix^{(\blockVar, j)} &\succeq 0 \quad \forall \blockVar \in \{0, \ldots,\characCutoff\}\label{eq:gram_constraint_1}\\
    \gramMatrix^{(\blockVar, j)}_{\AclassicalVal, \AclassicalVal'} &= \sqrt{\prob_{\blockVarReg_j|\Aclassical_j}(\blockVar|\AclassicalVal)\prob_{\blockVarReg_j|\Aclassical_j}(\blockVar|\AclassicalVal')}\bra{\tilde \Astate_{\AclassicalVal, \blockVar}^{(j)}}\ket{\tilde \Astate_{\AclassicalVal', \blockVar}^{(j)}} \quad \forall \AclassicalVal, \AclassicalVal'\in\AclassicalAlph, \AclassicalVal\neq \AclassicalVal', \blockVar \in \{0, \ldots,\characCutoff\} \label{eq:gram_constraint_2}  \\ 
    \gramMatrix^{(\blockVar, j)}_{\AclassicalVal, \AclassicalVal} &= \gramMatrix^{(\blockVar, j)}_{\AclassicalVal + |\AclassicalAlph|, \AclassicalVal + |\AclassicalAlph|} = \prob_{\blockVarReg_j|\Aclassical_j}(\blockVar|\AclassicalVal) \quad \forall \AclassicalVal\in\AclassicalAlph, \blockVar \in \{0, \ldots,\characCutoff\} \label{eq:gram_constraint_3}\\ 
    \gramMatrix^{(\blockVar, j)}_{\AclassicalVal + |\AclassicalAlph|, \AclassicalVal} &= \gramMatrix^{(\blockVar, j)}_{\AclassicalVal, \AclassicalVal + |\AclassicalAlph|} = 0 \quad \forall \AclassicalVal\in\AclassicalAlph, \blockVar \in \{0, \ldots,\characCutoff\}\,, \label{eq:gram_constraint_4}
\end{align}
and we can replace the occurrences in \cref{eq:marginal_for_opti} by the corresponding elements of the Gram matrices. Finally, we recall that the photon-number distribution $\prob_{\blockVarReg_j}$ is not exactly known by \cref{as:probability_bounds}. Therefore, we introduce the matrix $\probMatrix^{(\blockVar,j)}=\probMatrixVec_{\blockVar, j}\probMatrixVec_{\blockVar, j}^T$ where $\probMatrixVec_{\blockVar, j} = \sum_{\AclassicalVal\in\AclassicalAlph} \sqrt{\prob_{\blockVarReg_j|\Aclassical_j}(\blockVar|\AclassicalVal)} \hat{e}_\AclassicalVal$ as an optimization variable.\footnote{Technically, $\probMatrix^{(\blockVar,j)}$ is rank one. Since the rank-one constraint is nonconvex, we relax it by requiring only positive semi-definiteness.} It holds that $\probMatrix^{(\blockVar,j)} \succeq 0$, by definition, and by \cref{as:probability_bounds} we have the constraint
\begin{equation}
    \sqrt{\prob_{\blockVarReg_j|\Aclassical_j}^\mathrm{L}(\blockVar|\AclassicalVal)\prob_{\blockVarReg_j|\Aclassical_j}^\mathrm{L}(\blockVar|\AclassicalVal')}\leq \probMatrix^{(\blockVar,j)}_{\AclassicalVal,\AclassicalVal'} \leq \sqrt{\prob_{\blockVarReg_j|\Aclassical_j}^\mathrm{U}(\blockVar|\AclassicalVal)\prob_{\blockVarReg_j|\Aclassical_j}^\mathrm{U}(\blockVar|\AclassicalVal')} \quad \forall \AclassicalVal, \AclassicalVal'\in\AclassicalAlph, \blockVar\in\{0, \ldots, \tagCutoff\}\,. \label{eq:p_matrix_constraints}
\end{equation}
We can therefore replace the occurrences of the probability distribution $\prob_{\blockVarReg_j}$ by the corresponding elements in the matrix $\probMatrix^{(\blockVar,j)}$ and optimize it subject to the constraint given by \cref{eq:p_matrix_constraints}. This motivates defining the Gram matrices using the sub-normalized states as the probability distributions in \cref{eq:gram_constraint_2,eq:gram_constraint_3} can then be replaced by the corresponding entries of $\probMatrix^{(\blockVar,j)}$, while the known overlaps remain fixed, yielding linear constraints.

\subsubsection{Detector and statistics constraints}

\noindent We exploit the block-diagonality by plugging \cref{eq:block_diagonal_output_state} into the flag-state squashing and noise channel constraint 
\begin{equation}
\label{eq:fss_noise_constraint}
    \Tr\left[\widetilde{W}^{\Bmeas_j'}\rho_{\Amarg_j\BmeasSquash_j}\right] \geq \lambdaMin\left(1 - \sum_{\blockVar =0}^{\fssCutoff} \frac{1}{1 - \dtImpBnd[\blockVar, j]} \Tr\left[\Pi_\blockVar^{\BmeasSquash_j}\rho_{\Amarg_j\BmeasSquash_j}\right]\right)
\end{equation}
from \cref{eq:set_output_states_security_statement}, yielding the equivalent formulation in terms of sub-normalized states
\begin{align}
    &\sum_{\blockVar=0}^\tagCutoff \Tr\left[\widetilde{W}^{\Bmeas_j'}\rho_{\BmeasSquash_j\land\Ashield_j=\blockVar}\right] + \sum_{\AclassicalVal \in \AclassicalAlph} \Tr\left[\widetilde{W}^{\Bmeas_j'}\rho_{\BmeasSquash_j\land\Ashield_j=\AclassicalVal}\right] \nonumber\\ 
    &\geq \lambdaMin\left(1 - \sum_{\blockVar =0}^{\fssCutoff} \frac{1}{1 - \dtImpBnd[\blockVar, j]} \left(\sum_{\blockVar'=0}^\tagCutoff \Tr\left[\proj_\blockVar^{\BmeasSquash_j}\rho_{\BmeasSquash_j\land\Ashield_j=\blockVar'}\right] + \sum_{\AclassicalVal \in \AclassicalAlph} \Tr\left[\proj_\blockVar^{\BmeasSquash_j}\rho_{\BmeasSquash_j\land\Ashield_j=\AclassicalVal}\right]\right)\right)\,, \label{eq:detector_constraints}
\end{align}
where Alice's marginal $\Amarg_j$ is traced out as all operators act on $\BmeasSquash_j$. We can further exploit the block-diagonality by plugging \cref{eq:block_diagonal_output_state} into the expression for $\nu_{\varDecReg_j}$ from \cref{eq:def_nu_c_hat}, yielding
\begin{equation}
    \nu_{\varDecReg_j} = \sum_{\bar c \in \interAlph}\sum_{\substack{\AclassicalVal\in \AclassicalAlph \\ \BclassicalVal\in \BclassicalAlph}} \left(\annFunc[j](\bar c|\AclassicalVal,\BclassicalVal)\left(\sum_{\blockVar =0}^\tagCutoff   \Tr[\BpovmelTarg[j] \genDensity_{\BmeasSquash_j\land \Aclassical_j=\AclassicalVal\land\Ashield_j=\blockVar}] + \Tr[\BpovmelTarg[j] \genDensity_{\BmeasSquash_j\land\Ashield_j=\AclassicalVal}]\right)\unitStatsVec[\bar c]\right) \,,\label{eq:expanded_statistics_constraints}
\end{equation}
where we additionally use the fact that Alice's measurement $\{\Apovmel[j]\}_{\AclassicalVal\in \AclassicalAlph}$ is given by $\Apovmel[j] = \ketbra{\AclassicalVal}{\AclassicalVal}_{\Ameas_j}\otimes \identity_{\Ashield_j}$ to trace out Alice's marginal system $\Amarg_j$, and define the sub-normalized post-measurement state $\genDensity_{\BmeasSquash_j\land \Aclassical_j=\AclassicalVal} \coloneqq \Tr_{\Ameas_j}\left[\ketbra{\AclassicalVal}{\AclassicalVal}_{\Ameas_j}\genDensity_{\Ameas_j\BmeasSquash_j}\right]$. We further recall that Bob's POVMs are block diagonal and given by \cref{eq:def_new_virtual_povm_elements}, a property that can also be exploited in the numerical implementation.

\begin{remark}
    The flag-state squashing and noise channel constraint follows from the operator inequality defining the restricted set of attack channels in \cref{eq:marginal_definition_fss_noise}. It therefore holds for every output state of an allowed attack channel. Since conditioning on Alice's shield register (after it is traced out) commutes with the attack channel, we can tighten the optimization by imposing \cref{eq:fss_noise_constraint} separately on the $\blockVar$-photon states $\rho_{\Ameas_j\BmeasSquash_j|\Ashield_j=\blockVar}$ with $\blockVar\leq\characCutoff$.
\end{remark}

\subsubsection{Relaxed decoy constraints}
\label{ap:relaxed_decoy}

\noindent For decoy-state protocols, in the ideal case, Alice’s signal states (conditioned on sending $\blockVar$ photons) do not leak information about the intensity choice. However, due to source imperfections, cf. \cref{as:block_diagonal_source,as:imperfectly_characterized_source}, the signal states may partially leak this information (e.g. due to a Trojan-horse attack, cf. \cref{ap:trojan_horse_attack}). We therefore relax the usual decoy constraint to the case where the intensity leakage is bounded. For non-decoy protocols (or whenever only a single intensity is used), the constraints discussed in this section become trivial since the set of possible intensity choices $\intensityAlph$ contains only one element. In this way, both cases are handled uniformly. For simplicity, in the following we assume that the target states chosen in \cref{as:imperfectly_characterized_source} do not leak information about Alice’s intensity choice, i.e.
\begin{equation}
\label{eq:overlap_target_intensity_leakage}
    \left|\braket{\tilde \Astate_{\AclassicalVal_a\AclassicalVal_\mu,\blockVar}^{(j)}}{\tilde \Astate_{\AclassicalVal_a\AclassicalVal_{\mu'},\blockVar}^{(j)}}\right|^2 = 1
\end{equation}
for all $\AclassicalVal_a\in \tilde{\AclassicalAlph}$, $\AclassicalVal_\mu, \AclassicalVal_{\mu'}\in\intensityAlph$ and $\blockVar\leq \tagCutoff$. Note, however, that the argument can easily be extended to the case where only a lower bound on \cref{eq:overlap_target_intensity_leakage} is known, since the target states are explicitly known (see \cref{eq:relaxed_decoy_block} where this scenario is considered).
Then, by \cref{eq:definition_tau_x_m}, it follows that the overlap of the signal states for different intensities is bounded by
\begin{equation}
\label{eq:intensity_leakage_bound_overlap}
    \left|\braket{\AstateVirt_{\AclassicalVal_a\AclassicalVal_\mu,\blockVar}^{(j)}}{\AstateVirt_{\AclassicalVal_a\AclassicalVal_{\mu'},\blockVar}^{(j)}}\right|^2 \geq 1 - \left(\sqrt{\srcImpBnd[\AclassicalVal_a,\AclassicalVal_\mu,\blockVar,j]} + \sqrt{\srcImpBnd[\AclassicalVal_a,\AclassicalVal_{\mu'},\blockVar,j]}\right)^2
\end{equation}
for all $\AclassicalVal_a\in \tilde{\AclassicalAlph}$, $\AclassicalVal_\mu, \AclassicalVal_{\mu'}\in\intensityAlph$ and $\blockVar\leq \tagCutoff$. Recall that we write $\AclassicalVal = (\AclassicalVal_a, \AclassicalVal_\mu)$. This expression bounds how much information Alice’s signal states leak about her intensity choice. After applying the source maps from \cref{sec:problem_reductions}, the data-processing inequality for the fidelity implies that the fidelity between states forwarded to Bob with different intensities is bounded, 
\begin{equation}
\label{eq:fidelity_relaxed_decoy}
    F(\genDensity_{\BmeasSquash_j| \Aclassical_j=\AclassicalVal_a\AclassicalVal_\mu\land\Ashield_j=\blockVar}, \genDensity_{\BmeasSquash_j| \Aclassical_j=\AclassicalVal_a\AclassicalVal_{\mu'} \land \Ashield_j=\blockVar}) \geq 1 - \zeta_{\AclassicalVal_a,\AclassicalVal_\mu, \AclassicalVal_{\mu'},\blockVar,j}
\end{equation}
for all $\AclassicalVal_a\in \tilde{\AclassicalAlph}$, $\AclassicalVal_\mu, \AclassicalVal_{\mu'}\in\intensityAlph$ and $\blockVar\leq \tagCutoff$, where we define $\zeta_{\AclassicalVal_a,\AclassicalVal_\mu, \AclassicalVal_{\mu'},\blockVar,j} \coloneqq \left(\sqrt{\srcImpBnd[\AclassicalVal_a,\AclassicalVal_\mu,\blockVar,j]} + \sqrt{\srcImpBnd[\AclassicalVal_a,\AclassicalVal_{\mu'},\blockVar,j]}\right)^2$ for convenience. As we later optimize over sub-normalized states $\rho_{\Ameas_j\BmeasSquash_j\land\Ashield_j=\blockVar}$, we cannot directly write the constraint above as the probability distribution $\prob_{\blockVarReg_j}$ is not exactly known, following \cref{as:probability_bounds} (and we can therefore not relate the sub-normalized states to the normalized ones). However, following \cref{eq:fidelity_relaxed_decoy}, the statistics of the states with different intensity choices are related by
\begin{align}
    &  \frac{\Tr\left[\BpovmelTarg[j]\genDensity_{\BmeasSquash_j\land\Aclassical_j=\AclassicalVal_a\AclassicalVal_\mu\land \Ashield_j=\blockVar}\right]}{\prob_{\blockVarReg_j,\Aclassical_j}^{\mathrm{L}}(\blockVar,\AclassicalVal_a\AclassicalVal_\mu)} \geq \frac{\Tr\left[\BpovmelTarg[j]\genDensity_{\BmeasSquash_j\land\Aclassical_j=\AclassicalVal_a \AclassicalVal_\mu'\land\Ashield_j=\blockVar }\right]}{\prob_{\blockVarReg_j,\Aclassical_j}^{\mathrm{U}}(\blockVar,\AclassicalVal_a\AclassicalVal_\mu')} - \sqrt{\zeta_{\AclassicalVal_a,\AclassicalVal_\mu, \AclassicalVal_{\mu'},\blockVar,j}} \label{eq:relaxed_decoy_1}\\
    &  \frac{\Tr\left[\proj_{\blockVar'}^{\BmeasSquash_j}\genDensity_{\BmeasSquash_j\land\Aclassical_j=\AclassicalVal_a\AclassicalVal_\mu\land \Ashield_j=\blockVar}\right]}{\prob_{\blockVarReg_j,\Aclassical_j}^{\mathrm{L}}(\blockVar,\AclassicalVal_a\AclassicalVal_\mu)} \geq \frac{\Tr\left[\proj_{\blockVar'}^{\BmeasSquash_j}\genDensity_{\BmeasSquash_j\land\Aclassical_j=\AclassicalVal_a \AclassicalVal_\mu'\land\Ashield_j=\blockVar }\right]}{\prob_{\blockVarReg_j,\Aclassical_j}^{\mathrm{U}}(\blockVar,\AclassicalVal_a\AclassicalVal_\mu')} - \sqrt{\zeta_{\AclassicalVal_a,\AclassicalVal_\mu, \AclassicalVal_{\mu'},\blockVar,j}} \label{eq:relaxed_decoy_2}
\end{align}
for all $\BclassicalVal\in\BclassicalAlph, \blockVar\in\{0, \ldots,\characCutoff\},  \blockVar'\in\{0, \ldots,\fssCutoff\}, \AclassicalVal_a\in \tilde{\AclassicalAlph}, \AclassicalVal_\mu, \AclassicalVal_\mu' \in \intensityAlph$, where $\prob_{\blockVarReg_j, \Aclassical_j}^\mathrm{U/L}(\blockVar, \AclassicalVal) = \prob_{\Aclassical_j}(\AclassicalVal)\prob_{\blockVarReg_j|\Aclassical_j}^\mathrm{U/L}(\blockVar|\AclassicalVal)$ and $\prob_{\blockVarReg_j| \Aclassical_j}^\mathrm{U/L}(\blockVar|\AclassicalVal)$ is known by \cref{as:probability_bounds}. 

The source imperfections from \cref{as:imperfectly_characterized_source} are allowed to leak information about Alice's intensity choice, which results in the relaxed decoy constraints. If, instead, one assumes in \cref{as:imperfectly_characterized_source} that the source imperfections do not leak any information about the intensity choice, then \cref{eq:fidelity_relaxed_decoy} becomes trivial and the relaxed decoy constraints reduce to the standard decoy constraints,
\begin{equation}
    \genDensity_{\BmeasSquash_j| \Aclassical_j=\AclassicalVal_a\AclassicalVal_\mu\land\Ashield_j=\blockVar} = \genDensity_{\BmeasSquash_j| \Aclassical_j=\AclassicalVal_a\AclassicalVal_{\mu'}\land\Ashield_j=\blockVar}
\end{equation}
for all $\AclassicalVal_a\in \tilde{\AclassicalAlph}$, $\AclassicalVal_\mu, \AclassicalVal_{\mu'}\in\intensityAlph$, and $\blockVar\leq \tagCutoff$.

\begin{remark}
    In practice, a source imperfection may affect only certain degrees of freedom and need not leak any information about the intensity choice. This additional structure could be incorporated by introducing fidelity bounds between signal states of different intensities, thereby distinguishing between intensity-leaking and non-intensity-leaking source imperfections, but we do not treat this scenario explicitly in this work.
\end{remark}

\subsubsection{Simplifications}
\label{ap:simplifications}
\noindent In this section, we perform simplifications to reduce the numerical complexity and improve numerical stability, while only slightly affecting the tightness, similarly to \cite[Sec.~IX.A]{kamin_renyi_2025}. For simplicity, let the outcome $\widetilde{W}^{\Bmeas_j'}$ be one of Bob’s POVM elements, i.e. there exists $\BclassicalVal\in\BclassicalAlph$ such that $\BpovmelTarg[j]=\widetilde{W}^{\Bmeas_j'}$. We make use of the cutoff $\characCutoff \leq \tagCutoff$ to replace the detection statistics with yields. In particular, for all $\blockVar\in\{\characCutoff+1, \ldots, \tagCutoff\}$, we have
\begin{equation}
\label{eq:yield_introduction}
    \Tr[\BpovmelTarg[j] \genDensity_{\BmeasSquash_j| \Aclassical_j=\AclassicalVal_a\AclassicalVal_\mu\land\Ashield_j=\blockVar}] = \prob_{\BclassicalReg_j|\blockVarReg_j,\Aclassical_j}(\BclassicalVal|\blockVar, \AclassicalVal_a,\AclassicalVal_\mu) \eqqcolon (\yield^{\AclassicalVal_a, \AclassicalVal_\mu, \blockVar})_\BclassicalVal\,
\end{equation}
where, following the relaxed decoy constraint (\cref{eq:fidelity_relaxed_decoy}), the deviation between yields for different intensity choices is bounded as
\begin{equation}
    \left|(\yield^{\AclassicalVal_a, \AclassicalVal_\mu, \blockVar})_\BclassicalVal - (\yield^{\AclassicalVal_a, \AclassicalVal_{\mu'}, \blockVar})_\BclassicalVal\right| \leq \sqrt{\zeta_{\AclassicalVal_a,\AclassicalVal_\mu, \AclassicalVal_{\mu'},\blockVar,j}}\,.
\end{equation}
Additionally, we can bound the first sum in the r.h.s. of \cref{eq:detector_constraints} using
\begin{equation}
\label{eq:bound_with_q_max}
    \sum_{\blockVar=0}^\fssCutoff\frac{1}{1 - \dtImpBnd[\blockVar, j]}\sum_{\blockVar'=\characCutoff+1}^\tagCutoff \Tr\left[\proj_\blockVar^{\BmeasSquash_j}\rho_{\BmeasSquash_j\land\Ashield_j=\blockVar'}\right] \leq \sum_{\blockVar'=\characCutoff+1}^\tagCutoff \sum_{\AclassicalVal\in\AclassicalAlph} \prob_{\Aclassical_j}(\AclassicalVal)\frac{\probMatrix^{(\blockVar',j)}_{\AclassicalVal,\AclassicalVal}}{1 - \dtImpBndMax[j]}\,,
\end{equation}
where $\dtImpBndMax[j] \coloneqq \max_{\blockVar\in\{0,\ldots,\fssCutoff\}} \dtImpBnd[\blockVar, j]$. \cref{eq:yield_introduction,eq:bound_with_q_max} allow us to replace all terms involving $\rho_{\BmeasSquash_j\land\Ashield_j=\blockVar}$ for $\blockVar\in\{\characCutoff+1, \ldots, \tagCutoff\}$ by the corresponding bounds.

\begin{remark}
    Since the probability that Alice sends a flag is typically small, especially for large $\tagCutoff$, one can further simplify the optimization problem by using
    \begin{equation}
        \Tr[\genDensity_{\BmeasSquash_j\land\Ashield_j=\AclassicalVal}] = \prob_{\Aclassical_j}(\AclassicalVal)\left(1 -\sum_{\blockVar=0}^\tagCutoff \probMatrix^{(\blockVar,j)}_{\AclassicalVal,\AclassicalVal}\right)
    \end{equation}
    and therefore bound any trace involving $\genDensity_{\BmeasSquash_j\land \Ashield_j=\AclassicalVal}$. This removes the corresponding optimization variable with only a small penalty.
\end{remark}

\subsubsection{Convex optimization problem}

\begin{theorem}[Convex optimization problem for \cref{th:security_imperfect_devices}]
\label{th:convex_opti_numerics_imperfect}
    Consider the scenario in \cref{th:security_imperfect_devices}. Let $\characCutoff\leq \tagCutoff$ be a pre-defined cutoff and let the target states satisfy Eq.~\eqref{eq:overlap_target_intensity_leakage}. Then 
    \begin{equation}
        \kappaQKDVirt_j\left(\ftradeoff{1}{j-1},\setAliceMarginalsVirt[j],\gMapVirt[j]\right)  \leq \inf_{\genDensity_{\Amarg_j\BmeasSquash_j}\in\widehat \stateSet_j(\setAliceMarginalsVirt[j])}\fRenyiUpEnt{1}{j-1}(\secretReg_j|\eveCopyReg_j\varDecReg_j\evePurReg)_{\gMapVirt[j]\left[\purFunc(\genDensity_{\Amarg_j\BmeasSquash_j})\right]}
    \end{equation}
    holds for any $\kappaQKDVirt_j\left(\ftradeoff{1}{j-1},\setAliceMarginalsVirt[j],\gMapVirt[j]\right) $ chosen according to the convex optimization problem described in \cref{fig:opti_problem_final}, where
    \begin{equation}
        \renyiMblock(\nu_{|\mathtt{gen}}) \coloneqq \frac{1}{1-\alpha}\log\left(\sum_{\blockVar=0}^\characCutoff \left(2^{(1-\alpha) \renyidown(\secretReg_j|\eveCopyReg_j\evePurReg)_{\nu_{\land \blockVarReg_j=\blockVar | \mathtt{gen}}}} - \Tr[\rho_{\Ameas_j\BmeasSquash_j\land\Ashield_j=\blockVar|\mathtt{gen}}]\right) + 1\right) 
    \end{equation}
    and
    \begin{equation}
    \renyidown(\secretReg_j|\eveCopyReg_j\evePurReg)_{\nu_{\land\blockVarReg_j=\blockVar|\mathtt{gen}}} = \frac{1}{1-\alpha}\log\Tr\left[\left(\Tr_{\secretReg_j}\left[\Tr_{\eveCopyReg_j}\left[\left(V_j \genDensity_{\Ameas_j\BmeasSquash_j\land\Ashield_j= \blockVar|\mathtt{gen}} V_j^\dagger \right)\right]^{\frac{1}{\alpha}}\right]\right)^\alpha\right]\,.
    \end{equation}
    
\end{theorem}
\begin{proof}
    The statement follows from a combination of arguments from \cref{ap:expanding_objective,sec:expanding_constraints}. We obtain the objective function from \cref{ap:expanding_objective}, i.e. \cref{eq:expanded_entropy_f,eq:expand_entropy_block}, where we replace the states conditioned on the block label $\blockVar$ by sub-normalized states instead. The source constraints follow from \cref{eq:marginal_for_opti} where the Gram matrix is substituted and the Gram matrix constraints are given by \cref{eq:gram_constraint_1,eq:gram_constraint_2,eq:gram_constraint_3,eq:gram_constraint_4}. The imperfect photon-number distribution is captured by $\probMatrix^{(\blockVar,j)}$ as defined in \cref{sec:expanding_constraints} with constraints given by \cref{eq:p_matrix_constraints}.
    
    The detector constraints are given by \cref{eq:detector_constraints} where we substituted the yields from \cref{eq:yield_introduction} for Alice's blocks $\blockVar\in\{\characCutoff+1, \ldots, \tagCutoff\}$ and the simplification from \cref{eq:bound_with_q_max}. Similarly, the statistics constraints are given by \cref{eq:expanded_statistics_constraints} where we substituted the yields for blocks $\blockVar\in\{\characCutoff+1, \ldots, \tagCutoff\}$. Finally, the relaxed decoy constraints follow from \cref{ap:relaxed_decoy}, more specifically \cref{eq:relaxed_decoy_1,eq:relaxed_decoy_2,eq:yield_introduction}. 

    The feasible set in \cref{fig:opti_problem_final} is convex, since all constraints are affine equalities or inequalities, and positive-semidefinite constraints. The convexity of the objective function follows from \cite[App.~F]{kamin_renyi_2025}. Therefore, the optimization problem is convex.
    
\end{proof}

\begin{figure}[!htbp]
        \centering
        \caption{Convex optimization problem for \cref{th:convex_opti_numerics_imperfect}.}
        \vspace*{-0.2cm}
        \label{fig:opti_problem_final}
    \begin{constraintblock}{Objective}
    \[
    \begin{aligned}
        & \mathrm{minimize} \quad \frac{\alpha}{1-\alpha} \log\left((1-p_\mathtt{test}^{(j)}) 2^{\frac{1-\alpha}{\alpha} \left(\renyiMblock(\nu_{|\mathtt{gen}}) - f(\varDecGen)\right)}+ \sum_{\varDecVal \in \varDecAlph_\mathtt{test}} \nu(\varDecVal) 2^{\frac{\alpha - 1}{\alpha}f(\varDecVal)} \right)
    \end{aligned}
    \]
    \end{constraintblock}

    \begin{constraintblock}{Optimization variables}
    \[
    \begin{aligned}
        &\{\rho_{\Ameas_j\BmeasSquash_j\land\Ashield_j=\blockVar} \in \setDensity_\leq(\Ameas_j\BmeasSquash_j)\}_{\blockVar=0}^\characCutoff & \{\gramMatrix^{(\blockVar, j)} \in \setPos(2|\AclassicalAlph|)\}_{\blockVar \in \{0, \ldots,\characCutoff\}} & \qquad\nu_{\varDecReg_j} \in \probSimplex_{|\varDecAlph|}  \\ 
        & \{\rho_{\BmeasSquash_j\land\Ashield_j=\AclassicalVal}\in \setDensity_\leq(\BmeasSquash_j)\}_{\AclassicalVal\in\AclassicalAlph} & \{\yield^{\AclassicalVal,\blockVar}\in \probSimplex_{|\BclassicalAlph|}\}_{\AclassicalVal\in\AclassicalAlph, \blockVar \in \{\characCutoff+1, \ldots, \tagCutoff\}} & \qquad\{\probMatrix^{(\blockVar,j)} \in \setPos(|\AclassicalAlph|)\}_{\blockVar\in\{0,\ldots,\tagCutoff\}} &
    \end{aligned}
    \]
    \end{constraintblock}

    \begin{constraintblock}{Source constraints}
    \[
    \begin{aligned}
        & \Tr[\genDensity_{\BmeasSquash_j\land \Ashield_j=\AclassicalVal}] = \prob_{\Aclassical_j}(\AclassicalVal)\left(1 -\sum_{\blockVar=0}^\tagCutoff \probMatrix^{(\blockVar,j)}_{\AclassicalVal,\AclassicalVal}\right) & \forall \AclassicalVal\in\AclassicalAlph\\
        &\frac{\bra{\AclassicalVal'}_{\Ameas_j}\genDensity_{\Ameas_j\land\Ashield_j=\blockVar} \ket{\AclassicalVal}_{\Ameas_j}}{\sqrt{\prob_{\Aclassical_j}(\AclassicalVal)\prob_{\Aclassical_j}( \AclassicalVal')}} = \sqrt{\left(1-\srcImpBnd[\AclassicalVal,\blockVar,j]\right)\left(1-\srcImpBnd[\AclassicalVal',\blockVar,j]\right)}\gramMatrix^{(\blockVar, j)}_{\AclassicalVal, \AclassicalVal'} + \sqrt{\left(1-\srcImpBnd[\AclassicalVal,\blockVar,j]\right)\srcImpBnd[\AclassicalVal',\blockVar,j]} \gramMatrix^{(\blockVar, j)}_{\AclassicalVal, \AclassicalVal'+|\AclassicalAlph|} \hspace*{-5cm} \nonumber \\
        &+ \sqrt{\srcImpBnd[\AclassicalVal,\blockVar,j]\left(1-\srcImpBnd[\AclassicalVal',\blockVar,j]\right)} \gramMatrix^{(\blockVar, j)}_{\AclassicalVal+|\AclassicalAlph|, \AclassicalVal'} + \sqrt{\srcImpBnd[\AclassicalVal,\blockVar,j]\srcImpBnd[\AclassicalVal',\blockVar,j]} \gramMatrix^{(\blockVar, j)}_{\AclassicalVal+|\AclassicalAlph|, \AclassicalVal'+|\AclassicalAlph|}&\forall \AclassicalVal, \AclassicalVal'\in\AclassicalAlph, \blockVar \in \{0, \ldots,\characCutoff\}\nonumber \\
        & \gramMatrix^{(\blockVar, j)}_{\AclassicalVal, \AclassicalVal'} = \probMatrix^{(\blockVar,j)}_{\AclassicalVal,\AclassicalVal'}\bra{\tilde \Astate_{\AclassicalVal, \blockVar}^{(j)}}\ket{\tilde \Astate_{\AclassicalVal', \blockVar}^{(j)}} &\forall \AclassicalVal, \AclassicalVal'\in\AclassicalAlph, \AclassicalVal\neq \AclassicalVal', \blockVar \in \{0, \ldots,\characCutoff\} \nonumber  \\ 
        & \gramMatrix^{(\blockVar, j)}_{\AclassicalVal, \AclassicalVal} = \gramMatrix^{(\blockVar, j)}_{\AclassicalVal + |\AclassicalAlph|, \AclassicalVal + |\AclassicalAlph|} = \probMatrix^{(\blockVar,j)}_{\AclassicalVal,\AclassicalVal} &\forall \AclassicalVal\in\AclassicalAlph, \blockVar \in \{0, \ldots,\characCutoff\}\nonumber \\ 
        & \gramMatrix^{(\blockVar, j)}_{\AclassicalVal + |\AclassicalAlph|, \AclassicalVal} = \gramMatrix^{(\blockVar, j)}_{\AclassicalVal, \AclassicalVal + |\AclassicalAlph|} = 0 &\forall \AclassicalVal\in\AclassicalAlph, \blockVar \in \{0, \ldots,\characCutoff\}\nonumber \\
        & \probMatrix^{(\blockVar,j)}_{\AclassicalVal,\AclassicalVal'} \mathrel{\substack{\leq\\\geq}} \sqrt{\prob_{\blockVarReg_j|\Aclassical_j}^{\mathrm{U}/ \mathrm{L}}(\blockVar|\AclassicalVal)\prob_{\blockVarReg_j|\Aclassical_j}^{\mathrm{U}/\mathrm{L}}(\blockVar|\AclassicalVal')} &\forall \AclassicalVal, \AclassicalVal'\in\AclassicalAlph, \blockVar\in\{0, \ldots, \tagCutoff\} \nonumber
    \end{aligned}
    \]
    \end{constraintblock}

    \begin{constraintblock}{Detector constraints}
    \[
    \begin{aligned}
        &\sum_{\blockVar=0}^\characCutoff \Tr\left[\widetilde{W}^{\Bmeas_j'}\rho_{\BmeasSquash_j\land\Ashield_j=\blockVar} \right] + \sum_{\blockVar=\characCutoff+1}^\tagCutoff \sum_{\AclassicalVal\in\AclassicalAlph} \prob_{\Aclassical_j, \blockVarReg_j}^\mathrm{U}(\AclassicalVal, \blockVar) (\yield^{\AclassicalVal, \blockVar})_{\BclassicalVal=o}  + \sum_{\AclassicalVal \in \AclassicalAlph} \Tr\left[\widetilde{W}^{\Bmeas_j'}\rho_{\BmeasSquash_j\land\Ashield_j=\AclassicalVal}\right] \nonumber \\ 
        &\geq \lambdaMin\left(1 - \sum_{\blockVar =0}^{\fssCutoff} \frac{1}{1 - \dtImpBnd[\blockVar, j]} \left(\sum_{\blockVar'=0}^\characCutoff \Tr\left[\proj_\blockVar^{\BmeasSquash_j}\rho_{\BmeasSquash_j\land\Ashield_j=\blockVar'}\right]  + \sum_{\AclassicalVal \in \AclassicalAlph} \Tr\left[\proj_\blockVar^{\BmeasSquash_j}\rho_{\BmeasSquash_j\land\Ashield_j=\AclassicalVal}\right]\right) - \sum_{\blockVar'=\characCutoff+1}^\tagCutoff \sum_{\AclassicalVal\in\AclassicalAlph} \prob_{\Aclassical_j}(\AclassicalVal)\frac{\probMatrix^{(\blockVar',j)}_{\AclassicalVal,\AclassicalVal}}{1 - \dtImpBndMax[j]} \right)\nonumber 
    \end{aligned}
    \]
    \end{constraintblock}

    \begin{constraintblock}{Statistics constraints}
    \[
    \begin{aligned}
        \nu_{\varDecReg_j} \mathrel{\substack{\leq\\\geq}} \sum_{\bar c \in \interAlph}\sum_{\substack{\AclassicalVal\in \AclassicalAlph \\ \BclassicalVal\in \BclassicalAlph}} \Bigg(&\annFunc[j](\bar c|\AclassicalVal,\BclassicalVal)\Bigg(\sum_{\blockVar =0}^\characCutoff \Tr[\BpovmelTarg[j] \genDensity_{\BmeasSquash_j\land\Aclassical_j=\AclassicalVal\land \Ashield_j=\blockVar}]    \\
        &+ \sum_{\blockVar =\characCutoff+1}^\tagCutoff \prob_{\Aclassical_j, \blockVarReg_j}^{\mathrm{U}/\mathrm{L}}(\AclassicalVal, \blockVar) (\yield^{\AclassicalVal, \blockVar})_{\BclassicalVal}+\Tr[\BpovmelTarg[j] \genDensity_{\BmeasSquash_j\land \Ashield_j=\AclassicalVal}] \Bigg)\unitStatsVec[\bar c]\Bigg)
    \end{aligned}
    \]
    \end{constraintblock}

    \begin{constraintblock}{Decoy constraints}
    \[
    \begin{aligned}
        & \left|(\yield^{\AclassicalVal_a, \AclassicalVal_\mu, \blockVar})_{\BclassicalVal} - (\yield^{\AclassicalVal_a, \AclassicalVal_{\mu'}, \blockVar})_{\BclassicalVal}\right| \leq \sqrt{\zeta_{\AclassicalVal_a,\AclassicalVal_\mu, \AclassicalVal_{\mu'},\blockVar,j}} & \hspace*{-4cm}\forall \BclassicalVal\in\BclassicalAlph, \blockVar\in\{\characCutoff+1, \ldots,\tagCutoff\}, \AclassicalVal_a\in\tilde{\AclassicalAlph}, \AclassicalVal_\mu, \AclassicalVal_\mu' \in \intensityAlph \nonumber \\
        & \frac{\Tr\left[\BpovmelTarg[j]\genDensity_{\BmeasSquash_j\land\Aclassical_j=\AclassicalVal_a\AclassicalVal_\mu\land \Ashield_j=\blockVar}\right]}{\prob_{\blockVarReg_j,\Aclassical_j}^{\mathrm{L}}(\blockVar,\AclassicalVal_a\AclassicalVal_\mu)} \geq \frac{\Tr\left[\BpovmelTarg[j]\genDensity_{\BmeasSquash_j\land\Aclassical_j=\AclassicalVal_a \AclassicalVal_\mu'\land\Ashield_j=\blockVar }\right]}{\prob_{\blockVarReg_j,\Aclassical_j}^{\mathrm{U}}(\blockVar,\AclassicalVal_a\AclassicalVal_\mu')} - \sqrt{\zeta_{\AclassicalVal_a,\AclassicalVal_\mu, \AclassicalVal_{\mu'},\blockVar,j}}  \\ &\forall \BclassicalVal\in\BclassicalAlph, \blockVar\in\{0, \ldots,\characCutoff\}, \AclassicalVal_a\in \tilde{\AclassicalAlph}, \AclassicalVal_\mu, \AclassicalVal_\mu' \in \intensityAlph \nonumber \\
        & \frac{\Tr\left[\proj_{\blockVar'}^{\BmeasSquash_j}\genDensity_{\BmeasSquash_j\land\Aclassical_j=\AclassicalVal_a\AclassicalVal_\mu\land \Ashield_j=\blockVar}\right]}{\prob_{\blockVarReg_j,\Aclassical_j}^{\mathrm{L}}(\blockVar,\AclassicalVal_a\AclassicalVal_\mu)} \geq \frac{\Tr\left[\proj_{\blockVar'}^{\BmeasSquash_j}\genDensity_{\BmeasSquash_j\land\Aclassical_j=\AclassicalVal_a \AclassicalVal_\mu'\land\Ashield_j=\blockVar }\right]}{\prob_{\blockVarReg_j,\Aclassical_j}^{\mathrm{U}}(\blockVar,\AclassicalVal_a\AclassicalVal_\mu')} - \sqrt{\zeta_{\AclassicalVal_a,\AclassicalVal_\mu, \AclassicalVal_{\mu'},\blockVar,j}}  \\
        &\forall \blockVar'\in\{0, \ldots,\fssCutoff\}, \blockVar\in\{0, \ldots,\characCutoff\}, \AclassicalVal_a\in \tilde{\AclassicalAlph}, \AclassicalVal_\mu, \AclassicalVal_\mu' \in \intensityAlph
    \end{aligned}
    \]
    \end{constraintblock}
    \end{figure}

\section{Decoy-state BB84 with active basis choice}
\label{ap:active_basis_choice}
\noindent Subspace weight estimation is a necessary step when applying the flag-state squasher (see \cref{lem:flag_state_squasher}). However, for active basis choice on the detector side, determining an outcome \(W^{\Bmeas_j}\) and a constant \(\lambdaMin>0\) satisfying \(W^{\Bmeas_j}\geq\lambdaMin\proj_{>\fssCutoff}^{\Bmeas_j}\), where \(\proj_{>\fssCutoff}^{\Bmeas_j}\) denotes the projector onto the subspace with \(\blockVar>\fssCutoff\), remains an open problem. Once such a subspace weight estimate is available, the flag-state squasher can also be used for active setups.

Instead, here we describe an alternative approach using the \textit{simple squasher} \cite[Theorem~10]{PhysRevA.89.012325} to easily incorporate active detection setups in \cref{th:security_imperfect_devices} without subspace weight estimation. The simple squasher applies to ideal detectors. We therefore first apply the noise channel construction from \cref{lem:noise_channel}, which reduces the analysis to ideal active BB84 POVMs, and subsequently apply the simple squasher. Importantly, the simple squasher requires every double-click event to be mapped uniformly at random to one of the two bit outcomes in the selected basis. This approach to active detection setups is also discussed in \cite[Technical Aside~6.3 and App.~D.2]{nahar_proof-technique-independent_2026}.

\begin{lemma}
\label{lem:combination_noise_channel_simple_squasher}
Consider the decoy-state BB84 protocol $\big\{\big\{\Astate^{(j)}_{\Aclassical_j\Aprime_j},\{\Bpovmel[j]\}_{\BclassicalVal\in\BclassicalAlph},\annKeyMap[j]\big\}_{j=1}^{\totRounds},\ppMap\big\}$, where every double-click event is mapped uniformly at random to one of the two bit outcomes in the selected basis. Bob's measurement POVMs are block diagonal in Fock space, i.e. $\Bpovmel[j]=\oplus_{\blockVar=0}^{\infty}\BpovmelBlock{j}{\blockVar}$, and Bob performs the basis choice actively\footnote{Without loss of generality, we absorb any common loss for the POVM elements $\Bpovmel[j]$ into the quantum channel.}. Let $\{\BpovmelBlockTilde{j}{\blockVar}\}_{j,\BclassicalVal\in\BclassicalAlph,\blockVar}$ be the ideal active BB84 POVMs\footnote{The ideal POVMs describe unit efficiency threshold detectors without dark counts (see \cite[Sec.~V]{PhysRevA.89.012325}).} and assume that there exist $0\leq\dtImpBnd[0,j]<1$ and $0\leq\dtImpBnd[\geq1,j]<1$ such that
\begin{align}
    \BpovmelBlock{j}{0}-(1-\dtImpBnd[0,j])\BpovmelBlockTilde{j}{0}&\geq0\,,\label{eq:active_bb84_noise_vacuum}\\
    \BpovmelBlock{j}{\blockVar}-(1-\dtImpBnd[\geq1,j])\BpovmelBlockTilde{j}{\blockVar}&\geq0\,,\label{eq:active_bb84_noise_nonvacuum}
\end{align}
for all $\blockVar\geq1$ and all $\BclassicalVal\in\BclassicalAlph$. Define a new set of POVMs $\{\BpovmelTarg[j]\}_{j,\BclassicalVal\in\BclassicalAlph}$ where
\begin{equation}
\label{eq:target_povm_active_bb84}
    \BpovmelTarg[j]\coloneqq\BpovmelBlockTilde{j}{0}\oplus\BpovmelBlockTilde{j}{1}\oplus\ketbra{\BclassicalVal}{\BclassicalVal}_{\flagSpaceReg}\,.
\end{equation}
Then there exist squashing maps $\squashMap_j\in\cptp(\Bmeas_j, \BmeasSquash_j)$ such that
    \begin{align}
    \label{eq:squash_condition_active_bb84}
        \squashMap_j^\dagger \left[\BpovmelTarg[j]\right] &= \Bpovmel[j]\,,
    \end{align}
    and if the restricted attack channel is constructed as
    \begin{align}
    \label{eq:attack_channel_noise_simple_squasher}
         \setAttackCh_j = \left\{\attackCh_j \in \cptp(\Ereg_{j-1}, \BmeasSquash_j \Ereg_j) : \frac{1}{1 - \dtImpBnd[0, j]}\attackCh_j^\dagger\left[\proj_0^{\BmeasSquash_j} \otimes \identity_{\Ereg_j} \right] + \frac{1}{1 - \dtImpBnd[\geq1, j]}\attackCh_j^\dagger\left[\proj_1^{\BmeasSquash_j} \otimes \identity_{\Ereg_j} \right] = \identity_{\Ereg_{j-1}}\right\}
    \end{align}
    then
    \begin{equation}
        \setAttackCh_j \supset \squashMap_j \circ \cptp(\Ereg_{j-1}, \Bmeas_j \Ereg_j)
    \end{equation}
    for all rounds $j$.
\end{lemma}

\begin{proof}
Let $\noiseChMap_j$ denote the noise channel obtained from \cref{lem:noise_channel} using the ideal active BB84 POVMs as target POVMs. By \cref{eq:active_bb84_noise_vacuum,eq:active_bb84_noise_nonvacuum}, it satisfies $\noiseChMap_j^\dagger[\proj_0]=(1-\dtImpBnd[0,j])\proj_0^{\Bmeas_j}$ and $\noiseChMap_j^\dagger[\proj_{\geq1}]=(1-\dtImpBnd[\geq1,j])\proj_{\geq1}^{\Bmeas_j}$. The simple squasher \cite[Theorem~10]{PhysRevA.89.012325} maps the ideal active BB84 optical POVMs to the vacuum-plus-qubit POVMs $\BpovmelBlockTilde{j}{0}\oplus\BpovmelBlockTilde{j}{1}$, acts as the identity on the vacuum and single-photon subspace, and maps every higher photon-number subspace to the qubit subspace. Extend this map trivially on the flag space, then the composite map $\squashMap_j$ satisfies \cref{eq:squash_condition_active_bb84} and $\squashMap_j^\dagger[\proj_0^{\BmeasSquash_j}]=(1-\dtImpBnd[0,j])\proj_0^{\Bmeas_j}$, $\squashMap_j^\dagger[\proj_1^{\BmeasSquash_j}]=(1-\dtImpBnd[\geq1,j])\proj_{\geq1}^{\Bmeas_j}$, from which \cref{eq:attack_channel_noise_simple_squasher} follows.
\end{proof}

If the noise channel and simple squasher construction from \cref{lem:combination_noise_channel_simple_squasher} is used to handle active detection setups, the final optimization problem in \cref{fig:opti_problem_final} changes as follows:
\begin{enumerate}
    \item Replace Bob's POVMs by the target POVMs in \cref{eq:target_povm_active_bb84} and map double-click events uniformly at random to one of the two bit outcomes in the selected basis.
    \item Replace the detector constraint by the constraints induced by \cref{eq:attack_channel_noise_simple_squasher} (analogously to \cref{sec:construction_set_output_states}). Since Alice's state is block-diagonal in $\Ashield_j$, these constraints can be imposed separately on every block,
    \begin{align}
    \Tr[K_j\rho_{\Ameas_j\BmeasSquash_j\land\Ashield_j=\blockVar}] &= \Tr[\rho_{\Ameas_j\BmeasSquash_j\land\Ashield_j=\blockVar}] \quad \forall \blockVar\in\{0,\ldots,\characCutoff\} \\
    \Tr[K_j\rho_{\Ameas_j\BmeasSquash_j\land\Ashield_j=\AclassicalVal}] &= \Tr[\rho_{\Ameas_j\BmeasSquash_j\land\Ashield_j=\AclassicalVal}] \quad \forall \AclassicalVal\in\AclassicalAlph\,,
    \end{align}
    where $K_j \coloneqq \frac{1}{1 - \dtImpBnd[0, j]} \proj_0^{\BmeasSquash_j} + \frac{1}{1 - \dtImpBnd[\geq1, j]} \proj_1^{\BmeasSquash_j}$. Note that we do not need an additional constraint for the blocks from $\characCutoff+1$ to $\tagCutoff$, since they are not fully characterized and their total contribution is already bounded through the trace of the marginal $\rho_{\Ameas_j\BmeasSquash_j}$.
    \item In the decoy constraints, replace $\{\proj_{\blockVar}^{\BmeasSquash_j}\}_{\blockVar=0}^{\fssCutoff}$ by $\{\proj_0^{\BmeasSquash_j},\proj_1^{\BmeasSquash_j}\}$. Equivalently, one may effectively set $\fssCutoff=1$.
    \item All remaining constraints and the objective function remain unchanged.
\end{enumerate}

\begin{remark}
\label{rem:active_basis_choice_q_params}
    To apply \cref{lem:combination_noise_channel_simple_squasher}, we need the detector POVM deviations $1>\dtImpBnd[0, j]\geq 0$ and $1>\dtImpBnd[\geq1, j]\geq 0$. Assume Bob uses two threshold detectors described by POVM elements $\{M_{\BclassicalVal}^{B_{j,\blockVar}, \mathbf{d_B},\boldsymbol{\eta}}\}_{\BclassicalVal\in \BclassicalAlph, \blockVar}$ with (unknown) dark count probabilities $\mathbf{d_B}^{(j)} = (d_1^{(j)}, d_{2}^{(j)})$ and (unknown) detection efficiencies $\boldsymbol{\eta} = (\eta_1, \eta_{2})$ in each round $j$. The dark count probabilities are imperfectly characterized in the sense that we have upper bounds
    \begin{equation}
        d_k^{(j)} \leq \tilde d_k(1+\Delta_{\mathrm{d}_\mathrm{B}})
    \end{equation}
    for all $k\in\{1, 2\}$ where $\tilde d_k$ and $\Delta_{\mathrm{d}_\mathrm{B}}$ are known. Additionally, the detection efficiencies are imperfectly characterized in the sense that they are in some range
    \begin{equation}
        \eta_k \in [\tilde \eta_k(1-\Delta_{\eta}), \tilde \eta_k(1+\Delta_{\eta})]
    \end{equation}
    for all $k\in\{1, 2\}$, where $\tilde \eta_k$ and $\Delta_{\eta}$ are known. Then, following \cite[App.~D]{nahar_proof-technique-independent_2026}, we have
    \begin{align}
        \dtImpBnd[0, j] &= 1 - \left((1-\tilde d_1(1+\Delta_{\mathrm{d}_\mathrm{B}}))(1-\tilde d_2(1+\Delta_{\mathrm{d}_\mathrm{B}}))\right)\,,\\
        \dtImpBnd[\geq1, j] &= 1 - \eta_\mathrm{min}'\left(1 - \frac{d_\mathrm{max}}{2}\right)\,,
    \end{align}
    in \cref{lem:combination_noise_channel_simple_squasher}, where $\eta_\mathrm{min}' \coloneq \frac{\min_k \tilde \eta_k(1-\Delta_{\eta})}{\max_k \tilde \eta_k(1+\Delta_{\eta})}$ is a lower bound on the minimum efficiency after pulling out the common loss into the channel and $d_\mathrm{max} \coloneqq \max_{k} \tilde d_k(1+\Delta_{\mathrm{d}_\mathrm{B}})$.
\end{remark}

\section{Device imperfections}
\label{ap:models_device_imperfections}

\noindent In this section, we detail the models for the device imperfections considered in the example of the polarization-encoded decoy-state BB84 setup from \cref{sec:example_pol_encoding}, as well as a few additional imperfections, and show how they can be combined.

\subsection{Source imperfections}
\label{ap:models_source_imperfections}

\noindent Most independent source imperfections enter the analysis through the source fidelity bound in \cref{as:imperfectly_characterized_source} or bounds on the photon-number distribution from \cref{as:probability_bounds}. Therefore, to incorporate source imperfection, the aim is to find a set of target states $\{(\tilde \Astate_{\AclassicalVal,\blockVar}^{(j)})_{\Aprime_j}\}_{j, \AclassicalVal\in \AclassicalAlph, \blockVar\leq\tagCutoff}$, and the source fidelity bounds $\{\srcImpBnd[\AclassicalVal,\blockVar,j]\}_{\AclassicalVal\in\AclassicalAlph, \blockVar\leq \tagCutoff, j}$ such that \cref{as:imperfectly_characterized_source} holds, and find bounds $\prob_{\blockVarReg_j|\Aclassical_j}^{\mathrm{L}}(\blockVar|\AclassicalVal)$ and $\prob_{\blockVarReg_j|\Aclassical_j}^{\mathrm{U}}(\blockVar|\AclassicalVal)$ such that \cref{as:probability_bounds} holds. For simplicity, we choose the same set of target states each round and assume that the source fidelity bounds are identical but note that we can readily incorporate varying sets of target states as long as the corresponding source fidelity bounds are known. 

\subsubsection{Imperfections affecting the source fidelity bounds (\cref{as:imperfectly_characterized_source})}

\textbf{Imperfect state encoding:} When encoding her signals, Alice may prepare states that differ from the ideal $H, V, D$ and $A$ states. Recall from \cref{as:block_diagonal_source} that Alice's signal states are block-diagonal in photon number. The vacuum block is given by $\ket{\mathrm{vac}}$. In each $\blockVar$-photon block with $\blockVar \geq 1$, she prepares $\blockVar$ photons in the mode corresponding to the single-photon states $\{\ket{\AstatePrepared_\AclassicalVal}\}_{\AclassicalVal\in\AclassicalAlph}$.\footnote{For example, the single-photon component of Alice's signal states is given by $(\Astate_{\AclassicalVal,1}^{(j)})_{\Aprime_j} = \ketbra{\AstatePrepared_\AclassicalVal}{\AstatePrepared_\AclassicalVal}_{\Aprime_j}$ for setting choice $\AclassicalVal\in\AclassicalAlph$. Under the assumption that all photons in a given block occupy the same mode, as is naturally the case for the models considered here, it suffices to characterize the single-photon states to obtain source fidelity bounds for all higher photon-number blocks, as formalized below.} Let $\big\{\lvert\tilde{\AstatePrepared}_\AclassicalVal\rangle\big\}_{\AclassicalVal\in\AclassicalAlph}$ be a known set of target single-photon states. These targets need not be perfectly encoded and may themselves incorporate other imperfections. For example, in \cref{sec:example_pol_encoding}, we consider the following model where single-photon target states are imperfectly encoded and given by \cite[Eq.~(4)]{sixto2026finitekeysecurityanalysisdecoystate}
\begin{equation}
\label{eq:state_preparation_flaw_model}
    \lvert\tilde{\AstatePrepared}_\AclassicalVal\rangle = \cos(\theta_\AclassicalVal) \ket{H} + \sin(\theta_\AclassicalVal)\ket{V}\,,
\end{equation}
where $\theta_\AclassicalVal = (1+\delta_{\mathrm{SPF}}/\pi)\phi_\AclassicalVal/2$ with state-preparation flaw term $\delta_{\mathrm{SPF}}$ and $\phi_\AclassicalVal\in \{0,\pi,\pi/2,3\pi/2\}$ for $\AclassicalVal_a\in\{H, V, D, A\}$. Note, however, that any known target states can be chosen. Suppose that the prepared single-photon states are imperfectly characterized with angular uncertainty $\Delta\theta$ relative to the targets. Then they satisfy
\begin{equation}
\label{eq:relation_source_fidelity_delta_enc}
    \left|\braket{\tilde{\AstatePrepared}_\AclassicalVal}{\AstatePrepared_\AclassicalVal}\right|^2 \geq 1 - \sin^2(\Delta\theta) = 1-\srcImpBndEnc
\end{equation}
for all $\AclassicalVal\in\AclassicalAlph$, where we define $\srcImpBndEnc\coloneqq \sin^2(\Delta\theta)$. For any pure single-photon modes $\ket{\psi_\AclassicalVal}$ and $\lvert\tilde{\AstatePrepared}_\AclassicalVal\rangle$, the corresponding states with $\blockVar$ photons occupying the same mode satisfy
\begin{equation}
    \left|\braket{\blockVar_{\tilde \psi_\AclassicalVal}}{\blockVar_{\psi_\AclassicalVal}}\right|^2 = \left|\braket{\tilde\psi_\AclassicalVal}{\psi_\AclassicalVal}\right|^{2\blockVar}.
\end{equation}
Then, one may choose the source fidelity bound $\srcImpBndEnc[\blockVar]=1 - (1-\srcImpBndEnc)^\blockVar \leq \blockVar\srcImpBndEnc$, for the $\blockVar$-photon block, and \cref{as:imperfectly_characterized_source} is satisfied.

\textbf{Leakage from non-encoding degrees of freedom:} Alice may additionally leak information about her setting choice through non-encoding degrees of freedom. Such leakage is included in the previous model by taking the prepared and target single-photon states to describe both encoding and non-encoding degrees of freedom.

\textbf{Trojan-horse attack:} In a Trojan-horse attack, Eve injects light into Alice's system and analyses the back-reflections. Eve's light travels through Alice's system and is encoded in the same way as Alice's states, therefore leaking information about Alice's setting choice. Let $\tilde{E}_j^{\mathrm{THA}}$ denote Eve's back-reflected light. The model is detailed in \cref{ap:trojan_horse_attack}. We choose the target states to be
\begin{equation}
    \ketbra{\mathrm{vac}}{\mathrm{vac}}_{\tilde{E}_j^{\mathrm{THA}}} \otimes (\Astate_{x,\blockVar})_{\Aprime}
\end{equation}
for all $x\in\AclassicalAlph$, i.e. the target states correspond to the signal states Alice prepares without back-reflections. Assuming that the mean photon number of the back-reflected light is upper-bounded by $\srcImpBndTHA/2$ for all rounds $j$, we have, cf. \cref{ap:trojan_horse_attack},
\begin{equation}
    F\left((\Astate_{x,\blockVar})_{\Aprime\tilde{E}_j^{\mathrm{THA}}}, \ketbra{\mathrm{vac}}{\mathrm{vac}}_{\tilde{E}_j^{\mathrm{THA}}}\otimes (\Astate_{x,\blockVar})_{\Aprime} \right) \geq 1-\srcImpBndTHA
\end{equation}
for all $x\in\AclassicalAlph$ and all $\blockVar\leq\tagCutoff$. Then, the source fidelity bounds are given by $\srcImpBndTHA[\blockVar] = \srcImpBndTHA$ for all $\blockVar\leq \tagCutoff$ and \cref{as:imperfectly_characterized_source} is satisfied.

The previous imperfections modify the states prepared by Alice and therefore contribute to \cref{as:imperfectly_characterized_source}. These imperfections can be combined into a total source fidelity bound by successively applying fidelity bounds between intermediate target states. For the imperfections discussed above, we choose the total target state $\ketbra{\mathrm{vac}}{\mathrm{vac}}_{\tilde{E}_j^{\mathrm{THA}}}\otimes (\tilde{\Astate}_{\AclassicalVal,\blockVar})_{\Aprime}$ for setting choice $\AclassicalVal\in\AclassicalAlph$ with
\begin{align}
    (\tilde{\Astate}_{\AclassicalVal,0})_{\Aprime} & = \ketbra{\mathrm{vac}}{\mathrm{vac}}_{\Aprime}\\
    (\tilde{\Astate}_{\AclassicalVal,1})_{\Aprime} &= \lvert\tilde{\psi}_\AclassicalVal\rangle\langle\tilde{\psi}_\AclassicalVal\lvert_{\Aprime}\label{eq:mode_single_photon_target}
\end{align}
and the $m$-photon blocks $(\tilde{\Astate}_{\AclassicalVal,\blockVar})_{\Aprime}$ are given by $\blockVar$ photons in the mode corresponding to the single-photon state given by \cref{eq:mode_single_photon_target}.
The resulting source fidelity bound is then given by combining the previous fidelity bounds
\begin{equation}
\label{eq:formula_combination_source_fidelity_bound}
    \srcImpBnd_\blockVar= \left(\sqrt{\srcImpBndEnc[\blockVar]} + \sqrt{\srcImpBndTHA[\blockVar]}\right)^2 
\end{equation}
using the relation between fidelity and purified distance and the triangle inequality of the purified distance.

\subsubsection{Imperfections affecting the photon-number distribution (\cref{as:probability_bounds})}

\noindent The next two source imperfections affect \cref{as:probability_bounds} on the photon-number distribution. The goal is therefore to derive upper and lower bounds on the photon-number probability distribution.

\textbf{Intensity fluctuations:} Alice prepares optical pulses with photon-number distribution $\prob_{\blockVarReg_j|\Aclassical_j}$ (see \cref{as:block_diagonal_source}). For each intensity choice $\mu_i$, assume that the corresponding ``true'' intensity $\tilde{\mu}_i$ is characterized to lie within the interval
\begin{equation}
    \tilde{\mu}_i \in [\mu_i(1-\srcImpBndInt),\mu_i(1+\srcImpBndInt)]
\end{equation}
where $\mu_i$ denotes a known target intensity. We assume that, conditioned on an intensity choice, the intensity is not fixed but is independently sampled in each round from within this interval, denoted by the r.v. $I_{\AclassicalVal,j}$ (with unknown probability density function).\footnote{Note that if the probability density function of $I_{\AclassicalVal,j}$ is known, such as in \cite{sixto2026finitekeysecurityanalysisdecoystate}, then the photon-number distribution $\prob_{\blockVarReg_j|\Aclassical_j}$ can directly be computed.} Then, bounds can be obtained by minimizing the photon-number probability distribution over all possible intensities
\begin{align}
    \prob_{\blockVarReg_j|\Aclassical_j}^{\mathrm{L}}(\blockVar|\AclassicalVal_a\mu_i) &= \min_{\tilde{\mu}_i \in [\mu_i(1-\srcImpBndInt),\mu_i(1+\srcImpBndInt)]} \prob_{\blockVarReg_j|I_{\AclassicalVal,j}}(\blockVar|\tilde \mu_i)\,, \\
    \prob_{\blockVarReg_j|\Aclassical_j}^{\mathrm{U}}(\blockVar|\AclassicalVal_a\mu_i) &= \max_{\tilde{\mu}_i \in [\mu_i(1-\srcImpBndInt),\mu_i(1+\srcImpBndInt)]} \prob_{\blockVarReg_j|I_{\AclassicalVal,j}}(\blockVar|\tilde \mu_i)
\end{align}
for all $\blockVar\leq\tagCutoff$, all intensity choices $\mu_i$ and all $\AclassicalVal\in\AclassicalAlph$, and \cref{as:probability_bounds} is satisfied.

\begin{remark}
    Following the first and last source constraints in \cref{fig:opti_problem_final}, the probability that Alice emits a flag, conditioned on setting choice $\AclassicalVal\in\AclassicalAlph$, is upper bounded by
    \begin{equation}
        1-\sum_{\blockVar=0}^{\tagCutoff} \probMatrix^{(\blockVar,j)}_{\AclassicalVal,\AclassicalVal} \leq 1-\sum_{\blockVar=0}^{\tagCutoff} \prob_{\blockVarReg_j|\Aclassical_j}^{\mathrm{L}}(\blockVar|\AclassicalVal).
    \end{equation}
    However, for intensity fluctuations, we can exploit the additional structure to tighten this bound. In particular, the bound above may be very loose, since the minima of the individual photon-number probabilities may occur at different intensities within the interval $[\mu_i(1-\srcImpBndInt),\mu_i(1+\srcImpBndInt)]$. Whenever the weight in the flags is monotonically increasing in the intensity (which is the case for Poisson distributions), it is upper bounded by evaluating the tail at the largest allowed intensity. Therefore, we can additionally impose
    \begin{equation}
        1-\sum_{\blockVar=0}^{\tagCutoff} \probMatrix^{(\blockVar,j)}_{\AclassicalVal,\AclassicalVal} \leq 1-\sum_{\blockVar=0}^{\tagCutoff} \prob_{\blockVarReg_j|I_{\AclassicalVal,j}} (\blockVar|\mu_i(1+\srcImpBndInt))
    \end{equation}
    for all setting choices $\AclassicalVal=\AclassicalVal_a\mu_i\in\AclassicalAlph$ and all rounds $j$ in the optimization problem from \cref{fig:opti_problem_final}.
\end{remark}

\textbf{Imperfect photon-number distribution:} Additionally, let the photon-number distribution $\prob_{\blockVarReg_j|I_{\AclassicalVal, j}}$ be imperfectly characterized with known target distribution, i.e.
\begin{equation}
    \left| \prob_{\blockVarReg_j|I_{\AclassicalVal, j}}(\blockVar|\tilde \mu_i) - \tilde{\prob}_{\blockVarReg_j|I_{\AclassicalVal, j}}(\blockVar|\tilde \mu_i) \right| \leq \srcImpBndPN[\blockVar]
\end{equation}
for all $\blockVar\leq \tagCutoff$, all intensity choices $\tilde{\mu}_i \in [\mu_i(1-\srcImpBndInt),\mu_i(1+\srcImpBndInt)]$ and all $\AclassicalVal\in\AclassicalAlph$, where $\tilde{\prob}_{\blockVarReg_j|I_{\AclassicalVal, j}}$ denotes the known target distribution. Then, bounds can be obtained as
\begin{align}
    \prob_{\blockVarReg_j|\Aclassical_j}^{\mathrm{L}}(\blockVar|\AclassicalVal_a\mu_i) &= \min_{\tilde{\mu}_i \in [\mu_i(1-\srcImpBndInt),\mu_i(1+\srcImpBndInt)]} \tilde{\prob}_{\blockVarReg_j|I_{\AclassicalVal,j}}(\blockVar|\tilde \mu_i) - \srcImpBndPN[\blockVar]\,, \\
    \prob_{\blockVarReg_j|\Aclassical_j}^{\mathrm{U}}(\blockVar|\AclassicalVal_a\mu_i) &= \max_{\tilde{\mu}_i \in [\mu_i(1-\srcImpBndInt),\mu_i(1+\srcImpBndInt)]} \tilde{\prob}_{\blockVarReg_j|I_{\AclassicalVal,j}}(\blockVar|\tilde \mu_i) + \srcImpBndPN[\blockVar]
\end{align}
for all $\blockVar\leq\tagCutoff$, all intensity choices $\mu_i$ and all $\AclassicalVal\in\AclassicalAlph$, and \cref{as:probability_bounds} is satisfied.

\subsubsection{Imperfections affecting the block-diagonality (\cref{as:block_diagonal_source})} 

\textbf{Imperfectly block-diagonal source:} Imperfect block-diagonality is handled separately in \cref{sec:imperfect_block_diagonal}, as it constitutes a special case where \cref{as:block_diagonal_source} is no longer exactly satisfied, requiring more specialized arguments. For example, imperfect phase randomization in decoy-state protocols leads to states that are only approximately block-diagonal in the photon-number basis. Another example is an imperfect single-photon source exhibiting coherences between different photon-number blocks.

\subsection{Detector imperfections}
\label{ap:detector_imperfections}

\noindent Independent detector imperfections enter the analysis through the detector POVM deviation in \cref{as:imperfectly_characterized_detector}.
Therefore, to incorporate these imperfections, the aim is to find target POVM elements $\{\BpovmelBlockTilde{j}{\blockVar}\}_{\BclassicalVal\in \BclassicalAlph, \blockVar\leq\fssCutoff, j}$ and the detector POVM deviations $\{\dtImpBnd[\blockVar, j]\}_{\blockVar\leq \fssCutoff, j}$ such that \cref{as:imperfectly_characterized_detector} holds. As for source imperfection, for simplicity, we assume that Bob's target POVM elements and detector POVM deviations are the same each round, but note that we can readily incorporate varying sets of target POVM elements if the corresponding detector POVM deviations are known. We set $\fssCutoff = 1$ as it is sufficient to compute practical key rates (see \cref{sec:example_pol_encoding}).

\subsubsection{Imperfections affecting the detector POVM deviation (\cref{as:imperfectly_characterized_detector})}

\textbf{Dark-count rate mismatch:} We assume Bob uses $\nbDetec$ threshold detectors described by POVM elements $\{M_{\BclassicalVal}^{B_{j,\blockVar}, \mathbf{d_B}}\}_{\BclassicalVal\in \BclassicalAlph, \blockVar\leq 1}$ with (unknown) dark count probabilities $\mathbf{d_B}^{(j)} = (d_1^{(j)}, \ldots, d_{N_\mathrm{det}}^{(j)})$ in each round $j$. The dark count probabilities are imperfectly characterized in the sense that we have upper bounds
\begin{equation}
    d_k^{(j)} \leq \tilde d_k(1+\Delta_{\mathrm{d}_\mathrm{B}})
\end{equation}
for all $k\in\{1, \ldots, \nbDetec\}$ where $\tilde d_k$ and $\Delta_{\mathrm{d}_\mathrm{B}}$ are known. Following \cite[Theorem~1]{nahar_imperfect_2026}, by choosing the target POVM elements to be threshold detectors without dark counts $\{M_{\BclassicalVal}^{B_{j,\blockVar}}\}_{\BclassicalVal\in \BclassicalAlph, \blockVar\leq 1}$ we find
\begin{equation}
    M_{\BclassicalVal}^{B_{j,\blockVar}, \mathbf{d_B}} - (1-\dtImpBndDCR[\blockVar]) M_{\BclassicalVal}^{B_{j,\blockVar}} \geq 0
\end{equation}
for all $\BclassicalVal\in\BclassicalAlph$ and $\blockVar\leq 1$, where
\begin{equation}
\label{eq:q_param_dark_counts}
    \dtImpBndDCR[\blockVar=0] = \dtImpBndDCR[\blockVar=1] = 1 - \Pi_{k=1}^{\nbDetec}(1-d_k^{(j)}) \leq 1 - \Pi_{k=1}^{\nbDetec}(1 - \tilde d_k(1+\Delta_{\mathrm{d}_\mathrm{B}}))\,,
\end{equation}
and \cref{as:imperfectly_characterized_detector} is satisfied.

\textbf{Detection-efficiency mismatch:} We assume that the detectors described by POVM elements $\{M_{\BclassicalVal}^{B_{j,\blockVar}, \boldsymbol{\eta}}\}_{\BclassicalVal\in \BclassicalAlph, \blockVar\leq 1}$ have (unknown) detection efficiencies $\boldsymbol{\eta} = (\eta_1, \ldots \eta_{N_\mathrm{det}})$ which are imperfectly characterized in the sense that they are in some range
\begin{equation}
    \eta_k \in [\tilde \eta_k(1-\Delta_{\eta}), \tilde \eta_k(1+\Delta_{\eta})]
\end{equation}
for all $k\in\{1, \ldots, \nbDetec\}$, where $\tilde \eta_k$ and $\Delta_{\eta}$ are known. Following \cite[Theorem~2]{nahar_imperfect_2026}, for any value $\eta_* \in \left[\frac{\eta_\mathrm{min}}{1 - (\eta_\mathrm{max} - \eta_\mathrm{min})}, 1\right]$, where $\eta_\mathrm{min} \coloneq \min_{k} \tilde \eta_k(1-\Delta_{\eta})$ and $\eta_\mathrm{max} \coloneq \max_{k} \tilde \eta_k(1+\Delta_{\eta})$ we have
\begin{equation}
    M_{\BclassicalVal}^{B_{j,\blockVar}, \boldsymbol{\eta}} - (1-\dtImpBndEff[\blockVar, \eta_*]) M_{\BclassicalVal}^{B_{j,\blockVar}, \eta_*} \geq 0\,,
\end{equation}
for all $\BclassicalVal\in\BclassicalAlph$ and $\blockVar\leq 1$, where the POVM elements $M_{\BclassicalVal}^{B_{j,\blockVar}, \eta_*}$ describe detectors without dark counts and with common detection efficiency $\eta_*$ and
\begin{align}
    \dtImpBndEff[\blockVar = 0, \eta_*] &= 0\\
    \dtImpBndEff[\blockVar = 1, \eta_*] &= 1 - \frac{\eta_\mathrm{min}}{\eta_*}\,,
    \label{eq:q_param_efficiency}
\end{align}
and \cref{as:imperfectly_characterized_detector} is satisfied.

\begin{remark}
\label{rem:pulling_out_loss}
    A convenient normalization is to absorb a common detector loss $\eta_{\mathrm{max}}$ into the quantum channel before applying the noise channel and subspace weight estimation. This can always be done when the common loss commutes with the passive optical setup \cite{nahar_proof-technique-independent_2026}. The new detector efficiencies are $\eta'_k \coloneq \eta_k/\eta_{\mathrm{max}}$ and satisfy
    \begin{equation} 
        \eta'_k \in \left[\frac{\tilde\eta_k(1-\Delta_\eta)}{\eta_\mathrm{max}}, \frac{\tilde\eta_k(1+\Delta_\eta)}{\eta_\mathrm{max}}\right]. 
    \end{equation}
    Thus, the normalized bounds are given by
    \begin{equation} 
        \eta'_\mathrm{min} = \frac{\eta_\mathrm{min}}{\eta_\mathrm{max}}, \qquad \eta'_\mathrm{max} = 1. 
    \end{equation}
    Applying the analysis above with the new detector efficiencies and setting $\eta_* = 1$ yields
    \begin{align} 
        \dtImpBndEff[\blockVar=0,\eta_*=1] &= 0, \\ \dtImpBndEff[\blockVar=1,\eta_*=1] &= 1-\frac{\eta_\mathrm{min}}{\eta_\mathrm{max}}. 
    \end{align}
    All subsequent detector bounds, including the subspace-weight bound $\lambda_\mathrm{min}$, must then be evaluated for the normalized efficiencies.
\end{remark}

\textbf{Imperfect passive basis choice:} We assume that the basis choice is implemented using a beamsplitter with imperfectly characterized splitting ratio $s$, satisfying
\begin{equation}
    s \in [\tilde{s}(1-\Delta_s), \tilde{s}(1+\Delta_s)]\,,
\end{equation}
where $\tilde{s}$ and $\Delta_s$ are known. Consider detectors without dark counts and with common detection efficiency $\eta_*$. Choosing the target POVM elements to have splitting ratio $s'$, we have
\begin{equation}
    M_{\BclassicalVal}^{B_{j,\blockVar}, \eta_*, s} - (1-q_{\blockVar}^{s}) M_{\BclassicalVal}^{B_{j,\blockVar}, \eta_*, s'} \geq 0\,,
\end{equation}
for all $\BclassicalVal\in\BclassicalAlph$ and $\blockVar\leq 1$, where \cite[Sec.~5.5.5]{nahar_proof-technique-independent_2026}
\begin{align}
    q_{\blockVar=0}^{s} &= 0\,, \\
    q_{\blockVar=1}^{s} &= 1-\min\left\{\frac{\tilde s(1-\Delta_s)}{s'}, \frac{1-\tilde s(1+\Delta_s)}{1-s'}\right\} \,. \label{eq:q_param_splitting}
\end{align}
Thus, \cref{as:imperfectly_characterized_detector} is satisfied. 

\begin{remark}
    To minimize the detector POVM deviation in \cref{eq:q_param_splitting}, we can choose the target splitting ratio to be
    \begin{equation}
        s' = \frac{\tilde{s}(1-\Delta_s)}{1-2\tilde{s}\Delta_s}\,,
    \end{equation}
    then
    \begin{equation}
        q_{\blockVar=1}^{s} = 2\tilde{s}\Delta_s\,,
    \end{equation}
    while $q_{\blockVar=0}^{s}=0$ remains unchanged.
\end{remark}

We can combine dark-count rate mismatch, detection-efficiency mismatch, and imperfect passive basis choice following Ref.~\cite{nahar_imperfect_2026} to satisfy
\begin{equation}
    M_{\BclassicalVal}^{B_{j,\blockVar}, \mathbf{d_B}, \boldsymbol{\eta}, s} - \left(1-q_{\blockVar,\eta_*}^{\mathbf{d_B},\boldsymbol{\eta},s}\right) M_{\BclassicalVal}^{B_{j,\blockVar}, \eta_*, s'} \geq 0\,,
\end{equation}
for all $\BclassicalVal\in\BclassicalAlph$ and $\blockVar\leq 1$, as the detector POVM deviations need only be multiplied, resulting in
\begin{align}
    q_{\blockVar=0,\eta_*}^{\mathbf{d_B},\boldsymbol{\eta},s} &= 1-\prod_{k=1}^{\nbDetec}\left(1-\tilde d_k(1+\Delta_{\mathrm{d}_\mathrm{B}})\right)\,, \label{eq:q_param_0_full} \\
    q_{\blockVar=1,\eta_*}^{\mathbf{d_B},\boldsymbol{\eta},s} &= 1-\left(1-q_{\blockVar=1}^{s}\right)\frac{\eta_\mathrm{min}}{\eta_*}\prod_{k=1}^{\nbDetec}\left(1-\tilde d_k(1+\Delta_{\mathrm{d}_\mathrm{B}})\right)\,, \label{eq:q_param_1_full}
\end{align}
and \cref{as:imperfectly_characterized_detector} is satisfied. Note that if the common detector loss is first absorbed into the channel as described in \cref{rem:pulling_out_loss}, the factor $\eta_\mathrm{min}/\eta_*$ is replaced by $\eta_\mathrm{min}/\eta_\mathrm{max}$, and the target detectors have unit efficiency. 

\subsubsection{Imperfections affecting the subspace weight estimation (\cref{th:security_imperfect_devices})}

\noindent Commonly, for passive detection setups, the lower bound on the weight in the flag space is obtained through the multi-click POVM element. When detection-efficiency mismatch, dark counts and an imperfectly characterized splitting ratio are present, described respectively by $\boldsymbol{\eta}$, $\mathbf{d_B}$ and $s$, the corresponding multi-click POVM element is given by $W^{\Bmeas_{j}, \mathbf{d_B}, \boldsymbol{\eta},s}$, which depends on these parameters. Following \cite[App.~C]{wang2025phaseerrorestimationpassive}, the multi-click POVM element satisfies
\begin{equation}
    W^{\Bmeas_{j}, \mathbf{d_B}, \boldsymbol{\eta},s} \geq \lambdaMin^{\mathbf{d_B},\boldsymbol{\eta},s} \proj_{>\fssCutoff}^{\Bmeas_j}\,,
\end{equation}
where
\begin{equation}
    \label{eq:lambda_min_expression_imperfect}\lambdaMin^{\mathbf{d_B},\boldsymbol{\eta},s} \geq 2\eta_\mathrm{min}^2 \bar{s}(1-\bar{s})
\end{equation}
and
\begin{equation}
    \bar{s} \coloneq \frac{1}{2} + \max\left\{
        \left|\tilde{s}(1-\Delta_s)-\frac{1}{2}\right|,
        \left|\tilde{s}(1+\Delta_s)-\frac{1}{2}\right|
    \right\}\,.
\end{equation}
Here, $\bar{s}$ corresponds to the value of $s$ furthest from the ideal balanced splitting ratio $\frac{1}{2}$. Consequently, the worst-case bound on $\lambdaMin^{\mathbf{d_B},\boldsymbol{\eta}}$ is obtained by taking the most asymmetric splitting ratio compatible with the characterization uncertainty.

\subsection{Trojan-horse attack model}
\label{ap:trojan_horse_attack}

\noindent In this section, we detail the model used for the Trojan-horse attack (see \cref{sec:example_pol_encoding,ap:models_source_imperfections}) without assuming that the emitted Alice-Eve states have a product structure, in contrast to \cite[App.~C]{pereira_modified_2023}. In the $j$th round, Eve holds onto the state $\rho_{\Ereg_j}$. We can split it up into two parts $\Ereg_j = E_j^{\mathrm{THA}}\hat{E_j}$ where $E_j^{\mathrm{THA}}$ is the part that Eve injects into Alice's lab. We do not assume a specific optical state for the injected light.
However, we assume there exists a power limiting mechanism (e.g. due to the laser-induced fiber threshold) that limits how much power Eve can inject, i.e. we assume we have a bound $\nu_{\mathrm{max}}^{(j)}$ on the mean number of photons Eve can inject, i.e.
\begin{equation}
    \langle{N_{E_j^{\mathrm{THA}}}} \otimes \identity_{\hat{E_j}}\rangle_{\rho_{\Ereg_j}} \leq \nu_{\mathrm{max}}^{(j)}
    \label{eq:assumption_power_lim_mechanism}
\end{equation}
where $N_{E_j^{\mathrm{THA}}}$ is the photon number operator acting on $E_j^{\mathrm{THA}}$. Then, Eve's light travels through Alice's system, may become correlated with Alice's state setting choice $\AclassicalVal\in\AclassicalAlph$ and is back-reflected. We assume that the Trojan-horse attack does not affect the state Alice prepares in register $\Aprime_j$. Let $\tilde{E}_j^{\mathrm{THA}}$ be the system exiting Alice's lab. Alice therefore effectively prepares the new state 
\begin{equation}
    \sigma_{\Aclassical_j\Aprime_j\tilde{E}_j^{\mathrm{THA}}} = \sum_{\AclassicalVal\in\AclassicalAlph} \prob_{\Aclassical_j}(\AclassicalVal) \ketbra{\AclassicalVal}{\AclassicalVal}_{\Aclassical_j}\otimes(\sigma_\AclassicalVal^{(j)})_{\Aprime_j\tilde{E}_j^{\mathrm{THA}}}\,,
\end{equation}
with additional register $\tilde{E}_j^{\mathrm{THA}}$ due to the Trojan-horse attack. The application of the modified timing argument in \cref{lem:modified_timing_protocol} requires additional care in the presence of Trojan-horse attacks. If the protocol allows for on-the-fly announcements, Eve's injected light may depend on previous announcements. Consequently, the signal states above may depend on these announcements and cannot generally be regarded as states prepared in advance. Without on-the-fly announcements, this problem does not arise. For protocols with on-the-fly announcements, we therefore restrict the present analysis to independent Trojan-horse attacks in each round. Under this assumption, the signals states retain the round independence required by the modified timing argument. 

Let $\eta_{\mathrm{THA}}^{(j)}$ denote the attenuation experienced by system $E_j^{\mathrm{THA}}$ between the power-limiting mechanism and exiting Alice's lab. Due to this attenuation and \cref{eq:assumption_power_lim_mechanism}, we have
\begin{equation}
    \langle{N_{\tilde{E}_j^{\mathrm{THA}}}\otimes\identity_{\Aprime_j}}\rangle_{\sigma_\AclassicalVal^{(j)}} \leq \nu_{\mathrm{max}}^{(j)} \eta_{\mathrm{THA}}^{(j)}
\end{equation}
for all $x \in \mathcal{X}$. This allows us to bound the weight of the reduced state $\rho_{\tilde{E}_j^{\mathrm{THA}}}$ outside the vacuum subspace
\begin{align}
    \langle{N_{\tilde{E}_j^{\mathrm{THA}}}\otimes\identity_{\Aprime_j}}\rangle_{\sigma_\AclassicalVal^{(j)}} &= \sum_m m\Tr(\big(\sigma_\AclassicalVal^{(j)}\big)_{\Aprime_j\tilde{E}_j^{\mathrm{THA}}}\ketbra{m}{m}_{\tilde{E}_j^{\mathrm{THA}}}\otimes \identity_{\Aprime_j}) \\
    &= \sum_{m\neq 0}m\bra{m}\rho_{\tilde{E}_j^{\mathrm{THA}}}\ket{m} \\
    &\geq \sum_{m\neq 0}\bra{m}\rho_{\tilde{E}_j^{\mathrm{THA}}}\ket{m}\,.
\end{align}
Since $\sum_{m}\bra{m}\rho_{\tilde{E}_j^{\mathrm{THA}}}\ket{m}= 1$, it follows that
\begin{equation}
    \bra{\mathrm{vac}}\rho_{\tilde{E}_j^{\mathrm{THA}}}\ket{\mathrm{vac}} = p_{\mathrm{vac}} \geq 1-\nu_{\mathrm{max}}^{(j)} \eta_{\mathrm{THA}}^{(j)}\,,
\end{equation}
and we have
\begin{equation}
    \proj_{\mathrm{vac}}^{\tilde{E}_j^{\mathrm{THA}}}\big(\sigma_\AclassicalVal^{(j)}\big)_{\Aprime_j\tilde{E}_j^{\mathrm{THA}}}\proj_{\mathrm{vac}}^{\tilde{E}_j^{\mathrm{THA}}} = p_{\mathrm{vac}} \ketbra{\mathrm{vac}}{\mathrm{vac}}_{\tilde E_j^{\mathrm{THA}}}\otimes \big(\bar \sigma_\AclassicalVal^{(j)}\big)_{\Aprime_j}
\end{equation}
with 
\begin{equation}
    \big(\bar \sigma_\AclassicalVal^{(j)}\big)_{\Aprime_j}=\frac{\bra{\mathrm{vac}}_{\tilde E_j^{\mathrm{THA}}}\big(\sigma_\AclassicalVal^{(j)}\big)_{\Aprime_j\tilde{E}_j^{\mathrm{THA}}}\ket{\mathrm{vac}}_{\tilde E_j^{\mathrm{THA}}}}{p_{\mathrm{vac}}}
\end{equation}
and therefore $\big( \sigma_\AclassicalVal^{(j)}\big)_{\Aprime_j}\succeq p_{\mathrm{vac}}\big(\bar \sigma_\AclassicalVal^{(j)}\big)_{\Aprime_j}$, which implies $F\left(\big(\bar \sigma_\AclassicalVal^{(j)}\big)_{\Aprime_j},\big( \sigma_\AclassicalVal^{(j)}\big)_{\Aprime_j}\right)\geq p_{\mathrm{vac}}$. Using this, we find a lower bound
\begin{align}
    F\left((\sigma_\AclassicalVal^{(j)})_{\Aprime_j\tilde E_j^{\mathrm{THA}}}, \ketbra{\mathrm{vac}}{\mathrm{vac}}_{\tilde E_j^{\mathrm{THA}}}\otimes (\sigma_\AclassicalVal^{(j)})_{\Aprime_j}\right) &= p_{\mathrm{vac}} F\left(\big(\bar \sigma_\AclassicalVal^{(j)}\big)_{\Aprime_j},\big( \sigma_\AclassicalVal^{(j)}\big)_{\Aprime_j}\right)\\
    &\geq p_{\mathrm{vac}}^2 \geq \left(1-\nu_{\mathrm{max}}^{(j)} \eta_{\mathrm{THA}}^{(j)}\right)^2 \geq 1-2\nu_{\mathrm{max}}^{(j)} \eta_{\mathrm{THA}}^{(j)}\,.
    \label{eq:fidelity_bound_THA}
\end{align}
Then we can choose the set of target states $\left\{\ketbra{\mathrm{vac}}{\mathrm{vac}}_{\tilde E_j^{\mathrm{THA}}}\otimes (\sigma_\AclassicalVal^{(j)})_{\Aprime_j}\right\}_x$ with the bound $\srcImpBndTHA[j] = 2\nu_{\mathrm{max}}^{(j)} \eta_{\mathrm{THA}}^{(j)}$ given by \cref{eq:fidelity_bound_THA}. 

The only quantities that need to be characterized are the attenuation $\eta_{\mathrm{THA}}^{(j)}$ experienced by Eve's injected light and the maximum allowed input power $\nu_{\mathrm{max}}^{(j)}$. Determining $\eta_{\mathrm{THA}}^{(j)}$ only requires characterizing the total path attenuation $\eta_{\mathrm{THA}}^{(j)}(\omega)$ as a function of all relevant modes $\omega \in \mathcal{W}$, where $\omega$ may correspond to polarization, wavelength, or other degrees of freedom. Then,
\begin{equation}
    \eta_{\mathrm{THA}}^{(j)} = \sup_{\omega \in \mathcal{W}}\eta_{\mathrm{THA}}^{(j)}(\omega).
\end{equation}
This guarantees security of the QKD protocol against all Trojan-horse attacks with modes $\omega \in \mathcal{W}$ under the model described above. We note that it may not be straightforward to determine a tight bound on the maximum allowed input power $\nu_{\mathrm{max}}^{(j)}$ in practice. Indeed, the breakdown mechanism of optical fibers is difficult to model and depends strongly on the properties of the injected light. A detailed discussion is beyond the scope of this work, and we therefore treat $\nu_{\mathrm{max}}^{(j)}$ as a characterized parameter of the implementation.

\section{Explicit construction of the set of marginal states}
\label{ap:construction_marginal_constraint}

\begin{lemma}
    Consider the \nameref{prot:entanglement_qkd_protocol} $\big\{\big\{\AstateVirt_{\Aclassical_j\Aprime''_j}^{(j)}, \big\{\BpovmelTarg[j]\big\}_{\BclassicalVal\in\BclassicalAlph}, \annKeyMap[j]\big\}_{j=1}^\totRounds, \ppMap\big\}$ from \cref{lem:security_virtual_after_reductions}, which resulted from the sequence of source and squashing maps from \cref{sec:problem_reductions}, with (known) POVMs given by \cref{eq:def_new_virtual_povm_elements} and imperfectly characterized signal states satisfying \cref{eq:def_new_virtual_states_alice}. Let $\AstateVirt_{\Amarg_j\Aprime''_j}^{(j)}$ denote the corresponding source-replaced state in the $j$th round and $\setAliceMarginalsVirt[j]$ be a set of possible marginal states. Then, $\AstateVirt^{(j)}_{\Amarg_j} \in \setAliceMarginalsVirt[j]$ if the set of possible marginal states is constructed as
    \begin{align}
        \setAliceMarginalsVirt[j] = \Big\{& \genDensity_{\Amarg_j} \in \setDensity_=(\Amarg_j) : \genDensity_{\Amarg_j} = \sum_{\AclassicalVal, \AclassicalVal'\in\AclassicalAlph} \sqrt{\prob_{\Aclassical_j}(\AclassicalVal) \prob_{\Aclassical_j}(\AclassicalVal')} \ketbra{\AclassicalVal}{\AclassicalVal'}_{\Ameas_j} \otimes \sum_{\blockVar=0}^{\tagCutoff}  \sqrt{\prob_{\blockVarReg_j|\Aclassical_j}(\blockVar|\AclassicalVal)\prob_{\blockVarReg_j|\Aclassical_j}(\blockVar|\AclassicalVal')} \braket{\AstateVirt_{\AclassicalVal',\blockVar}^{(j)}}{\AstateVirt_{\AclassicalVal,\blockVar}^{(j)}} \ketbra{\blockVar}{\blockVar}_{\Ashield_j}\nonumber \\
        & + \sum_{\AclassicalVal\in\AclassicalAlph} \prob_{\Aclassical_j}(\AclassicalVal) \ketbra{\AclassicalVal}{\AclassicalVal}_{\Ameas_j} \otimes \left(1 - \sum_{\blockVar=0}^{\tagCutoff} \prob_{\blockVarReg_j|\Aclassical_j}(\blockVar|\AclassicalVal)\right) \ketbra{\AclassicalVal}{\AclassicalVal}_{\Ashield_j}\,, \label{eq:explicit_construction_marginal_constraint} \\
        & \braket{\AstateVirt_{\AclassicalVal,\blockVar}^{(j)}}{\AstateVirt_{\AclassicalVal',\blockVar}^{(j)}} = \sqrt{\left(1-\srcImpBnd[\AclassicalVal,\blockVar,j]\right)\left(1-\srcImpBnd[\AclassicalVal',\blockVar,j]\right)}\bra{\tilde \Astate_{\AclassicalVal,\blockVar}^{(j)}}\ket{\tilde \Astate_{\AclassicalVal',\blockVar}^{(j)}} + \sqrt{\left(1-\srcImpBnd[\AclassicalVal,\blockVar,j]\right)\srcImpBnd[\AclassicalVal',\blockVar,j]} \bra{\tilde \Astate_{\AclassicalVal,\blockVar}^{(j)}}\ket{\tilde \Astate_{\AclassicalVal',\blockVar}^{(j)\perp}} \nonumber\\
        &+ \sqrt{\srcImpBnd[\AclassicalVal,\blockVar,j]\left(1-\srcImpBnd[\AclassicalVal',\blockVar,j]\right)} \bra{\tilde \Astate_{\AclassicalVal,\blockVar}^{(j)\perp}}\ket{\tilde \Astate_{\AclassicalVal',\blockVar}^{(j)}} + \sqrt{\srcImpBnd[\AclassicalVal,\blockVar,j]\srcImpBnd[\AclassicalVal',\blockVar,j]} \bra{\tilde \Astate_{\AclassicalVal,\blockVar}^{(j)\perp}}\ket{\tilde \Astate_{\AclassicalVal',\blockVar}^{(j)\perp}}\nonumber \\
        & \bra{\tilde \Astate_{\AclassicalVal,\blockVar}^{(j)}}\ket{\tilde \Astate_{\AclassicalVal,\blockVar}^{(j)\perp}} = 0 \text{ for all } \AclassicalVal\in\AclassicalAlph, \blockVar\leq\tagCutoff\,,  \nonumber \\
        & \prob_{\blockVarReg_j|\Aclassical_j}^{\mathrm{L}}(\blockVar|\AclassicalVal) \leq \prob_{\blockVarReg_j|\Aclassical_j}(\blockVar|\AclassicalVal) \leq \prob_{\blockVarReg_j|\Aclassical_j}^{\mathrm{U}}(\blockVar|\AclassicalVal) \text{ for all } \AclassicalVal\in\AclassicalAlph, \blockVar\leq\tagCutoff\Big\}\,, \nonumber
    \end{align}
    where $\left\{\ket{\tilde \Astate_{\AclassicalVal,\blockVar}^{(j)}}_{\Aprime''_j}\right\}_{j, \AclassicalVal\in \AclassicalAlph, \blockVar\leq\tagCutoff}$ is defined in terms of any set of purifications of the target states $\big\{\tilde \Astate_{\AclassicalVal,\blockVar}^{(j)}\big\}_{j, \AclassicalVal\in \AclassicalAlph, \blockVar\leq\tagCutoff}$ from \cref{as:imperfectly_characterized_source}, while the states $\ket{\tilde \Astate_{\AclassicalVal,\blockVar}^{(j)\perp}}_{\Aprime''_j}$ are unknown and $\bra{\tilde \Astate_{\AclassicalVal,\blockVar}^{(j)}}\ket{\tilde \Astate_{\AclassicalVal,\blockVar}^{(j)\perp}} = 0$ for all $\AclassicalVal\in\AclassicalAlph$, $\blockVar\leq\tagCutoff$, and all rounds $j$. The source fidelity bounds $\{\srcImpBnd[\AclassicalVal,\blockVar,j]\}_{\AclassicalVal\in\AclassicalAlph,\blockVar\leq\tagCutoff,j}$
    are given by \cref{as:imperfectly_characterized_source} while the upper and lower bounds on the photon-number distributions $\prob_{\blockVarReg_j|\Aclassical_j}^{\mathrm{U}}(\blockVar|\AclassicalVal)$ and $\prob_{\blockVarReg_j|\Aclassical_j}^{\mathrm{L}}(\blockVar|\AclassicalVal)$ are given by \cref{as:probability_bounds}. 
\end{lemma}
\begin{proof}
    The source-replaced states $\AstateVirt_{\Amarg_j\Aprime''_j}^{(j)} = \ketbra{\AstateVirt^{(j)}}{\AstateVirt^{(j)}}_{\Amarg_j\Aprime''_j}$ corresponding to $\AstateVirt_{\Aclassical_j\Aprime''_j}^{(j)}$ are given by 
    \begin{align}
    \ket{\AstateVirt^{(j)}}_{\Amarg_j\Aprime''_j} &= \sum_{\AclassicalVal\in\AclassicalAlph} \sqrt{\prob_{\Aclassical_j}(\AclassicalVal)}\ket{\AclassicalVal}_{\Ameas_j}\ket{\AstateVirt_\AclassicalVal^{(j)}}_{\Ashield_j\Aprime''_j} \label{eq:source_replaced_virtual_state} \\
    \ket{\AstateVirt_\AclassicalVal^{(j)}}_{\Ashield_j\Aprime''_j} &= \sum_{\blockVar=0}^\tagCutoff \sqrt{\prob_{\blockVarReg_j|\Aclassical_j}(\blockVar|\AclassicalVal)} \ket{\blockVar}_{\Ashield_j}\ket{\AstateVirt_{\AclassicalVal,\blockVar}^{(j)}}_{\Aprime''_j} + \sqrt{1 - \sum_{\blockVar=0}^\tagCutoff \prob_{\blockVarReg_j|\Aclassical_j}(\blockVar|\AclassicalVal)}\ket{\AclassicalVal}_{\Ashield_j}\ket{\AclassicalVal}_{\Aprime''_j} \\
    \ket{\AstateVirt_{\AclassicalVal,\blockVar}^{(j)}}_{\Aprime''_j} &= \sqrt{1-\srcImpBnd[\AclassicalVal,\blockVar,j]}\ket{\tilde \Astate_{\AclassicalVal,\blockVar}^{(j)}}_{\Aprime''_j} +\sqrt{\srcImpBnd[\AclassicalVal,\blockVar,j]} \ket{\tilde \Astate_{\AclassicalVal,\blockVar}^{(j)\perp}}_{\Aprime''_j}\,,
    \end{align}
    as follows from the source-replacement scheme (\cref{lem:source_replacement_scheme}) and the expression for Alice's signal states given by \cref{eq:def_new_virtual_states_alice}. Then, the expression for the set of possible marginal states \cref{eq:explicit_construction_marginal_constraint} directly follows from tracing out the register $\Aprime''_j$ in \cref{eq:source_replaced_virtual_state}.
\end{proof}

\section{Imperfectly block-diagonal source}
\label{sec:imperfect_block_diagonal}

\noindent In this section, we no longer assume Alice's signal states are block diagonal in Fock space and therefore replace the source \cref{as:block_diagonal_source,as:imperfectly_characterized_source,as:probability_bounds} by the model described in \cref{ap:source_model_block}, while keeping the detector \cref{as:block_diagonal_detector,as:imperfectly_characterized_detector}. The new source model allows for imperfectly characterized non-block-diagonal signal states. It captures, for instance, imperfect phase randomization of weak coherent pulses (as illustrated in \cref{sec:imperfect_phase_randomization}),\footnote{The analysis can also incorporate imperfect single-photon sources which exhibit photon-number coherence (see \cref{rem:single_photon_sources}).} combined with intensity fluctuations and source and detector imperfections, which, to the best of our knowledge, no previous work handles simultaneously. Our treatment of imperfectly block-diagonal sources is inspired by Refs.~\cite{kamin_renyi_2025,nahar_imperfect_2023}, although the approach taken in this work differs significantly, in part due to the additional device imperfections. Moreover, the non-block-diagonal source model introduced is strictly more general than in Refs.~\cite{kamin_renyi_2025,nahar_imperfect_2023}, as it neither assumes coherent states nor restricts the form of the uncertainty about the coherences (e.g. also including imperfectly characterized single-photon sources).

The overall strategy is the following. Since Alice's signal states are not block diagonal in the same basis, Eve can no longer be assumed to perform a QND measurement of the block label, which the numerical analysis of \cref{ap:numerics} heavily relies on. Moreover, the photon-number coherence matrix of Alice's signal states may itself be imperfectly known.\footnote{For example, due to imperfect characterization or intensity fluctuations.} We therefore first construct nearby states with a known photon-number coherence matrix in \cref{ap:continuity_bound_block}, while keeping the underlying $n$-photon states imperfectly characterized.

We then apply the MEAT to the protocol with Alice's original signal states in \cref{ap:apply_meat_block}. At the level of the resulting single-round optimization, we replace these states by the states with known coherence matrix, and account for this replacement through a continuity bound and a relaxed constraint. The reason for keeping the $n$-photon states imperfectly characterized is that their uncertainty can be treated through the Gram matrix constraints, as in the main text, without paying an additional continuity bound cost.

Then, we essentially repeat the same analysis as in the main text. The main difference is that the source maps are now applied at the level of the single-round entropy, rather than at the level of the trace-norm (\cref{lem:epsilon_security_source_maps,lem:epsilon_security_squashing}). In \cref{ap:source_map_block}, we first introduce a source map which restores block diagonality. We then apply the tagging source map and partial characterization source map following the same steps as in \cref{sec:problem_reductions}. The resulting virtual states have the same block-diagonal form as those considered in the main text, so the numerical analysis of \cref{ap:numerics} applies, yielding the convex optimization problem presented in \cref{ap:convex_optimization_problem_block}, which is very similar to the one for block-diagonal source (\cref{th:convex_opti_numerics_imperfect}).

\subsection{Model assumptions}
\label{ap:source_model_block}
\noindent We replace the source \cref{as:block_diagonal_source,as:imperfectly_characterized_source,as:probability_bounds} by the following three assumptions (which can be viewed as generalizations of the previous assumptions, with just one additional restriction) while keeping the detector \cref{as:block_diagonal_detector,as:imperfectly_characterized_detector}.

\begin{assumptionbis}[Source Fock basis decomposition]
\label{as:source_fock_decomp}
    For each setting choice $\AclassicalVal\in\AclassicalAlph$, Alice's signal state in round $j$ can be expanded in the Fock basis as
    \begin{equation}
    \label{eq:signal_states_imperf_block_diagonal}
        (\bar{\Astate}_\AclassicalVal^{(j)})_{\Aprime_j} = \sum_{n,n'=0}^{\infty} \bar L_{n,n'|\AclassicalVal,j}\, \ketbra{\Astate_{\AclassicalVal,n}^{(j)}}{\Astate_{\AclassicalVal,n'}^{(j)}}_{\Aprime_j}\,,
    \end{equation}
    where $\ket{\Astate_{\AclassicalVal,n}^{(j)}}_{\Aprime_j}$ denotes Alice's $n$-photon state for setting choice $\AclassicalVal\in\AclassicalAlph$. We refer to the coefficient matrix $\bar L_{n,n'|\AclassicalVal,j}$ as the photon-number coherence matrix. Its diagonal elements describe the photon-number distribution, while its off-diagonal elements describe coherences between different photon-number blocks. 
\end{assumptionbis}

\begin{assumptionbis}[Imperfectly characterized photon-number states]
\label{as:imperfect_charac_nonblock}
    Alice's $n$-photon states are imperfectly characterized up to a cutoff $\phaseCutoff$, i.e. let $\ket{\tilde \Astate_{\AclassicalVal,n}^{(j)}}_{\Aprime_j}$ be known target states satisfying
    \begin{equation}
    \label{eq:mode_characterization_block}
        \left|\braket{\Astate_{\AclassicalVal,n}^{(j)}}{\tilde \Astate_{\AclassicalVal,n}^{(j)}}\right|^2 \geq 1-\srcImpBndMode[\AclassicalVal,n,j]
    \end{equation}
    for all $\AclassicalVal\in\AclassicalAlph$ and all $n\leq \phaseCutoff$.
\end{assumptionbis}

\begin{assumptionbis}[Imperfectly characterized coherence matrix]
\label{as:imperfect_charac_coherence_matrix}
    The coherence matrices $\bar L_{n,n'|\AclassicalVal,j}$ are imperfectly characterized, i.e. they belong to known sets $\mathcal{L}_{\AclassicalVal,j}$.
\end{assumptionbis}

\cref{as:block_diagonal_source} is relaxed by \cref{as:source_fock_decomp} by allowing coherences between different photon numbers through the coherence matrix $\bar L_{n,n'|\AclassicalVal,j}$ in \cref{eq:signal_states_imperf_block_diagonal}. The imperfect characterization condition from \cref{as:imperfectly_characterized_source} is retained in \cref{eq:mode_characterization_block}, with the only additional restriction that the $n$-photon states are pure.\footnote{We note that previous models for imperfect phase randomization also assume purity in each photon-number block \cite{nahar_imperfect_2023,kamin_renyi_2025}. These works however do not handle imperfect state preparation, intensity fluctuations, or detector imperfections combined.} This covers most source imperfections (including state preparation flaws), but not, for instance, Trojan-horse attacks, where Alice's $n$-photon states may be mixed due to Eve's injected light (see \cref{ap:trojan_horse_attack}). 

The set $\mathcal{L}_{\AclassicalVal,j}$ describes the uncertainty about both the photon-number distribution and the coherences, which may arise, for example, from imperfect characterization or intensity fluctuations. Thus, the imperfect characterization of the photon-number distribution in \cref{as:probability_bounds} is generalized to imperfect characterization of the full coherence matrix in \cref{as:imperfect_charac_coherence_matrix}.

\subsection{Target state and coherence cutoff}
\label{ap:continuity_bound_block}

\noindent In order to numerically diagonalize Alice's signal states and apply the analysis from \cref{ap:numerics}, we must address two issues. First, the coherence matrix is only imperfectly characterized (by \cref{as:imperfect_charac_coherence_matrix}), so we replace each signal state by a state with known coherence matrix. Second, the states may be infinite-dimensional (by \cref{as:source_fock_decomp}), in which case we introduce a cutoff $\phaseCutoff$ and remove the coherences between the $\leq\phaseCutoff$ and $>\phaseCutoff$ photon-number blocks. We can then numerically diagonalize the finite-dimensional $\leq\phaseCutoff$ block of the new states.

We remove the uncertainty about the coherence matrix by choosing a target state $(\Astate_{\AclassicalVal}^{(j)})_{\Aprime_j}$ given by \cref{eq:signal_states_imperf_block_diagonal} but with known coherence matrix $L_{n,n'|\AclassicalVal,j} \in \mathcal{L}_{\AclassicalVal,j}$,\footnote{In other words, the target states are chosen to have the same $n$-photon states as in \cref{eq:signal_states_imperf_block_diagonal}, but the corresponding coherence matrices are now known.} and bounding its distance to the states Alice actually prepares,
\begin{equation}
\label{eq:bound_uncertainty_block}
    \frac{1}{2}\norm{(\Astate_{\AclassicalVal}^{(j)})_{\Aprime_j} - (\bar{\Astate}_\AclassicalVal^{(j)})_{\Aprime_j}}_1  \leq \epsilon^\mathrm{coh}_{\AclassicalVal,j}\,.
\end{equation}
The trace distance can be evaluated by computing the maximum distance between any possible state given by \cref{eq:signal_states_imperf_block_diagonal} with $\bar L_{n,n'|\AclassicalVal,j} \in \mathcal{L}_{\AclassicalVal,j}$ from the target state. Next, similarly to \cite[Sec.~XII]{kamin_renyi_2025}, we use the cutoff $\phaseCutoff$ and define
\begin{equation}
\label{eq:def_coherence_cutoff_map}
    \mathcal{Z}^\mathrm{cut}\big[(\Astate_{\AclassicalVal}^{(j)})_{\Aprime_j}\big] \coloneqq \proj_{\leq\phaseCutoff}(\Astate_{\AclassicalVal}^{(j)})_{\Aprime_j}\proj_{\leq\phaseCutoff} + \proj_{>\phaseCutoff}(\Astate_{\AclassicalVal}^{(j)})_{\Aprime_j}\proj_{>\phaseCutoff}\,,
\end{equation}
which removes the coherences between the $\leq\phaseCutoff$ and $>\phaseCutoff$ photon-number blocks, and bound
\begin{equation}
\label{eq:bound_cutoff_block}
    \frac{1}{2}\norm{(\Astate_{\AclassicalVal}^{(j)})_{\Aprime_j} - \mathcal{Z}^\mathrm{cut}\big[(\Astate_{\AclassicalVal}^{(j)})_{\Aprime_j}\big]}_1 \leq \epsilon^{\mathrm{cut}}_{\AclassicalVal,j}\,.
\end{equation}
Both $\epsilon^\mathrm{coh}_{\AclassicalVal,j}$ and $\epsilon^{\mathrm{cut}}_{\AclassicalVal,j}$ are computable from the source model, since the coherence matrices $L_{n,n'|\AclassicalVal,j}$ of the target states and the set of possible coherence matrices $\mathcal{L}_{\AclassicalVal,j}$ for the actual states are known. We show how to compute these bounds for imperfectly phase-randomized weak-coherent pulses in \cref{sec:imperfect_phase_randomization}. By the triangle inequality, the states Alice prepare satisfy
\begin{equation}
\label{eq:bound_total_block}
    \frac{1}{2}\norm{(\bar{\Astate}_\AclassicalVal^{(j)})_{\Aprime_j} - \mathcal{Z}^\mathrm{cut}\big[(\Astate_{\AclassicalVal}^{(j)})_{\Aprime_j}\big]}_1 \leq \epsilon^\mathrm{coh}_{\AclassicalVal,j} + \epsilon^{\mathrm{cut}}_{\AclassicalVal,j}\,.
\end{equation}
where the coherence matrix of $\mathcal{Z}^\mathrm{cut}\big[(\Astate_{\AclassicalVal}^{(j)})_{\Aprime_j}\big]$ is fully known in Fock space, and coherence between the blocks $\leq\phaseCutoff$ and $>\phaseCutoff$ removed, which is what allows us to diagonalize it in the $\leq \phaseCutoff$ block.

\begin{remark}
    We remove the uncertainty about the coherence matrix (and account for the corresponding replacement through the continuity bound in \cref{ap:apply_meat_block}) because a known coherence matrix is required to diagonalize the reference state. By contrast, we keep the uncertainty about the $n$-photon signal states described by \cref{eq:mode_characterization_block}. Rather than removing this uncertainty as well, and paying an additional penalty as continuity bound depending on $\srcImpBndMode[\AclassicalVal,n,j]$, we later treat it through the Gram matrix constraints, analogously to the main text. 
    
    This treatment of intensity fluctuations is very conservative and leads to a substantial penalty in the continuity bound in \cref{ap:apply_meat_block}. Consequently, this approach is considerably less robust to intensity fluctuations than the analysis for block-diagonal sources in the main text. Developing a tighter treatment of intensity fluctuations with non-block-diagonal source is a natural direction for future work.
\end{remark}

\subsection{Applying the MEAT}
\label{ap:apply_meat_block}

\noindent Consider the protocol $\big\{\big\{\bar{\Astate}^{(j)}_{\Aclassical_j\Aprime_j}, \big\{\Bpovmel[j]\big\}_{\BclassicalVal\in\BclassicalAlph}, \annKeyMap[j]\big\}_{j=1}^\totRounds, \ppMap\big\}$ (see \cref{def:tuple_protocol}), where Alice's signal states satisfy \cref{as:source_fock_decomp,as:imperfect_charac_nonblock,as:imperfect_charac_coherence_matrix} and Bob's measurement POVMs satisfy \cref{as:block_diagonal_detector,as:imperfectly_characterized_detector}. Let $\bar{\Astate}^{(j)}_{\Amarg_j\Aprime_j}$ denote the corresponding source-replaced state (\cref{lem:source_replacement_scheme}) and $\setAliceMarginalsBar[j]$ denote the set of possible marginals (defined in terms of the source model, \cref{ap:source_model_block}). As in \cref{sec:problem_reductions}, we apply the flag-state squasher (\cref{lem:flag_state_squasher}) and the noise channel (\cref{lem:noise_channel}) to reduce the analysis to a new protocol $\big\{\big\{\bar{\Astate}^{(j)}_{\Aclassical_j\Aprime_j}, \big\{\BpovmelTarg[j]\big\}_{\BclassicalVal\in\BclassicalAlph}, \annKeyMap[j]\big\}_{j=1}^\totRounds, \ppMap\big\}$ (see \cref{lem:epsilon_security_squashing,lem:combination_FSS_noise_channel}) with the POVM elements from \cref{eq:def_new_virtual_povm_elements}, and with Eve restricted to the attack set $\setAttackChVirt_j$ given by \cref{eq:restricted_attack_set_fss_noise}. 

In contrast to the main text, we do not yet apply any source map. Instead, we apply the MEAT via \cite[Theorem~8.2 and Definition~8.3]{tupkary_rigorous_2026} directly to this protocol, and keep the formulation in terms of attack channels (instead of the marginal of attack channels). This is done so that the continuity bound below can be applied. 

\begin{remark}
    Note that Alice's signal states are still infinite dimensional at this stage, and thus \textit{technically} the MEAT cannot be applied directly. However, \cite[Theorem~A.1 and Remark~A.2]{tupkary_rigorous_2026} show that a security statement holding for all finite-dimensional truncations of the relevant registers extends to the infinite-dimensional setting. Therefore, we ignore this technical obstacle in this work, and leave its rigorous resolution to later work.
\end{remark}  Since we do not know $\bar{\Astate}^{(j)}_{\Amarg_j}$, we obtain a lower bound by optimizing over all possible marginals in $\setAliceMarginalsBar[j]$, analogously to \cref{sec:construction_set_output_states}. Applying \cite[Theorem~8.2 and Definition~8.3]{tupkary_rigorous_2026} therefore yields the single-round optimization
\begin{equation}
\label{eq:opti_imperfect_block_attack_channel}
    \inf_{\nu\in\stateSet_j(\setAliceMarginalsBar[j],\setAttackChVirt_j)} \fRenyiUpEnt{1}{j-1}(\secretReg_j|\eveCopyReg_j\varDecReg_j\Ereg_j\widetilde{E})_\nu\,,
\end{equation}
with
\begin{equation}
\label{eq:state_set_imperfect_block}
    \stateSet_j(\setAliceMarginalsBar[j],\setAttackChVirt_j) \coloneqq \left\{ \gMapVirt[j]\circ\attackCh_j\big[\omega_{\Amarg_j\Ereg_{j-1}\widetilde{E}}\big] : \Tr_{\Ereg_{j-1}\widetilde{E}}\big[\omega_{\Amarg_j\Ereg_{j-1}\widetilde{E}}\big] \in\setAliceMarginalsBar[j],\ \attackCh_j\in\setAttackChVirt_j \right\}\,,
\end{equation}
where $\widetilde{E}$ purifies $\Amarg_j\Ereg_{j-1}$. We expand the objective as in \cref{eq:expanded_entropy_f}, but without the block register $\blockVarReg_j$ (since we don't have block-diagonal structure yet), yielding
\begin{equation}
\label{eq:expanded_entropy_imperfect_block}
    \fRenyiUpEnt{1}{j-1}(\secretReg_j|\eveCopyReg_j\varDecReg_j\Ereg_j\widetilde{E})_\nu \geq \frac{\alpha}{1-\alpha}\log\left( (1-\ptest^{(j)}) 2^{\frac{1-\alpha}{\alpha}\left( \renyiup(\secretReg_j|\eveCopyReg_j\Ereg_j\widetilde{E})_{\nu_{|\mathtt{gen}}} - \ftradeoff{1}{j-1}(\varDecGen)\right)} + \sum_{\varDecVal\in\varDecAlph_{\mathtt{test}}} \nu_{\varDecReg_j}(\varDecVal)\, 2^{-\frac{1-\alpha}{\alpha}\ftradeoff{1}{j-1}(\varDecVal)} \right)\,.
\end{equation}

We now replace Alice's signal states by the states $\mathcal{Z}^\mathrm{cut}\big[(\Astate_{\AclassicalVal}^{(j)})_{\Aprime_j}\big]$ given by \cref{eq:def_coherence_cutoff_map} using the continuity bound for sandwiched Rényi entropies \cite{bluhm_unified_2026}. Applying \cite[Lemma~20]{kamin_renyi_2025} to \cref{eq:bound_total_block}, for every $\nu\in\stateSet_j(\setAliceMarginalsBar[j],\setAttackChVirt_j)$ there exists a state $\nu'\in\stateSet_j(\setAliceMarginalsPrime[j],\setAttackChVirt_j)$, such that
\begin{align}
\label{eq:trace_distance_full}
    \frac{1}{2}\norm{\nu - \nu'}_1 &\leq \sum_{\AclassicalVal\in\AclassicalAlph} \prob_{\Aclassical_j}(\AclassicalVal) \left(\epsilon^\mathrm{coh}_{\AclassicalVal,j}+\epsilon^{\mathrm{cut}}_{\AclassicalVal,j}\right) \eqqcolon \epsilon^{\mathrm{full}}_j\,, \\
    \label{eq:trace_distance_gen}
    \frac{1}{2}\norm{\nu_{|\mathtt{gen}} - \nu'_{|\mathtt{gen}}}_1 &\leq \sum_{\AclassicalVal\in\AclassicalAlph} \prob_{\Aclassical_j}(\AclassicalVal|\mathtt{gen}) \left(\epsilon^\mathrm{coh}_{\AclassicalVal,j}+\epsilon^{\mathrm{cut}}_{\AclassicalVal,j}\right) \eqqcolon \epsilon^{\mathrm{gen}}_j\,,
\end{align}
where \cref{eq:trace_distance_gen} is obtained by conditioning on generation rounds, and with the set $\setAliceMarginalsPrime[j]$ of possible marginals of the source-replaced states corresponding to $\mathcal{Z}^\mathrm{cut}[(\Astate_{\AclassicalVal}^{(j)})_{\Aprime_j}]$. Applying the continuity bound for sandwiched Rényi entropies \cite{bluhm_unified_2026} to \cref{eq:trace_distance_gen} yields
\begin{equation}
\label{eq:continuity_bound_block}
    \renyiup(\secretReg_j|\eveCopyReg_j\Ereg_j\widetilde{E})_{\nu_{|\mathtt{gen}}} \;\geq\; \renyiup(\secretReg_j|\eveCopyReg_j\Ereg_j\widetilde{E})_{\nu'_{|\mathtt{gen}}} - g(\epsilon^{\mathrm{gen}}_j)
\end{equation}
where we use the same improvement on the continuity bound from \cite[Eq.~(215)]{kamin_renyi_2025} since the register $\secretReg_j$ is classical, i.e.
\begin{equation}
\label{eq:continuity_bound_function}
    g(\epsilon) \coloneqq \min\begin{cases}
        \log(1+\epsilon) + \dfrac{1}{\alpha-1}\log\left(1 + \epsilon\, |\secretAlph|^{\alpha-1} - \dfrac{\epsilon^{\alpha}}{(1+\epsilon)^{\alpha-1}}\right)\,, \\
        \dfrac{\alpha}{\alpha-1} \log\left(1 + \epsilon\, |\secretAlph|^{\frac{\alpha-1}{\alpha}}\right)\,, \\
        \log(1+\epsilon) + \dfrac{\alpha}{\alpha-1}\log\left(1 + \epsilon\, |\secretAlph|^{\frac{\alpha-1}{\alpha}} - \dfrac{\epsilon^{2-\frac{1}{\alpha}}}{(1+\epsilon)^{\frac{\alpha-1}{\alpha}}}\right)\,.
    \end{cases}
\end{equation}
We can also relate the announcement distributions by restricting \cref{eq:trace_distance_full} to the register $\varDecReg_j$ which gives
\begin{equation}
\label{eq:announcement_trace_distance}
    \frac{1}{2}\norm{\nu_{\varDecReg_j}-\nu'_{\varDecReg_j}}_1 \leq \epsilon^{\mathrm{full}}_j\,.
\end{equation}
This directly links the distribution $\nu_{\varDecReg_j}$ of the actual source, which is the one the tradeoff function is evaluated on in \cref{eq:expanded_entropy_imperfect_block}, to the distribution $\nu'_{\varDecReg_j}$ generated by the source preparing $\mathcal{Z}^\mathrm{cut}\big[(\Astate_{\AclassicalVal}^{(j)})_{\Aprime_j}\big]$. This differs from Ref.~\cite{kamin_renyi_2025}, where the two distributions are related through the eigenvalues and eigenvectors of the two states. Here, the relation is expressed purely as a constraint on the observed statistics.

At this point, we pass to the formulation with marginal $\marginalSetVirt_j$ of the set of attack channels. Applying \cite[Lemma~8.8]{tupkary_rigorous_2026} (with zero tradeoff function) yields

\begin{align}
    \inf_{\nu'\in\stateSet_j(\setAliceMarginalsPrime[j],\setAttackChVirt_j)}\renyiup(\secretReg_j|\eveCopyReg_j\Ereg_j\widetilde{E})_{\nu'_{|\mathtt{gen}}} &\geq \inf_{\genDensity_{\Amarg_j\BmeasSquash_j}\in\stateSet'_j(\setAliceMarginalsPrime[j],\marginalSetVirt_j)}\renyiup   (\secretReg_j|\eveCopyReg_j\evePurReg)_{\gMapVirt[j][\purFunc(\genDensity_{\Amarg_j\BmeasSquash_j})]_{|\mathtt{gen}}}\label{eq:opti_problem_virtual_block} \\
    \stateSet'_j(\setAliceMarginalsPrime[j],\marginalSetVirt_j)&=\left\{\marginalMap_j[\genDensity_{\Amarg_j\Aprime_j}]:\genDensity_{\Amarg_j}\in\setAliceMarginalsPrime[j],\marginalMap_j\in\marginalSetVirt_j\right\}\,,
\end{align}
where $\purFunc$ is a purifying function of $\Amarg_j\BmeasSquash_j$ onto $\evePurReg$ and $\marginalSetVirt_j$ denotes the marginal of the restricted set of attack channels $\setAttackChVirt_j$, defined analogously to \cref{eq:marginal_definition_fss_noise} (but on the register $\Aprime_j$).

We have now replaced Alice's original signal states $(\bar{\Astate}_{\AclassicalVal}^{(j)})_{\Aprime_j}$ by the states $\mathcal{Z}^{\mathrm{cut}}[(\Astate_{\AclassicalVal}^{(j)})_{\Aprime_j}]$, with known coherence matrices and the coherences between the $\leq\phaseCutoff$ and $>\phaseCutoff$ photon-number blocks have been removed. For convenience, define $\nu' \coloneqq \gMapVirt[j]\big[\purFunc(\genDensity_{\Amarg_j\BmeasSquash_j})\big]\,$. Combining \cref{eq:expanded_entropy_imperfect_block,eq:continuity_bound_block,eq:announcement_trace_distance,eq:opti_problem_virtual_block}, we obtain the following lower bound on the single-round optimization in \cref{eq:opti_imperfect_block_attack_channel}
\begin{align}
    \inf_{\substack{\genDensity_{\Amarg_j\BmeasSquash_j}\in\stateSet'_j(\setAliceMarginalsPrime[j],\marginalSetVirt_j)\\ \nu_{\varDecReg_j}\in\probSimplex_{|\varDecAlph|}\\ \frac{1}{2}\norm{\nu_{\varDecReg_j}-\nu'_{\varDecReg_j}}_1\leq\epsilon^{\mathrm{full}}_j}} \frac{\alpha}{1-\alpha}\log&\Bigg((1-\ptest^{(j)})2^{\frac{1-\alpha}{\alpha}\left(\renyiup(\secretReg_j|\eveCopyReg_j\evePurReg)_{\nu'_{|\mathtt{gen}}}-g(\epsilon^{\mathrm{gen}}_j)-\ftradeoff{1}{j-1}(\varDecGen)\right)}+\sum_{\varDecVal\in\varDecAlph_{\mathtt{test}}}\nu_{\varDecReg_j}(\varDecVal)\,2^{-\frac{1-\alpha}{\alpha}\ftradeoff{1}{j-1}(\varDecVal)}\Bigg)
    \label{eq:block_opti_middle_step}
\end{align}
Here, $\nu'_{\varDecReg_j}$ is determined by $\genDensity_{\Amarg_j\BmeasSquash_j}$ (which only includes the nice states with known coherence matrix and finite cutoff), while $\nu_{\varDecReg_j}$ is an optimization variable describing the announcement distribution of the original protocol and is optimized subject to the trace distance constraint.

The remaining steps broadly repeat the analysis from the main text. We first diagonalize the known coherence matrix within the $\leq\phaseCutoff$ photon-number block and introduce a source map which restores the block-diagonal structure. We then apply the tagging source map and partial characterization source map as in \cref{sec:problem_reductions}, but now at the level of the single-round entropy. This yields virtual states of the same form as those considered in the main text, to which the numerical analysis of \cref{ap:numerics} applies.

\subsection{Application of source maps}
\label{ap:source_map_block}

\noindent We now introduce a series of source maps, similarly to \cref{sec:problem_reductions} but where we first introduce the block-tagging source map to ensure the signal states are block-diagonal in the same basis. Since the coherence matrix of the states $\mathcal{Z}^\mathrm{cut}\big[(\Astate_{\AclassicalVal}^{(j)})_{\Aprime_j}\big]$ is known and coherences between the $\leq\phaseCutoff$ and $>\phaseCutoff$ blocks removed, we may diagonalize it, yielding
\begin{equation}
\label{eq:diagonalization_nice_state}
    \mathcal{Z}^\mathrm{cut}\big[(\Astate_{\AclassicalVal}^{(j)})_{\Aprime_j}\big] = \sum_{\blockVar=0}^{\infty}\omega_{\blockVar|\AclassicalVal,j}\, \ketbra{\omega_{\AclassicalVal,\blockVar}^{(j)}}{\omega_{\AclassicalVal,\blockVar}^{(j)}}_{\Aprime_j}\,,
\end{equation}
where the eigenvalues $\omega_{\blockVar|\AclassicalVal,j}$ and the coefficients of the eigenvectors $\ket{\omega_{\AclassicalVal,\blockVar}^{(j)}}$ with $\blockVar\leq\phaseCutoff$ can be obtained numerically.\footnote{The eigenvectors themselves remain unknown, since only the coefficients $c_{n|\AclassicalVal,\blockVar,j}$ in \cref{eq:eigenstate_expansion} are known, whereas the $n$-photon states are not.} We refer to $\blockVar$ as the block label (in general, it does not correspond to the photon number).

\textbf{Block-tagging source map:} Let $A_{B_j}$ be a register with orthonormal basis $\{\ket{\blockVar}\}_\blockVar$, and
\begin{equation}
\label{eq:block_tagged_state}
    \sum_{\blockVar=0}^{\infty}\omega_{\blockVar|\AclassicalVal,j}\, \ketbra{\omega_{\AclassicalVal,\blockVar}^{(j)}}{\omega_{\AclassicalVal,\blockVar}^{(j)}}_{\Aprime_j} \otimes\ketbra{\blockVar}{\blockVar}_{A_{B_j}}\,.
\end{equation}
be Alice's new signal states for setting choice $\AclassicalVal\in\AclassicalAlph$. Then $\Tr_{A_{B_j}}$ is a source map from \cref{eq:block_tagged_state} to \cref{eq:diagonalization_nice_state}. By construction, Alice's signal states are now block-diagonal in the same basis with respect to the label $\blockVar$, which is precisely the structure required by the tagging source map.

\textbf{Tagging source map:} Applying \cref{lem:tagging_source_map} with cutoff $\tagCutoff\leq\phaseCutoff$, where the two blocks of \cref{lem:tagging_source_map} are taken to be $\blockVar\leq\tagCutoff$ and $\blockVar>\tagCutoff$, yields
\begin{equation}
\label{eq:tagged_state_imperfect_block}
    (\Astate'^{(j)}_{\AclassicalVal})_{\Aprime'_j} = \sum_{\blockVar=0}^{\tagCutoff}\omega_{\blockVar|\AclassicalVal,j}\, \ketbra{\omega_{\AclassicalVal,\blockVar}^{(j)}}{\omega_{\AclassicalVal,\blockVar}^{(j)}}_{\Aprime'_j} + \left(1-\sum_{\blockVar=0}^{\tagCutoff}\omega_{\blockVar|\AclassicalVal,j}\right) \ketbra{\AclassicalVal}{\AclassicalVal}_{\Aprime'_j}\,,
\end{equation}
where we have absorbed the block register $A_{B_j}$ into $\Aprime'_j$, and where $\{\ket{\AclassicalVal}\}_{\AclassicalVal\in\AclassicalAlph}$ is orthonormal and orthogonal to the span of the retained eigenstates. This is exactly the form of the tagged states in \cref{sec:problem_reductions}, with the photon-number distribution replaced by $\omega_{\blockVar|\AclassicalVal,j}$. In contrast to \cref{sec:problem_reductions}, these weights are known exactly, since any uncertainty about them was removed through the continuity bound in \cref{eq:continuity_bound_block}.

\textbf{Partial characterization source map:} The retained eigenstates lie in the $\leq\phaseCutoff$ photon-number subspace and can therefore be written as
\begin{equation}
\label{eq:eigenstate_expansion}
    \ket{\omega_{\AclassicalVal,\blockVar}^{(j)}}_{\Aprime_j} = \sum_{n=0}^{\phaseCutoff} c_{n|\AclassicalVal,\blockVar,j} \ket{\Astate_{\AclassicalVal,n}^{(j)}}_{\Aprime_j}\,,
\end{equation}
with coefficients $c_{n|\AclassicalVal,\blockVar,j}$ known from the diagonalization. We choose the corresponding target states to have the same photon-number coefficients,
\begin{equation}
\label{eq:target_eigenstate}
    \ket{\tilde{\omega}_{\AclassicalVal,\blockVar}^{(j)}}_{\Aprime_j} \coloneqq \sum_{n=0}^{\phaseCutoff} c_{n|\AclassicalVal,\blockVar,j} \ket{\tilde \Astate_{\AclassicalVal,n}^{(j)}}_{\Aprime_j}\,,
\end{equation}
where we recall that the $n$-photon states $\ket{\tilde \Astate_{\AclassicalVal,n}^{(j)}}_{\Aprime_j}$ are known (see \cref{ap:source_model_block}). Hence,
\begin{equation}
\label{eq:overlap_eigenstates}
    \braket{\tilde{\omega}_{\AclassicalVal,\blockVar}^{(j)}}{\omega_{\AclassicalVal,\blockVar}^{(j)}} = \sum_{n=0}^{\phaseCutoff}\left|c_{n|\AclassicalVal,\blockVar,j}\right|^2 \braket{\tilde \Astate_{\AclassicalVal,n}^{(j)}}{\Astate_{\AclassicalVal,n}^{(j)}}\,,
\end{equation}
and, using the source fidelity bound from \cref{eq:mode_characterization_block} yields
\begin{align}
\label{eq:fidelity_bound_eigenstates}
    \left|\braket{\tilde{\omega}_{\AclassicalVal,\blockVar}^{(j)}}{\omega_{\AclassicalVal,\blockVar}^{(j)}}\right|^2 &\geq \inf_{\substack{ \{v_{\AclassicalVal,n,j}\}_{n=0}^{\phaseCutoff} \\ \sqrt{1-\srcImpBndMode[\AclassicalVal,n,j]}\leq |v_{\AclassicalVal,n,j}| \leq 1}}\left|\sum_{n=0}^{\phaseCutoff} \left|c_{n|\AclassicalVal,\blockVar,j}\right|^2 v_{\AclassicalVal,n,j}\right|^2 \\
    &= \max\left\{0, \max_{0\leq n\leq \phaseCutoff}\left(\left|c_{n|\AclassicalVal,\blockVar,j}\right|^2\left(1+\sqrt{1-\srcImpBndMode[\AclassicalVal,n,j]}\right)-1\right)\right\}^2
    \eqqcolon 1-\srcImpBnd[\AclassicalVal,\blockVar,j]\,.
\end{align}
The inequality directly follows from \cref{eq:overlap_eigenstates,eq:mode_characterization_block} while the equality follows from the reverse triangle inequality and the normalization $\sum_{n=0}^{\phaseCutoff}|c_{n|\AclassicalVal,\blockVar,j}|^2=1$. Intuitively, the bound above accounts for possible interference between the contributions from different photon-number states (as their phases may not align, see \cref{rem:phase_characterization}). 

\begin{remark}
\label{rem:phase_characterization}
    In contrast to the block-diagonal source model considered in the main text (and in previous analyses \cite{sixto2026finitekeysecurityanalysisdecoystate,navarrete2026numericalsecurityanalysispractical,pereira_optimal_2025}), setting $\srcImpBndMode[\AclassicalVal,n,j]=0$ does not imply that the source is perfectly characterized when coherences between different photon-number states are present. Indeed, the condition in \cref{eq:mode_characterization_block} determines each $n$-photon state only up to a phase, which may depend on the photon number. The corresponding contributions to the eigenstate overlap in \cref{eq:fidelity_bound_eigenstates} may therefore interfere destructively, such that $\srcImpBnd[\AclassicalVal,\blockVar,j]$ need not vanish even when $\srcImpBndMode[\AclassicalVal,n,j]=0$ for every $n$.
    
    If additional information about the phases of the overlaps is available, including restrictions on how they may vary with the photon number, it can be incorporated by restricting the infimum in \cref{eq:fidelity_bound_eigenstates}, thereby yielding a tighter bound. Here, we assume that only the squared overlaps in \cref{eq:mode_characterization_block} are known (no information about the relative phase) and therefore optimize over all compatible phases.
\end{remark}

Following \cref{eq:fidelity_bound_eigenstates}, the fidelity of the eigenstates with respect to known target states is bounded. We may therefore apply the partial characterization source map from \cref{lem:source_map_imperfect_source}, yielding
\begin{equation}
\label{eq:final_virtual_state_imperfect_block}
    (\AstateVirt_{\AclassicalVal}^{(j)})_{\Aprime''_j} = \sum_{\blockVar=0}^{\tagCutoff}\omega_{\blockVar|\AclassicalVal,j}\, \ketbra{\AstateVirt_{\AclassicalVal,\blockVar}^{(j)}}{\AstateVirt_{\AclassicalVal,\blockVar}^{(j)}}_{\Aprime''_j} + \left(1-\sum_{\blockVar=0}^{\tagCutoff}\omega_{\blockVar|\AclassicalVal,j}\right) \ketbra{\AclassicalVal}{\AclassicalVal}_{\Aprime''_j}\,,
\end{equation}
with
\begin{equation}
\label{eq:final_virtual_eigenstate_imperfect_block}
    \ket{\AstateVirt_{\AclassicalVal,\blockVar}^{(j)}}_{\Aprime''_j} = \sqrt{1-\srcImpBnd[\AclassicalVal,\blockVar,j]} \ket{\tilde{\omega}_{\AclassicalVal,\blockVar}^{(j)}}_{\Aprime''_j} + \sqrt{\srcImpBnd[\AclassicalVal,\blockVar,j]} \ket{\tilde{\omega}_{\AclassicalVal,\blockVar}^{(j)\perp}}_{\Aprime''_j}\,,
\end{equation}
where $\braket{\tilde{\omega}_{\AclassicalVal,\blockVar}^{(j)}}{\tilde{\omega}_{\AclassicalVal,\blockVar}^{(j)\perp}} = 0$ for all $\AclassicalVal\in\AclassicalAlph$, $\blockVar\leq\tagCutoff$ and all rounds $j$. Note that \cref{eq:final_virtual_state_imperfect_block,eq:final_virtual_eigenstate_imperfect_block} have exactly the form of \cref{eq:def_new_virtual_states_alice}, with the photon-number distribution replaced by $\omega_{\blockVar|\AclassicalVal,j}$ and with the target states and source fidelity bound given by \cref{eq:target_eigenstate,eq:fidelity_bound_eigenstates} instead. Denoting by $\srcMap_j^\mathrm{tot}$ the composition of the block-tagging source map, the tagging source map (\cref{lem:tagging_source_map}) and the partial characterization source map (\cref{lem:source_map_imperfect_source}), we have
\begin{equation}
\label{eq:source_map_composition_block}
    \srcMap_j^\mathrm{tot}\big[(\AstateVirt_{\AclassicalVal}^{(j)})_{\Aprime''_j}\big] = \mathcal{Z}^\mathrm{cut}\big[(\Astate_{\AclassicalVal}^{(j)})_{\Aprime_j}\big]
\end{equation}
for every $\AclassicalVal\in\AclassicalAlph$ and every round $j$. We now wish to apply this source map argument to \cref{eq:opti_problem_virtual_block}, where we have a restricted set of attack channels (in contrast to the argument in \cref{sec:problem_reductions}, where we first applied the source maps and then the squashing maps).

\begin{lemma}[Source maps with restricted marginal channels]
\label{lem:source_map_restricted_attacks}
    Let $\srcMap_j\in\cptp(\Aprime'_j,\Aprime_j)$ be a source map relating $\srcMap_j\big[(\sigma_\AclassicalVal^{(j)})_{\Aprime'_j}\big] = (\rho_\AclassicalVal^{(j)})_{\Aprime_j}$ for all $\AclassicalVal\in\AclassicalAlph$. Let $\sigma^{(j)}_{\Amarg_j\Aprime'_j}$ and $\rho^{(j)}_{\Amarg_j\Aprime_j}$ be the corresponding source-replaced states. Let $H_j$ be a Hermitian operator on $\BmeasSquash_j$ and let $\marginalMap_j\in\cptp(\Aprime_j,\BmeasSquash_j)$ satisfy $\marginalMap_j^\dagger[H_j]\succeq0$. Then the marginal channel $\widetilde{\marginalMap}_j\coloneqq\marginalMap_j\circ\srcMap_j$ satisfies the same constraint,
    \begin{equation}
    \label{eq:lemma_constraint_preserved}
        \widetilde{\marginalMap}_j^\dagger[H_j] = \srcMap_j^\dagger\big[\marginalMap_j^\dagger[H_j]\big]\succeq0\,.
    \end{equation}
    Let 
    \begin{equation}
        \nu' \coloneqq \gMapVirt[j]\big[\purFunc(\marginalMap_j\big[\rho^{(j)}_{\Amarg_j\Aprime_j}\big])\big]\,,
        \qquad
        \widetilde{\nu} \coloneqq \gMapVirt[j]\big[\purFunc(\widetilde{\marginalMap}_j\big[\sigma^{(j)}_{\Amarg_j\Aprime'_j}\big])\big]\,,
    \end{equation}
    then
    \begin{equation}
    \label{eq:announcement_identity}
        \nu'_{\varDecReg_j} = \widetilde{\nu}_{\varDecReg_j}
    \end{equation}
    and
    \begin{equation}
    \label{eq:lemma_entropy_inequality}
        \renyiup(\secretReg_j|\eveCopyReg_j\evePurReg)_{\nu'_{|\mathtt{gen}}} \geq \renyiup(\secretReg_j|\eveCopyReg_j\evePurReg)_{\widetilde{\nu}_{|\mathtt{gen}}}\,,
    \end{equation}
    where $\evePurReg$ denotes a purifying register.
\end{lemma}

\begin{proof}
    \cref{eq:lemma_constraint_preserved} follows because the adjoint of a completely positive map preserves positivity. The states on Bob's register satisfy
    \begin{equation}
    \label{eq:attack_relation_virtual}
        \widetilde{\marginalMap}_j\big[(\sigma_\AclassicalVal^{(j)})_{\Aprime'_j}\big] = \marginalMap_j\circ\srcMap_j\big[(\sigma_\AclassicalVal^{(j)})_{\Aprime'_j}\big] = \marginalMap_j\big[(\rho_\AclassicalVal^{(j)})_{\Aprime_j}\big]
    \end{equation}
    for every $\AclassicalVal\in\AclassicalAlph$. Since the two source-replaced states have the same setting choice probabilities, and following \cref{eq:attack_relation_virtual}, Bob receives the same state, the Alice-Bob states (after Alice's measurement) are identical and produce the same distribution on $\varDecReg_j$, from which \cref{eq:announcement_identity} follows. 
    
    The argument for \cref{eq:lemma_entropy_inequality} is similar to \cite[Theorem~18]{kamin_renyi_2025}. Let $V_{\marginalMap_j}:\Aprime_j\to\BmeasSquash_jE_j^{\marginalMap}$ be a Stinespring isometry of $\marginalMap_j$, $V_{\srcMap_j}\colon\Aprime'_j\to\Aprime_jE_j^{\mathrm{src}}$ be a Stinespring isometry of $\srcMap_j$, and define $V_{\widetilde{\marginalMap}_j}\coloneqq V_{\marginalMap_j}V_{\srcMap_j}$.
    For every $\AclassicalVal\in\AclassicalAlph$,
    \begin{equation}
    \label{eq:trace_relation_source_map_meat}
        \Tr_{E_j^{\mathrm{src}}}\Big[V_{\widetilde{\marginalMap}_j}(\sigma_\AclassicalVal^{(j)})_{\Aprime'_j}V_{\widetilde{\marginalMap}_j}^\dagger\Big] = V_{\marginalMap_j}(\rho_\AclassicalVal^{(j)})_{\Aprime_j}V_{\marginalMap_j}^\dagger
    \end{equation}
    is the state Bob and Eve share. Therefore, after Alice measures the marginal $\Amarg_j$, we have
    \begin{equation}
        \Tr_{E_j^{\mathrm{src}}}\left[\sum_{\AclassicalVal\in\AclassicalAlph} \prob_{\Aclassical_j}(\AclassicalVal)\ketbra{\AclassicalVal}{\AclassicalVal}_{\Aclassical_j} \otimes V_{\widetilde{\marginalMap}_j}(\sigma_\AclassicalVal^{(j)})_{\Aprime'_j}V_{\widetilde{\marginalMap}_j}^\dagger \right] = \sum_{\AclassicalVal\in\AclassicalAlph} \prob_{\Aclassical_j}(\AclassicalVal)\ketbra{\AclassicalVal}{\AclassicalVal}_{\Aclassical_j} \otimes V_{\marginalMap_j}(\rho_\AclassicalVal^{(j)})_{\Aprime_j}V_{\marginalMap_j}^\dagger
    \end{equation}
    and as the remaining protocol steps only act on Alice and Bob's registers, the corresponding protocol outputs are related by
    \begin{equation}
        \nu' = \Tr_{E_j^{\mathrm{src}}}\big[\widetilde{\nu}\big]\,.
    \end{equation}
    Giving $E_j^{\mathrm{src}}$ to Eve and applying data processing \cite[Corollary~5.5]{tomamichel_quantum_2016} gives \cref{eq:lemma_entropy_inequality}, where $E_j^{\mathrm{src}}$ is included in the purifying register $\evePurReg$ of the output $\widetilde{\nu}$. 
\end{proof}

We now apply \cref{lem:source_map_restricted_attacks} to the source map $\srcMap_j^{\mathrm{tot}}$ in \cref{eq:source_map_composition_block} to lower bound \cref{eq:block_opti_middle_step}. Let $\setAliceMarginalsVirt[j]$ denote the set of possible source-replaced marginals corresponding to the virtual signal states in \cref{eq:final_virtual_state_imperfect_block}. For every $\marginalMap_j\in\marginalSetVirt_j$, \cref{lem:source_map_restricted_attacks} shows that the induced marginal channel $\widetilde{\marginalMap}_j \coloneqq \marginalMap_j\circ\srcMap_j^{\mathrm{tot}}$ belongs to $\marginalSetTilde_j$ defined as in \cref{eq:marginal_definition_fss_noise}, i.e. it satisfies the same constraints as $\marginalSetVirt_j$. Then, following \cref{eq:announcement_identity,eq:lemma_entropy_inequality}, for every fixed $\nu_{\varDecReg_j}$ we have
\begin{align}
\label{eq:source_map_infimum_block}
    \inf_{\substack{\genDensity_{\Amarg_j\BmeasSquash_j}\in\stateSet'_j(\setAliceMarginalsPrime[j],\marginalSetVirt_j)\\ \frac{1}{2}\norm{\nu_{\varDecReg_j}-\nu'_{\varDecReg_j}}_1\leq\epsilon^{\mathrm{full}}_j}}\renyiup(\secretReg_j|\eveCopyReg_j\evePurReg)_{\nu'_{|\mathtt{gen}}} \geq\inf_{\substack{\genDensity_{\Amarg_j\BmeasSquash_j}\in\stateSet'_j(\setAliceMarginalsVirt[j],\marginalSetTilde_j)\\ \frac{1}{2}\norm{\nu_{\varDecReg_j}-\nu'_{\varDecReg_j}}_1\leq\epsilon^{\mathrm{full}}_j}}\renyiup(\secretReg_j|\eveCopyReg_j\evePurReg)_{\nu'_{|\mathtt{gen}}} 
\end{align}
after enlarging the set on the r.h.s to include all possible output states in $\stateSet'_j\big(\setAliceMarginalsVirt[j],\marginalSetTilde_j\big)$ (and not just the attacks restricted to applying the source map first) which satisfy the constraint
\begin{equation}
\label{eq:virtual_announcement_trace_distance}
    \frac{1}{2}\norm{\nu_{\varDecReg_j}-\nu'_{\varDecReg_j}}_1\leq\epsilon_j^{\mathrm{full}}\,,
\end{equation}
where we recall that $\nu' = \gMapVirt[j]\big[\purFunc(\genDensity_{\Amarg_j\BmeasSquash_j})\big]\,$. Finally, applying \cref{eq:source_map_infimum_block} and taking the infimum over $\nu_{\varDecReg_j}$ directly yields a lower bound on \cref{eq:block_opti_middle_step}.

\subsection{Relaxed decoy constraints}
\label{eq:relaxed_decoy_block}

\noindent Before stating the final optimization problem, we derive the relaxed decoy constraints analogously to \cref{ap:relaxed_decoy}. In contrast to \cref{eq:overlap_target_intensity_leakage}, due to the non-block-diagonal source, the target states defined in \cref{eq:target_eigenstate} cannot in general be assumed to be independent of the intensity choice. Thus, even for the same encoding choice, the target states corresponding to different intensity choices may differ. The triangle inequality for the purified distance, together with \cref{eq:final_virtual_eigenstate_imperfect_block}, implies that
\begin{equation}
\label{eq:relaxed_decoy_bound_imperfect_block}
    \left|\braket{\AstateVirt_{\AclassicalVal_a\AclassicalVal_\mu,\blockVar}^{(j)}} {\AstateVirt_{\AclassicalVal_a\AclassicalVal_{\mu'},\blockVar}^{(j)}}\right|^2 \geq 1 -\tilde{\zeta}_{\AclassicalVal_a,\AclassicalVal_\mu,\AclassicalVal_{\mu'},\blockVar,j}\,,
\end{equation}
for all $\AclassicalVal_a\in\tilde{\AclassicalAlph}$, $\AclassicalVal_\mu,\AclassicalVal_{\mu'}\in\intensityAlph$, and $\blockVar\leq\tagCutoff$, where we define
\begin{equation}
    \tilde{\zeta}_{\AclassicalVal_a,\AclassicalVal_\mu,\AclassicalVal_{\mu'},\blockVar,j}\coloneqq\left(\sqrt{\srcImpBnd[\AclassicalVal_a,\AclassicalVal_\mu,\blockVar,j]}+\sqrt{1-\left|\bra{\tilde\omega_{\AclassicalVal_a,\AclassicalVal_\mu,\blockVar}^{(j)}}\ket{\tilde\omega_{\AclassicalVal_a,\AclassicalVal_{\mu'},\blockVar}^{(j)}}\right|^2}+\sqrt{\srcImpBnd[\AclassicalVal_a,\AclassicalVal_{\mu'},\blockVar,j]}\right)^2
\end{equation}
for convenience. Compared with \cref{eq:fidelity_relaxed_decoy}, the second term accounts for the possible dependence of the target states on the intensity choice. The overlap $\left|\bra{\tilde\omega_{\AclassicalVal_a,\AclassicalVal_\mu,\blockVar}^{(j)}}\ket{\tilde\omega_{\AclassicalVal_a,\AclassicalVal_{\mu'},\blockVar}^{(j)}}\right|^2$ can be computed directly from \cref{eq:target_eigenstate}, since the coefficients of the target states are known from the diagonalization. Consequently, the relaxed decoy constraints derived in \cref{ap:relaxed_decoy} remain valid with the bound $\zeta_{\AclassicalVal_a,\AclassicalVal_\mu,\AclassicalVal_{\mu'},\blockVar,j}$ replaced by $\tilde{\zeta}_{\AclassicalVal_a,\AclassicalVal_\mu,\AclassicalVal_{\mu'},\blockVar,j}$.

\subsection{Security statement and convex optimization problem}
\label{ap:convex_optimization_problem_block}

\noindent We now combine the previous argument to obtain a security statement and the corresponding convex optimization problem for imperfectly block-diagonal sources.

\begin{theorem}[Security with an imperfectly block-diagonal source and imperfectly characterized devices]
\label{th:convex_opti_numerics_imperfect_block}
Consider the \nameref{prot:generic_qkd_protocol} given by $\big\{\big\{\bar{\Astate}^{(j)}_{\Aclassical_j\Aprime_j},\big\{\Bpovmel[j]\big\}_{\BclassicalVal\in\BclassicalAlph},\annKeyMap[j]\big\}_{j=1}^{\totRounds},\ppMap\big\}$, where Alice's signal states satisfy the source model from \cref{ap:source_model_block} and \cref{as:block_diagonal_detector,as:imperfectly_characterized_detector} are satisfied on the detector side. Let $W^{\Bmeas_j}$ be an outcome and $\lambdaMin>0$ satisfying $W^{\Bmeas_j}\geq\lambdaMin\proj_{>\fssCutoff}^{\Bmeas_j}$, where $\proj_{>\fssCutoff}^{\Bmeas_j}$ is the projector onto the subspace with more than $\fssCutoff$ photons. For each round $j$ and every value $\varDecVal_1^{j-1}$, let $\ftradeoff{1}{j-1}$ be a tradeoff function on the register $\varDecReg_j$. Then the protocol is $(\epsSecr+\epsCor)$-secure for all $\alpha\in(1,2)$ if the final key length $\keyl$ satisfies
\begin{equation}
\keyl(\varDecVal_1^\totRounds)=\max\left\{0,\left\lfloor\fFull(\varDecVal_1^\totRounds)-\ECcost(\varDecVal_1^\totRounds)-\left\lceil\log\left(\frac{1}{\epsCor}\right)\right\rceil-\frac{\alpha}{\alpha-1}\log\left(\frac{1}{\epsSecr}\right)+2\right\rfloor\right\}\,,
\end{equation}
where
\begin{equation}
\fFull(\varDecVal_1^\totRounds)\coloneqq\sum_{j=1}^\totRounds\left(\ftradeoff{1}{j-1}(\varDecVal_j)+\kappaQKDVirt_j\left(\ftradeoff{1}{j-1},\setAliceMarginalsVirt[j],\gMapVirt[j]\right)\right)\,,
\end{equation}
and
\begin{equation}
\label{eq:optimization_imperfect_block_devices}
\kappaQKDVirt_j\left(\ftradeoff{1}{j-1},\setAliceMarginalsVirt[j],\gMapVirt[j]\right)\leq\inf_{\nu\in\stateSet_j(\setAliceMarginalsBar[j],\setAttackChVirt_j)}\fRenyiUpEnt{1}{j-1}(\secretReg_j|\eveCopyReg_j\varDecReg_j\Ereg_j\widetilde{E})_\nu\,.
\end{equation}
The quantity $\kappaQKDVirt_j\left(\ftradeoff{1}{j-1},\setAliceMarginalsVirt[j],\gMapVirt[j]\right)$ may be chosen according to the convex optimization problem described in \cref{fig:opti_problem_imperfect_block}. Here, the block weights $\omega_{\blockVar|\AclassicalVal,j}$ are given by \cref{eq:diagonalization_nice_state}, the target states $\ket{\tilde{\omega}_{\AclassicalVal,\blockVar}^{(j)}}$ are given by \cref{eq:target_eigenstate}, the source fidelity bounds $\srcImpBnd[\AclassicalVal,\blockVar,j]$ are given by \cref{eq:fidelity_bound_eigenstates}, and the deviations $\epsilon_j^{\mathrm{gen}}$ and $\epsilon_j^{\mathrm{full}}$ are given by \cref{eq:trace_distance_gen,eq:trace_distance_full}. We let $\characCutoff\leq\tagCutoff\leq\phaseCutoff$ be fixed cutoffs and use the shorthand
\begin{equation}
\prob_{\Aclassical_j,\blockVarReg_j}(\AclassicalVal,\blockVar)\coloneqq\prob_{\Aclassical_j}(\AclassicalVal)\omega_{\blockVar|\AclassicalVal,j}\,.
\end{equation}
Furthermore, $\renyiMblock(\nu'_{|\mathtt{gen}})$ is defined as in \cref{th:convex_opti_numerics_imperfect}, with $\nu$ replaced by $\nu'$. The state $\bar{\Astate}^{(j)}_{\Amarg_j\Aprime_j}$ denotes the source-replaced state in the $j$th round corresponding to $\bar{\Astate}^{(j)}_{\Aclassical_j\Aprime_j}$, with marginal contained in $\setAliceMarginalsBar[j]$, and $\widetilde{E}$ purifies $\Amarg_j\Ereg_{j-1}$. The set $\setAliceMarginalsVirt[j]$ contains the source-replaced marginals corresponding to the virtual signal states in \cref{eq:final_virtual_state_imperfect_block}. The protocol maps $\{\gMapVirt[j]\}_j$ from the \nameref{prot:entanglement_qkd_protocol} are given by the squashed POVMs $\{\BpovmelTarg[j]\}_{j,\BclassicalVal\in\BclassicalAlph}$ from \cref{eq:def_new_virtual_povm_elements}. Finally, $\squashMap_j^\dagger\big[\widetilde{W}^{\Bmeas_j'}\big]=W^{\Bmeas_j}$, where $\squashMap_j$ is given by \cref{eq:explicit_noise_channel_construction}.
\end{theorem}

\begin{proof}
    The proof for $\epsSecu$-security follows the same argument as \cref{th:security_imperfect_devices}. The right-hand side of \cref{eq:optimization_imperfect_block_devices} is lower bounded by \cref{eq:block_opti_middle_step}. Applying \cref{eq:source_map_infimum_block} for each fixed $\nu_{\varDecReg_j}$ and then minimizing over $\nu_{\varDecReg_j}$ yields a lower bound in terms of the virtual states in \cref{eq:final_virtual_state_imperfect_block}, subject to the announcement constraint in \cref{eq:virtual_announcement_trace_distance}.
    
    The virtual states in \cref{eq:final_virtual_state_imperfect_block} have the same block-diagonal form as the states considered in \cref{th:convex_opti_numerics_imperfect}. The source, detector, and statistics constraints therefore follow by the same argument, with the photon-number probabilities replaced by the exact block weights $\omega_{\blockVar|\AclassicalVal,j}$ and with the target states and source fidelity bounds given by \cref{eq:target_eigenstate,eq:fidelity_bound_eigenstates}.
    
    The relaxed decoy constraints follow from \cref{eq:relaxed_decoy_bound_imperfect_block} by the same data-processing and measurement arguments as in \cref{ap:relaxed_decoy}, with the bound $\zeta_{\AclassicalVal_a,\AclassicalVal_\mu,\AclassicalVal_{\mu'},\blockVar,j}$ replaced by $\tilde{\zeta}_{\AclassicalVal_a,\AclassicalVal_\mu,\AclassicalVal_{\mu'},\blockVar,j}$. Convexity follows as in \cref{th:convex_opti_numerics_imperfect}. Indeed, the additional statistics constraint is a norm inequality, while the modified source and decoy constraints are affine and positive-semidefinite constraints.
\end{proof}

\begin{figure}[!htbp]
    \centering
    \caption{Convex optimization problem for \cref{th:convex_opti_numerics_imperfect_block}.}
    \vspace*{-0.2cm}
    \label{fig:opti_problem_imperfect_block}
\begin{constraintblock}{Objective}
\[
\begin{aligned}
    &\mathrm{minimize}\quad\frac{\alpha}{1-\alpha}\log\left((1-p_\mathtt{test}^{(j)})2^{\frac{1-\alpha}{\alpha}\left(\renyiMblock(\nu'_{|\mathtt{gen}})-g(\epsilon^{\mathrm{gen}}_j)-f(\varDecGen)\right)}+\sum_{\varDecVal\in\varDecAlph_\mathtt{test}}\nu_{\varDecReg_j}(\varDecVal)\,2^{\frac{\alpha-1}{\alpha}f(\varDecVal)}\right)
\end{aligned}
\]
\end{constraintblock}
\begin{constraintblock}{Optimization variables}
\[
\begin{aligned}
    &\{\rho_{\Ameas_j\BmeasSquash_j\land\Ashield_j=\blockVar}\in\setDensity_\leq(\Ameas_j\BmeasSquash_j)\}_{\blockVar=0}^{\characCutoff} &\{\gramMatrix^{(\blockVar,j)}\in\setPos(2|\AclassicalAlph|)\}_{\blockVar\in\{0,\ldots,\characCutoff\}} &\qquad\nu_{\varDecReg_j}\in\probSimplex_{|\varDecAlph|} \\
    &\{\rho_{\BmeasSquash_j\land\Ashield_j=\AclassicalVal}\in\setDensity_\leq(\BmeasSquash_j)\}_{\AclassicalVal\in\AclassicalAlph} &\{\yield^{\AclassicalVal,\blockVar}\in\probSimplex_{|\BclassicalAlph|}\}_{\AclassicalVal\in\AclassicalAlph,\blockVar\in\{\characCutoff+1,\ldots,\tagCutoff\}} &
\end{aligned}
\]
\end{constraintblock}
\begin{constraintblock}{Source constraints}
\[
\begin{aligned}
    &\Tr[\genDensity_{\BmeasSquash_j\land\Ashield_j=\AclassicalVal}]=\prob_{\Aclassical_j}(\AclassicalVal)\left(1-\sum_{\blockVar=0}^{\tagCutoff}\omega_{\blockVar|\AclassicalVal,j}\right) &\forall\AclassicalVal\in\AclassicalAlph \\
    &\frac{\bra{\AclassicalVal'}_{\Ameas_j}\genDensity_{\Ameas_j\land\Ashield_j=\blockVar}\ket{\AclassicalVal}_{\Ameas_j}}{\sqrt{\prob_{\Aclassical_j}(\AclassicalVal)\prob_{\Aclassical_j}(\AclassicalVal')}}=\sqrt{\left(1-\srcImpBnd[\AclassicalVal,\blockVar,j]\right)\left(1-\srcImpBnd[\AclassicalVal',\blockVar,j]\right)}\gramMatrix^{(\blockVar,j)}_{\AclassicalVal,\AclassicalVal'}+\sqrt{\left(1-\srcImpBnd[\AclassicalVal,\blockVar,j]\right)\srcImpBnd[\AclassicalVal',\blockVar,j]}\gramMatrix^{(\blockVar,j)}_{\AclassicalVal,\AclassicalVal'+|\AclassicalAlph|} \hspace*{-5cm}\nonumber \\
    &+\sqrt{\srcImpBnd[\AclassicalVal,\blockVar,j]\left(1-\srcImpBnd[\AclassicalVal',\blockVar,j]\right)}\gramMatrix^{(\blockVar,j)}_{\AclassicalVal+|\AclassicalAlph|,\AclassicalVal'}+\sqrt{\srcImpBnd[\AclassicalVal,\blockVar,j]\srcImpBnd[\AclassicalVal',\blockVar,j]}\gramMatrix^{(\blockVar,j)}_{\AclassicalVal+|\AclassicalAlph|,\AclassicalVal'+|\AclassicalAlph|} &\forall\AclassicalVal,\AclassicalVal'\in\AclassicalAlph,\blockVar\in\{0,\ldots,\characCutoff\}\nonumber \\
    &\gramMatrix^{(\blockVar,j)}_{\AclassicalVal,\AclassicalVal'}=\sqrt{\omega_{\blockVar|\AclassicalVal,j}\omega_{\blockVar|\AclassicalVal',j}}\,\bra{\tilde\omega_{\AclassicalVal,\blockVar}^{(j)}}\ket{\tilde\omega_{\AclassicalVal',\blockVar}^{(j)}} &\forall\AclassicalVal,\AclassicalVal'\in\AclassicalAlph,\AclassicalVal\neq\AclassicalVal',\blockVar\in\{0,\ldots,\characCutoff\}\nonumber \\
    &\gramMatrix^{(\blockVar,j)}_{\AclassicalVal,\AclassicalVal}=\gramMatrix^{(\blockVar,j)}_{\AclassicalVal+|\AclassicalAlph|,\AclassicalVal+|\AclassicalAlph|}=\omega_{\blockVar|\AclassicalVal,j} &\forall\AclassicalVal\in\AclassicalAlph,\blockVar\in\{0,\ldots,\characCutoff\}\nonumber \\
    &\gramMatrix^{(\blockVar,j)}_{\AclassicalVal+|\AclassicalAlph|,\AclassicalVal}=\gramMatrix^{(\blockVar,j)}_{\AclassicalVal,\AclassicalVal+|\AclassicalAlph|}=0 &\forall\AclassicalVal\in\AclassicalAlph,\blockVar\in\{0,\ldots,\characCutoff\}
\end{aligned}
\]
\end{constraintblock}
\begin{constraintblock}{Detector constraints}
\[
\begin{aligned}
    &\sum_{\blockVar=0}^{\characCutoff}\Tr\left[\widetilde{W}^{\Bmeas_j'}\rho_{\BmeasSquash_j\land\Ashield_j=\blockVar}\right]+\sum_{\blockVar=\characCutoff+1}^{\tagCutoff}\sum_{\AclassicalVal\in\AclassicalAlph}\prob_{\Aclassical_j,\blockVarReg_j}(\AclassicalVal,\blockVar)(\yield^{\AclassicalVal,\blockVar})_{\BclassicalVal=o}+\sum_{\AclassicalVal\in\AclassicalAlph}\Tr\left[\widetilde{W}^{\Bmeas_j'}\rho_{\BmeasSquash_j\land\Ashield_j=\AclassicalVal}\right]\nonumber \\
    &\geq\lambdaMin\left(1-\sum_{\blockVar=0}^{\fssCutoff}\frac{1}{1-\dtImpBnd[\blockVar,j]}\left(\sum_{\blockVar'=0}^{\characCutoff}\Tr\left[\proj_\blockVar^{\BmeasSquash_j}\rho_{\BmeasSquash_j\land\Ashield_j=\blockVar'}\right]+\sum_{\AclassicalVal\in\AclassicalAlph}\Tr\left[\proj_\blockVar^{\BmeasSquash_j}\rho_{\BmeasSquash_j\land\Ashield_j=\AclassicalVal}\right]\right)-\sum_{\blockVar'=\characCutoff+1}^{\tagCutoff}\sum_{\AclassicalVal\in\AclassicalAlph}\frac{\prob_{\Aclassical_j,\blockVarReg_j}(\AclassicalVal,\blockVar')}{1-\dtImpBndMax[j]}\right)
\end{aligned}
\]
\end{constraintblock}
\begin{constraintblock}{Statistics constraints}
\[
\begin{aligned}
    \frac{1}{2}\norm{\nu_{\varDecReg_j}-\nu'_{\varDecReg_j}}_1\leq\epsilon^{\mathrm{full}}_j
\end{aligned}
\]
\[
\begin{aligned}
    \nu'_{\varDecReg_j}\coloneqq\sum_{\bar c\in\interAlph}\sum_{\substack{\AclassicalVal\in\AclassicalAlph\\\BclassicalVal\in\BclassicalAlph}}\Bigg(&\annFunc[j](\bar c|\AclassicalVal,\BclassicalVal)\Bigg(\sum_{\blockVar=0}^{\characCutoff}\Tr[\BpovmelTarg[j]\genDensity_{\BmeasSquash_j\land\Aclassical_j=\AclassicalVal\land\Ashield_j=\blockVar}] \\
    &+\sum_{\blockVar=\characCutoff+1}^{\tagCutoff}\prob_{\Aclassical_j,\blockVarReg_j}(\AclassicalVal,\blockVar)(\yield^{\AclassicalVal,\blockVar})_{\BclassicalVal}+\Tr[\BpovmelTarg[j]\genDensity_{\BmeasSquash_j\land\Ashield_j=\AclassicalVal}]\Bigg)\unitStatsVec[\bar c]\Bigg)
\end{aligned}
\]
\end{constraintblock}

\begin{constraintblock}{Decoy constraints}
\[
\begin{aligned}
    & \left|(\yield^{\AclassicalVal_a,\AclassicalVal_\mu,\blockVar})_{\BclassicalVal}-(\yield^{\AclassicalVal_a,\AclassicalVal_{\mu'},\blockVar})_{\BclassicalVal}\right|\leq\sqrt{\tilde{\zeta}_{\AclassicalVal_a,\AclassicalVal_\mu,\AclassicalVal_{\mu'},\blockVar,j}} & \hspace*{-4cm}\forall \BclassicalVal\in\BclassicalAlph,\blockVar\in\{\characCutoff+1,\ldots,\tagCutoff\},\AclassicalVal_a\in\tilde{\AclassicalAlph},\AclassicalVal_\mu,\AclassicalVal_{\mu'}\in\intensityAlph \nonumber\\
    & \frac{\Tr\left[\BpovmelTarg[j]\genDensity_{\BmeasSquash_j\land\Aclassical_j=\AclassicalVal_a\AclassicalVal_\mu\land\Ashield_j=\blockVar}\right]}{\prob_{\Aclassical_j,\blockVarReg_j}(\AclassicalVal_a\AclassicalVal_\mu,\blockVar)}\geq\frac{\Tr\left[\BpovmelTarg[j]\genDensity_{\BmeasSquash_j\land\Aclassical_j=\AclassicalVal_a\AclassicalVal_{\mu'}\land\Ashield_j=\blockVar}\right]}{\prob_{\Aclassical_j,\blockVarReg_j}(\AclassicalVal_a\AclassicalVal_{\mu'},\blockVar)}-\sqrt{\tilde{\zeta}_{\AclassicalVal_a,\AclassicalVal_\mu,\AclassicalVal_{\mu'},\blockVar,j}} \\ 
    &\forall \BclassicalVal\in\BclassicalAlph,\blockVar\in\{0,\ldots,\characCutoff\},\AclassicalVal_a\in\tilde{\AclassicalAlph},\AclassicalVal_\mu,\AclassicalVal_{\mu'}\in\intensityAlph \nonumber\\
    & \frac{\Tr\left[\proj_{\blockVar'}^{\BmeasSquash_j}\genDensity_{\BmeasSquash_j\land\Aclassical_j=\AclassicalVal_a\AclassicalVal_\mu\land\Ashield_j=\blockVar}\right]}{\prob_{\Aclassical_j,\blockVarReg_j}(\AclassicalVal_a\AclassicalVal_\mu,\blockVar)}\geq\frac{\Tr\left[\proj_{\blockVar'}^{\BmeasSquash_j}\genDensity_{\BmeasSquash_j\land\Aclassical_j=\AclassicalVal_a\AclassicalVal_{\mu'}\land\Ashield_j=\blockVar}\right]}{\prob_{\Aclassical_j,\blockVarReg_j}(\AclassicalVal_a\AclassicalVal_{\mu'},\blockVar)}-\sqrt{\tilde{\zeta}_{\AclassicalVal_a,\AclassicalVal_\mu,\AclassicalVal_{\mu'},\blockVar,j}} \\
    & \forall \blockVar'\in\{0,\ldots,\fssCutoff\},\blockVar\in\{0,\ldots,\characCutoff\},\AclassicalVal_a\in\tilde{\AclassicalAlph},\AclassicalVal_\mu,\AclassicalVal_{\mu'}\in\intensityAlph
\end{aligned}
\]
\end{constraintblock}

\end{figure}

\subsection{Example: imperfect phase randomization}
\label{sec:imperfect_phase_randomization}
\noindent We illustrate the application of \cref{th:convex_opti_numerics_imperfect_block} for an imperfectly phase-randomized weak-coherent source with intensity fluctuations and imperfectly characterized modes satisfying \cref{eq:mode_characterization_block}. We consider the same detector imperfections as in the main text, i.e. \cref{as:block_diagonal_detector,as:imperfectly_characterized_detector} hold. To apply the theorem, we require bounds on \cref{eq:bound_uncertainty_block,eq:bound_cutoff_block}, which we present in this section. 

Assume that the unknown phase distribution is independent of the intensity fluctuations and has a probability density lower bounded by $\delta^{\mathrm{ph}}_j/(2\pi)$ for some known $\delta^{\mathrm{ph}}_j\in[0,1]$. Following \cite[Sec.~III.A]{nahar_imperfect_2023}, we introduce the model states
\begin{equation}
\label{eq:phase_model_state}
    \big(\Astate_{\AclassicalVal_a\tilde{\AclassicalVal}_\mu}^{\mathrm{mod},(j)}\big)_{\Aprime_j} \coloneqq \delta^{\mathrm{ph}}_j\sum_{n=0}^{\infty}e^{-\tilde{\AclassicalVal}_\mu}\frac{\tilde{\AclassicalVal}_\mu^n}{n!}\ketbra{n_{\AclassicalVal_a}}{n_{\AclassicalVal_a}}_{\Aprime_j}+(1-\delta^{\mathrm{ph}}_j)\ketbra{\sqrt{\tilde{\AclassicalVal}_\mu},\AclassicalVal_a}{\sqrt{\tilde{\AclassicalVal}_\mu},\AclassicalVal_a}_{\Aprime_j}\,,
\end{equation}
where $\ket{n_a}$ denotes an $n$-photon Fock state in the optical mode corresponding to the encoding choice $a$ and $\ket{\sqrt{\mu},a}$ denotes a coherent state with mean photon number $\mu$ in the optical mode corresponding to the encoding choice $a$. 

Then, there exists a source map $\srcMap_j^{\mathrm{phase}}$ mapping the model states in \cref{eq:phase_model_state} to the actual imperfectly phase-randomized states \cite{nahar_imperfect_2023,kamin_renyi_2025} for all $\AclassicalVal\in\AclassicalAlph$. Since the phase and intensity fluctuations are independent, the same source map relates the intensity-averaged states. We may therefore apply \cref{lem:epsilon_security_source_maps} and consider the model states in the following. Let $I_{\AclassicalVal,j}$ denote the r.v. describing the intensity fluctuations (and imperfect characterization) with unknown distribution and support in $[\AclassicalVal_\mu(1-\srcImpBndInt),\AclassicalVal_\mu(1+\srcImpBndInt)]$. We define the intensity-averaged model state
\begin{equation}
    (\bar{\Astate}_{\AclassicalVal_a,\AclassicalVal_\mu}^{\mathrm{mod},(j)})_{\Aprime_j} \coloneqq \int_{\tilde{\AclassicalVal}_\mu\in[\AclassicalVal_\mu(1-\srcImpBndInt),\AclassicalVal_\mu(1+\srcImpBndInt)]}(\Astate_{\AclassicalVal_a\tilde{\AclassicalVal}_\mu}^{\mathrm{mod},(j)})_{\Aprime_j}\,d\prob_{I_{\AclassicalVal,j}}(\tilde{\AclassicalVal}_\mu)\,.
\end{equation}
Choosing the model state at the intensity $\AclassicalVal_\mu$ as the reference state, convexity of the trace distance and the coherent state overlap give
\begin{align}
    \frac{1}{2}\norm{(\bar{\Astate}_{\AclassicalVal_a,\AclassicalVal_\mu}^{\mathrm{mod},(j)})_{\Aprime_j}-(\Astate_{\AclassicalVal_a,\AclassicalVal_\mu}^{\mathrm{mod},(j)})_{\Aprime_j}}_1
    &\leq \sup_{\tilde{\AclassicalVal}_\mu\in[\AclassicalVal_\mu(1-\srcImpBndInt),\AclassicalVal_\mu(1+\srcImpBndInt)]}\sqrt{1-\exp\left[-\left(\sqrt{\tilde{\AclassicalVal}_\mu}-\sqrt{\AclassicalVal_\mu}\right)^2\right]} \\
    &= \sqrt{1-\exp\left[-\AclassicalVal_\mu\left(1-\sqrt{1-\srcImpBndInt}\right)^2\right]} \eqqcolon \epsilon^{\mathrm{int}}_{\AclassicalVal,j}\,.
\label{eq:intensity_uncertainty_bound_coherent_simplified}
\end{align}
For the coherence cutoff in \cref{eq:bound_cutoff_block}, following \cite{nahar_imperfect_2023}, we may choose
\begin{equation}
\label{eq:phase_cutoff_bound}
    \epsilon^{\mathrm{cut}}_{\AclassicalVal,j} \coloneqq (1-\delta^{\mathrm{ph}}_j)\sqrt{W_{\AclassicalVal,j}}\,,
\end{equation}
where we define the weight contained outside the $\leq\phaseCutoff$ subspace
\begin{equation}
    W_{\AclassicalVal,j} \coloneqq 1-e^{-\AclassicalVal_\mu}\sum_{n=0}^{\phaseCutoff}\frac{\AclassicalVal_\mu^n}{n!}\,.
\end{equation}
Consequently, the bounds required in \cref{eq:bound_uncertainty_block,eq:bound_cutoff_block} hold with $\epsilon^\mathrm{coh}_{\AclassicalVal,j}=\epsilon^{\mathrm{int}}_{\AclassicalVal,j}$ and $\epsilon^{\mathrm{cut}}_{\AclassicalVal,j}$ given by \cref{eq:phase_cutoff_bound}, and \cref{th:convex_opti_numerics_imperfect_block} applies.

\begin{remark}[Application to imperfect single-photon sources]
\label{rem:single_photon_sources}
\cref{th:convex_opti_numerics_imperfect_block} also applies to imperfectly characterized single-photon sources. In this case, imperfect characterization of the coherence matrix can be used to derive a bound on \cref{eq:bound_uncertainty_block} with some reference coherence matrix, directly yielding $\epsilon^\mathrm{coh}_{\AclassicalVal,j}$. If the reference coherence matrix is finite-dimensional, it can be diagonalized exactly, and no coherence cutoff is required, i.e. one may choose $\epsilon^{\mathrm{cut}}_{\AclassicalVal,j}=0$.
\end{remark}

\section{Miscellaneous} \label{app:misc}

\noindent Classical registers are represented by Hilbert spaces with computational basis indexed by the corresponding alphabet, as formalized in \cref{def:cq_state} below.

\begin{definition}[Classical-quantum state]
\label{def:cq_state}
    A state $\genDensity_{\genClassReg\genQReg} \in \setDensity_\leq(\genClassReg\genQReg)$ is called a \textit{classical-quantum (cq) state} with respect to the classical register $\genClassReg$ if it can be written as
    \begin{equation}
        \genDensity_{\genClassReg\genQReg} = \sum_{\genClassIndex \in \genClassAlph} \prob_{\genClassReg}(\genClassIndex) |\genClassIndex\rangle\langle\genClassIndex|_{\genClassReg} \otimes \genDensity_{\genQReg|\genClassReg=\genClassIndex}\,,
    \end{equation}
    where $\prob_{\genClassReg}$ is a probability distribution over $\genClassAlph$ and $\genDensity_{\genQReg|\genClassReg=\genClassIndex} \in \setDensity_\leq(\genQReg)$ for all $\genClassIndex \in \genClassAlph$.
\end{definition}

\begin{definition}[Conditioning on classical events]
    Let $\genDensity_{\genClassReg\genQReg} \in \setDensity_\leq(\genClassReg\genQReg)$ be classical on $\genClassReg$ with alphabet $\genClassAlph$ and let $\Omega\subseteq\genClassAlph$ be an event on register $\genClassReg$. We define the corresponding \textit{conditional state} as
    \begin{equation}
        \genDensity_{\genQReg\genClassReg|\Omega} \coloneqq \frac{\Tr[\genDensity_{\genQReg\genClassReg}]}{\Tr[\genDensity_{\genQReg\genClassReg\land\Omega}]} \genDensity_{\genQReg\genClassReg\land\Omega}\,,
    \end{equation}
    with \textit{sub-normalized conditional state}
    \begin{equation}
        \genDensity_{\genQReg\genClassReg\land\Omega} \coloneqq \sum_{\genClassIndex\in\Omega} \prob_{\genClassReg}(\genClassIndex) \ketbra{\genClassIndex}{\genClassIndex}_\genClassReg \otimes \genDensity_{\genQReg|\genClassReg=\genClassIndex}\,.
    \end{equation}
    The operator $\land$ has higher precedence than the operator $|$, i.e. it acts before the operator $|$. An expression such as $\genDensity_{\genQReg\genClassReg|\Omega_1\land \Omega_2}$ is to be interpreted as $\genDensity_{\genQReg\genClassReg|(\Omega_1\land \Omega_2)}$.
\end{definition}

\begin{definition}[Probability simplex]
\label{def:probablity_simplex}
For $d\in\mathbb{N}$, we define the probability simplex of dimension $d$ as
\begin{equation}
    \probSimplex_d \coloneqq \left\{\boldsymbol{p}\in\mathbb{R}^{d} : p_i \geq 0 \text{ for all } i\in\{1,\dots,d\}, \sum_{i=1}^{d} p_i = 1 \right\}.
\end{equation}
\end{definition}

We manipulate Rényi entropies throughout this work as they form a core part of the MEAT framework, and therefore define some of the relevant quantities for convenience below.
\begin{definition}[Minimal quantum Rényi divergence {\protect\cite[Def.~4.8]{tomamichel_quantum_2016}}]
\label{def:min_renyi_div}
    Let $\alpha \in (0, 1) \cup (1,\infty)$, $\genDensity, \genDensityAlt \in \setDensity_\leq(\genQReg)$ with $\Tr[\genDensity]\neq 0$. The \textit{minimal quantum Rényi divergence} of $\genDensity$ with $\genDensityAlt$ is defined as
    \begin{align}
        \renyidiv(\genDensity\Vert \genDensityAlt) \coloneqq \frac{1}{\alpha-1} \log\frac{\Tr\left[\right(\genDensityAlt^{\frac{1-\alpha}{2\alpha}}\genDensity \genDensityAlt^{\frac{1-\alpha}{2\alpha}}\left)^\alpha \right]}{\Tr[\genDensity]}
    \end{align}
    if $(\alpha < 1 \land \genDensity \not\perp \genDensityAlt) \lor \mathrm{supp}(\genDensity) \subseteq \mathrm{supp}(\genDensityAlt)$ and $\infty$ otherwise.
\end{definition}

\begin{definition}[Quantum conditional Rényi entropy {\protect\cite[Def.~5.2]{tomamichel_quantum_2016}}]
\label{def:condtional_renyi_entropy}
    Let $\alpha \geq 0$. The \textit{quantum conditional Rényi entropies} of $\genQReg$ given $\genQReg'$ for a state $\genDensity\in\setDensity_=(\genQReg\genQReg')$ are defined as
    \begin{align}
        \renyidown(\genQReg|\genQReg')_\genDensity &\coloneqq - \renyidiv(\genDensity_{\genQReg\genQReg'}\Vert \identity_\genQReg \otimes \genDensity_{\genQReg'})\\
        \renyiup(\genQReg|\genQReg')_\genDensity &\coloneqq\sup_{\genDensityAlt_{\genQReg'} \in \setDensity_=(\genQReg')} - \renyidiv(\genDensity_{\genQReg\genQReg'}\Vert \identity_\genQReg \otimes \genDensityAlt_{\genQReg'})\,.
    \end{align}
\end{definition}

\begin{definition}[$f$-weighted Rényi entropies {\protect\cite{inprep_weightentropy,arqand_marginal-constrained_2025}}]
\label{def:f_weighted_renyi_entropy}
    Let $\genDensity_{\genQReg\genQReg'\genClassReg} \in \setDensity_\leq(\genQReg\genQReg'\genClassReg)$ be classical on $\genClassReg$. Let $f:\genClassReg \rightarrow \mathbb{R}$ be a function, which is called a \textit{tradeoff function} in this context. Then, for any $\alpha \in (1, \infty]$, we define the following version of the \textit{$f$-weighted Rényi entropies}
    \begin{align}
        \fRenyiUpGen(\genQReg | \genQReg'\genClassReg)_\genDensity \coloneqq \frac{\alpha}{1 -\alpha}\log\left(\sum_{\genClassIndex\in\genClassAlph} \genDensity_{\genClassReg}(\genClassIndex)2^{\frac{1-\alpha}{\alpha} \left(\renyiup(\genQReg|\genQReg')_{\genDensity_{|\genClassIndex}} - f(\genClassIndex)\right)}\right)\,.
    \end{align}
    where $\genDensity_{\genClassReg}(\genClassIndex)\genDensity_{\genQReg\genQReg'|\genClassReg=\genClassIndex} = \Tr_\genClassReg[\Pi_\genClassIndex \genDensity_{\genQReg\genQReg'\genClassReg} \Pi_\genClassIndex]$ with projectors $\Pi_\genClassIndex = \ketbra{\genClassIndex}{\genClassIndex}_\genClassReg$ and $\genDensity_{\genClassReg}(\genClassIndex)$ denotes the probability distribution on $\genClassReg$ induced by $\genDensity_{\genQReg\genQReg'\genClassReg}$.
\end{definition}

In the \nameref{prot:entanglement_qkd_protocol}, instead of viewing Alice as randomly choosing a setting and preparing the corresponding signal state, we equivalently view her as preparing a single fixed entangled state, sending one half to Bob and keeping the other, and only later measuring her half to determine which setting was chosen. This standard argument is known as the \textit{source-replacement scheme} and formalized below.

\begin{lemma}[Source-replacement scheme {\protect\cite{PhysRevLett.68.557, PhysRevLett.92.217903, tupkary_rigorous_2026}}]
\label{lem:source_replacement_scheme}
    Consider a scenario where an $\totRounds$-round state is prepared as follows
    \begin{align}
        \genDensity_{\genClassReg_1^\totRounds (\genQReg')_1^\totRounds} &= \bigotimes_{j=1}^\totRounds \genDensity^{(j)}_{\genClassReg_j \genQReg_j'} \\
        \genDensity^{(j)}_{\genClassReg_j \genQReg_j'} &= \sum_{\genClassIndex\in \genClassAlph}\prob_{\genClassReg_j}(\genClassIndex)\ketbra{\genClassIndex}{\genClassIndex}_{\genClassReg_j} \otimes (\genDensity_\genClassIndex^{(j)})_{\genQReg'_j}\,,\label{eq:source_replacement_original_form}
    \end{align}
    and the register $(\genQReg')_1^\totRounds$ is sent through a channel. Then, the state $\genDensity_{\genClassReg_1^\totRounds (\genQReg')_1^\totRounds}$ can be obtained by preparing
    \begin{align}
        \genDensity_{\bar \genQReg_1^\totRounds \hat \genQReg_1^\totRounds (\genQReg')_1^\totRounds} &= \bigotimes_{j=1}^\totRounds \genDensity^{(j)}_{\bar \genQReg_j \hat \genQReg_j \genQReg_j'} \\
        \genDensity^{(j)}_{\bar \genQReg_j \hat \genQReg_j \genQReg_j'} &= \sum_{\genClassIndex, \genClassIndex'\in \genClassAlph}\sqrt{\prob_{\genClassReg_j}(\genClassIndex)\prob_{\genClassReg_j}(\genClassIndex')}\ketbra{\genClassIndex}{\genClassIndex'}_{\bar \genQReg_j} \otimes \ketbra{\genDensity_\genClassIndex^{(j)}}{\genDensity_{\genClassIndex'}^{(j)}}_{\hat \genQReg_j \genQReg_j'}\,,
    \end{align}
    where $\ket{\genDensity_{\genClassIndex}^{(j)}}_{\hat \genQReg_j \genQReg_j'}$ is a purification of $(\genDensity_\genClassIndex^{(j)})_{\genQReg'_j}$, $\{\ket{\genClassIndex}_{\bar \genQReg_j}\}_\genClassIndex$ is an orthonormal basis for $\bar \genQReg_j$ and $\hat \genQReg_j$ is called the shield system. Then the register $(\genQReg')_1^\totRounds$ is sent through the channel. Finally, for each $j\in\{1, \ldots,\totRounds\}$, the system $\bar \genQReg_j \hat \genQReg_j$ is measured using the POVM $\big\{\ketbra{\genClassIndex}{\genClassIndex}_{\bar \genQReg_j} \otimes \identity_{\hat \genQReg_j}\big\}$ and the outcome stored in a register $\genClassReg_j$.
\end{lemma}

\section{Notation and symbols}
\label{app:notation}

\begin{table}[h]
\centering
\renewcommand{\arraystretch}{1.4}
\begin{tabular}{ll}
\hline
\textbf{Symbol} & \textbf{Meaning} \\
\hline

$\totRounds$ & Total number of protocol rounds \\

$j$ & Protocol round index, $j \in \{1,\ldots,\totRounds\}$ \\

$\attackCh_j\in\setAttackCh_j$ & Eve's attack channel in the $j$th round \\

$\setAttackCh_j \subseteq \cptp(\Ereg_{j-1}, \Bmeas_j \Ereg_j')$ & Eve's set of possible attacks in the $j$th round\\

$\gMapTildeFull_j\in\cptp(\Aclassical_j\Bmeas_j, \secretReg_j\Aclassical_j\BclassicalReg_j\varDecReg_j\eveCopyReg_j)$ & Full protocol map for the modified protocol (\cref{lem:modified_timing_protocol}) \\

$\widetilde{\gMap}_j\in \cptp(\Aclassical_j\Bmeas_j, \secretReg_j \varDecReg_j \widetilde C_j)$ & Protocol map for the modified protocol (\cref{lem:modified_timing_protocol}), given by $\Tr_{\Aclassical_j\BclassicalReg_j}\circ \gMapTildeFull_j$ \\

$\widetilde{\qkdMap}_j\in \cptp(\Ereg_{j-1}\Aclassical_j, \secretReg_j\varDecReg_j \Ereg_j)$ & Full QKD map for the modified protocol (\cref{lem:modified_timing_protocol}), given by $\widetilde{\gMap}_j\circ\attackCh_j$ \\

$\qkdMap_j\in \cptp(\Ereg_{j-1}\Amarg_j, \secretReg_j\varDecReg_j \Ereg_j)$ & Full QKD map for the \nameref{prot:entanglement_qkd_protocol} \\

$\gMap_j\in \cptp(\Amarg_j\Bmeas_j, \secretReg_j \varDecReg_j \widetilde C_j)$ & Protocol map for the \nameref{prot:entanglement_qkd_protocol} \\

$\ppMap\in\cptp(\Aclassical_1^\totRounds\BclassicalReg_1^\totRounds\secretReg_1^\totRounds\varDecReg_1^\totRounds, \AkeyReg\BkeyReg\cppReg\varDecReg_1^\totRounds)$ & Classical post-processing map \\

$\annFunc[j]$ & Stochastic announcement map $\annFunc[j]:\AclassicalAlph\times\BclassicalAlph\rightarrow\probSimplex_{|\varDecAlph|}$ \\

$\keymapFunc[j]$ & Key map \\

$\annKeyMap[j] \in \cptp(\Aclassical_j\BclassicalReg_j, \secretReg_j \widetilde C_j)$ & \makecell[l]{Map representing the combination of the announcements map and key\\ map (\cref{def:tuple_protocol})} \\

$\epsSecu$ & Security parameter \\

$\epsCor$ & Correctness parameter \\

$\epsSecr$ & Secrecy parameter \\

$\keyl(\varDecVal^\totRounds_1)$ & Secure key length function \\

$\ECcost(\varDecVal^\totRounds_1)$ & Bits leaked for error correction \\

$\ptest$ & Test round probability \\

$\srcImpBnd[\AclassicalVal,\blockVar]$ & \makecell[l]{Source fidelity bound for setting choice $\AclassicalVal$ and source photon-number \\ block $\blockVar$ (\cref{as:imperfectly_characterized_source})} \\

$\dtImpBnd[\blockVar]$ & \makecell[l]{Detector POVM deviation for detector photon-number block $\blockVar$ \\(\cref{as:imperfectly_characterized_detector})} \\

\hline
\end{tabular}
\caption{Notation for the QKD protocol.}
\label{tab:protocol_notation}
\end{table}

\begin{table}[]
\centering
\renewcommand{\arraystretch}{1.4}
\begin{tabular}{ll}
\hline
\textbf{Symbol} & \textbf{Meaning} \\
\hline

$\genQReg$ & Generic quantum register \\

$\genClassReg$, $\genClassIndex$, $\genClassAlph$ & Generic classical register, label, and alphabet \\

$\hilbert_\genQReg$ & Hilbert space corresponding to register $\genQReg$ \\

$\cptp(\genQReg, \genQReg')$ & Set of completely positive and trace-preserving maps from $\genQReg$ to $\genQReg'$ \\

$\setDensity_=(\genQReg)$ & Set of density operators on $\genQReg$ \\

$\setDensity_\leq(\genQReg)$ & Set of sub-normalized density operators on $\genQReg$ \\

$\setPos(d)$ & Set of positive semi-definite $d\cross d$ matrices \\

$\probSimplex_d$ & Probability simplex of dimension $d$ (\cref{def:probablity_simplex}) \\

$\srcMap$ & Source map (\cref{def:source_map}) \\

$\squashMap$ & Squashing map (\cref{def:squashing_map}) \\

$\renyidiv$ & Minimal quantum Rényi divergence (\cref{def:min_renyi_div}) \\

$\renyidown$ & Down-arrow sandwiched Rényi entropy (\cref{def:condtional_renyi_entropy}) \\

$\renyiup$ & Up-arrow sandwiched Rényi entropy (\cref{def:condtional_renyi_entropy}) \\

$\ftradeoff{1}{j-1}$ & Tradeoff function for round $j$ \\

$\fFull$ & Global tradeoff function (\cref{th:security_imperfect_devices}) \\

$\fRenyiUpGen$ & $f$-weighted sandwiched Rényi entropy (\cref{def:f_weighted_renyi_entropy})\\

\hline
\end{tabular}
\caption{Generic notation used in this work.}
\label{tab:generic_notation}
\end{table}

\begin{table}[]
\centering
\renewcommand{\arraystretch}{1.4}
\begin{tabular}{ll}
\hline
\textbf{Symbol} & \textbf{Meaning} \\
\hline

$\Aclassical_j$, $\AclassicalVal$, $\AclassicalAlph$ & Alice's setting choice register, value, and alphabet \\

$\prob_{\Aclassical_j}$ & Probability distribution according to which Alice chooses her setting \\

$(\Astate_\AclassicalVal^{(j)})_{\Aprime_j}$ & Alice's signal states for setting choice $\AclassicalVal$ \\

$\Amarg_j$ & Alice's source-replacement marginal register \\

$\Aprime_j$ & Alice's signal state register \\

$\Ameas_j$ & Alice's source-replacement measurement register \\

$\Ashield_j$ & Alice's shield system \\

$\AkeyReg$ & Alice's key register \\

$\Apovmel[j]$ & Alice's POVM element after source-replacement for setting choice $\AclassicalVal$\\

$\BclassicalReg_j$, $\BclassicalVal$, $\BclassicalAlph$ & Bob's measurement outcome register, value, and alphabet \\

$\Bmeas_j$ & Bob's measurement input register \\

$\Bpovmel$ & Bob's POVM element for outcome $\BclassicalVal$\\

$\BpovmelTarg$ & Bob's target POVM element after squashing for outcome $\BclassicalVal$\\

$W^{\Bmeas_j}$ & Observable used for non-preserved subspace weight estimation\\

$\BkeyReg$ & Bob's key register \\

$\Ereg_j$ & Eve's side-information \\

$\eveCopyReg_j$ & Eve's copy of the public announcements \\

$\evePurReg$ & Eve's purifying register \\

$\varDecReg_j$, $\varDecVal$, $\varDecAlph$ & Public announcements register, value, and alphabet \\

$\secretReg_j$, $\secretVal$, $\secretAlph$ & Secret register, value, and alphabet \\ 

$\cppReg$ & Register storing the public announcements during classical post-processing\\

\hline
\end{tabular}
\caption{Registers and operators denoting Alice, Bob and Eve.}
\label{tab:notation_alice_bob_eve}
\end{table}

\makeatletter
\let\addcontentsline\oldaddcontentsline
\makeatother

\end{document}